\documentclass[twocolumn,prx,longbibliography,nofootinbib, superscriptaddress]{revtex4-2}
\usepackage{mathrsfs}
\usepackage[dvips]{graphicx} % for figures
\usepackage{amsfonts}
\usepackage{amssymb}
\usepackage{amscd}
\usepackage{amsmath}    % need for subequations	

\usepackage{amsthm}
\usepackage{appendix}
\usepackage{bbm}
\usepackage{dsfont}
\usepackage{mathptmx}
\usepackage{enumerate}
\usepackage{enumitem}
\usepackage{epsfig}
\usepackage{float}
\usepackage[colorlinks=true,linkcolor=teal,citecolor=teal,urlcolor=teal]{hyperref}
\usepackage{makecell}
\usepackage{pifont}
\usepackage{subfigure}
\usepackage{subfloat}
\usepackage{xcolor}		
\usepackage{physics}
\usepackage[most]{tcolorbox}
\usepackage{tabularx}
\usepackage{MnSymbol}
\usepackage{soul}
\usepackage{tikz}
\usepackage{tikzpeople}
\usepackage{pgfplots}
\usepackage[normalem]{ulem}
\usepackage{newtxtext}
\usepackage{titletoc}
\usepackage{multirow}

\usetikzlibrary{quantikz}
\usetikzlibrary{arrows}
\usetikzlibrary{shapes,fadings,snakes}
\usetikzlibrary{decorations.pathmorphing,patterns}
\usetikzlibrary{calc}
\usetikzlibrary{positioning}
\usepackage[linesnumbered,ruled,vlined]{algorithm2e}

\newtheorem{theorem}{Theorem}

\newtheorem{lemma}{Lemma}
\newtheorem{corollary}{Corollary}

\newtheorem{observation}{Observation}
\newtheorem{definition}{Definition}
\newtheorem{proposition}{Proposition}

\newcommand{\bI}{\mathbb{I}}

\newcommand{\cT}{\mathcal{T}}

\newcommand{\E}{\mathbb{E}}
\newcommand{\Cl}{\mathrm{Cl}}

\newcommand{\ot}{\otimes}
\newcommand{\C}{\mathbb{C}}

\newcommand{\spec}{\operatorname{spec}}
\newcommand{\Sep}{\mathrm{Sep}}
\newcommand{\ME}{\mathrm{ME}}
\newcommand{\Wg}{\mathrm{Wg}}

\newcommand{\esssup}{\operatorname*{esssup}}

\newcommand{\G}{\mathcal{G}}

\newcommand{\Fset}{\mathcal{F}}

\newcommand{\STAB}{\mathrm{STAB}}
\newcommand{\Ufree}{\mathcal{U}_\Fset}

\newcommand{\floor}[1]{\left\lfloor #1\right\rfloor}
\newcommand{\Coeff}[2]{\left[#1\right]#2}

\definecolor{ingo}{rgb}{1,.7,.02}

\usepackage[most]{tcolorbox}
\newtcolorbox[auto counter]{mybox}[2][]{
	enhanced,
	breakable,
	colback=blue!5!white,
	colframe=blue!75!black,
	fonttitle=\bfseries,
	title=Box \thetcbcounter: #2,#1
}

\makeatletter

\newif\ifappendixtocrecord
\appendixtocrecordfalse

\newcommand{\appendixtableofcontents}{%
  \section*{Contents of Appendices}%
  \setcounter{tocdepth}{3}%
  \@starttoc{atoc}%
  \appendixtocrecordtrue
}

\let\oldsection\section
\let\oldsubsection\subsection
\let\oldsubsubsection\subsubsection

\renewcommand{\section}{%
  \@ifstar{\app@sectionstar}{\@ifnextchar[{\app@sectionopt}{\app@sectionnoopt}}%
}

\newcommand{\app@sectionstar}[1]{%
  \oldsection*{#1}%
}

\newcommand{\app@sectionnoopt}[1]{%
  \oldsection{#1}%
  \ifappendixtocrecord
    \addcontentsline{atoc}{section}{\protect\numberline{\thesection}#1}%
  \fi
}

\def\app@sectionopt[#1]#2{%
  \oldsection[#1]{#2}%
  \ifappendixtocrecord
    \addcontentsline{atoc}{section}{\protect\numberline{\thesection}#1}%
  \fi
}

\renewcommand{\subsection}{%
  \@ifstar{\app@subsectionstar}{\@ifnextchar[{\app@subsectionopt}{\app@subsectionnoopt}}%
}

\newcommand{\app@subsectionstar}[1]{%
  \oldsubsection*{#1}%
}

\newcommand{\app@subsectionnoopt}[1]{%
  \oldsubsection{#1}%
  \ifappendixtocrecord
    \addcontentsline{atoc}{subsection}{\protect\numberline{\thesubsection}#1}%
  \fi
}

\def\app@subsectionopt[#1]#2{%
  \oldsubsection[#1]{#2}%
  \ifappendixtocrecord
    \addcontentsline{atoc}{subsection}{\protect\numberline{\thesubsection}#1}%
  \fi
}

\renewcommand{\subsubsection}{%
  \@ifstar{\app@subsubsectionstar}{\@ifnextchar[{\app@subsubsectionopt}{\app@subsubsectionnoopt}}%
}

\newcommand{\app@subsubsectionstar}[1]{%
  \oldsubsubsection*{#1}%
}

\newcommand{\app@subsubsectionnoopt}[1]{%
  \oldsubsubsection{#1}%
  \ifappendixtocrecord
    \addcontentsline{atoc}{subsubsection}{\protect\numberline{\thesubsubsection}#1}%
  \fi
}

\def\app@subsubsectionopt[#1]#2{%
  \oldsubsubsection[#1]{#2}%
  \ifappendixtocrecord
    \addcontentsline{atoc}{subsubsection}{\protect\numberline{\thesubsubsection}#1}%
  \fi
}

\makeatother

\begin{document}
\title{Witness expansion: \\
A unified framework for analytical and measurable mixed-state resource detection}

\author{Yifan Tang}
\thanks{\href{mailto:yifta@zedat.fu-berlin.de}{yifta@zedat.fu-berlin.de}}
\affiliation{Dahlem Center for Complex Quantum Systems, Freie Universität Berlin, 14195 Berlin, Germany}

\author{Chengkai Zhu}
\thanks{\href{mailto:zhuchengkai7@gmail.com}{zhuchengkai7@gmail.com}}
\affiliation{Thrust of Artificial Intelligence, Information Hub, The Hong Kong University of Science and Technology (Guangzhou), Guangdong 511453, China}

\author{Yuzhen Zhang}
\affiliation{Department of Physics, University of California, Santa Barbara, CA 93106, USA}

\author{Jens Eisert}
\affiliation{Dahlem Center for Complex Quantum Systems, Freie Universität Berlin, 14195 Berlin, Germany}
\affiliation{Helmholtz-Zentrum Berlin für Materialien und Energie, 14109 Berlin, Germany}

\author{\mbox{Zi-Wen Liu}}
\affiliation{Yau Mathematical Sciences Center, Tsinghua University, Beijing 100084, China}  

\author{Ingo Roth}
\affiliation{Quantum Research Center, Technology Innovation Institute (TII), Abu Dhabi, United Arab Emirates}

\author{Otfried Gühne}
\affiliation{Naturwissenschaftlich-Technische Fakultät, Universität Siegen, Walter-Flex-Straße 3, 57068 Siegen, Germany}

\author{Xin Wang}
\thanks{\href{mailto:wangxinfelix@gmail.com}{wangxinfelix@gmail.com}}
\affiliation{Thrust of Artificial Intelligence, Information Hub, The Hong Kong University of Science and Technology (Guangzhou), Guangdong 511453, China}

\author{Zhenhuan Liu}
\thanks{\href{mailto:qubithuan@gmail.com}{qubithuan@gmail.com}}
\affiliation{Quantum Research Center, Technology Innovation Institute (TII), Abu Dhabi, United Arab Emirates}

\begin{abstract}
Quantum information science aims to harness different kinds of quantum resources to accomplish specific information-processing tasks. These resources also play an increasingly important role in addressing fundamental questions concerning quantum phases and dynamics. Therefore, developing powerful and practical methods for identifying and detecting quantum resources is of great significance, with applications ranging from benchmarking quantum devices to understanding the fundamental structure of quantum theory. In this work, we propose \emph{witness expansion}, a unified framework for constructing nonlinear criteria for detecting quantum resources that are associated with a well-defined group of free unitaries. 
These criteria apply to both pure and mixed quantum states and are based on polynomial functions of the target state, which can be estimated experimentally using multiple copies of the state and evaluated analytically in certain physical models. 
We show how several well-known resource-detection quantities naturally emerge from our framework, including the $l_2$ norm of coherence, partial-transpose moments for entanglement, stabilizer entropy for nonstabilizerness (quantum magic), and fermionic antiflatness for fermionic non-Gaussianity. Beyond recovering these existing structures, our framework also yields new criteria for detecting qubit and qudit magic states, substantially enhancing witness-based detection capabilities.
In addition, it gives, to the best of our knowledge, the first analytical criterion for detecting mixed-state fermionic non-Gaussianity with respect to the convex hull of pure fermionic Gaussian states that remains nontrivial for arbitrary numbers of qubits, demonstrating the broad applicability and conceptual unifying power of the framework.

\end{abstract}

\maketitle

\section{Introduction}\label{sec:intro}

Since the late 20th century, the integration of quantum mechanics with information theory has given rise to quantum information science, which opens new perspectives on computation, communication, and quantum many-body physics.
At its core, quantum information science concerns the manipulation of different quantum resources for different information-processing tasks~\cite{chitambar2019resource,gour2025QuantumResourceTheories}.  
For example, quantum coherence serves as a key resource for quantum random number generation~\cite{herrero2017qrng,baumgratz2014Quantifying}; entanglement is essential in quantum secure communication, quantum networks, and quantum simulation~\cite{ekert1991QuantumCryptographyBased,Horodecki2009entanglement}; and quantum magic distinguishes quantum states that enable computational advantages from those that do not~\cite{howard2017ApplicationResourceTheory,gottesman1998HeisenbergRepresentationQuantum}.  
A quantum computer capable of performing certain tasks must therefore be able to generate, preserve, and manipulate the corresponding resources.  
In recent years, quantum resources have also emerged as valuable probes of quantum many-body physics, exhibiting distinctive signatures across different phases of matter~\cite{amico2008entanglement,jens2010area,tarabunga2023manybody,wei2025long}. 
Therefore, identifying and detecting quantum resources is of both practical and theoretical significance, with applications ranging from benchmarking quantum devices 
in the quantum technologies 
\cite{Eisert2020certification,Kliesch2021Certification} to probing the foundational structure of quantum theory.

Despite their importance, detecting quantum resources is often highly nontrivial because many resources are defined through mathematically complicated free sets.
Prominent examples include entanglement and magic, where resource states are specified by exclusion: entangled states are those that are not separable, and magic states are those that are not stabilizer states.
The difficulty becomes even more pronounced for mixed states, where one often has to decide whether a density matrix lies inside the convex hull of pure free states.
The geometric feature of convexity makes many mathematically well-defined pure-state resource criteria difficult to extend directly to mixed states~\cite{bermejo2025CharacterizingQuantumResourcefulness}, and often requires dedicated numerical optimization methods~\cite{Zhu2025unified}.
Fermionic non-Gaussianity provides a representative example, since the convex hull of pure fermionic Gaussian states is naturally connected to classical simulation complexity~\cite{Valiant2002QuantumCircuits}.
However, existing mixed-state criteria either rely on formulations other than this convex-hull structure~\cite{bittel2025OptimalTraceDistanceBoundsa,haug2026practical}, or lack a general method applicable to arbitrary system sizes without case-by-case optimization~\cite{oszmaniec2014ClassicalSimulationFermionic,vershynina2014CompleteCriterionConvexGaussianstate}.
Indeed, even when the full classical description of a density matrix is given, deciding whether it is resourceful is sometimes computationally hard~\cite{gurvits2003complexity,leone2026UnbearableHardnessDeciding}.
Thus, a central first step toward practical resource detection is to identify simpler, more transparent mathematical signatures of resourcefulness, or of broad, physically relevant subclasses of resource states.

Entanglement theory provides a paradigmatic example of this strategy.
Although no simple criterion gives a complete characterization of mixed-state entanglement of any size, a rich family of entanglement criteria 
has been developed, including entanglement witnesses, the \emph{positive partial transpose} (PPT) criterion, and moment-based criteria~\cite{peres1996separability,terhal2000BellInequalitiesSeparability,GUHNE2009detection,chruscinski2014EntanglementWitnessesConstruction}.
These criteria do not merely provide practical detection tools; they also reveal important structural and operational aspects of entanglement.
For instance, the 
PPT criterion is closely related to entanglement distillation and bound entanglement~\cite{vidal2002computable,PhD}, 
while witness-based and moment-based criteria have enabled a variety of practical detection protocols in experiments~\cite{pan2012multiphoton,elben2020mix,Brydges2019Probing,yu2021optimal,rico2024poly,cieslinski2024AnalysingQuantumSystems}. 
Moreover, concise entanglement quantifiers such as entanglement entropy and negativity often admit analytical expressions in quantum field theories and conformal field theories~\cite{Pasquale2004qft,calabrese2012negativity}, thereby deepening our understanding of many-body quantum physics.
These successes suggest that, even when a resource is difficult to characterize exactly, suitably designed detection criteria can capture its essential physical and operational content.

However, extending the success of entanglement detection criteria to other quantum resources is far from straightforward.
A fundamental obstacle is that different resources are characterized by genuinely different mathematical structures.
Many entanglement criteria are designed around the tensor-product structure of multipartite systems and therefore do not directly provide a general methodology for resources with different free state structures.
Witnesses provide a universal linear approach: as separating hyperplanes in the state space, they can be defined for any convex free state set and form a complete detection scheme.
However, the detection capability of any single linear witness is typically limited, and effective detection often requires access to a large family of witnesses~\cite{liu2022FundamentalLimitationDetectability}.
General resource-theoretic quantifiers, such as robustness measures and distillable resources~\cite{vidal1999RobustnessEntanglement,napoli2016RobustnessCoherenceOperational,campbell2017RoadsFaulttolerantUniversal}, can in principle be defined across different resource theories and detect all resource states.
However, they often require difficult optimizations over the free set or asymptotic operational tasks, and are therefore not always analytically tractable or directly computable.
Thus, rather than aiming at a complete resource quantifier, we seek a unified framework for constructing practical and analytical resource detection criteria.
The key is to identify simple common structures shared across different resources, while avoiding direct optimization over the full set of free states whenever possible.
This leads us to the following design principles:
\begin{enumerate}
\item \textbf{Universal:} Applicable to both pure and mixed states, and to various resources.
\item \textbf{Analytical:} Capable of yielding closed-form or analytically tractable expressions for concrete resources.
\item \textbf{Measurable:} Experimentally accessible through multi-copy measurement protocols.
\item \textbf{Adjustable:} Equipped with a tunable parameter that balances the detection capability and the practical cost.
\item \textbf{Pure-state faithful:} Capable of detecting all pure resource states with a single nonlinear criterion.
\end{enumerate}

\begin{figure}
\centering
\includegraphics[width=1\linewidth]{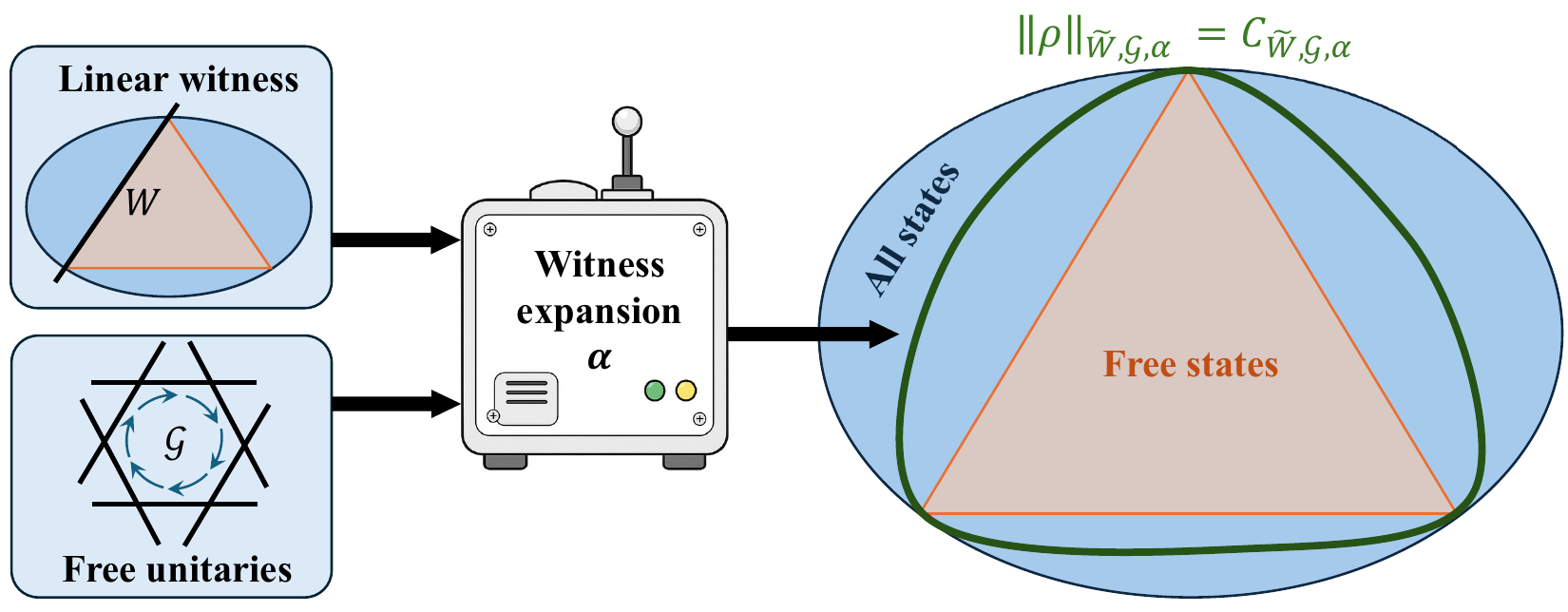}
\caption{ \textbf{Conceptual illustration of \emph{witness expansion} (WE).} The WE framework provides a unified methodology for quantum resource detection. By taking a linear seed witness $W$ and a corresponding group of free unitaries $\mathcal{G}$ as inputs, it systematically constructs a nonlinear detection criterion (the green curve). This generated criterion features a tunable parameter, $\alpha$, which governs the trade-off between the test's detection capability and its complexity.}
\label{fig:overview}
\end{figure}

In this work, we introduce \emph{witness expansion} (WE), a unified and versatile framework for systematically constructing nonlinear resource detection criteria that satisfy all the requirements listed above. 
As illustrated in Fig.~\ref{fig:overview}, the input of the framework consists of a linear seed resource witness $W$ and a set of free unitaries $\mathcal{G}$, preferably forming a compact group, associated with the resource theory under consideration. 
A nonlinear criterion is then generated by evaluating the $L^\alpha$-norm of the witness expectations over the free unitary orbit of $W$ and comparing it with the maximal value attainable by free states. 
The parameter $\alpha$ provides a tunable tradeoff between detection power and practical cost: larger $\alpha$ generally yields stronger detection, while smaller $\alpha$ often leads to simpler analytical form and experimental implementations. 
Importantly, apart from the choice of $W$ and $\mathcal{G}$, the construction does not require detailed resource-specific information about the full geometry of the free state set. 
The framework can also be optimized by restricting to suitable subgroups of free unitaries or by choosing a linear operator whose expectation value is highly sensitive to the target resource.
We note that nonlinear improvements of linear witnesses have been studied for entanglement detection~\cite{moroder2008IterationsNonlinearEntanglement,hyllus2006OptimalEntanglementWitnesses}, but these approaches rely on the special mathematical structure of entanglement and apply to restricted classes of detectable states.

We demonstrate the universality of WE by first showing that several well-known nonlinear resource quantifiers arise naturally within this framework.
These include coherence norms for coherence~\cite{baumgratz2014Quantifying}, R\'enyi entanglement entropy and partial-transpose moments for entanglement~\cite{peres1996separability}, stabilizer R\'enyi entropy for pure-state magic~\cite{Leone2022sre}, and fermionic antiflatness for fermionic non-Gaussianity of pure states~\cite{sierant2026FermionicMagicResources}.
The fact that quantities associated with such different resources can be generated from the 
same framework highlights the unifying power of WE.
In these examples, the nonlinear criteria generated by WE also exhibit stronger detection capability than the original linear seed witnesses; in particular, a single criterion constructed by WE can detect all pure resource states in the corresponding settings.
Beyond recovering existing structures, we derive new mixed-state detection criteria for qubit magic, odd-prime-dimensional qudit magic, and fermionic non-Gaussianity.
In particular, to the best of our knowledge, we obtain the first nontrivial arbitrary-size 
criterion for detecting mixed-state fermionic non-Gaussianity with respect to the convex hull 
of pure fermionic Gaussian states, while remaining faithful on all pure fermionic states.
These examples show that WE is not only a conceptual unification of known resource quantifiers, 
but also a practical method for designing new resource detection schemes.

The remainder of this work is organized as follows.
In Section~\ref{sec:framework}, we introduce the motivation and mathematical formulation of WE, and discuss possible extensions and optimization strategies.
In Section~\ref{sec:coherence}, we use the structure of the diagonal unitary group to show how WE relates to the coherence norm.
In Section~\ref{sec:entanglement}, we exploit a tensor product unitary group to demonstrate how entanglement quantities such as R\'enyi entanglement entropy and partial-transpose moments naturally emerge from the framework.
In Section~\ref{sec:qudit_magic}, we combine WE with Wigner operators to derive new mixed-state magic detection criteria for odd-prime-dimensional qudit systems and illustrate how restricting to suitable subgroups can enhance the detection capability.
In Section~\ref{sec:qubit_magic}, we use fourth-order Clifford twirling to recover stabilizer R\'enyi entropy for pure states, derive new mixed-state qubit magic detection criteria, and illustrate how replacing a linear seed witness with a resource-sensitive projector can improve the detection capability of WE.
In Section~\ref{sec:fermion}, we use the structure of the matchgate group to show how WE reproduces fermionic antiflatness for pure states and yields new mixed-state criteria for detecting fermionic non-Gaussianity outside the convex hull of pure Gaussian states.
Finally, in Section~\ref{sec:outlook}, we summarize our results and discuss several promising directions for future work.

\section{Framework}\label{sec:framework}
\subsection{Nonlinear criterion from linear witnesses}
The motivation behind WE comes from a simple observation: by applying a vector norm to 
the expectation values of many linear witnesses, one can 
combine them into a single nonlinear witness, which has a stronger detection capability than any individual linear witness in the limiting case.
Let $\mathcal{H}$ be a finite-dimensional Hilbert space, and let $\mathcal{D}(\mathcal{H})$ denote the set of (finite-dimensional)  quantum states, so the set of density operators.
Throughout this work, we consider a quantum resource theory that specifies a convex set of free states $\mathcal{F}\subseteq\mathcal{D}(\mathcal{H})$, together with free operations that preserve $\mathcal{F}$~\cite{chitambar2019resource}.
We will mainly use a compact subgroup $\mathcal{G}$ of all free unitaries $\Ufree$.

In quantum resource theories, witnesses constitute a fundamental tool for detecting resource states.
A Hermitian operator $W$ is a resource witness if $\Tr(W\sigma)\geq0$ for all $\sigma\in\mathcal{F}$, while $\Tr(W\rho)<0$ for at least one resource state $\rho$.
For later convenience, we convert each witness into a positive semi-definite \emph{standard witness}
\begin{equation}\label{eq:def_tilde_W}
    \widetilde{W}
    =
    \bI-\frac{W}{\norm{W}_\infty},
\end{equation}
where $\mathbb{I}$ represents the identity matrix and $\norm{\cdot}_\infty$ denotes the spectral or operator norm.
This normalization has two useful properties: ${\rm Tr}(\widetilde{W}\rho)\geq0$ for every quantum state $\rho$, and ${\rm Tr}(\widetilde{W}\sigma)\leq1$ for every free state $\sigma\in\mathcal{F}$.
Therefore, if a state $\rho$ satisfies ${\rm Tr}(\widetilde{W}\rho)>1$, then $\rho$ must be resourceful.

Let $(\Omega,\mu)$ be a probability space, and let $\widetilde{W}_\Omega=\{\widetilde{W}_m\}_{m\in\Omega}$ be a measurable family of standard resource witnesses.
For a density matrix $\rho$, define the expectation value of each witness as a function $f_{\widetilde{W}_\Omega,\rho}:\Omega\rightarrow\mathbb{R}_{\geq0}$ given by
\begin{equation}\label{eq:witness_expectation_function}
f_{\widetilde{W}_\Omega,\rho}(m)=\Tr(\widetilde{W}_m\rho).
\end{equation}
For $1\leq \alpha<\infty$, notice that $\widetilde{W}_m\geq0$, we aggregate the witness expectations using the $L^\alpha(\mu)$-norm
\begin{equation}
    \norm{f_{\widetilde{W}_\Omega,\rho}}_{L^\alpha(\mu)}
    =
    \left(
    \int_\Omega
    \Tr(\widetilde{W}_m\rho)^\alpha
    \mathrm{d}\mu(m)
    \right)^{1/\alpha}.
    \label{eq:def_R_alpha_mu}
\end{equation}
Because $\mu$ is a probability measure, every free state $\sigma$ satisfies by definition $\|{f_{\widetilde{W}_\Omega,\sigma}}\|_{L^\alpha(\mu)}\leq1$.
Thus, we obtain the nonlinear detection criterion
\begin{equation}\label{eq:WE_1}
    \norm{f_{\widetilde{W}_\Omega,\rho}}_{L^\alpha(\mu)}>1
    \quad\Longrightarrow\quad
    \rho\notin\mathcal{F}.
\end{equation}

The parameter $\alpha$ controls the detection capability of the criterion.
Since $L^\alpha$-norms are monotone in $\alpha$ for probability measures, increasing $\alpha$ can only strengthen the criterion.
In the limit $\alpha\rightarrow\infty$, the norm becomes the essential supremum taken over witness expectations as
\begin{equation}
\begin{aligned}
    \norm{f_{\widetilde{W}_\Omega,\rho}}_{L^\infty(\mu)}
    &\coloneqq
    \esssup_{m\in\Omega}
    \Tr(\widetilde{W}_m\rho)
    \\
    &=
    \inf\left\{
    t\in\mathbb{R}
    \middle|
    \mu\left(
    \left\{
    m\in\Omega
    \middle|
    \Tr(\widetilde{W}_m\rho)>t
    \right\}
    \right)=0
    \right\}.
\end{aligned}
\end{equation}
Consequently, the $\alpha=\infty$ criterion detects exactly those states that are detected by at least one witness in the ensemble, up to measure-zero exceptions.
A proof of this statement is given in Appendix~\ref{app:framework_nonlinear}.

For a fixed value of $\alpha$, the threshold $1$ in Eq.~\eqref{eq:WE_1} is not necessarily optimal.
Indeed, although every free state $\sigma$ satisfies $\|{f_{\widetilde{W}_\Omega,\sigma}}\|_{L^\alpha(\mu)}\leq 1$, the maximal value achievable by free states can be strictly smaller than $1$ for some chosen witness ensemble and measure.
This observation allows us to tighten the threshold and leads to the refined criterion
\begin{equation}\tag{\ensuremath{\diamond}}\label{eq:criterion_alpha_C}
\norm{f_{\widetilde{W}_\Omega,\rho}}_{L^\alpha(\mu)}
>
C_{\widetilde{W}_\Omega,\mu,\alpha}
\coloneqq
\sup_{\sigma\in\mathcal{F}}
\norm{f_{\widetilde{W}_\Omega,\sigma}}_{L^\alpha(\mu)}
\quad\Longrightarrow\quad
\rho\notin\mathcal{F}.
\end{equation}
Geometrically, this criterion combines the family of linear witness tests into a single nonlinear detection boundary, illustrated by the red curve in Fig.~\ref{fig:comparison}.
By construction, $C_{\widetilde{W}_\Omega,\mu,\alpha}\leq 1$, therefore, Eq.~\eqref{eq:criterion_alpha_C} is always at least as strong as Eq.~\eqref{eq:WE_1} for the same $\alpha$.

\begin{figure}
\centering
\includegraphics[width=0.9\linewidth]{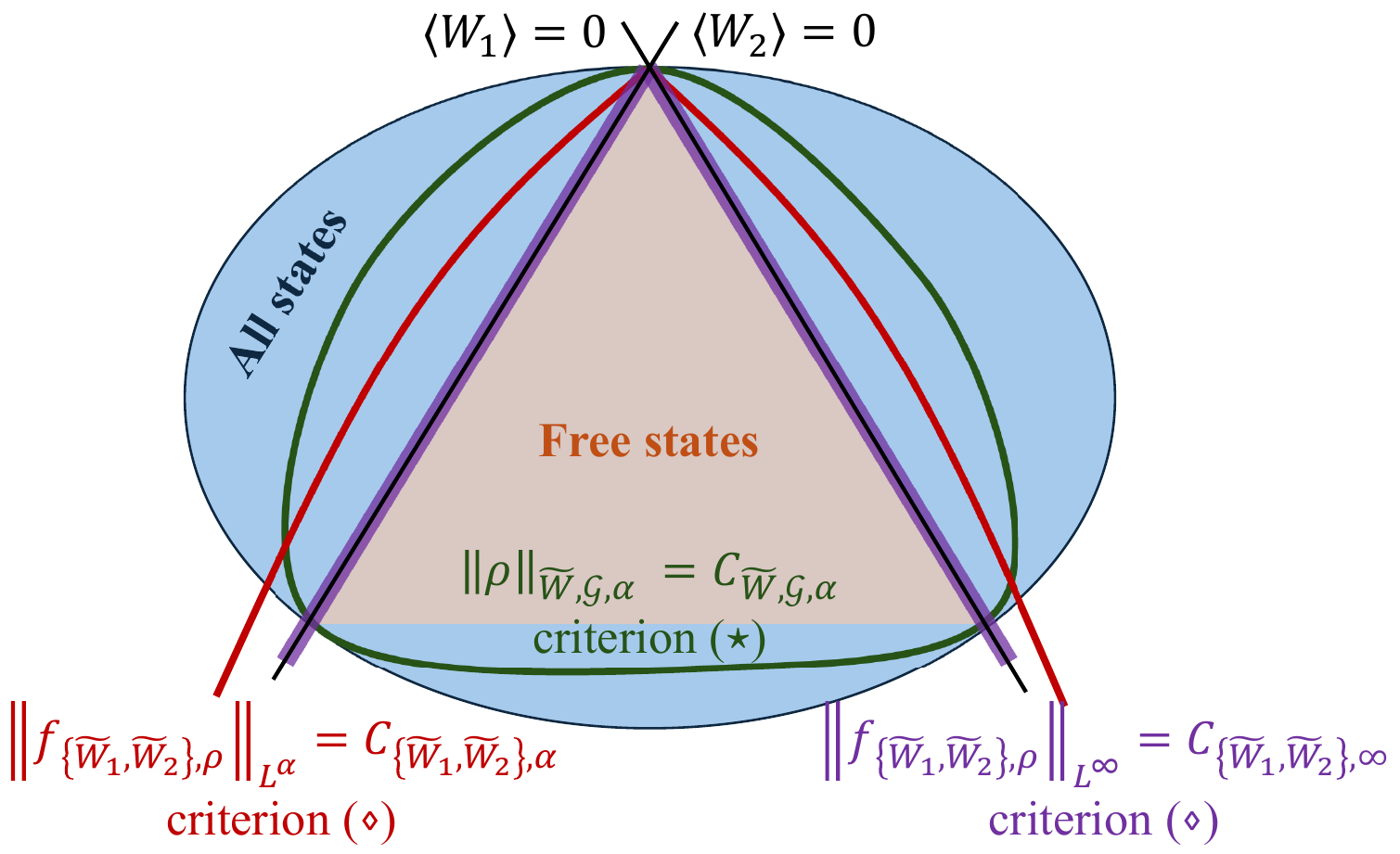}
\caption{
\textbf{Comparison of nonlinear criteria constructions.} The nonlinear criterion~\eqref{eq:criterion_alpha_C} (the red curve) combines a predetermined set of linear witnesses (here, $W_1$ and $W_2$) and can detect all resource states on the opposite side of the red curve from free states. The $\alpha \to \infty$ limit (the purple curve) recovers the union of their detection boundaries and detects all states that can be detected by any of the linear witnesses. The WE criterion~\eqref{eq:WE_alpha_criterion} (the green curve) generates a broader witness family through the free unitary twirling of a seed witness, detecting states that the seed witness misses. In all the examples we analyze, the WE norm in criterion~\eqref{eq:WE_alpha_criterion} takes the same value for all pure free states.
}
\label{fig:comparison}
\end{figure}

The same improvement persists in the limit $\alpha\rightarrow\infty$.
Since the refined threshold is never larger than the original one, the $\alpha=\infty$ version of Eq.~\eqref{eq:criterion_alpha_C} detects all states detectable by a positive-measure subset of the witness ensemble, and may detect additional states due to the tighter free state threshold, as illustrated by the purple lines in Fig.~\ref{fig:comparison}.
The construction also applies directly to a discrete witness ensemble, in which case the integral in Eq.~\eqref{eq:def_R_alpha_mu} is simply replaced by an average uniformly over the chosen witnesses.

\subsection{Witness expansion criterion}\label{subsec:WE_intro}
The refined criterion in Eq.~\eqref{eq:criterion_alpha_C} is still very general, and its use in practice faces two difficulties.
First, constructing a useful resource witness is already a nontrivial task, and constructing a large family of independent witnesses can be even more difficult.
Second, the threshold $C_{\widetilde{W}_\Omega,\mu,\alpha}$ requires an optimization over all free states.
In general, it is not clear which free state maximizes the $L^\alpha$-norm, making the threshold difficult to compute.

The key idea of the WE framework is to exploit a fundamental property of resource witnesses, which allows us to address both difficulties simultaneously: 
If $W$ is a valid resource witness and $U$ is a free unitary, then
$U^\dagger WU$ is again a valid resource witness.
Indeed, for every free state $\sigma$, the state $U\sigma U^\dagger$ is still free, and, therefore,
\begin{equation}
\Tr(U^\dagger WU\sigma)=\Tr(WU\sigma U^\dagger)\geq0. 
\end{equation}
Also, if $W$ detects a resource state $\rho$, then $U^\dagger WU$ can detect the resource state $U^\dagger \rho U$.
Thus, starting from a single standard seed witness $\widetilde{W}$, a whole family of standard witnesses can be generated by the free unitary orbit
\begin{equation}
    \left\{
    U^\dagger \widetilde{W}U
    \middle|
    U\in\mathcal{G}
    \right\},
\end{equation}
where $\mathcal{G}$ is a compact subgroup of the free unitary group $\mathcal{U}_{\mathcal F}$.
Equipped with the normalized Haar measure $\mu_\mathcal{G}$ on $\mathcal{G}$, the general construction in Eq.~\eqref{eq:def_R_alpha_mu} becomes
\begin{equation}\label{eq:def_R_alpha_mu_Haar}
    \left(
    \int_{\mathcal{G}}
    \Tr(U^\dagger\widetilde{W}U\rho)^\alpha
    \mathrm{d}\mu_\mathcal{G}(U)
    \right)^{1/\alpha}.
\end{equation}
In practice, we would like the WE criterion to apply uniformly to systems of arbitrary size.
This imposes a natural requirement that the seed witnesses should form a scalable family rather than have a limitation on system sizes.
Otherwise, the resulting WE criterion would inherit the same size restriction and could not provide a scalable resource test.

For practical applications, we mainly focus on positive integers $\alpha\in\mathbb{N}_+$.
In this case, the above quantity can be written as a unitary twirling over $\alpha$ tensor-copies of the state.
We define the \emph{WE $\alpha$-norm} associated with the standard seed witness $\widetilde{W}$ and the free unitary subgroup $\mathcal{G}$ as
\begin{equation}\label{eq:def_WE_alpha_norm}
\norm{\rho}_{\widetilde{W},\mathcal{G},\alpha}
\coloneqq
\left[
\Tr\left(
\mathbb{E}_{U\sim\mu_{\mathcal{G}}}
\left[
{U^\dagger}^{\otimes\alpha}
\widetilde{W}^{\otimes\alpha}
U^{\otimes\alpha}
\right]
\rho^{\otimes\alpha}
\right)
\right]^{1/\alpha}.
\end{equation}
We use the term ``norm'' in the sense that the WE $\alpha$-norm is the $L^\alpha$-norm of the witness expectation function associated to $\rho$ (see Eq.~\eqref{eq:def_R_alpha_mu}), and it need not be a norm on $\rho$ itself.
We then introduce the \emph{WE $\alpha$-criterion}
\begin{equation}\tag{\ensuremath{\star}}\label{eq:WE_alpha_criterion}
\norm{\rho}_{\widetilde{W},\mathcal{G},\alpha}
>
C_{\widetilde{W},\mathcal{G},\alpha}
\coloneqq
\sup_{\sigma\in\mathcal{F}}
\norm{\sigma}_{\widetilde{W},\mathcal{G},\alpha}
\quad\Longrightarrow\quad
\rho\notin\mathcal{F}.
\end{equation}

This construction has the following advantages.
First, instead of searching for many unrelated witnesses, one only needs to choose a single standard seed witness $\widetilde{W}$ and a physically motivated compact subgroup $\mathcal{G}$ of free unitaries.
Second, the choice of $\mathcal{G}$ can greatly simplify the threshold $C_{\widetilde{W},\mathcal{G},\alpha}$ when it captures enough symmetry of the set of free states.
The quantity $\norm{\rho}_{\widetilde{W},\mathcal{G},\alpha}$ is invariant under every unitary in $\mathcal{G}$.
In particular, if the free set is the convex hull of pure free states that form a single orbit under this group, then the threshold no longer requires a global optimization over all free states.
By Minkowski's inequality, the semi-positivity of $\widetilde{W}$ ensures the convexity of the mapping
\begin{equation}
\rho\mapsto\norm{\rho}_{\widetilde{W},\mathcal G,\alpha}, 
\end{equation}
thus it is sufficient to evaluate the threshold value on any one pure free state.

\begin{observation}[Threshold from any pure free state]\label{obs:threshold_single_orbit}
Suppose that the free set is the convex hull of pure free states $\mathcal F_{\rm pure}$, i.e.,
\begin{equation}
    \mathcal F=\mathrm{Conv}(\mathcal F_{\rm pure}),
\end{equation}
and that all pure free states belong to a single orbit under $\mathcal G$.
Then, for $1\leq\alpha\leq\infty$ and any representative pure free state $\psi\in\mathcal F_{\rm pure}$,
\begin{equation}
    C_{\widetilde{W},\mathcal G,\alpha}
    =
    \norm{\psi}_{\widetilde{W},\mathcal G,\alpha}.
\end{equation}
\end{observation}

In the concrete resource theories studied later, the relevant choices of $\mathcal{G}$ have particularly clear mathematical structures.
For example, local Haar unitaries naturally appear in entanglement theory, Clifford groups in magic resource theory, and matchgate groups in fermionic non-Gaussianity.
Their corresponding twirling formulas are known or tractable for the finite moments considered here, allowing both $\norm{\rho}_{\widetilde{W},\mathcal{G},\alpha}$ and the free state threshold $C_{\widetilde{W},\mathcal{G},\alpha}$ to be evaluated analytically in many cases.
Moreover, since Eq.~\eqref{eq:def_WE_alpha_norm} is an $\alpha$-copy polynomial function of $\rho$, it can be estimated experimentally using measurements on $\alpha$ copies.
Finally, in the limit $\alpha\to\infty$, the criterion detects all states detectable by at least one witness in the orbit $\{U^\dagger WU\}_{U\in\mathcal{G}}$, and may detect additional states due to the optimized threshold in Eq.~\eqref{eq:WE_alpha_criterion}; see Appendix~\ref{app:framework_nonlinear} for the proof.

The criteria established here are practically minded and operationally meaningful.
In particular, a large WE value certifies not only that the state lies outside the free set $\mathcal F$, but also that it is robustly separated from $\mathcal F$.
This robustness gives the criteria predictive power in the presence of small errors and allows WE norms to provide quantitative information about the amount of resource present in a state.
To make this statement precise, we consider the \emph{free robustness} for a state w.r.t. a quantum resource with free states $\mathcal{F}$~\cite{vidal1999RobustnessEntanglement,howard2017ApplicationResourceTheory,PhysRevLett.123.020401,liu2022magic}, defined as
\begin{equation}
\mathcal{R}(\rho)
=
\inf\left\{
t\ge0\middle\mid
\exists \tau\in\mathcal F
\text{ s.t. }
\frac{\rho+t\tau}{1+t}\in\mathcal F
\right\}.
\end{equation}
This quantity measures the minimum amount of any free state that must be mixed with $\rho$ so that the resulting state becomes free.
Thus $\mathcal{R}(\rho)=0$ for all free states and $\mathcal{R}(\rho)>0$ certifies that $\rho\notin\mathcal F$.
Moreover, free robustness is a natural resource-theoretic quantifier, since it is monotone under free operations that map free states to free states~\cite{howard2017ApplicationResourceTheory}.
Therefore, it provides an operational way to interpret the strength of a violation of the WE $\alpha$-criterion in Eq.~\eqref{eq:WE_alpha_criterion}: a larger violation implies that the state remains resourceful even after mixing with a larger amount of free noise. 
The proof of the following result is given in Appendix~\ref{app:free_robustness_WE}, and one can similarly derive a lower bound for $\mathcal{R}(\rho)$ using the violation of criterion~\eqref{eq:criterion_alpha_C}.

\begin{proposition}[Lower bound of free robustness from WE violation]\label{prop:lower_bound_free_robustness}
For $\alpha\in\mathbb{N}_+$, any positive semi-definite operator $\widetilde{W}$ and any compact subgroup $\mathcal{G}$ of free unitaries, the free robustness of any state $\rho$ has a lower bound when $C_{\widetilde{W},\mathcal{G},\alpha}>0$:
\begin{equation}
\mathcal{R}(\rho)
\ge
\frac{\norm{\rho}_{\widetilde{W},\mathcal{G},\alpha}
-
C_{\widetilde{W},\mathcal{G},\alpha}}{2C_{\widetilde{W},\mathcal{G},\alpha}},
\end{equation}
where the WE $\alpha$-norm $\norm{\rho}_{\widetilde{W},\mathcal{G},\alpha}$ and the threshold value $C_{\widetilde{W},\mathcal{G},\alpha}$ for free states are defined in Eq.~\eqref{eq:def_WE_alpha_norm} and Eq.~\eqref{eq:WE_alpha_criterion}.
\end{proposition}

\subsection{Enhancing the detection capability}\label{subsec:enhancing_WE}

The detection capability of the WE criterion in Eq.~\eqref{eq:WE_alpha_criterion} can be improved in two direct ways.
The first is to optimize the standard seed witness $\widetilde W$.
Since WE builds a nonlinear criterion by expanding the free unitary orbit of a single linear witness, a seed witness with strong linear detection capability generally leads to a stronger expanded criterion.
This point is illustrated in Fig.~\ref{fig:comparison}: the nonlinear boundary generated by expansion depends crucially on the initial hyperplane from which the witness orbit is generated.

In some resource theories, the quality of a linear witness can be geometrically characterized.
For instance, in entanglement theory, the detectable volume of a witness is closely related to the distance between the maximally mixed state and the hyperplane defined by $\langle W\rangle=0$~\cite{liu2022FundamentalLimitationDetectability}.
In Section~\ref{sec:entanglement}, we choose the SWAP operator as the seed witness, whose corresponding hyperplane is tangent to the largest separable ball around the maximally mixed state in Frobenius norm~\cite{gurvits2002largest}.
We further show that, within the second-order WE construction, the SWAP witness yields the strongest detection capability among all possible choices of linear seed witnesses.

The second way is to optimize the free unitary set $\mathcal G$ used in the expansion.
Taking $\mathcal G$ to be the full free unitary group automatically generates a large family of rotated linear witnesses.
As $\alpha$ increases, this allows the criterion to approach the union of the detection regions of these witnesses.
However, for small or moderate $\alpha$, the same twirling also introduces an averaging effect, which may weaken the detection capability of highly mixed states of the resulting nonlinear test.
Thus, a smaller subgroup $\mathcal G\subsetneq\Ufree$ can sometimes be advantageous: a judicious choice of $\mathcal{G}$ allows the WE criterion to emphasize the witness directions most relevant to the states of interest, rather than averaging over less informative directions.
In the limiting case where $\mathcal G$ contains only one element, the WE criterion reduces to that of a single linear witness.
We will see an explicit example in the qudit magic setting in Section~\ref{sec:qudit_magic}, where restricting the expansion to a suitable subgroup can yield a strictly stronger WE $2$-criterion than the one obtained from the full free unitary group.

\subsection{Generalization}\label{subsec:generalization}
Beyond optimizing the seed witness and the free unitary set within the WE framework, one can further extend the framework by relaxing the choice of the seed operator itself.
Using a standardized resource witness is a natural route to Eq.~\eqref{eq:WE_alpha_criterion}, but the construction only requires a resource-sensitive seed whose induced functional is convex on the free set.
A positive semi-definite seed operator gives a simple sufficient condition for this convexity for arbitrary $\alpha\ge 1$, and it need not arise from standardizing a linear witness.
A typical choice is the projector onto a highly resourceful pure state.
This is motivated by the fact that, as the system size grows, such states often have decreasing maximal fidelity with the free set, leading to directly computable free state thresholds.
We use this idea in Section~\ref{subsec:magic_projector} and Section~\ref{subsec:fermionic_projector}, where highly resourceful pure states serve as seed operators for new WE criteria detecting qubit magic and fermionic non-Gaussianity.
In principle, such positive semi-definite resource-sensitive operators may also be viewed as arising from a suitable standardization of linear witnesses, but working with them directly avoids the difficult optimization problem involved in explicitly constructing the corresponding witnesses.

A second extension becomes available when $\alpha$ is an even integer.
In this case, the convexity argument underlying Observation~\ref{obs:threshold_single_orbit} does not rely on the positive semi-definiteness of the seed operator, because the relevant functional is built from even powers of linear expectation values.
Therefore, the seed can be chosen from a broader class of Hermitian operators, not necessarily positive operators or standardized witnesses.
This flexibility is exploited in Section~\ref{subsec:fermionic_pauli}, where a special family of Pauli operators is used as the seed.
The resulting WE $2$-criterion provides a strong test for fermionic non-Gaussianity, illustrating that useful nonlinear resource criteria can also arise from algebraically structured observables rather than from the fidelity with positive semi-definite operators.

A still more general direction is to go beyond the tensor-power WE ansatz and optimize directly over the $\alpha$-copy observables.
Such a criterion has the general form
\begin{equation}\label{eq:nonlinear_alpha_criterion}
\Tr(W_\alpha\rho^{\otimes\alpha})
>
\sup_{\sigma\in\mathcal F}
\Tr(W_\alpha\sigma^{\otimes\alpha})
\quad\Longrightarrow\quad
\rho\notin\mathcal F.
\end{equation}
For analytical tractability, it is convenient to require $\Tr(W_\alpha\rho^{\otimes\alpha})$ to be invariant under free unitary conjugation $\rho\mapsto U\rho U^\dagger$ for all $U\in\Ufree$, or equivalently to take $W_\alpha$ from the $\alpha$-th order commutant of the free unitary group, i.e., ${U^\dagger}^{\otimes\alpha}
W_\alpha
U^{\otimes\alpha}=W_\alpha$ for all $U\in\Ufree$~\cite{Mele2024introductiontohaar,bittel2025CompleteTheoryClifford,haug2026EfficientWitnessingTesting,sierant2026TheoryMatchgateCommutanta,braccia2026CommutantFermionicGaussian,sierant2026FermionicMagicResources}.
WE realizes a natural subfamily of these commutants, where
\begin{equation}\label{eq:nonlinear_witness_twirling}
W_\alpha
=
\mathbb E_{U\sim\mu_{\Ufree}}
\left[
{U^\dagger}^{\otimes\alpha}
\widetilde W^{\otimes\alpha}
U^{\otimes\alpha}
\right]\in\mathrm{Comm}(\Ufree^{\otimes\alpha}).
\end{equation}
More general and even more powerful nonlinear criteria can be obtained by optimizing over the operators in the 
commutant.
The price to pay is that the free state threshold may no longer be determined by a 
single representative 
of pure free states as in Observation~\ref{obs:threshold_single_orbit}, and often 
has to be analyzed case by case.

In the following sections, we provide concrete examples of the WE criterion~\eqref{eq:WE_alpha_criterion} applied to several quantum resources: coherence, entanglement, magic, and fermionic non-Gaussianity.

\section{Coherence}\label{sec:coherence}
The first example is the detection of quantum coherence~\cite{streltsov2017coherence}, which plays an essential role in quantum random number generation~\cite{yuan2015randomness}, quantum thermodynamics\cite{cwiklinski2015Limitations,lostaglio2015DescriptionQuantumCoherence}, and quantum metrology~\cite{napoli2016RobustnessCoherenceOperational,ahnefeld2026coherence}.
Protocols for detecting 
and quantifying coherence have been both theoretically explored~\cite{baumgratz2014Quantifying,napoli2016RobustnessCoherenceOperational,ma2021DetectingEstimatingCoherence} and experimentally conducted~\cite{ringbauer2018CertificationQuantificationMultilevel}.
We will show in this section how a typical coherence quantifier can be generated via WE by choosing a suitable witness and a free unitary set.

In the resource theory of quantum coherence, the set of free states consists of the incoherent states
\begin{equation}
\Fset=\mathcal{I}=\left\{\sigma\in\mathcal{D}(\mathcal{H})\middle\mid \bra{j}\sigma\ket{k}=0,\forall j\neq k\right\}
\end{equation}
which are diagonal in a given fixed reference basis $\{\ket{j}\}_{j=0}^{d-1}$ of a Hilbert space with dimension $d\geq2$.
To fit into the WE framework, we select the coherence witness and its standardized version as
\begin{equation}\label{eq:coherence_witness}
W=\Xi=\frac1{d-1}\sum_{j\neq k}\ketbra{k}{j}, \qquad\widetilde{W}=\widetilde{\Xi}=\bI-\Xi.
\end{equation}
One can easily check that $\Tr(\Xi\sigma)=0$, $\forall\sigma\in\Fset$ and $\Xi$ can detect the coherent state $\rho=(\ket{0}-\ket{1})(\bra{0}-\bra{1})/2$.
We choose the subgroup of free unitaries to be all diagonal unitaries
\begin{equation}\label{eq:coherence_free_unitaries}
\mathcal{G}=\mathcal{U}_{\textup{diag}}(d)=\left\{\sum_{k=0}^{d-1}e^{i\theta_k}\ket{k}\bra{k}\middle\mid\theta_k\in[0,2\pi)\right\}.
\end{equation}
We provide the analytical form of WE $\alpha$-norm for any $\alpha\in\mathbb{N}_+$ in the following theorem (see the proof in Appendix~\ref{app:coherence}).
\begin{theorem}[WE criterion for coherence]\label{thm:WE_criterion_coherence}
For any state $\rho\in\mathcal{D}(\mathcal{H})$ with $\dim\mathcal{H}=d$ and integer $\alpha\in\mathbb{N}_+$, the WE $\alpha$-norm in criterion~\eqref{eq:WE_alpha_criterion} with standard coherence witness $\widetilde{\Xi}$ in Eq.~\eqref{eq:coherence_witness} and group $\mathcal{U}_{\textup{diag}}(d)$ in Eq.~\eqref{eq:coherence_free_unitaries} is given by
\begin{equation}\label{eq:R_evaluation_coherence}
\norm{\rho}_{\widetilde{\Xi},\mathcal{U}_{\textup{diag}}(d),\alpha}^\alpha=\sum_{\ell=0}^\alpha\binom{\alpha}\ell\left(-\frac1{d-1}\right)^\ell\ell!\sum_{M\in\mathcal{E}_\ell}\prod_{\substack{j,k=0\\j\neq k}}^{d-1}\frac{\bra{j}\rho\ket{k}^{M_{j,k}}}{M_{j,k}!},
\end{equation}
where $\mathcal{E}_\ell$ denotes the set of all $d \times d$ matrices $M$ with entries $M_{j,k}\in\mathbb{N}$, satisfying: (i) $M_{j,j} = 0$ for all $j$. (ii) $\sum_k M_{j,k} = \sum_k M_{k,j}$ for all $j$. (iii) $\sum_{j,k} M_{j,k} = \ell$. The threshold value for all incoherent states is $C_{\widetilde{\Xi},\mathcal{U}_{\textup{diag}}(d),\alpha}=1$.
\end{theorem}
Specifically, the $\alpha=1$ criterion 
is 
trivial and when $\alpha=2$, criterion~\eqref{eq:WE_alpha_criterion} reduces to
\begin{equation}
\sum_{\substack{j,k=0\\j\neq k}}^{d-1}\abs{\bra{j}\rho\ket{k}}^2>0\quad\Longrightarrow\quad\rho\notin\mathcal{I}.
\end{equation}
At this point,
we see that $\norm{\rho}_{\widetilde{\Xi},\mathcal{U}_{\textup{diag}}(d),2}^2$ is a linear function of the $l_2$ norm of coherence~\cite{baumgratz2014Quantifying} and the WE 2-criterion already characterizes all the coherent states $\rho$, so taking larger $\alpha$ values is not necessary for detecting coherence.

\section{Entanglement}\label{sec:entanglement}
Entanglement is a foundational feature of quantum mechanics and one of the most profound operational distinctions between quantum and classical theories~\cite{einstein1935CanQuantumMechanicalDescription,schrodinger1935DiscussionProbabilityRelations,bell1964EinsteinPodolskyRosen,clauser1969ProposedExperimentTest}. It serves as the indispensable quantum resource for tasks including quantum teleportation~\cite{bennett1993TeleportingUnknownQuantum,yin2012QuantumTeleportationEntanglement}, quantum cryptography~\cite{ekert1991QuantumCryptographyBased,bennett1984quantum}, and quantum error correction~\cite{shor1995SchemeReducingDecoherence,bennett1996MixedstateEntanglementQuantum,brun2006CorrectingQuantumErrors}. This resource-theoretic viewpoint motivated a broad range of methods of certifying, quantifying, and manipulating entanglement~\cite{peres1996separability,horodecki1996SeparabilityMixedStates,brandao2005QuantifyingEntanglementWitness,hyllus2006OptimalEntanglementWitnesses,eisert2007QuantitativeEntanglementWitnesses,GUHNE2009detection,Horodecki2009entanglement,elben2020mix,imai2021bound,yu2021optimal,liu2022correlation,rico2024poly,Eisert2020certification,Kliesch2021Certification}.
We focus on the quantum entanglement between two subsystems with the same dimension, so that $\mathcal H_A\cong\mathcal H_B\cong\mathbb C^d$.
% and leave the discussion about non-equal subsystem dimensions to Appendix~\ref{app:entanglement}.

The free states in entanglement resource theory are those separable under this bipartition
\begin{equation}
\Fset=\Sep=\left\{\sum_kp_k\sigma_A^k\otimes\sigma_B^k\middle\mid \sum_kp_k=1,p_k\geq0,\sigma_\bullet^k\in\mathcal{D}(\mathcal{H}_\bullet),\forall k\right\}.   
\end{equation}
We select the SWAP operator between subsystems $A$ and $B$ as the witness to be expanded and its standardized version
\begin{equation}\label{eq:entanglement_witness}
W=\mathbb{S}=\sum_{j,k=0}^{d-1}\ket{j,k}\bra{k,j},\qquad\widetilde{W}=\widetilde{\mathbb{S}}=\bI-\mathbb{S}.
\end{equation}
It is easy to verify that for all separable state $\sigma_{A,B}$, $\Tr(\mathbb{S}\sigma_{A,B})=\sum_kp_k\Tr(\sigma_A^k\sigma_B^k)\ge0$ and $\mathbb{S}$ can detect the entangled state $\rho=(\ket{0,1}-\ket{1,0})(\bra{0,1}-\bra{1,0})/2$.
The reason for choosing $\mathbb{S}$ as the seed witness is due to its strong detection capability.
With the standard Hilbert--Schmidt normalization, the hyperplane of $\Tr(\mathbb{S}\rho)=0$ is tangent to the largest separable ball around the
maximally mixed state, and $\mathbb{S}$ detects an entangled state with the minimal purity~\cite{gurvits2002largest}.
The free unitaries are chosen to be the tensor products of local unitaries on each subsystem
\begin{equation}\label{eq:entanglement_free_unitaries}
\mathcal{G}=\mathcal{U}_{\textup{loc}}(d,d)=\left\{U_A\otimes U_B\middle\mid U_\bullet\in\mathcal{U}(\mathcal{H}_\bullet)\right\}.
\end{equation}
For $\alpha\in\mathbb{N}_+$, the analytical expression of $\norm{\rho}_{\widetilde{\mathbb{S}},\mathcal{U}_{\textup{loc}}(d,d),\alpha}$ can be calculated using the Schur-Weyl duality~\cite{collins2006IntegrationRespectHaar,Mele2024introductiontohaar}, its analytical expression is provided in the following theorem (see Appendix~\ref{app:proof_WE_criterion_entanglement} for a proof).

\begin{theorem}[WE criterion for entanglement]\label{thm:WE_criterion_entanglement}
For any state $\rho\in\mathcal{D}(\mathcal{H}_A\otimes\mathcal{H}_B)$ with $\dim\mathcal{H}_A=\dim\mathcal{H}_B=d$ and integer $\alpha\in\mathbb{N}_+$, the WE $\alpha$-norm in criterion~\eqref{eq:WE_alpha_criterion} with standard entanglement witness $\widetilde{\mathbb{S}}$ in Eq.~\eqref{eq:entanglement_witness} and group of free unitaries $\mathcal{U}_{\textup{loc}}(d,d)$ in Eq.~\eqref{eq:entanglement_free_unitaries} is given by
\begin{equation}\label{eq:R_evaluation_entanglement}
\begin{aligned}
\norm{\rho}_{\widetilde{\mathbb{S}},\mathcal{U}_{\textup{loc}}(d,d),\alpha}^\alpha=&1+\sum_{\ell=1}^\alpha\binom{\alpha}\ell(-1)^\ell\sum_{\pi,\omega\in S_\ell}\Wg_d(\pi\omega)
\\&\Tr\left[(\mathbb{P}_\pi^{(\ell),A}\otimes \mathbb{P}_\omega^{(\ell),B})\rho^{\otimes\ell}\right],
\end{aligned}
\end{equation}
where $\Wg_d$ denotes the unitary Weingarten function and $\mathbb{P}_\pi^{(\ell)}$ is the permutation operator for $\pi\in S_\ell$ on $\ell$ copies
\begin{equation}
\mathbb{P}_\pi^{(\ell)}\left(\ket{\psi_1}\otimes\cdots\otimes\ket{\psi_\ell}\right)
=\ket{\psi_{\pi^{-1}(1)}}\otimes\cdots\otimes\ket{\psi_{\pi^{-1}(\ell)}}.
\end{equation}
The threshold value for all separable states is 
\begin{equation}
C_{\widetilde{\mathbb{S}},\mathcal{U}_{\textup{loc}}(d,d),\alpha}=
\left(\frac{d-1}{d+\alpha-1}\right)^{1/\alpha}.
\end{equation}
\end{theorem}

As an example, the $\alpha=1$ criterion is again trivial. We provide the exact calculations for $\alpha=2$ and $3$.

\begin{corollary}[WE $2,3$-norms for entanglement]
\label{cor:WE_2_3_norm_entanglement}
For every bipartite state $\rho$ with subsystems $A$ and $B$ of dimension $d$ each, 
we have
\begin{equation}
\label{eq:R2_entanglement}
\norm{\rho}_{\widetilde{\mathbb{S}},\mathcal{U}_{\textup{loc}}(d,d),2}^2=
1-\frac{2}{d}
+
\frac{1+r_2-(r_{2,A}+r_{2,B})/d}{d^2-1},
\end{equation}
and when $d\geq3$
\begin{equation}
\label{eq:R3_entanglement}
\begin{aligned}
\norm{\rho}_{\widetilde{\mathbb{S}},\mathcal{U}_{\textup{loc}}(d,d),3}^3
=&
1-\frac{3}{d}
+\frac{3\left(1+r_2-(r_{2,A}+r_{2,B})/d\right)}{d^2-1}\\
&
-\frac{\Delta_3(\rho)}{d(d^2-1)(d^2-4)},
\end{aligned}
\end{equation}
where
\begin{equation}
\label{eq:Delta3new}
\begin{aligned}
\Delta_3(\rho)
&=(d^2-2)\left(1+3r_2+2\tau_3\right)\\
&\quad-3d\left(r_{2,A}+r_{2,B}+2u_A+2u_B\right)\\
&\quad+4\left(r_{3,A}+r_{3,B}+3v+r_3\right),
\end{aligned}
\end{equation}
and
\begin{equation}\label{eq:trace_abbr}
\begin{aligned}
r_2&=\Tr(\rho^2),
& r_3&=\Tr(\rho^3),\\
\rho_A&=\Tr_B(\rho),&\rho_B&=\Tr_A(\rho),\\
r_{2,A}&=\Tr(\rho_A^2),
& r_{2,B}&=\Tr(\rho_B^2),\\
r_{3,A}&=\Tr(\rho_A^3),
& r_{3,B}&=\Tr(\rho_B^3),\\
u_A&=\Tr\left[\rho^2(\rho_A\otimes \bI)\right],&
u_B&=\Tr\left[\rho^2(\bI\otimes \rho_B)\right],\\
\tau_3&=\Tr\left[(\rho^{\mathrm{T}_B})^3\right],&
v&=\Tr\left[\rho(\rho_A\otimes \rho_B)\right].
\end{aligned}
\end{equation}
\end{corollary}

The proof of Corollary~\ref{cor:WE_2_3_norm_entanglement} is given in Appendix~\ref{app:proof_WE_criterion_entanglement}.
For $\alpha=2$, Corollary~\ref{cor:WE_2_3_norm_entanglement} gives the WE 2-criterion
\begin{equation}\label{eq:WE_2_criterion_entanglement}
\Tr(\rho^2)>
\frac{1}{d}\left(\Tr(\rho_A^2)+\Tr(\rho_B^2)\right)+1-\frac{2}{d}
\quad\Longrightarrow\quad
\rho\notin\Sep.
\end{equation}
This criterion detects all pure entangled states, and therefore gives a substantial improvement over the original linear SWAP witness.
Moreover, for a pure state $\rho=\ketbra{\psi}{\psi}$, the criterion depends on the subsystem purity 
${\rm Tr}(\rho_A^2)={\rm Tr}(\rho_B^2)$, or equivalently on the R\'enyi-$2$ entanglement entropy of $\ket{\psi}$.
Thus, in the entanglement setting, WE naturally produces a well-known entanglement quantifier from the local unitary expansion of the SWAP witness.
However, as discussed in Section~\ref{sec:framework}, WE also introduces an averaging effect over the local unitary orbit of the seed witness.
In Eq.~\eqref{eq:WE_2_criterion_entanglement}, this averaging effect is reflected by the factor $1/d$ in front of the local purities ${\rm Tr}(\rho_A^2)$ and ${\rm Tr}(\rho_B^2)$.
As a consequence, the criterion becomes relatively insensitive to the imbalance between global and local purities, and it cannot detect many highly mixed entangled states.

A sharper two-copy criterion is given by the purity criterion~\cite{horodecki1996QuantumAentropyInequalities,nielsen2001SeparableStatesAre}
\begin{equation}\label{eq:purity_criterion_entanglement}
\Tr(\rho^2)>
\min\left\{\Tr(\rho_A^2),\Tr(\rho_B^2)\right\}
\quad\Longrightarrow\quad
\rho\notin\Sep.
\end{equation}
This criterion is optimal among all criteria that use only the three quantities $\left\{{\rm Tr}(\rho^2),{\rm Tr}(\rho_A^2),{\rm Tr}(\rho_B^2)\right\}$, in the sense that any other criterion can only detect a subset of states detectable by the purity criterion. Equivalently, within the general two-copy nonlinear witness form in Eq.~\eqref{eq:nonlinear_alpha_criterion}, the combined use of $W_\alpha=\mathbb{S}_{A}\mathbb{S}_B-\mathbb{S}_A$ and $\mathbb{S}_{A}\mathbb{S}_B-\mathbb{S}_B$ give the strongest criterion based on these purity data, where $\mathbb{S}_A$ ($\mathbb{S}_B$)
denotes the SWAP operator between two copies of subsystem $A$ ($B$).

This comparison also clarifies the limitation of optimizing only within the WE ansatz.
In Appendix~\ref{app:optimal_2_criterion_entanglement}, we prove that, if the nonlinear witness is restricted to the twirled form in Eq.~\eqref{eq:nonlinear_witness_twirling}, then choosing the SWAP operator as the seed witness in Eq.~\eqref{eq:entanglement_witness} is already optimal for $\alpha=2$, i.e., the states detectable by WE 2-criterion constructed by any seed operator can be detected by the criterion constructed by $\widetilde{\mathbb{S}}$.
Therefore, improving the two-copy entanglement criterion beyond Eq.~\eqref{eq:WE_2_criterion_entanglement} requires leaving the linear seed WE construction.
One must instead use nonlinear witnesses tailored more specifically to the structure of separable states, such as the purity criterion above.

The WE $3$-criterion contains several higher-order state functions beyond the global and local purities.
These include the third partial-transpose moment ${\rm Tr}[(\rho^{\mathrm T_B})^3]$, which has been widely used in entanglement detection~\cite{elben2020mix,yu2021optimal,miller2026DetectingEntanglementFew}.
The WE $3$-norm also relates to third-order moments of reduced states such as ${\rm Tr}(\rho_A^3)$ and ${\rm Tr}(\rho_B^3)$, which appear naturally in 
moment-based criteria for detecting mixed-state and 
bound entanglement~\cite{imai2021bound}.
In addition, the criterion involves the correlation functional ${\rm Tr}[\rho(\rho_A\otimes\rho_B)]$, which has been used to characterize correlations in multipartite quantum systems~\cite{liu2022correlation}.
Therefore, the WE $3$-criterion provides a new entanglement 
test built from quantities that are 
already known to be informative for entanglement detection.
The important point is that these quantities arise here from a single, systematic construction.
Previous criteria based on partial-transpose moments, moments of 
reduced states, or correlation functions usually exploit 
specific mathematical structures of entanglement.
By contrast, WE generates these higher-order quantities from the local unitary expansion of the SWAP witness.
This provides a more universal route to such moment-based entanglement criteria and suggests that analogous structures can emerge in other resource theories by changing the seed witness and the free unitary group.

A more interesting observation is the connection between WE and the PPT criterion~\cite{peres1996separability,horodecki1996SeparabilityMixedStates}.
The PPT criterion is one of the most important tools in entanglement detection:
It combines strong detection capability~\cite{aubrun2012phase}, a simple and efficiently computable mathematical structure, and deep connections to operational tasks such as entanglement 
distillation~\cite{vidal2002computable,PureBipartiteBennett}.
We show that the $\alpha\to\infty$ limit of the WE criterion gives a PPT-like condition.
This condition is weaker than the original PPT criterion, but it has a very similar structure: it tests the partial transpose $\rho^{\mathrm T_B}$ against all maximally entangled states.
If the search over maximally entangled states in the following theorem is replaced by a search over all normalized pure states, then one recovers exactly the PPT condition.

\begin{theorem}[PPT-like WE $\infty$-criterion for entanglement]
\label{thm:WE_infinity_criterion_entanglement}
Define the maximally entangled state vector as $|\Phi_d\rangle:= \sum_{j=0}^{d-1}|j,j\rangle/{\sqrt d}$, and the maximally entangled orbit
as
\begin{equation}
\ME_d:=\left\{(U_A\otimes U_B)|\Phi_d\rangle\middle\mid U_A,U_B\in\mathrm{U}(d)\right\}.
\label{eq-max-ent-orbit}
\end{equation}
For every bipartite state $\rho\in\mathcal{D}(\mathcal{H}_A\otimes\mathcal{H}_B)$,
\begin{equation}
\lim_{\alpha\to\infty}
\norm{\rho}_{\widetilde{\mathbb{S}},\mathcal{U}_{\textup{loc}}(d,d),\alpha}
=
1-d\min_{\ket{\Phi}\in\ME_d}
\bra{\Phi}\rho^{\mathrm{T}_B}\ket{\Phi}.
\end{equation}
Moreover,
\begin{equation}
\lim_{\alpha\to\infty}
C_{\widetilde{\mathbb{S}},\mathcal{U}_{\textup{loc}}(d,d),\alpha}
=
1.
\end{equation}
Consequently,
\begin{equation}
\exists\ket{\Phi}\in\ME_d
\textup{ s.t. }
\bra{\Phi}\rho^{\mathrm{T}_B}\ket{\Phi}<0
\quad\Longrightarrow\quad
\rho\notin\Sep.
\label{eq-swaplu}
\end{equation}
\end{theorem}

The proof of Theorem~\ref{thm:WE_infinity_criterion_entanglement} is provided in Appendix~\ref{app:proof_WE_infinity_criterion_entanglement} and can be intuitively understood as follows. Due to the ``transpose trick'' $(U \otimes \bI) \ket{\Phi_d}=(\bI \otimes U^\mathrm{T}) \ket{\Phi_d}$ of the maximally entangled  state, it suffices in Eq.~\eqref{eq-max-ent-orbit} to consider unitaries on subsystem $A$ alone and one can set $U_B= \bI.$
Then the criterion in Eq.~\eqref{eq-swaplu} states that
$\Tr[\rho^{\mathrm{T}_B}  \ketbra{\Phi}{\Phi}]
= \Tr[\rho (U_A \otimes \bI) (\ketbra{\Phi_d}{\Phi_d})^{\mathrm{T}_B} 
(U_A^\dagger \otimes \bI)]
= \Tr[\rho (U_A \otimes \bI) \mathbb{S}
(U_A^\dagger \otimes \bI)]/d < 0$. This means that in the limit $\alpha \rightarrow \infty$ the criterion exactly checks all SWAP witnesses up to local unitaries.

\section{Qudit magic for an odd prime dimension}\label{sec:qudit_magic}
% Nonstabilizerness, or quantum magic, captures the states and operations beyond the stabilizer formalism, where the stabilizer states evolved by Clifford operations remain efficiently classical simulable
Nonstabilizerness, or quantum magic, captures the states and operations beyond the stabilizer formalism, where the stabilizer states evolved by Clifford operations remain efficiently classical simulable~\cite{gottesman1998HeisenbergRepresentationQuantum,aaronson2004stabilizer}.
In this sense, magic is regarded as an operational resource, formalized in the resource theory of magic~\cite{veitch2014ResourceTheoryStabilizer,wang2019ResourceTheoryMagic}, and is essential in tasks including classical simulation~\cite{howard2014ContextualitySuppliesMagic,bravyi2016ImprovedClassicalSimulation,bravyi2019SimulationQuantumCircuits}, fault-tolerant quantum computation~\cite{bravyi2005UniversalQuantumComputation,howard2017ApplicationResourceTheory,campbell2017RoadsFaulttolerantUniversal}, and quantum contextuality~\cite{howard2014ContextualitySuppliesMagic}.
The methods for detecting and quantifying magic resources have been developed in the past decade, such as magic witnesses~\cite{liu2025MagicCriterionAlmost,haug2026EfficientWitnessingTesting}, stabilizer rank~\cite{bravyi2016ImprovedClassicalSimulation,bravyi2019SimulationQuantumCircuits,huang2019ApproximateStabilizerRank,qassim2021ImprovedUpperBounds}, robustness of magic~\cite{howard2017ApplicationResourceTheory,heinrich2019RobustnessMagicSymmetries,liu2022magic,hamaguchi2024HandbookQuantifyingRobustness}, and stabilizer entropy~\cite{Leone2022sre,leone2024StabilizerEntropiesAre,bittel2026OperationalInterpretationStabilizer}.
Nevertheless, compared with the rich toolbox available for entanglement detection, methods for detecting mixed-state magic remain relatively limited. The following examples demonstrate that the WE framework offers a systematic route toward constructing new mixed-state magic criteria.

We first consider the magic of qudit system with an odd prime dimension and denote the local dimension as $q$.
Let $\omega_q=e^{i2\pi/q}$ be a primitive $q$-th root of unity, define the Weyl operators for $z,x\in\mathbb{Z}_q^n$ by their action on basis vectors $\left\{\ket{t}\right\}_{t\in\mathbb{Z}_q^n}$ of $\mathcal{H}=(\mathbb C^q)^{\otimes n}$,
\begin{equation}
\begin{gathered}
Z(z)\ket{t}=\omega_q^{z\cdot t}\ket{t},\qquad X(x)\ket{t}=\ket{t+x},\\
T_{z,x}=\omega_q^{-2^{-1}z\cdot x}Z(z)X(x),
\end{gathered}
\end{equation}
where all additions and inner products are taken modulo $q$, and $2^{-1}$ denotes the inverse of $2$ in $\mathbb Z_q$.
The $n$-qudit Heisenberg-Weyl group consists of all Weyl operators and phases,
\begin{equation}
\mathsf{P}_{n,q}=\left\{\omega_q^aT_{z,x}\middle\mid a\in\mathbb{Z}_q,z,x\in\mathbb{Z}_q^n\right\}.
\end{equation}
The set of free unitaries is the $n$-qudit Clifford group, which is the normalizer of $\mathsf{P}_{n,q}$,
\begin{equation}\label{eq:qudit_magic_free_unitaries}
\Ufree=\Cl_{n,q}=\{U\in \mathrm{U}(d)\mid U\mathsf{P}_{n,q}U^\dagger=\mathsf{P}_{n,q}\},
\end{equation}
where $d=q^n$ is the global dimension.
The pure stabilizer states are those obtainable from $\ket{0^n}$ by applying a Clifford unitary, denoted as
\begin{equation}
\Sigma_{n,q}=\left\{U\ketbra{0^n}{0^n}U^\dagger\middle| U\in \Cl_{n,q}\right\},
\end{equation}
and the set of free states ($n$-qudit stabilizer polytope) is the convex hull of $\Sigma_{n,q}$,
\begin{equation}
\Fset=\STAB_{n,q}=\textup{Conv}(\Sigma_{n,q}).
\end{equation}

For odd-prime-dimensional qudit magic, a natural choice of the seed operator for the WE framework is the phase-space point operator, or Wigner operator, defined at point $(\mu,\nu)\in V\coloneqq\mathbb Z_q^n\times\mathbb Z_q^n$ as
\begin{equation}
A_{\mu,\nu}
=
\frac{1}{d}
\sum_{z,x\in\mathbb Z_q^n}
\omega_q^{\nu\cdot z-\mu\cdot x}
T_{z,x}^{\dagger}.
\end{equation}
Each $A_{\mu,\nu}$ is a valid magic witness, since stabilizer states have nonnegative discrete Wigner functions and therefore satisfy $\Tr(A_{\mu,\nu}\sigma)\ge 0$ for all $\sigma\in\STAB_{n,q}$~\cite{gross2006Hudson}.
Moreover, the whole family $\{A_{\mu,\nu}\}_{\mu,\nu}$ can be generated from a single phase-space point operator by Clifford conjugation.
Thus, the Wigner operators provide a canonical witness orbit for applying the WE framework to odd-prime-dimensional qudit magic. We choose
\begin{equation}\label{eq:qudit_magic_witness}
W=A_{0^{2n}}=\frac{1}{d}
\sum_{z,x\in\mathbb Z_q^n}
T_{z,x}^{\dagger},\qquad\widetilde{W}=\widetilde{A}_{0^{2n}}=\bI-A_{0^{2n}}.
\end{equation}

Before introducing our results, we first note a general obstruction that applies to magic detection in both qudit and qubit stabilizer resource theories.

\begin{observation}[State moments cannot detect magic]
\label{obs:state_moments_cannot_detect_magic}
Any criterion depending only on the state moments $\{{\rm Tr}(\rho^k)\}_k$ cannot detect a magic state $\rho$, in both qudit and qubit stabilizer resource theories.
\end{observation}

Indeed, the stabilizer polytope contains all states diagonal in the computational basis.
Therefore, for every density operator $\rho$, one can construct a stabilizer state with the same spectrum by mixing computational basis states with probabilities equal to the eigenvalues of $\rho$.
Such a stabilizer state has the same values of all moments ${\rm Tr}(\rho^k)$ as $\rho$, 
so spectral moments alone cannot distinguish magic states from stabilizer states.

For the WE framework, this gives a simple obstruction.
Whenever the $\alpha$-th order group twirl coincides with the corresponding Haar twirl, the resulting quantity can depend on $\rho$ only through spectral moments, and hence cannot yield a nontrivial magic criterion.
In the odd-prime-dimensional qudit case, $\Cl_{n,q}$ is a unitary 2-design~\cite{divincenzo2002QuantumDataHiding}.
Thus, the WE $1,2$-criteria with full Clifford group are trivial: the first order only retains $\Tr(\rho)$, while the second order can depend on $\rho$ only through $\Tr(\rho^2)$.
Since $\Cl_{n,q}$ is not a unitary 3-design~\cite{gross2021SchurWeylClifford,bittel2025CompleteTheoryClifford}, the first potentially nontrivial full-Clifford WE criterion starts at $\alpha=3$.
\begin{theorem}[WE 3-criterion with Clifford group for odd-prime-dimensional qudit magic]\label{thm:WE_3_qudit_magic}
For any $n$-qudit state $\rho$ with odd-prime local dimension $q$, the WE $3$-norm in criterion~\eqref{eq:WE_alpha_criterion} with qudit magic witness $A_{0^{2n}}$ in Eq.~\eqref{eq:qudit_magic_witness} and $n$-qudit Clifford group $\Cl_{n,q}$ in Eq.~\eqref{eq:qudit_magic_free_unitaries} is given by
\begin{equation}\label{eq:R_3_qudit_magic}
\norm{\rho}_{\widetilde{A}_{0^{2n}},\Cl_{n,q},3}^3
=
1-\frac{3}{d}+\frac{3}{d}\Tr(\rho^2)-d\sum_{u\in V}W_\rho(u)^3,
\end{equation}
where $d=q^n$ and the discrete Wigner function evaluated at point $u\in V$ is defined as
\begin{equation}
W_{\rho}(u)=\frac1{d}\Tr(A_u\rho).
\end{equation}
The threshold value for all stabilizer states and $\alpha\geq1$ is
\begin{equation}\label{eq:C_alpha_qudit_magic}
C_{\widetilde{A}_{0^{2n}},\Cl_{n,q},\alpha}=\left(1-\frac1d\right)^{\frac1\alpha}.
\end{equation}
The WE $3$-criterion detects all pure non-stabilizer states, as for any $n$-qudit pure state $\psi=\ketbra{\psi}{\psi}$,
\begin{equation}
\norm{\psi}_{\widetilde{A}_{0^{2n}},\Cl_{n,q},3}>C_{\widetilde{A}_{0^{2n}},\Cl_{n,q},3}\quad\Longleftrightarrow\quad\psi\notin\Sigma_{n,q}.
\end{equation}
\end{theorem}
The detailed proof is left in Appendix~\ref{app:qudit_magic}.
Although $\alpha=3$ leads to a nontrivial magic detection criterion, it is still desirable to lower the copy number and ease the experimental complexity.
The triviality of the full-Clifford WE $2$-criterion can be understood as a consequence of excessive averaging.
This suggests breaking the full Clifford symmetry and expanding the seed witness only over a smaller subgroup.

A simple choice is the subgroup generated by phase-space translations in one direction,
\begin{equation}\label{eq:qudit_magic_subgroup_free_unitaries}
    \mathcal G_X
    =
    \left\{X(x)\middle|x\in\mathbb Z_q^n\right\}.
\end{equation}
Under this subgroup, the seed phase-space point operator is rotated only within a single column of phase space,
\begin{equation}
X(x)^\dagger A_{0^{2n}}X(x)=A_{0^n,-x}.
\end{equation}
Thus the averaging is restricted to a proper subset of Wigner witnesses rather than the full set.
As shown below, this construction already yields a nontrivial second-order magic criterion (see the proof in Appendix~\ref{app:qudit_magic}).
\begin{theorem}[WE 2-criterion with subgroup for odd-prime-dimensional qudit magic]\label{thm:WE_2_subgroup_qudit_magic}
For any $n$-qudit state $\rho$ with odd-prime local dimension $q$, the WE $2$-norm in criterion~\eqref{eq:WE_alpha_criterion} with qudit magic witness $A_{0^{2n}}$ in Eq.~\eqref{eq:qudit_magic_witness} and a subgroup of free unitaries $\mathcal{G}_X$ in Eq.~\eqref{eq:qudit_magic_subgroup_free_unitaries} is given by
\begin{equation}
\begin{aligned}
\norm{\rho}_{\widetilde{A}_{0^{2n}},\mathcal{G}_X,2}^2&=1-2\sum_{x\in \mathbb{Z}_q^n}W_\rho((0,x))+d\sum_{x\in \mathbb{Z}_q^n}W_\rho((0,x))^2,
\end{aligned}
\end{equation}
where $d=q^n$.
The threshold value for all stabilizer states is
$C_{\widetilde{A}_{0^{2n}},\mathcal{G}_X,2}=1$. Especially, for any $n$ and $q$, there exists an $n$-qudit state $\rho_2$ that is detectable by this second-order criterion, i.e., $\norm{\rho_2}_{\widetilde{A}_{0^{2n}},\mathcal{G}_X,2}>C_{\widetilde{A}_{0^{2n}},\mathcal{G}_X,2}$, but cannot be detected by the third-order criterion in Theorem~\ref{thm:WE_3_qudit_magic}, i.e., $\norm{\rho_2}_{\widetilde{A}_{0^{2n}},\Cl_{n,q},3}\leq C_{\widetilde{A}_{0^{2n}},\Cl_{n,q},3}$. On the contrary, there exists an $n$-qudit pure non-stabilizer state vector 
$\ket{\psi_3}$ with 
\begin{equation}
    \norm{\ketbra{\psi_3}{\psi_3}}_{\widetilde{A}_{0^{2n}},\mathcal{G}_X,2}\leq C_{\widetilde{A}_{0^{2n}},\mathcal{G}_X,2},  
\end{equation}
thus the third-order criterion can detect it while the second-order criterion cannot.
\end{theorem}

This example illustrates the tradeoff involved in choosing the expansion set $\mathcal G$.
For a fixed copy number $\alpha$, averaging over a larger free unitary set can improve the coverage of pure resource states, as in Theorem~\ref{thm:WE_3_qudit_magic}, where the full Clifford expansion detects all pure magic states.
However, the same averaging can also wash out useful structure and weaken the detection of mixed states; in the extreme case, the full-Clifford second-order criterion becomes completely trivial.
Restricting $\mathcal G$ to a smaller subgroup sacrifices the symmetry, and the resulting criterion is not complete for pure states, but it can retain witness directions that are more sensitive to particular mixed-state resources.
Theorem~\ref{thm:WE_2_subgroup_qudit_magic} provides such an example: by averaging only over a column of Wigner witnesses, the second-order WE criterion becomes nontrivial and detects states that are missed by the third-order WE criterion.

It is worth emphasizing again that the price of using a restricted expansion set is that the threshold generally becomes more resource-specific.
For the full free unitary group, symmetry often makes $C_{\widetilde{W},\mathcal G,\alpha}$ easy to evaluate, when it is sufficient to evaluate on a single representative of pure free states.
For smaller subgroups, this simplification may fail, and the threshold has to be analyzed case by case.

\section{Qubit magic}\label{sec:qubit_magic}

Now we consider the $n$-qubit magic state resource theory. Denote the set of $n$-qubit phase-free Pauli operators as
\begin{equation}
\mathcal{P}_n=\left\{\bI,X,Y,Z\right\}^{\otimes n},
\end{equation}
and the $n$-qubit Pauli group as $\mathsf{P}_n$
\begin{equation}
\mathsf{P}_n=\left\{i^aP\middle\mid a\in\{0,1,2,3\},P\in\mathcal{P}_n\right\}.
\end{equation}
Recall the $n$-qubit Clifford group $\mathrm{Cl}_n$ (also the group of free unitaries) is the normalizer of the  $\mathsf{P}_n$, 
\begin{equation}\label{eq:qubit_magic_free_unitaries}
\Ufree=\Cl_n=\{U\in \mathrm{U}(d)\mid U\mathsf{P}_nU^\dagger=\mathsf{P}_n\},
\end{equation}
where $d=2^n$ is the global dimension.
The pure stabilizer states are those that can be generated from $\ket{0^n}$ using a Clifford unitary, denoted as
\begin{equation}
\Sigma_{n}=\left\{U\ketbra{0^n}{0^n}U^\dagger\middle| U\in \Cl_{n}\right\}.
\end{equation}
The set of free states is the $n$-qubit stabilizer polytope, i.e., the convex hull of $\Sigma_n$,
\begin{equation}
\Fset=\STAB_n=\textup{Conv}(\Sigma_n).
\end{equation}

\subsection{WE 4-norm for qubit magic}
Since the multi-qubit Clifford group is a unitary 3-design~\cite{zhu2017multiqubit} but not a unitary 4-design~\cite{zhu2016CliffordGroupFails}, the first potentially nontrivial full-Clifford WE criterion starts at $\alpha=4$ by Observation~\ref{obs:state_moments_cannot_detect_magic}.
We therefore apply the fourth-order Clifford twirling formula~\cite{roth2018RecoveringQuantumGates} and obtain the following general result, with the proof given in Appendix~\ref{app:qubit_magic_WE_4-norm}.
Instead of restricting to a standard witness $\widetilde{W}$, we derive the formula with a positive semi-definite operator $O$ and leave the specific discussion of the choice of $O$ to a later part of the manuscript.

\begin{theorem}[WE $4$-norm for qubit magic]
\label{thm:R4_magic}
For every $n$-qubit state $\rho$, the WE $4$-norm in criterion~\eqref{eq:WE_alpha_criterion} with a positive semi-definite operator $O$ and $n$-qubit Clifford group $\Cl_n$ in Eq.~\eqref{eq:qubit_magic_free_unitaries} is given by
\begin{equation}
\label{eq:R4-general-inline}
\norm{\rho}_{O,\Cl_n,4}^4
=
\sum_{\lambda\vdash4:\ell(\lambda)\le d}
d_\lambda
\left[
\frac{\widehat q_\lambda(\rho)a_\lambda}{D_\lambda^+}
+
\frac{
\left(\widehat p_\lambda(\rho)-\widehat q_\lambda(\rho)\right)b_\lambda
}{D_\lambda^-}
\right].
\end{equation}
Here $d_\lambda$ and $D_\lambda^\pm$ are the dimension factors of the fourth-order Clifford commutant, and the terms with zero denominator should be omitted in the above summation. Constants $a_\lambda,b_\lambda$ are determined by the operator $O$.
The state-dependent terms $\widehat p_\lambda(\rho)$ and $\widehat q_\lambda(\rho)$ are, respectively, linear combinations of the state moments
\begin{equation}
\label{eq:p-moments-main}
    1,\quad
    r_2,\quad
    r_2^2,\quad
    r_3,\quad
    r_4,
    \qquad
    r_k=\Tr(\rho^k),
\end{equation}
and of the Pauli moments
\begin{equation}
\label{eq:q-moments-main}
    q_\mu(\rho)
    =
    \Tr\left(\mathbb{P}_\pi^{(4)}Q\rho^{\otimes4}\right),
    \qquad
    \mu\vdash4,
\end{equation}
where $\mathbb{P}_\pi^{(4)}$ is the permutation operator associated to $\pi\in S_4$ with cycle type $\mu$ and
\begin{equation}
    Q=\frac{1}{d^2}\sum_{P\in\mathcal P_n}P^{\otimes4},
\end{equation}
with $d=2^n$.
The explicit definition of functions
$\widehat p_\lambda(\rho)$ and $\widehat q_\lambda(\rho)$ are given as Eq.~\eqref{eq:pqhat-def-inline} in Appendix~\ref{app:qubit_magic_WE_4-norm}.
\end{theorem}

Although qubit and odd-prime-dimensional qudit magic are defined in a formally similar way, the qubit case has a substantially more intricate mathematical structure.
A major difference is that the discrete Wigner representation no longer provides a nonnegative characterization of stabilizer states for qubits.
Therefore, instead of using Wigner operators as canonical magic witnesses, we need to choose a suitable magic witness as the seed in the WE framework.

\subsection{WE 4-criterion with triangle witness}
As emphasized in Section~\ref{subsec:WE_intro}, the seed witness used in the WE framework should preferably form a scalable family with an analytical expression, so that the resulting WE criterion can be defined uniformly for arbitrary system sizes.
This requirement is particularly restrictive for mixed-state magic.
Compared with entanglement, where many systematic witness construction methods are available~\cite{GUHNE2009detection}, the toolbox for mixed-state magic witnesses is still limited, and many existing constructions rely on numerical optimization or are tailored to small numbers of qubits~\cite{warmuz2025magic,Alberto2026polytope,varela2026predicting}.
We therefore use the witness introduced with the  magic Triangle Criterion~\cite{liu2025MagicCriterionAlmost} as a scalable seed witness.
A triangle witness has the form $\ketbra{\psi_1}{\psi_1}+\ketbra{\psi_2}{\psi_2}-\ketbra{\psi_3}{\psi_3}$, where $\ketbra{\psi_1}{\psi_1}$, $\ketbra{\psi_2}{\psi_2}$, and $\ketbra{\psi_3}{\psi_3}$ are pure stabilizer states with pairwise fidelity $1/2$.
For concreteness, we choose three nearest-neighbor single-qubit stabilizer states tensored with $\ket{0^{n-1}}$, which gives after normalization
\begin{equation}\label{eq:qubit_magic_witness}
    W=T=
    \frac{\bI+X-Y+Z}{1+\sqrt3}
    \otimes
    \ket{0^{n-1}}\bra{0^{n-1}},
    \qquad
    \widetilde W=\widetilde T=\bI-T.
\end{equation}
All triangle witnesses of this type can be generated from Eq.~\eqref{eq:qubit_magic_witness} by Clifford conjugation. The coefficients related to $\widetilde{T}$ appearing as $a_\lambda$ and $b_\lambda$ in Theorem~\ref{thm:R4_magic} are provided in Appendix~\ref{app:qubit_magic_T_coefficients}.

Although the explicit form of the WE $4$-criterion with standard triangle witness $\widetilde{T}$ for qubit magic is rather involved, it has a simple and suggestive pure-state reduction.
For a pure $n$-qubit state $\psi=\ketbra{\psi}{\psi}$, its WE $4$-norm $\norm{\psi}_{\widetilde{T},\Cl_n,4}$ depends only on the Pauli fourth moment of $\psi$, or equivalently on its stabilizer $2$-R\'enyi entropy~\cite{Leone2022sre}.
Thus, in the qubit magic setting, the WE framework establishes a direct connection between two seemingly different magic criteria: the Triangle Criterion and the stabilizer R\'enyi entropy characterization of pure-state magic.
This provides another example of how the WE framework can unify resource indicators that arise from quite different mathematical constructions.
\begin{corollary}[Pure-state reduction of WE $4$-criterion with triangle witness for qubit magic]\label{cor:pure-C4_magic}
Let $\psi=|\psi\rangle\langle\psi|$ be a pure $n$-qubit state and $d=2^n$. Denote its Pauli fourth moment as 
\begin{equation}\label{eq:def_P2}
\mathfrak{P}_2(\psi)=d\Tr(Q\psi^{\otimes4})=\frac{1}{d}\sum_{P\in\mathcal P_n}|\langle\psi|P|\psi\rangle|^4.
\end{equation}
Then its WE $4$-norm in criterion~\eqref{eq:WE_alpha_criterion} with standard triangle witness $\widetilde{T}$ in Eq.~\eqref{eq:qubit_magic_witness} and $n$-qubit Clifford group $\Cl_n$ in Eq.~\eqref{eq:qubit_magic_free_unitaries} is given by
\begin{equation}\label{eq:qubit_magic_norm}
\norm{\psi}_{\widetilde{T},\mathrm{Cl}_n,4}^4=\frac{24(7-4\sqrt{3})(d-8)}{d(d-1)(d+1)(d+2)(d+4)}\mathfrak{P}_2(\psi)+C_4,
\end{equation}
where $C_4$ is a constant related to $d$, given in Eq.~\eqref{eq:c_4} in Appendix~\ref{app:qubit_magic_T_coefficients}.
In particular, every pure stabilizer state $\phi$ satisfies $\mathfrak{P}_2(\phi)=1$, and
the stabilizer threshold is given by
\begin{equation}\label{eq:magic_threshold}
C_{\tilde{T},\mathrm{Cl}_n,4}^4=\frac{24(7-4\sqrt{3})(d-8)}{d(d-1)(d+1)(d+2)(d+4)}+C_4.
\end{equation}
\end{corollary}
See the proof of Corollary~\ref{cor:pure-C4_magic} in Appendix~\ref{app:qubit_magic_T_coefficients}.
This pure-state reduction exposes a limitation of the WE construction based on the triangle witness.
The fourth-order criterion is nontrivial only for small system sizes, mainly because of the factor $(d-8)$ appearing in the numerator on the right-hand side of Eq.~\eqref{eq:qubit_magic_norm}.
This factor of the coefficient in front of $\mathfrak{P}_2$ vanishes at $d=8$ and changes sign when $d>8$, so the resulting criterion only works for $n=1,2$ qubits.

There are two related reasons for this limitation.
First, it reflects the averaging effect intrinsic to the WE framework: twirling over a large free unitary orbit can suppress the resource sensitive part of the seed witness and thereby weaken mixed-state detection.
A similar phenomenon already appeared in the entanglement example, Eq.~\eqref{eq:WE_2_criterion_entanglement}, where the WE $2$-criterion for entanglement contains dimension suppressed coefficients in front of the subsystem purities, whereas the optimal purity-based criterion has constant coefficients.
Second, this result also suggests that the triangle witness may not be the best seed for WE in the qubit-magic setting.
In particular, $\widetilde T$ may not be sufficiently sensitive to magic, especially to multi-qubit magic.
It is therefore natural to ask whether other choices of seed operators can overcome these limitations and lead to stronger mixed-state magic criteria.

\subsection{WE 4-criterion with projector of a magic state}\label{subsec:magic_projector}
From the discussions in Section~\ref{subsec:generalization}, the role of standard witness $\widetilde{W}$ can be replaced by any positive semi-definite operator that is sufficiently sensitive to the target resource, in the sense that its expectation value for free states is bounded while resource states can exceed this threshold.
A particular candidate is a rank-$1$ projector onto a highly magic pure state vector $\ket{\Phi}$,
say
\begin{equation}
    \Phi=\ketbra{\Phi}{\Phi}.
\end{equation}
If $\Phi$ is highly resourceful (magic), then its maximal fidelity with free (stabilizer) states decreases as the system size grows, while the maximal expectation value over all quantum states equals one.
Thus, $\Phi$ provides a natural resource-sensitive positive seed operator.
The resulting WE $4$-criterion is given as follows.

\begin{corollary}[WE $4$-criterion with rank-$1$ projector for qubit magic]
\label{cor:projector_seed_magic}
For any $n\in\mathbb{N}_+$ and let $d=2^n$, there exists a pure $n$-qubit state $\Phi=\ketbra{\Phi}{\Phi}$ such that $\mathfrak{P}_2(\Phi)<4/(d+3)$ in Eq.~\eqref{eq:def_P2}. For every $n$-qubit state $\rho$, its WE $4$-norm in criterion~\eqref{eq:WE_alpha_criterion} with rank-$1$ projector $\Phi$ and $n$-qubit Clifford group $\Cl_n$ in Eq.~\eqref{eq:qubit_magic_free_unitaries} is given by
\begin{equation}\label{eq:projector_seed_magic_norm}
\begin{aligned}
&\norm{\rho}_{\Phi,\Cl_n,4}^4
=
\frac{6\mathfrak{P}_2(\Phi)}{d(d+1)(d+2)}
\widehat q_{[4]}(\rho)
\\
&\quad+
\frac{24(d-\mathfrak{P}_2(\Phi))}{d(d-1)(d+1)(d+2)(d+4)}
\left(
\widehat p_{[4]}(\rho)-\widehat q_{[4]}(\rho)
\right),
\end{aligned}
\end{equation}
where $\widehat p_{[4]}(\rho)$ and $\widehat q_{[4]}(\rho)$ are functions defined in Eq.~\eqref{eq:pqhat-def-inline} in Appendix~\ref{app:qubit_magic_WE_4-norm}.
The stabilizer threshold is
\begin{equation}\label{eq:qubit_magic_projector_threshold}
    C_{\Phi,\Cl_n,4}^4
    =
    \frac{6(\mathfrak{P}_2(\Phi)+4)}{d(d+1)(d+2)(d+4)}.
\end{equation}
This WE $4$-criterion detects all $n$-qubit pure magic states, as for any pure $n$-qubit state $\psi=\ketbra{\psi}{\psi}$, we have
\begin{equation}
\begin{aligned}
\norm{\psi}_{\Phi,\Cl_n,4}^4
-
C_{\Phi,\Cl_n,4}^4
&=
\frac{
6[4-\mathfrak{P}_2(\Phi)(d+3)]
}{
d(d-1)(d+1)(d+2)(d+4)
}
\\
&\quad\times
\left(
1-\mathfrak{P}_2(\psi)
\right),
\end{aligned}
\label{eq:projector_seed_pure_magic}
\end{equation}
so
\begin{equation}\label{eq:qubit_magic_WE_4_projector_pure}
\norm{\psi}_{\Phi,\Cl_n,4}>C_{\Phi,\Cl_n,4}\quad\Longleftrightarrow\quad\psi\notin\Sigma_{n}.
\end{equation}

\end{corollary}

The proof is given in Appendix~\ref{app:qubit_magic_Phi_coefficients}.
This corollary highlights the advantage of using a resource-sensitive positive semi-definite operator as the seed of WE.
First, it overcomes the pure-state limitation of dimension for the triangle witness construction in Corollary~\ref{cor:pure-C4_magic}.
Indeed, the criterion in Eq.~\eqref{eq:projector_seed_pure_magic} is nontrivial as long as $\mathfrak{P}_2(\Phi)<4/(d+3)$. 
Such pure states with high magic exist for each $n\in\mathbb{N}_+$~\cite{Leone2022sre}, and we provide a construction in Appendix~\ref{app:qubit_magic_Phi_coefficients} with 
\begin{equation}
4/(d+3)-\mathfrak{P}_2(\Phi)=\Theta(d^{-2})
\end{equation}
More importantly, Eq.~\eqref{eq:projector_seed_magic_norm} shows that the dependence of the criterion on the seed state $\Phi$ enters only through the single scalar quantity $\mathfrak{P}_2(\Phi)$.
Therefore, to construct a nontrivial WE criterion, it is sufficient to identify an admissible value of $\mathfrak{P}_2(\Phi)$ satisfying the above bound, without the explicit wavefunction of $\Phi$.
This substantially reduces the complexity of searching for a useful seed operator.
In this sense, this construction provides a concrete example of the generalizing strategy discussed in Section~\ref{subsec:generalization}: replacing a standardized linear witness with a more resource-sensitive positive semi-definite operator can lead to stronger magic criteria in the WE framework.

\section{Fermionic non-Gaussianity}\label{sec:fermion}

Mixed-state fermionic non-Gaussianity is important from the perspective of fermionic quantum computation.
This viewpoint is especially natural for quantum computing architectures and simulation tasks built from fermionic degrees of freedom, such as electronic systems~\cite{bravyi2002FermionicQuantumComputation}.
At the same time, circuits composed only of Gaussian (free-fermionic) states and operations, much like stabilizer circuits, are efficiently classically simulable~\cite{jozsa2008MatchgatesClassicalSimulation}.
Thus, detecting fermionic non-Gaussianity, especially relative to the convex hull of pure Gaussian states, is closely related to identifying the resources responsible for quantum computational advantage~\cite{hebenstreit2019AllPureFermionic}.
However, compared with stabilizer magic, the available tools for characterizing mixed-state fermionic non-Gaussianity remain limited.
Most existing criteria either depend on numerical optimization~\cite{vershynina2014CompleteCriterionConvexGaussianstate}, focus on pure states~\cite{lyu2024FermionicGaussianTesting,coffman2025MeasuringNonGaussianMagic,sierant2026FermionicMagicResources}, apply only to special low-dimensional cases~\cite{oszmaniec2014ClassicalSimulationFermionic}, or characterize mixed states with respect to notions other than the convex hull of pure Gaussian states~\cite{bittel2025OptimalTraceDistanceBoundsa,haug2026practical}.
As a result, the general problem of detecting membership in the convex hull, as an analogy to the stabilizer polytope in magic resource theory, remains much less understood.
Therefore, constructing new mixed-state fermionic non-Gaussianity tests through WE is of direct practical value for identifying non-Gaussian resources in fermionic quantum computation.

We encode an $n$-mode fermionic system into $n$ qubits through the Jordan--Wigner transformation~\cite{jordan1928BerPaulischeQuivalenzverbot} and introduce the $2n$ Majorana operators
\begin{equation}
    \gamma_{2k-1}
    =
    \left(\prod_{j=1}^{k-1} Z_j\right)X_k,
    \quad
    \gamma_{2k}
    =
    \left(\prod_{j=1}^{k-1} Z_j\right)Y_k,
    \quad k\in[n].
\end{equation}
They satisfy $\gamma_j=\gamma_j^\dagger$, $\gamma_j^2=\mathbb I$, and $\{\gamma_j,\gamma_k\}=2\delta_{j,k}\mathbb I$.
For every subset $S=\{s_1<\cdots<s_{\abs{S}}\}\subseteq[2n]$, we denote the Majorana product as $\gamma_S$ and its Hermitian version as $\widehat{\gamma}_S$, they are defined by
\begin{equation}\label{eq:def_Majorana_product}
    \gamma_S=\prod_{j=1}^{\abs{S}}\gamma_{s_j},
    \qquad
    \widehat{\gamma}_S=\widehat{\gamma}_S^\dagger=(-i)^{\abs{S}(\abs{S}-1)/2}\gamma_S.
\end{equation}
Especially $\gamma_\emptyset=\widehat{\gamma}_\emptyset=\bI$. The operators $\{\widehat{\gamma}_S\}_{S\subseteq[2n]}$ form an orthogonal operator basis of $\mathcal{L}((\mathbb{C}^2)^{\otimes n})$ w.r.t.\  the Hilbert--Schmidt inner product. 

Due to the superselection rule, a physical fermionic state is parity-preserving, i.e., its density matrix commutes with the parity operator $Z^{\otimes n}$.
Equivalently, the parity-preserving fermionic states are those whose Majorana expansion contains only even-degree terms, denoted as
\begin{equation}
\begin{aligned}
\mathcal D_F
&=\left\{\rho\in\mathcal D((\mathbb C^2)^{\otimes n})
\middle\mid
[\rho,Z^{\otimes n}]=0
\right\}
\\&=
\left\{\rho\in\mathcal D((\mathbb C^2)^{\otimes n})
\middle\mid
\Tr(\gamma_S\rho)=0,\forall\abs{S}\textup{ is odd}
\right\}.
\end{aligned}
\end{equation} 
For a physical fermionic state $\rho\in\mathcal{D}_F$, we define the \emph{even Majorana sector purities} for $\ell\in\{0,1,\cdots,n\}$ as
\begin{equation}\label{eq:Majorana_sector_purity_main}
B_\ell(\rho)
:=
\sum_{\abs{S}=2\ell}
\Tr(\widehat{\gamma}_S\rho)^2.
\end{equation}
This quantity measures the total squared, unnormalized Hilbert--Schmidt overlap of $\rho$ with Hermitian Majorana products of degree $2\ell$.
Equivalently, if $\Pi_{2\ell}$ denotes the Hilbert--Schmidt projection onto the degree-$2\ell$ Majorana sector, then (see the proof in Appendix~\ref{app:FNG_WE_projector_criterion})
\begin{equation}
\begin{aligned}
B_\ell(\rho)
=
2^n\Tr\left(\Pi_{2\ell}(\rho)^2\right).
\end{aligned}
\end{equation}

The free unitaries are the Gaussian (free-fermionic) unitaries, or equivalently matchgate unitaries~\cite{Valiant2002QuantumCircuits}, which act linearly on Majorana operators.
For every $R\in \mathrm{O}(2n)$ in the orthogonal group of degree $2n$, a Gaussian unitary $U_R$ is defined as to satisfy
\begin{equation}
    U_R^\dagger \gamma_j U_R
    =
    \sum_{k=1}^{2n} R_{j,k}\gamma_k.
\end{equation}
We denote the set of $n$-qubit Gaussian unitaries by $\mathrm M_n$
\begin{equation}\label{eq:FNG_free_unitaries}
\Ufree=\mathrm{M}_n=\left\{U_R\mid R\in\mathrm{O}(2n)
\right\} .
\end{equation}
The pure Gaussian (free-fermionic) states are
\begin{equation}
    \mathcal G_n
    =
    \left\{
    U\ketbra{0^n}{0^n}U^\dagger
    \middle\mid
    U\in\mathrm{M}_n
    \right\},
\end{equation}
and the free states are the convex hull of 
all pure Gaussian states (convex Gaussian states), 
\begin{equation}
\mathcal{F}=\mathrm{Conv}(\mathcal G_n).
\end{equation}
A state outside this convex hull is fermionic non-Gaussian. It is worth noting that in some literature the definition of mixed Gaussian states is narrower, i.e., the thermal state of a quadratic Hamiltonian in the Majorana operators~\cite{bittel2025OptimalTraceDistanceBoundsa,haug2026practical}, as they are contained in $\mathrm{Conv}(\mathcal G_n)$.

\subsection{Wick's theorem}

Wick's theorem certifies all pure Gaussian states~\cite{wick1950EvaluationCollisionMatrix}.
For a pure Gaussian state $\phi=\ketbra{\phi}{\phi}$, define its two-point correlation matrix $\Gamma(\phi)$ as
\begin{equation}\label{eq:covariant_matrix_main}
\Gamma_{j,k}(\phi)
=
-\frac{i}{2}\Tr\left([\gamma_j,\gamma_k]\phi\right),\qquad j,k\in[2n].
\end{equation}
Wick's theorem relates all higher-degree Majorana correlation functions of a Gaussian state to the Pfaffian of the correlation matrix. For $S\subseteq[2n]$ with even $\abs{S}$, the expectation value of Hermitian Majorana product $\widehat{\gamma}_S$ is given by
\begin{equation}
    \Tr(\widehat{\gamma}_S\phi)
    =
    \operatorname{Pf}\left(\Gamma_S(\phi)\right),
\end{equation}
where $\operatorname{Pf}$ is short for Pfaffian and $\Gamma_S(\phi)$ is the restriction of the correlation matrix to the indices in $S$.
Especially, all odd-degree Majorana correlators vanish,
\begin{equation}
    \Tr(\widehat{\gamma}_S\phi)=0,
    \qquad \textup{for odd }\abs{S}.
\end{equation}
This consequence of Wick's theorem already gives a family of linear witnesses with odd $\abs{S}$
\begin{equation}\label{eq:FNG_odd_majorana_witness}
W=\widehat{\gamma}_S,\qquad\widetilde{W}=\widetilde{\widehat{\gamma}}_S=\bI-\widehat{\gamma}_S .
\end{equation}
Indeed, if $\Tr(\widehat{\gamma}_S\psi)\neq 0$ for some odd $S$ and pure state $\psi$, then $\psi$ cannot be a Gaussian state.
By linearity, the same statement holds for mixed states and the convex hull of pure Gaussian states.

However, this simplest odd Majorana witness turns out to be trivial for detecting fermionic non-Gaussianity in the physical fermionic state space. As shown in the following theorem, the WE $2$-norm with an odd-degree Majorana product can only test the absence of odd Majorana components of the chosen degree, and collecting all odd degrees gives a test of whether a density matrix is a physical fermionic state (see the proof in Appendix~\ref{app:FNG_Wick_theorem}).

\begin{proposition}[WE $2$-criterion with odd-degree Majorana product for fermionic non-Gaussianity]
\label{thm:wick_odd_majorana_second_order}
Let $S\subseteq[2n]$ with $\abs{S}$ odd.
For any $n$-qubit state $\rho$, the WE $2$-criterion in criterion~\eqref{eq:WE_alpha_criterion} with $\widetilde{\widehat{\gamma}}_S$ in Eq.~\eqref{eq:FNG_odd_majorana_witness} and Gaussian unitaries $\mathrm{M}_n$ in Eq.~\eqref{eq:FNG_free_unitaries} is given by
\begin{equation}
\norm{\rho}_{\widetilde{\widehat{\gamma}}_S,\mathrm{M}_n,2}^2
=
1+\binom{2n}{\abs{S}}^{-1}
\sum_{\substack{T\subseteq[2n]\\ |T|=\abs{S}}}
\Tr(\widehat{\gamma}_T\rho)^2.
\end{equation}
The threshold value for all convex Gaussian states is $C_{\widetilde{\widehat{\gamma}}_S,\mathrm{M}_n,2}=1$. In particular, $\norm{\rho}_{\widetilde{\widehat{\gamma}}_S,\mathrm{M}_n,2}=1$ for every parity-preserving fermionic state $\rho\in\mathcal D_F$.
\end{proposition}

Therefore, once we restrict to physical fermionic states, $\norm{\rho}_{\widetilde{\widehat{\gamma}}_S,\mathrm{M}_n,2}$ are all identical.
In this sense, the odd-degree Majorana products result in only a parity test, not a genuine test of fermionic non-Gaussianity.
To obtain a nontrivial mixed-state fermionic non-Gaussianity criterion, we need a stronger seed witness.

\subsection{WE criteria with projector of a non-Gaussian state}\label{subsec:fermionic_projector}
We therefore turn to a seed operator based on a fidelity bound. Since for $n\leq3$, all physical fermionic states are convex Gaussian~\cite{bravyi2005ClassicalCapacityFermionic}, we assume $n\geq4$ in the following text. Let
\begin{equation}
n=4m+r,\quad\textup{with } r\in\{0,1,2,3\}\textup{ and } m\geq1.
\end{equation}
From the $n$-qubit GHZ state vector $\ket{\textup{GHZ}_n}=(\ket{0^n}+\ket{1^n})/\sqrt{2}$, we define the following pure fermionic non-Gaussian states 
\begin{equation}
\begin{aligned}
\ket{\eta_0}&=\ket{\textup{GHZ}_4},\qquad&
\ket{\eta_1}&=\ket{\textup{GHZ}_4}\otimes\ket{0},\\
\ket{\eta_2}&=\ket{\textup{GHZ}_6},&
\ket{\eta_3}&=\ket{\textup{GHZ}_6}\otimes\ket{0},
\end{aligned}
\end{equation}
and the rank-$1$ projector to a generic $n$-qubit non-Gaussian state
\begin{equation}\label{eq:FNG_projector}
P_n=\ket{\eta_r}\bra{\eta_r}\otimes\ket{\mathrm{GHZ}_4}^{\otimes(m-1)}\bra{\mathrm{GHZ}_4}^{\otimes(m-1)}.
\end{equation}
Using Lemma~1 in Ref.~\cite{reardon-smith2024FermionicLinearOptical}, we show that every $n$-qubit convex Gaussian state $\sigma\in\mathrm{Conv}(\G_n)$ satisfies (see Appendix~\ref{app:FNG_WE_projector_criterion})
\begin{equation}\label{eq:overlap_convex_Gaussian_GHZ_r}
    \Tr(\sigma P_n)
    \leq
    2^{-\lfloor n/4\rfloor}.
\end{equation}
This immediately implies $P_n$ is a scalable, resource-sensitive positive semi-definite operator for fermionic non-Gaussianity. The WE $1$-criterion is trivial, and the second- and third-order criteria are summarized in the following theorem (see the proof in Appendix~\ref{app:FNG_WE_projector_criterion}).

\begin{theorem}[WE $2,3$-criteria with rank-$1$ projector for fermionic non-Gaussianity]
\label{thm:R1-3_FNG_main}
For an $n$-qubit fermionic state $\rho\in\mathcal D_F$, the WE $2,3$-norms in criterion~\eqref{eq:WE_alpha_criterion} with rank-$1$ projector $P_n$ in Eq.~\eqref{eq:FNG_projector} and Gaussian unitaries $\mathrm{M}_n$ in Eq.~\eqref{eq:FNG_free_unitaries} are given by
\begin{equation}\label{eq:R1-3_FNG}
\begin{aligned}
\norm{\rho}_{P_n,\mathrm{M}_n,2}^2
&=\frac1{d^2}
\sum_{\ell=0}^{n}
\frac{B_\ell(P_n)B_\ell(\rho)}{\binom{2n}{2\ell}},
\\
\norm{\rho}_{P_n,\mathrm{M}_n,3}^3
&=
\frac{1}{d^3}
\sum_{\substack{\ell_1,\ell_2,\ell_3\ge 0\\ \ell_1+\ell_2+\ell_3\le n}}
\frac{\cT_{\ell_1,\ell_2,\ell_3}(P_n)\cT_{\ell_1,\ell_2,\ell_3}(\rho)}{\binom{2n}{2\ell_1,2\ell_2,2\ell_3,2n-2(\ell_1+\ell_2+\ell_3)}},
\end{aligned}
\end{equation}
where $d=2^n$, $B_\ell(\rho)$ is defined in Eq.~\eqref{eq:Majorana_sector_purity_main} and $\cT_{\ell_1,\ell_2,\ell_3}(\rho)$ is defined in Eq.~\eqref{eq:def_Pi_T} in Appendix~\ref{app:FNG_WE_projector_criterion}. The coefficients $B_\ell(P_n)$ and $\cT_{\ell_1,\ell_2,\ell_3}(P_n)$ are determined in Lemma~\ref{lem:Pi_T_P_n} in Appendix~\ref{app:FNG_WE_projector_criterion}.
The threshold values for convex Gaussian states are given by
\begin{equation}\label{eq:FNG_thresholds_main}
\begin{aligned}
C_{P_n,\mathrm{M}_n,2}^2
&=
\frac1{d^2}\sum_{\ell=0}^{n}
\frac{\binom{n}{\ell}B_\ell(P_n)}
{\binom{2n}{2\ell}},
\\
C_{P_n,\mathrm{M}_n,3}^3
&=
\frac{1}{d^3}
\sum_{\substack{\ell_1,\ell_2,\ell_3\ge 0\\ \ell_1+\ell_2+\ell_3\le n}}
\frac{(-1)^{\ell_1+\ell_2+\ell_3}
\binom{n}{\ell_1,\ell_2,\ell_3,n-(\ell_1+\ell_2+\ell_3)}}
{
\binom{2n}{2\ell_1,2\ell_2,2\ell_3,2n-2(\ell_1+\ell_2+\ell_3)}
}
\\&\hspace{6.4em}\times\cT_{\ell_1,\ell_2,\ell_3}(P_n).
\end{aligned}
\end{equation}
\end{theorem}
Although the expressions in Theorem~\ref{thm:R1-3_FNG_main} are more involved than the criterion induced by Wick's theorem, they detect mixed-state fermionic non-Gaussianity, rather than merely testing whether a state is parity preserving.
There is also a useful pure-state reduction, analogous to the relation between the WE criterion for qubit magic and stabilizer R\'enyi entropy.
For a fixed-parity pure fermionic state $\psi=\ketbra{\psi}{\psi}$, denote its correlation matrix in Eq.~\eqref{eq:covariant_matrix_main} as $\Gamma(\psi)$.
The \emph{first-order fermionic antiflatness}  (FAF) is defined as~\cite{sierant2026FermionicMagicResources}
\begin{equation}\label{eq:def_FAF1}
    \mathfrak F_1(\psi)
    =
    n-\frac12\Tr\left(\Gamma(\psi)^{\mathrm T}\Gamma(\psi)\right).
\end{equation}
This quantity is a faithful pure-state fermionic non-Gaussianity measure: $\mathfrak F_1(\psi)\geq0$, with equality if and only if $\ket{\psi}$ is Gaussian~\cite{sierant2026FermionicMagicResources}.
In the pure-state regime analyzed below, the WE $2$-norm directly relates to FAF for small system sizes.
Thus, the same mixed-state WE construction recovers a standard pure-state fermionic non-Gaussianity test as a special case (see the proof in Appendix~\ref{app:FNG_WE_projector_criterion}).
\begin{corollary}[WE $2$-criterion with rank-$1$ projector is faithful for pure states of small size]\label{cor:fermion_WE2}
Let $\psi=\ketbra{\psi}{\psi}$ be a physical pure fermionic state of $4\leq n\leq7$ qubits, $\norm{\psi}_{P_n,\mathrm{M}_n,2}^2$ in Eq.~\eqref{eq:R1-3_FNG} is a linear function of the first-order FAF $\mathfrak F_1(\psi)$ in Eq.~\eqref{eq:def_FAF1}, explicitly,
\begin{equation}
\begin{aligned}
\norm{\psi}_{P_4,\mathrm{M}_4,2}^2=&\dfrac{8+\mathfrak F_1(\psi)}{640},\qquad&\norm{\psi}_{P_5,\mathrm{M}_5,2}^2&=\dfrac{40+\mathfrak F_1(\psi)}{11520},\\
\norm{\psi}_{P_6,\mathrm{M}_6,2}^2=&\dfrac{52+\mathfrak F_1(\psi)}{59136},\qquad&\norm{\psi}_{P_7,\mathrm{M}_7,2}^2&=\dfrac{364+\mathfrak F_1(\psi)}{1537536}.
\end{aligned}
\end{equation}   
Thus for any $n$-qubit pure physical fermionic state $\psi$ with $4\leq n\leq7$, $\norm{\psi}_{P_n,\mathrm{M}_n,2}>C_{P_n,\mathrm{M}_n,2}\Longleftrightarrow\psi\notin\mathcal{G}_n$.
\end{corollary}
Although both constructions in Corollary~\ref{cor:projector_seed_magic} and in Theorem~\ref{thm:R1-3_FNG_main} use a pure resource state as the seed for WE, they are conceptually different.
The difference comes from the distinct mathematical structures of the Clifford group and the matchgate group.
In the qubit-magic case, the fourth-order Clifford-twirled response of a pure projector seed depends on the seed only through the single scalar quantity $\mathfrak{P}_2(\Phi)$.
As a result, one can specify the criterion by choosing an admissible value of $\mathfrak{P}_2(\Phi)$, without needing to write down the seed state explicitly.
In the fermionic case, however, the second-order matchgate twirl generally involves several Majorana sectors with different degrees once $n\geq 8$.
These additional terms are not related to a single scalar quantity analogous to $\mathfrak{P}_2(\Phi)$, but depend on several correlated properties of the seed state.
Therefore, unlike the qubit magic criterion in Corollary~\ref{cor:projector_seed_magic}, the WE criterion for fermionic non-Gaussianity generally requires an explicit construction of the seed state $P_n$.
This is why we use the GHZ-based family in Eq.~\eqref{eq:FNG_projector}, whose structure allows the required WE norms and convex Gaussian thresholds to be evaluated analytically. 

It is worth noting that the criteria in Theorem~\ref{thm:R1-3_FNG_main} do not guarantee to remain nontrivial for any $n$. Numerical calculation validates that WE $2,3$-criteria successfully detect non-Gaussian states $P_n$ for $n\leq10$, but both fail when $n=11$ (see Appendix~\ref{app:FNG_WE_projector_criterion}).

\subsection{WE 2-criterion with an even-degree Majorana product}
\label{subsec:fermionic_pauli}

The discussion above shows that the projector seed already gives mixed-state criteria, but these criteria are not guaranteed to remain nontrivial for arbitrary system sizes.
We now show that a much simpler family of seed operators, namely Pauli strings of even Majorana degree, can already lead to nontrivial WE criteria for fermionic non-Gaussianity. Remarkably, the WE $2$-criterion with a degree-$(2\floor{n/2})$ Majorana product remains nontrivial for any $n$ and detects all $n$-qubit pure fermionic non-Gaussian states.

Following the discussion in Section~\ref{subsec:generalization}, the seed operator need not be positive semi-definite when $\alpha$ is an even integer.
In particular, for the WE $2$-norm, the relevant functional is an $L^2$-norm of linear expectation values and is therefore convex.
Hence, its maximum over $\mathrm{Conv}(\mathcal G_n)$ is attained on pure Gaussian states.
For $1\leq \ell\leq n$, define the Hermitian Majorana product
\begin{equation}
\label{eq:FNG_pauli_general}
    \widehat{\gamma}_{T_\ell}
    =
    \prod_{j=1}^{\ell} Z_j,
    \qquad
    T_\ell=[2\ell].
\end{equation}
Equivalently, $\widehat{\gamma}_{T_\ell}$ is a Majorana product of degree $2\ell$.

The resulting WE $2$-criterion is as follows; see Appendix~\ref{app:FNG_Pauli} for the proof.

\begin{theorem}[WE $2$-criterion with an even-degree Majorana product for fermionic non-Gaussianity]
\label{thm:FNG_Pauli}
For an $n$-qubit fermionic state $\rho\in\mathcal D_F$ and any $1\leq \ell\leq n$, the WE $2$-norm in criterion~\eqref{eq:WE_alpha_criterion} with an even-degree Majorana product $\widehat{\gamma}_{T_\ell}$ in Eq.~\eqref{eq:FNG_pauli_general} and Gaussian unitaries $\mathrm{M}_n$ in Eq.~\eqref{eq:FNG_free_unitaries} is given by
\begin{equation}
\label{eq:FNG_WE_2_norm_Pauli_general}
    \norm{\rho}_{\widehat{\gamma}_{T_\ell},\mathrm{M}_n,2}^2
    =
    \binom{2n}{2\ell}^{-1} B_\ell(\rho),
\end{equation}
where $B_\ell(\rho)$, defined in Eq.~\eqref{eq:Majorana_sector_purity_main}, is the degree-$2\ell$ Majorana sector purity of $\rho$.
The exact convex Gaussian threshold is
\begin{equation}
\label{eq:FNG_WE_2_threshold_Pauli_general}
    C_{\widehat{\gamma}_{T_\ell},\mathrm{M}_n,2}^2
    =
    \frac{\binom{n}{\ell}}{\binom{2n}{2\ell}} .
\end{equation}
Consequently,
\begin{equation}
\label{eq:WE_2_criterion_Pauli_FNG_general}
    B_\ell(\rho)>\binom{n}{\ell}
    \quad\Longrightarrow\quad
    \rho\notin\mathrm{Conv}(\mathcal G_n).
\end{equation}
\end{theorem}

The contrast with Proposition~\ref{thm:wick_odd_majorana_second_order} highlights the importance of choosing an appropriate seed operator in the WE framework.
There, the seed is a standard witness constructed from an odd-degree Majorana product, and the resulting second-order criterion only tests whether a state is parity preserving.
Thus, once restricted to physical fermionic states, it does not give a genuine test of fermionic non-Gaussianity.
By contrast, the even-degree Pauli string in Theorem~\ref{thm:FNG_Pauli} yields a nontrivial mixed-state criterion, whose value is simply the Majorana sector purity $B_\ell(\rho)$ normalized by the number of degree-$2\ell$ Majorana products.

The detection capability of this criterion depends strongly on the choice of $\ell$.
This dependence is already visible on pure states.
For $\ell=1$, Eq.~\eqref{eq:FNG_WE_2_norm_Pauli_general} gives
\begin{equation}
\label{eq:FNG_Pauli_l1_FAF}
    \norm{\psi}_{\widehat{\gamma}_{T_1},\mathrm{M}_n,2}^2
    =
    \frac{B_1(\psi)}{\binom{2n}{2}}
    =
    \frac{n-\mathfrak F_1(\psi)}{\binom{2n}{2}},
\end{equation}
where $\mathfrak F_1$ is the first-order FAF in Eq.~\eqref{eq:def_FAF1}.
Since $\mathfrak F_1(\psi)\geq0$ for all pure physical fermionic states and vanishes exactly on pure Gaussian states, the $\ell=1$ criterion is negatively correlated with fermionic non-Gaussianity and cannot detect any pure non-Gaussian state.
Increasing $\ell$ makes the criterion sensitive to higher-degree Majorana correlations.
The most useful choice turns out to be the central sector,
\begin{equation}
    \ell_*=\floor{n/2}.
\end{equation}
For this choice, the corresponding Hermitian Majorana product is
\begin{equation}
\label{eq:FNG_pauli_central}
    \widehat{\gamma}_{T_{\ell_*}}
    =
    \prod_{j=1}^{\floor{n/2}} Z_j,
    \qquad
    T_{\ell_*}
    =
    \begin{cases}
    [n], & n\ \text{is even},\\
    [n-1], & n\ \text{is odd}.
    \end{cases}
\end{equation}
Remarkably, the WE criterion associated with this central Majorana sector is faithful on all fixed-parity pure states.
Therefore, it remains nontrivial for an arbitrary number of qubits.

\begin{corollary}[WE $2$-criterion with central sector Majorana product is faithful for all pure fermionic states]
\label{cor:FNG_Pauli_pure}
Let $\ell_*=\floor{n/2}$.
The WE $2$-criterion with  $\widehat{\gamma}_{T_{\ell_*}}$ in Eq.~\eqref{eq:FNG_pauli_central} and Gaussian unitaries $\mathrm{M}_n$ in Eq.~\eqref{eq:FNG_free_unitaries} detects all pure fermionic non-Gaussian states, i.e., for every pure physical fermionic state $\psi=\ketbra{\psi}{\psi}\in\mathcal D_F$,
\begin{equation}
    \norm{\psi}_{\widehat{\gamma}_{T_{\ell_*}},\mathrm{M}_n,2}
    >
    C_{\widehat{\gamma}_{T_{\ell_*}},\mathrm{M}_n,2}
    \quad\Longleftrightarrow\quad
    \psi\notin\mathcal G_n .
\end{equation}
Equivalently, the central Majorana sector purity satisfies
\begin{equation}
    B_{\floor{n/2}}(\psi)
    \geq
    \binom{n}{\floor{n/2}},
\end{equation}
with equality if and only if $\psi\in\mathcal G_n$.
\end{corollary}

This corollary shows that the WE $2$-criterion with a degree-$(2\floor{n/2})$ Majorana product overcomes the main limitations of the three criteria discussed above while retaining their useful features.
Compared with Proposition~\ref{thm:wick_odd_majorana_second_order}, it gives a nontrivial test of fermionic non-Gaussianity rather than merely checking the parity-preserving condition.
Compared with the projector-based criteria in Theorem~\ref{thm:R1-3_FNG_main}, it has a much simpler mathematical form and remains nontrivial for arbitrary numbers of qubits.
Compared with FAF, it is not restricted to pure states, but provides a genuine mixed-state criterion for detecting states outside the convex hull of pure fermionic Gaussian states.
To the best of our knowledge, this is the first fermionic non-Gaussianity detection criterion for arbitrary system size with all these properties, further demonstrating the power of the WE framework for constructing resource detection protocols.

\section{Summary and outlook}\label{sec:outlook}

In this work, we introduced the witness expansion framework as a general method for turning a linear resource witness into a family of nonlinear, multi-copy resource detection criteria.
Starting from only a linear seed witness $W$ and a set of free unitaries $\mathcal G$, the framework generates an $\alpha$-copy criterion whose detection capability can be tuned by $\alpha$ and whose value can be analytically computed from the corresponding twirling formula.
The construction is largely independent of the detailed geometry of the free state set: the free states enter mainly through the threshold $C_{\widetilde{W},\mathcal G,\alpha}$, which in many symmetric cases can be evaluated on a representative pure free state.
This minimal input makes the framework broadly applicable to different resource theories.
We demonstrated this universality through coherence, entanglement, qudit and qubit magic, and fermionic non-Gaussianity, recovering several known quantities and obtaining new mixed-state criteria within the same formalism.
In all these examples, a single WE criterion can detect all pure resource states, showing that the nonlinear expansion can go beyond the detection capability of the original linear witness while retaining analytical tractability.
We also showed that the detection capability of WE can be further improved by optimizing the seed witness and by choosing suitable free unitary subgroups.

A first open direction concerns the experimental accessibility of the WE criteria.
In this work, we have mainly focused on constructing functions for resource detection, rather than on designing concrete measurement protocols.
Although the polynomial quantities produced by WE can in principle be represented as observables on multiple copies of the state, this does not automatically imply low measurement complexity.
The actual cost depends on the specific resource theory and on the structure of the resulting polynomial.
For instance, some entanglement-related quantities can be measured efficiently, such as the purity and partial-transpose moment, while nonlinear magic quantities such as stabilizer R\'enyi entropies generally require more sophisticated protocols~\cite{Haug2024efficient}.
It would therefore be important to analyze the measurement complexity of the magic and fermionic non-Gaussianity criteria derived using the WE framework.

Another direction is about how to further enhance the detection capability of WE.
Although our discussion of entanglement mainly served to show that WE can reproduce several known entanglement detection structures, it also reveals two important directions for improvement.
First, WE need not be restricted to unitary orbits.
The essential requirement is that the expansion maps a valid witness to another valid witness.
For entanglement, this points to possible extensions based on local filters $K_A\otimes K_B$ instead of local unitaries.
Such generalized expansion families could produce richer nonlinear criteria and may strengthen the large-$\alpha$ limit toward the full PPT criterion, rather than the restricted version obtained in Theorem~\ref{thm:WE_infinity_criterion_entanglement}.
Second, the same example shows that symmetry averaging, while analytically useful, can also weaken detection capability by producing criteria with small coefficients in front of the relevant nonlinear invariants.
The comparison between the WE $2$-criterion and the optimal purity-based entanglement criterion indicates that these coefficients may have substantial room for improvement.
This suggests a complementary strategy: use WE as a principled starting point to identify useful nonlinear invariants, and then optimize their coefficients for a given resource theory, such as magic and fermionic non-Gaussianity.
Such an optimized criterion could provide stronger mixed-state detection and better robustness.

Finally, it would be interesting to extend WE beyond the finite-dimensional, information-theoretic setting considered in this work.
One possible direction is to develop WE for continuous-variable systems, where the relevant twirling operations may be ill-defined or analytically intractable~\cite{Iosue2024continuous}.
This may require either an appropriate continuous-variable formulation of the framework or a controlled finite-dimensional truncation before applying WE.
Another related direction is to formulate resource detection within the emerging framework of computational quantum resource theories.
Recent findings show that information-theoretic resource measures, such as certain entanglement monotones, may lose their operational meaning once the allowed computational resources are required to be efficient~\cite{leone2025EntanglementTheoryLimited}.
It is therefore important to design detection protocols and identify resource monotones that remain efficiently computable, and to distinguish them from quantities that are formally well defined but computationally intractable.

\begin{acknowledgments}
We thank Satoya Imai, Ryuji Takagi, Salvatore Francesco Emanuele Oliviero, Lorenzo Leone, Lennart Bittel, Antonio Anna Mele, Chenfeng Cao, Jonas Kitzinger, and Tobias Haug for valuable discussions.
Y.~T.~is supported by the Quantum Flagship MILLENION.
Z.~-W.~L.~is supported in part by NSFC under Grant No.~12475023, Dushi Program, and startup funding from YMSC.
Y.~Z.~acknowledges support from the U.S. National Science Foundation under Grant No.\  DMR-2441671.
O.~G.~is supported by the Deutsche Forschungsgemeinschaft (DFG, German Research Foundation, project number 563437167), the Sino-German Center for Research Promotion (Project M-0294), and the German Federal Ministry of Research, Technology and Space (Project QuKuK, Grant 
No.\  16KIS1618K and Project BeRyQC, Grant No.\ 13N17292).
C.~Z.~and X.~W.~are supported by the National Natural Science Foundation of China (Grant No.~92576114, 12447107), the Guangdong Provincial Quantum Science Strategic Initiative (Grant No.~GDZX2403008, GDZX2503001), and the Guangdong Provincial Key Lab of Integrated Communication, Sensing and Computation for Ubiquitous Internet of Things (Grant No.~2023B1212010007). 
J.~E.~is supported by the BMFTR (DAQC, MuniQC-Atoms, Hybrid++, QuSol), the Munich Quantum Valley, Berlin Quantum,
the Clusters of Excellence (ML4Q, MATH+),
the Quantum Flagship programmes MILLENION 
and PASQUANS2, the DFG (CRC 183, SPP 2514), 
and the European Research Council (DebuQC).
\end{acknowledgments}

\newpage
\bibliography{ref_WE_ISO4_arxiv}

\clearpage

\appendix
\onecolumngrid
\appendixtableofcontents

\section{Preliminaries of the Witness expansion framework}\label{app:preliminaries}
In this appendix, we prove some basic properties of the framework for constructing nonlinear criteria and relate the violation of the WE criterion to the free robustness of a state. 

\subsection{Framework for constructing nonlinear criteria}\label{app:framework_nonlinear}
We first review the notations and definitions.
For a finite-dimensional Hilbert space $\mathcal{H}$ and density operators $\mathcal{D}(\mathcal{H})$, denote the set of free states as $\mathcal{F}\subseteq\mathcal{D}(\mathcal{H})$ and the set of free unitaries as $\mathcal{U}_\mathcal{F}$. Assume we have a probabilistic witness ensemble consisting of a measurable space $\Omega$ equipped with a probability measure $\mu$ and a family of witnesses $\left\{W_m\right\}_{m\in\Omega}$, with $\Tr(W_m\sigma)\geq0,\forall\sigma\in\mathcal{F}$ and $\Tr(W_m\rho)<0$ for some $\rho\in\mathcal{D}(\mathcal{H})\setminus\mathcal{F}$. For each witness $W_m$, denote its standardized version as $\widetilde{W}_m=\bI-W_m/\norm{W_m}_\infty$, such that $\Tr(\widetilde{W}_m\rho)\geq0$ for all states $\rho$ and $\Tr(\widetilde{W}_m\sigma)\leq1$ when $\sigma$ is a free state. Denote $\widetilde{W}_\Omega=\{\widetilde{W}_m\}_{m\in\Omega}$ and define the function $f_{\widetilde{W}_\Omega,\rho}:\Omega\rightarrow\mathbb{R}_{\geq0}$,
\begin{equation}
f_{\widetilde{W}_\Omega,\rho}(m)=\Tr(\widetilde{W}_m\rho).
\end{equation}
For any $1\leq\alpha<\infty$, noticing that each $\widetilde{W}_m\geq0$, we denote the $L^\alpha(\mu)$-norm of $f_{\widetilde{W}_\Omega,\rho}$ as
\begin{equation}\label{eq:def_R_alpha_mu_app}
\norm{f_{\widetilde{W}_\Omega,\rho}}_{L^\alpha(\mu)}
    =\left(\int_\Omega \Tr(\widetilde{W}_m\rho)^\alpha\mathrm{d}\mu(m)\right)^{1/\alpha}.
\end{equation}
When $\alpha\rightarrow\infty$, the limiting value is defined as the essential supremum
\begin{equation}
\begin{aligned}
\norm{f_{\widetilde{W}_\Omega,\rho}}_{L^\infty(\mu)}&=\esssup_{m\in\Omega}\Tr(\widetilde{W}_m\rho)
\\
&=\inf\left\{t\in\mathbb{R}\middle\mid\mu\left(\left\{m\mid\Tr(\widetilde{W}_m\rho)>t\right\}\right)=0\right\}.
\end{aligned}
\end{equation}
For each $\alpha\in[1,\infty]$, we denote the threshold value of $\norm{f_{\widetilde{W}_\Omega,\sigma}}_{L^\alpha(\mu)}$ achievable by any free state $\sigma$ as
\begin{equation}\label{def_C_W_Omega_app}
C_{\widetilde{W}_\Omega,\mu,\alpha}=\sup_{\sigma\in\mathcal{F}}\norm{f_{\widetilde{W}_\Omega,\sigma}}_{L^\alpha(\mu)}.
\end{equation}
Denote the set of states with values beyond this threshold as
\begin{equation}
D_{\widetilde{W}_\Omega,\mu,\alpha}\coloneqq\left\{\rho\middle\mid\norm{f_{\widetilde{W}_\Omega,\rho}}_{L^\alpha(\mu)}>C_{\widetilde{W}_\Omega,\mu,\alpha}\right\}.
\end{equation}

\begin{proposition}[Basic properties of the probabilistic witness ensemble]\label{prop:property_witness_ensemble}
For a probabilistic witness ensemble $\left\{\widetilde{W}_m\right\}$ equipped with $\mu$ introduced above, the following statements hold:
\begin{enumerate}
\item For each $\alpha\in[1,\infty]$, 
\begin{equation}
D_{\widetilde{W}_\Omega,\mu,\alpha}\cap\mathcal{F}=\emptyset,
\qquad
0\leq C_{\widetilde{W}_\Omega,\mu,\alpha}\leq 1.
\label{eq:Calpha-bound-general}
\end{equation}
\item For each fixed state $\rho$, the map $\alpha\mapsto \norm{f_{\widetilde{W}_\Omega,\rho}}_{L^\alpha(\mu)}$ is nondecreasing on $[1,\infty)$ and
\begin{equation}
\norm{f_{\widetilde{W}_\Omega,\rho}}_{L^\alpha(\mu)}\uparrow \norm{f_{\widetilde{W}_\Omega,\rho}}_{L^\infty(\mu)}
\qquad (\alpha\to\infty).
\label{eq:Ralpha-to-Rinfty-general}
\end{equation}
\item The thresholds satisfy
\begin{equation}
C_{\widetilde{W}_\Omega,\mu,\alpha}\uparrow C_{\widetilde{W}_\Omega,\mu,\infty}
\qquad (\alpha\to\infty).
\label{eq:Calpha-to-Cinfty-general}
\end{equation}
\end{enumerate}
\end{proposition}

\begin{proof}[Proof of Proposition~\ref{prop:property_witness_ensemble}]
From the definition Eq.~\eqref{eq:def_R_alpha_mu_app}, we have $\norm{f_{\widetilde{W}_\Omega,\rho}}_{L^\alpha(\mu)}\geq0$ for all states and all $\alpha<\infty$, while $\norm{f_{\widetilde{W}_\Omega,\sigma}}_{L^\alpha(\mu)}\leq 1$ for every free state $\sigma\in\mathcal{F}$. So $0\leq C_{\widetilde{W}_\Omega,\mu,\alpha}\leq 1$ and by definition of $C_{\widetilde{W}_\Omega,\mu,\alpha}$ in Eq.~\eqref{def_C_W_Omega_app}, the validity of the criterion is proven.
For $1\leq \alpha<\beta<\infty$, the monotonicity of $L^\alpha(\mu)$ norms on a probability space yields
\begin{equation}
\norm{f_{\widetilde{W}_\Omega,\rho}}_{L^\alpha(\mu)}\leq \norm{f_{\widetilde{W}_\Omega,\rho}}_{L^\beta(\mu)}\leq \norm{f_{\widetilde{W}_\Omega,\rho}}_{L^\infty(\mu)}.
\end{equation}
Hence $\norm{f_{\widetilde{W}_\Omega,\rho}}_{L^\alpha(\mu)}$ is 
nondecreasing and bounded above by $\norm{f_{\widetilde{W}_\Omega,\rho}}_{L^\infty(\mu)}$. Let $s=\norm{f_{\widetilde{W}_\Omega,\rho}}_{L^\infty(\mu)}$ and fix $\varepsilon>0$. By the definition of essential supremum,
\begin{equation}
A_{\varepsilon}=\left\{m\in \Omega\mid\Tr(\widetilde{W}_m\rho)>s-\varepsilon\right\}
\end{equation}
has positive measure. Therefore
\begin{equation}
\norm{f_{\widetilde{W}_\Omega,\rho}}_{L^\alpha(\mu)}
\geq
\left(\int_{A_{\varepsilon}} \Tr(\widetilde{W}_m\rho)^\alpha\mathrm{d}\mu(m)\right)^{1/\alpha}
\geq
\mu(A_{\varepsilon})^{1/\alpha}(s-\varepsilon).
\end{equation}
Letting $\alpha\to\infty$ yields
\begin{equation}
\liminf_{\alpha\to\infty}\norm{f_{\widetilde{W}_\Omega,\rho}}_{L^\alpha(\mu)}\geq s-\varepsilon.
\end{equation}
Since $\varepsilon>0$ is arbitrary, Eq.~\eqref{eq:Ralpha-to-Rinfty-general} follows.
Finally, since $\norm{f_{\widetilde{W}_\Omega,\sigma}}_{L^\alpha(\mu)}$ is nondecreasing in $\alpha$ for each fixed $\sigma$, the sequence $C_{\widetilde{W}_\Omega,\mu,\alpha}$ is nondecreasing and bounded above by $C_{\widetilde{W}_\Omega,\mu,\infty}$. Moreover,
\begin{equation}
\sup_{\alpha<\infty}C_{\widetilde{W}_\Omega,\mu,\alpha}
=
\sup_{\alpha<\infty}\sup_{\sigma\in\mathcal{F}}\norm{f_{\widetilde{W}_\Omega,\sigma}}_{L^\alpha(\mu)}
=
\sup_{\sigma\in\mathcal{F}}\sup_{\alpha<\infty}\norm{f_{\widetilde{W}_\Omega,\sigma}}_{L^\alpha(\mu)}
=
\sup_{\sigma\in\mathcal{F}}\norm{f_{\widetilde{W}_\Omega,\sigma}}_{L^\infty(\mu)}
=
C_{\widetilde{W}_\Omega,\mu,\infty},
\end{equation}
which proves Eq.~\eqref{eq:Calpha-to-Cinfty-general}.
\end{proof}

Although the detection ability of the $\norm{f_{\widetilde{W}_\Omega,\rho}}_{L^\alpha(\mu)}>C_{\widetilde{W}_\Omega,\mu,\alpha}$ criterion is not monotonic with $\alpha$ in general, we relate the states that are detectable by each single witness $W_m$ (denoted as $D_m=\{\rho\mid\Tr(W_m\rho)<0\}$) to those detectable in the infinite-$\alpha$ case. For a general probabilistic witness ensemble, the literal union of the single-witness detectable sets is not suitable to compare with the $L^\infty(\mu)$ criterion, since the latter is governed by the essential supremum and is therefore insensitive to $\mu$-null subsets of witnesses. We instead consider
\begin{equation}
D_{\widetilde{W}_\Omega,\mu,\mathrm{ess}}
=
\left\{
\rho\in\mathcal{D}(\mathcal{H})\mid
\mu\left(\{m\in \Omega\mid\rho\in D_m\}\right)>0
\right\},
\end{equation}
that is, the set of states detected by a positive-measure subset of the ensemble. The following proposition identifies its relation to $D_{\widetilde{W}_\Omega,\mu,\infty}$.
\begin{proposition}[Essential single-witness detection]\label{prop:single_witness_essential}
For a general probabilistic witness ensemble one has
\begin{equation}
D_{\widetilde{W}_\Omega,\mu,\mathrm{ess}}\subseteq D_{\widetilde{W}_\Omega,\mu,\infty}.
\label{eq:Dess-in-Dinfty-general}
\end{equation}
If, in addition, $C_{\widetilde{W}_\Omega,\mu,\infty}=1$, then
\begin{equation}
D_{\widetilde{W}_\Omega,\mu,\mathrm{ess}}=D_{\widetilde{W}_\Omega,\mu,\infty}.
\label{eq:Dess-equals-Dinfty-general}
\end{equation}
\end{proposition}

\begin{proof}[Proof of Proposition~\ref{prop:single_witness_essential}]
Let $\rho\in D_{\widetilde{W}_\Omega,\mu,\mathrm{ess}}$. Then
\begin{equation}
\mu\left(\left\{m\mid\Tr(\widetilde{W}_m\rho)>1\right\}\right)>0.
\end{equation}
Since
\begin{equation}
\left\{m\mid\Tr(\widetilde{W}_m\rho)>1\right\}
=
\bigcup_{n=1}^{\infty}\left\{m\mid\Tr(\widetilde{W}_m\rho)>1+1/n\right\},
\end{equation}
there exists $n\in\mathbb{N}_+$ such that
\begin{equation}
\mu\left(\left\{m\mid\Tr(\widetilde{W}_m\rho)>1+1/n\right\}\right)>0.
\end{equation}
Hence
\begin{equation}
\norm{f_{\widetilde{W}_\Omega,\rho}}_{L^\infty(\mu)}=\esssup_m \Tr(\widetilde{W}_m\rho)\geq 1+1/n>1\geq C_{\widetilde{W}_\Omega,\mu,\infty},
\end{equation}
so we have $\rho\in D_{\widetilde{W}_\Omega,\mu,\infty}$.

Assume now that $C_{\widetilde{W}_\Omega,\mu,\infty}=1$ and let $\rho\in D_{\widetilde{W}_\Omega,\mu,\infty}$. Then $\norm{f_{\widetilde{W}_\Omega,\rho}}_{L^\infty(\mu)}>1$. For any $a$ such that $1<a<\norm{f_{\widetilde{W}_\Omega,\rho}}_{L^\infty(\mu)}$, if $\mu\left(\left\{m\mid\Tr(\widetilde{W}_m\rho)>a\right\}\right)=0$, then $a$ would be an essential upper bound for $\Tr(\widetilde{W}_m\rho)$, contradicting the definition of $\norm{f_{\widetilde{W}_\Omega,\rho}}_{L^\infty(\mu)}$. Thus
\begin{equation}
0<\mu\left(\left\{m\mid\Tr(\widetilde{W}_m\rho)>a\right\}\right)\leq\mu\left(\left\{m\mid\Tr(\widetilde{W}_m\rho)>1\right\}\right),
\end{equation}
which implies $\rho\in D_{\widetilde{W}_\Omega,\mu,\mathrm{ess}}$.
\end{proof}

Furthermore, if a resource state $\rho$ is detectable by the infinite-$\alpha$ criterion, it must also be detected by the criterion of some finite $\alpha$ and all sufficiently large $\alpha$.
\begin{corollary}[Eventual detection of every state in $D_{\widetilde{W}_\Omega,\mu,\infty}$]\label{cor:eventual_detection}
If $\rho\in D_{\widetilde{W}_\Omega,\mu,\infty}$, then there exists $\alpha_0(\rho)<\infty$ such that
\begin{equation}
\rho\in D_{\widetilde{W}_\Omega,\mu,\alpha},\qquad \forall\alpha\ge \alpha_0(\rho).
\label{eq:eventual-detection-general}
\end{equation}
\end{corollary}

\begin{proof}[Proof of Corollary~\ref{cor:eventual_detection}]
Choose $\delta>0$ such that $3\delta<\norm{f_{\widetilde{W}_\Omega,\rho}}_{L^\infty(\mu)}-C_{\widetilde{W}_\Omega,\mu,\infty}$.
By Proposition~\ref{prop:property_witness_ensemble}, for all sufficiently large $\alpha$, we have
$\norm{f_{\widetilde{W}_\Omega,\rho}}_{L^\alpha(\mu)}>\norm{f_{\widetilde{W}_\Omega,\rho}}_{L^\infty(\mu)}-\delta$ and
$C_{\widetilde{W}_\Omega,\mu,\alpha}<C_{\widetilde{W}_\Omega,\mu,\infty}+\delta$,
hence $\norm{f_{\widetilde{W}_\Omega,\rho}}_{L^\alpha(\mu)}>C_{\widetilde{W}_\Omega,\mu,\alpha}$.
\end{proof}

As a special case, the WE criterion uses the unitary twirling over a compact subgroup $\mathcal{G}$ of all free unitaries. We show that the $\alpha\to\infty$ criterion covers the ordinary union of single-witness detection regions, not merely the essential union. We first prove two useful lemmas.

\begin{lemma}[Full support of Haar measure]\label{lem:full_support}
Let $\mathcal{G}$ be a compact subgroup of the free unitaries $\Ufree$. Every nonempty open subset of $\G$ has strictly positive $\mu$-measures.
\end{lemma}

\begin{proof}[Proof of Lemma~\ref{lem:full_support}]
Let $\mathcal{O}\subseteq\G$ be nonempty and open, and choose $U_0\in\mathcal{O}$. For every $V\in\G$, the left translate $VU_0^{-1}\mathcal{O}$ is an open neighborhood of $V$. Thus, the family of translates of $\mathcal{O}$ covers $\G$. Since $\G$ is compact, finitely many such translates suffice:
\begin{equation}
\G=\bigcup_{j=1}^N V_jU_0^{-1}\mathcal{O}.
\end{equation}
If $\mu(\mathcal{O})=0$, then left-invariance of the Haar measure gives $\mu(V_jU_0^{-1}\mathcal{O})=0$ for all $j$, and renders a contradiction
\begin{equation}
1=\mu(\G)\leq \sum_{j=1}^N\mu(V_jU_0^{-1}\mathcal{O})=0.
\end{equation}
So we must have $\mu(\mathcal{O})>0$.
\end{proof}

\begin{lemma}[Essential supremum equals supremum]\label{lem:esssup_equals_sup}
Let $\mathcal{G}$ be a compact subgroup of the free unitaries $\Ufree$. For every fixed state $\rho$, the map $U\mapsto \Tr(U^\dagger\widetilde{W}U\rho)$ is continuous on $\G$, and therefore
\begin{equation}
\norm{f_{\widetilde{W}_\Omega,\rho}}_{L^\infty(\mu)}=\esssup_{U\in\G}\Tr(U^\dagger\widetilde{W}U\rho)=\sup_{U\in\G}\Tr(U^\dagger\widetilde{W}U\rho).
\label{eq:esssup-equals-sup-twirling}
\end{equation}
\end{lemma}

\begin{proof}[Proof of Lemma~\ref{lem:esssup_equals_sup}]
Continuity of $U\mapsto\Tr(U^\dagger\widetilde{W}U\rho)$ follows from continuity of multiplication and adjoint in a finite-dimensional system. Assume this function attains its maximum at $U_\rho$. Fix $\varepsilon>0$ and define
\begin{equation}
\mathcal{O}_{\varepsilon}=\left\{U\in\G\middle\mid\Tr(U^\dagger\widetilde{W}U\rho)>\Tr(U_\rho^\dagger\widetilde{W}U_\rho\rho)-\varepsilon\right\},
\end{equation}
it is a nonempty open neighborhood of $U_\rho$, hence $\mu(\mathcal{O}_{\varepsilon})>0$ by Lemma~\ref{lem:full_support}. Therefore
\begin{equation}
\esssup_{U\in\G}\Tr(U^\dagger\widetilde{W}U\rho)\geq \Tr(U_\rho^\dagger\widetilde{W}U_\rho\rho)-\varepsilon.
\end{equation}
Since $\varepsilon>0$ is arbitrary and $\esssup\leq\sup$ always holds, Eq.~\eqref{eq:esssup-equals-sup-twirling} follows.
\end{proof}

Now we prove that the detectable states by the WE infinite-$\alpha$ criterion actually cover all the states that could be detected by any witness constructed by $W$ under free unitary rotation.
\begin{proposition}[Recovery of the ordinary single-witness union under twirling]\label{prop:inf_single_witness_unitary}
For each $U\in\G$, denote the set of states that could be detected by witness $U^\dagger WU$ as $D_U$, then
\begin{equation}
\bigcup_{U\in\G}D_U\subseteq D_{\widetilde{W}_\Omega,\mu,\infty}.
\label{eq:union-in-Dinfty-twirling}
\end{equation}
If, in addition, $C_{\widetilde{W}_\Omega,\mu,\infty}=1$, then
\begin{equation}
\bigcup_{U\in\G}D_U=D_{\widetilde{W}_\Omega,\mu,\infty}.
\label{eq:equality-in-Dinfty-twirling}
\end{equation}
\end{proposition}

\begin{proof}[Proof of Proposition~\ref{prop:inf_single_witness_unitary}]
Let $\rho\in\bigcup_{U\in\G}D_U$, it is detected by some witness $U_0^\dagger WU_0$, by Lemma~\ref{lem:esssup_equals_sup},
\begin{equation}
\norm{f_{\widetilde{W}_\Omega,\rho}}_{L^\infty(\mu)}=\sup_{U\in\G}\Tr(U^\dagger\widetilde{W}U\rho)\geq \Tr(U_0^\dagger\widetilde{W}U_0\rho)>1.
\end{equation}
On the other hand, $\norm{f_{\widetilde{W}_\Omega,\sigma}}_{L^\infty(\mu)}\leq 1$ for every $\sigma\in\Fset$, hence $C_{\widetilde{W}_\Omega,\mu,\infty}\leq 1$ and we have $\rho\in D_{\widetilde{W}_\Omega,\mu,\infty}$.

If $C_{\widetilde{W}_\Omega,\mu,\infty}=1$ and $\rho\in D_{\widetilde{W}_\Omega,\mu,\infty}$, then $\norm{f_{\widetilde{W}_\Omega,\rho}}_{L^\infty(\mu)}>1$. By Lemma~\ref{lem:esssup_equals_sup} there exists $U_0\in\G$ such that $\Tr(U_0^\dagger\widetilde{W}U_0\rho)=\norm{f_{\widetilde{W}_\Omega,\rho}}_{L^\infty(\mu)}>1$, which means $\rho\in D_{U_0}\subseteq\bigcup_{U\in\G}D_U$.
\end{proof}

\subsection{Proof of Proposition~\ref{prop:lower_bound_free_robustness}}
\label{app:free_robustness_WE}

In this appendix, we prove that the amount by which the WE $\alpha$-criterion is violated gives a lower bound on the free robustness of the state.
Recall that the free robustness of $\rho$ with respect to the free set $\mathcal F$ is defined as
\begin{equation}
\mathcal{R}(\rho)
=
\inf\left\{
t\ge0\middle\mid
\exists \tau\in\mathcal F
\text{ s.t. }
\frac{\rho+t\tau}{1+t}\in\mathcal F
\right\}.
\end{equation}
This quantity measures the minimum amount of any free state that must be mixed with $\rho$ so that the resulting state becomes free. We then prove Proposition~\ref{prop:lower_bound_free_robustness} in the main text, which provides a lower bound of $\mathcal{R}(\rho)$ using the WE $\alpha$-norm $\norm{\rho}_{\widetilde{W},\mathcal{G},\alpha}$ and the threshold value $C_{\widetilde{W},\mathcal{G},\alpha}$ for free states.

\begin{proof}[Proof of Proposition~\ref{prop:lower_bound_free_robustness}]
Let $t\ge0$ be any feasible value in the definition of $\mathcal{R}(\rho)$.
Then there exists a free state $\tau\in\mathcal F$ such that
\begin{equation}
\sigma
=
\frac{\rho+t\tau}{1+t}
\in\mathcal F.
\end{equation}
Equivalently,
\begin{equation}
\rho
=
(1+t)\sigma-t\tau.
\end{equation}
For any positive semi-definite operator $\widetilde{W}$ and any compact subgroup $\mathcal{G}$ of free unitaries, we define the corresponding expectation function for state $\rho$ w.r.t. the family of positive semi-definite operators $\mathcal{G}^\dagger\widetilde{W}\mathcal{G}=\left\{U^\dagger\widetilde{W}U\middle\mid U\in\mathcal{G}\right\}$ as (cf. Eq.~\eqref{eq:witness_expectation_function} in the main text)
\begin{equation}
f_{\mathcal{G}^\dagger\widetilde{W}\mathcal{G},\rho}(U)
=
\Tr(U^\dagger \widetilde W U\rho).
\end{equation}
By linearity in the state, we have
\begin{equation}
f_{\mathcal{G}^\dagger\widetilde{W}\mathcal{G},\rho}
=
(1+t)f_{\mathcal{G}^\dagger\widetilde{W}\mathcal{G},\sigma}-tf_{\mathcal{G}^\dagger\widetilde{W}\mathcal{G},\tau}.
\end{equation}
For any $\alpha\in\mathbb{N}_+$, the WE $\alpha$-norm of $\rho$ is defined as the $L^\alpha(\mu_{\mathcal G})$-norm of $f_{\mathcal{G}^\dagger\widetilde{W}\mathcal{G},\rho}$ (cf. Eq.~\eqref{eq:def_WE_alpha_norm} in the main text), Minkowski's inequality gives
\begin{equation}
\norm{\rho}_{\widetilde W,\mathcal G,\alpha}
=
\norm{f_{\mathcal{G}^\dagger\widetilde{W}\mathcal{G},\rho}}_{L^\alpha(\mu_{\mathcal G})}
\leq
(1+t)\norm{f_{\mathcal{G}^\dagger\widetilde{W}\mathcal{G},\sigma}}_{L^\alpha(\mu_{\mathcal G})}
+
t\norm{f_{\mathcal{G}^\dagger\widetilde{W}\mathcal{G},\tau}}_{L^\alpha(\mu_{\mathcal G})}.
\end{equation}
We have,
\begin{equation}
\norm{\rho}_{\widetilde W,\mathcal G,\alpha}
\leq
(1+t)\norm{\sigma}_{\widetilde W,\mathcal G,\alpha}
+
t\norm{\tau}_{\widetilde W,\mathcal G,\alpha}\leq
(1+2t)C_{\widetilde W,\mathcal G,\alpha},
\end{equation}
since both $\sigma$ and $\tau$ are free states, their WE $\alpha$-norms are bounded by the free state threshold.
This simplifies to
\begin{equation}
t
\ge
\frac{\norm{\rho}_{\widetilde{W},\mathcal{G},\alpha}
-
C_{\widetilde{W},\mathcal{G},\alpha}}{2C_{\widetilde W,\mathcal G,\alpha}},
\end{equation}
for every feasible $t$.
Taking the infimum over all feasible $t$ in the definition of $\mathcal{R}(\rho)$ gives
\begin{equation}
\mathcal{R}(\rho)
\ge
\frac{\norm{\rho}_{\widetilde{W},\mathcal{G},\alpha}
-
C_{\widetilde{W},\mathcal{G},\alpha}}{2C_{\widetilde W,\mathcal G,\alpha}}.
\end{equation}
This finishes the proof.
\end{proof}

\section{Coherence}\label{app:coherence}
Fix the reference basis $\{\ket{j}\}_{j=0}^{d-1}$ of a $d$-dimensional Hilbert space $\mathcal H$ with $d\geq2$. The set of free states is
\begin{equation}
\Fset=\mathcal{I}=\left\{\sigma\in\mathcal{D}(\mathcal H)\middle|\bra{j}\sigma\ket{k}=0\ \textup{for all }j\neq k\right\},
\end{equation}
and we choose free unitaries as the diagonal unitaries
\begin{equation}
\mathcal{G}=\mathcal{U}_{\textup{diag}}(d)=\left\{\sum_{k=0}^{d-1} e^{i\theta_k}\ket{k}\bra{k}\middle| \theta_k\in[0,2\pi)\right\}.
\end{equation}
We parametrize a generic free unitary by
\begin{equation}
U_{\boldsymbol\theta}=\sum_{k=0}^{d-1} e^{i\theta_k}\ket{k}\bra{k},
\qquad
\boldsymbol\theta=(\theta_0,\cdots,\theta_{d-1})\in[0,2\pi)^d,
\end{equation}
and we write $\mu_{\mathcal{U}_{\textup{diag}}}$ for the normalized Haar measure on $\mathcal{U}_{\textup{diag}}$, namely
\begin{equation}
\E_{U\sim\mu_{\mathcal{U}_{\textup{diag}}}}[f(U)]
=\int_{[0,2\pi)^d} f(U_{\boldsymbol\theta})\frac{\mathrm{d}\boldsymbol\theta}{(2\pi)^d}.
\end{equation}

We choose the witness and its standardized version as
\begin{equation}\label{eq:witness-pair}
W=\Xi=\frac1{d-1}\sum_{j\neq k}\ket{k}\bra{j},
\qquad
\widetilde W=\widetilde{\Xi}=\bI-\Xi.
\end{equation}
For every $\sigma\in\Fset$ we have $\Tr(\Xi\sigma)=0$, while for the coherent state
\begin{equation}
\rho_-=\frac{(\ket{0}-\ket{1})(\bra{0}-\bra{1})}{2}
\end{equation}
we get
\begin{equation}
\Tr(\Xi\rho_-)= -\frac{1}{d-1}<0,
\end{equation}
so $\Xi$ is a valid coherence witness.

Denote the random off-diagonal witness expectation as
\begin{equation}\label{eq:omega-def}
\xi(U_{\boldsymbol\theta},\rho)
=(d-1)\Tr(U_{\boldsymbol\theta}^\dagger \Xi U_{\boldsymbol\theta}\rho)
=\sum_{\substack{j,k=0\\ j\neq k}}^{d-1}\bra{j}\rho\ket{k}e^{i(\theta_j-\theta_k)},
\end{equation}
and the corresponding standardized witness expectation is
\begin{equation}\label{eq:Xi-def}
\Tr(U_{\boldsymbol\theta}^\dagger\widetilde{\Xi}U_{\boldsymbol\theta}\rho)
=1-\frac{\xi(U_{\boldsymbol\theta},\rho)}{d-1}.
\end{equation}
In this way, we have
\begin{equation}\label{eq:cR-def_coherence}
\norm{\rho}_{\widetilde{\Xi},\mathcal{U}_{\textup{diag}}(d),\alpha}^\alpha
=\E_{U\sim\mu_{\mathcal{U}_{\textup{diag}}}}\left[\left(1-\frac{\xi(U_{\boldsymbol\theta},\rho)}{d-1}\right)^\alpha\right].
\end{equation}

We now derive the closed form of $\norm{\rho}_{\widetilde{\Xi},\mathcal{U}_{\textup{diag}}(d),\alpha}$ in terms of $\rho$ for an arbitrary $\alpha\in\mathbb{N}_+$. We first introduce a useful definition.

\begin{definition}[Eulerian adjacency matrices]
For $\ell\in\mathbb{N}$, let $\mathcal E_\ell$ be the set of all matrices $M=(M_{j,k})_{j,k=0}^{d-1}$ in $\mathcal M_d(\mathbb{N})$ such that
\begin{equation}\label{eq:En-def}
M_{j,j}=0, \forall j,
\qquad
\sum_{k=0}^{d-1} M_{j,k}=\sum_{k=0}^{d-1} M_{k,j}, \forall j,
\qquad
\sum_{j,k=0}^{d-1} M_{j,k}=\ell.
\end{equation}
Equivalently, $\mathcal E_\ell$ is the set of adjacency matrices of directed Eulerian multigraphs with $\ell$ edges and no self-loops.
\end{definition}
The expectation value of $\ell$-th moment of $\xi(U,\rho)$ over the Haar measure on $\mathcal{U}_{\textup{diag}}(d)$ is given as follows.
\begin{proposition}[Eulerian moment formula]\label{prop:eulerian}
For any state $\rho$, let $\xi(U,\rho)$ be defined as in Eq.~\eqref{eq:omega-def}, then for integer $\ell\in\mathbb{N}$,
\begin{equation}\label{eq:eulerian-formula}
\E_{U\sim\mu_{\mathcal{U}_{\textup{diag}}}}\left[\xi(U,\rho)^\ell\right]
=\ell!\sum_{M\in\mathcal E_\ell}
\prod_{\substack{j,k=0\\ j\neq k}}^{d-1}
\frac{\bra{j}\rho\ket{k}^{M_{j,k}}}{M_{j,k}!}.
\end{equation}
\end{proposition}

\begin{proof}[Proof of Proposition~\ref{prop:eulerian}]
We use the multinomial theorem and introduce running index $M_{j,k}$ to denote how many times $\bra{j}\rho\ket{k}e^{i(\theta_j-\theta_k)}$ appears in each term, 
\begin{equation}\label{eq:omega_l}
\xi(U_{\boldsymbol\theta},\rho)^\ell
=\sum_{\sum_{j\neq k} M_{j,k}=\ell}
\frac{\ell!}{\prod_{j\neq k}M_{j,k}!}
\prod_{j\neq k}\left(\bra{j}\rho\ket{k}e^{i(\theta_j-\theta_k)}\right)^{M_{j,k}}.
\end{equation}
Taking the Haar average over the random diagonal unitary group $\mathcal{U}_{\textup{diag}}(d)$ gives
\begin{equation}
\E_{U\sim\mu_{\mathcal{U}_{\textup{diag}}}}[\xi^\ell]
=\sum_{\sum_{j\neq k} M_{j,k}=\ell}
\frac{\ell!\prod_{j\neq k}\bra{j}\rho\ket{k}^{M_{j,k}}}{\prod_{j\neq k}M_{j,k}!}
\E_{\boldsymbol\theta}\left[
\exp\left(i\sum_{j\neq k} M_{j,k}(\theta_j-\theta_k)\right)
\right].
\end{equation}
The exponent can be regrouped as
\begin{equation}
\sum_{j\neq k} M_{j,k}(\theta_j-\theta_k)
=\sum_{r=0}^{d-1} \theta_r\left(\sum_{k=0}^{d-1}M_{r,k}-\sum_{j=0}^{d-1}M_{j,r}\right).
\end{equation}
Because the phases are independent and uniformly distributed, the average is nonzero if and only if each coefficient of $\theta_r$ vanishes, namely
\begin{equation}
\sum_{j=0}^{d-1}M_{j,r}=\sum_{k=0}^{d-1}M_{r,k},
\qquad\forall r\in\{0,\cdots,d-1\}.
\end{equation}
Together with $M_{j,j}=0$ and $\sum_{j,k}M_{j,k}=\ell$, this is exactly the defining condition for $M\in\mathcal E_\ell$. Hence, the sum collapses to Eq.~\eqref{eq:eulerian-formula}.
\end{proof}
Together with Proposition~\ref{prop:eulerian}, we obtain a closed form evaluation for $\norm{\rho}_{\widetilde{\Xi},\mathcal{U}_{\textup{diag}}(d),\alpha}$ for any state $\rho$ and positive integer $\alpha$.
We take the binomial expansion in Eq.~\eqref{eq:Xi-def} inside Eq.~\eqref{eq:cR-def_coherence}:
\begin{equation}\label{eq:cR-expansion}
\left(1-\frac{\xi(U,\rho)}{d-1}\right)^\alpha
=\sum_{\ell=0}^\alpha \binom{\alpha}{\ell}\left(-\frac{1}{d-1}\right)^\ell\xi(U,\rho)^\ell.
\end{equation}
Averaging over $U\sim\mu_{\mathcal{U}_{\textup{diag}}}$ and inserting Eq.~\eqref{eq:omega_l} yields 
\begin{equation}\label{eq:R_evaluation_coherence_app}
\norm{\rho}_{\widetilde{\Xi},\mathcal{U}_{\textup{diag}}(d),\alpha}^\alpha=\sum_{\ell=0}^\alpha\binom{\alpha}\ell\left(-\frac1{d-1}\right)^\ell\ell!\sum_{M\in\mathcal{E}_\ell}\prod_{\substack{j,k=0\\j\neq k}}^{d-1}\frac{\bra{j}\rho\ket{k}^{M_{j,k}}}{M_{j,k}!},
\end{equation}
which is Eq.~\eqref{eq:R_evaluation_coherence} in the main text.

We show that the threshold value of $\norm{\sigma}_{\widetilde{\Xi},\mathcal{U}_{\textup{diag}}(d),\alpha}$ achievable by any incoherent state $\sigma$ is simply 1, which completes the statement of WE $\alpha$-criterion.

\begin{theorem}[Exact incoherent threshold]\label{thm:free-threshold_coherence}
For any $\alpha\geq 1$ and every free state $\sigma\in\mathcal{I}$, we have
\begin{equation}
\norm{\sigma}_{\widetilde{\Xi},\mathcal{U}_{\textup{diag}}(d),\alpha}=1.
\end{equation}
Consequently, the exact threshold for incoherent states is
\begin{equation}
C_{\widetilde{\Xi},\mathcal{U}_{\textup{diag}}(d),\alpha}=1,\qquad\forall\alpha\geq1.
\end{equation}    
\end{theorem}

\begin{proof}[Proof of Theorem~\ref{thm:free-threshold_coherence}]
If $\sigma\in\mathcal{I}$, Eq.~\eqref{eq:omega-def} gives $\xi(U,\sigma)=0,
\forall U\in\mathcal{U}_{\textup{diag}}(d)$, hence 
\begin{equation}
\norm{\sigma}_{\widetilde{\Xi},\mathcal{U}_{\textup{diag}}(d),\alpha}=\left(\E_{U\sim\mu_{\mathcal{U}_{\textup{diag}}}}[1^\alpha]\right)^{1/\alpha}=1.
\end{equation}
\end{proof}

\begin{proof}[Proof of Theorem~\ref{thm:WE_criterion_coherence}]
Combining the expression for WE $\alpha$-norm in Eq.~\eqref{eq:R_evaluation_coherence_app} and Theorem~\ref{thm:free-threshold_coherence}, we finish the proof of Theorem~\ref{thm:WE_criterion_coherence}.
\end{proof}

Explicitly, the WE $1$-criterion is trivial and the explicit WE $2$-criterion is provided as below.
\begin{corollary}[WE $2$-criterion for coherence]\label{cor:WE_2_criterion_coherence}
For any state $\rho$, we have
\begin{equation}\label{eq:omega2}
\E_{U\sim\mu_{\mathcal{U}_{\textup{diag}}}}\left[\xi(U,\rho)^2\right]
=\sum_{\substack{j,k=0\\ j\neq k}}^{d-1}\abs{\bra{j}\rho\ket{k}}^2.
\end{equation}
Consequently, the WE 2-criterion for coherence is given by
\begin{equation}\label{eq:cR2}
\norm{\rho}_{\widetilde{\Xi},\mathcal{U}_{\textup{diag}}(d),2}^2
=1+\frac{1}{(d-1)^2}\sum_{\substack{j,k=0\\ j\neq k}}^{d-1}\abs{\bra{j}\rho\ket{k}}^2,
\qquad
C_{\widetilde{\Xi},\mathcal{U}_{\textup{diag}}(d),2}=1.
\end{equation}
\end{corollary}

\begin{proof}[Proof of Corollary~\ref{cor:WE_2_criterion_coherence}]
For $\ell=2$, the only matrices in $\mathcal E_2$ are the directed 2-cycles $j\leftrightarrow k$ with $j\neq k$. Therefore Eq.~\eqref{eq:eulerian-formula} becomes
\begin{equation}
\E_{U\sim\mu_{\mathcal{U}_{\textup{diag}}}}[\xi(U,\rho)^2]
=\sum_{\substack{j,k=0\\ j\neq k}}^{d-1}\bra{j}\rho\ket{k}\bra{k}\rho\ket{j}
=\sum_{\substack{j,k=0\\ j\neq k}}^{d-1}\abs{\bra{j}\rho\ket{k}}^2,
\end{equation}
using Hermiticity of $\rho$. Substituting this into Eq.~\eqref{eq:cR-expansion} and Eq.~\eqref{eq:cR-def_coherence} with $\alpha=2$ gives Eq.~\eqref{eq:cR2}.
\end{proof}
For the resource theory of coherence, the $\alpha=2$ criterion already characterizes all coherent states, showing the power of the WE framework.

\section{Entanglement}\label{app:entanglement}

\subsection{Proof of Theorem~\ref{thm:WE_criterion_entanglement} and Corollary~\ref{cor:WE_2_3_norm_entanglement}}\label{app:proof_WE_criterion_entanglement}
Let $\mathcal H_A\cong\mathcal H_B\cong\mathbb C^d$, and let $\rho$ be a density operator on $\mathcal H_A\otimes\mathcal H_B$. The free states are those separable under the bipartition $A$ and $B$
\begin{equation}
\Fset=\Sep=\left\{\sum_kp_k\sigma_A^k\otimes\sigma_B^k\middle\mid \sum_kp_k=1,p_k\geq0,\sigma_\bullet^k\in\mathcal{D}(\mathcal{H}_\bullet),\forall k\right\},   
\end{equation}
and the group of free unitaries are chosen as the tensor products of local unitaries on each subsystem
\begin{equation}
\mathcal{G}=\mathcal{U}_{\textup{loc}}(d,d)=\left\{U_A\otimes U_B\middle\mid U_\bullet\in\mathcal{U}(\mathcal{H}_\bullet)\right\}.
\end{equation}
We choose the SWAP operator $W=\mathbb{S}=\sum_{j,k=0}^{d-1}\ket{j,k}\bra{k,j}$ between the corresponding basis of subsystems $A$ and $B$ as the entanglement witness and set the standardized witness $\widetilde W=\widetilde{\mathbb{S}}=\bI-\mathbb{S}$, since $\mathbb{S}$ has eigenvalues $\pm1$. We denote
\begin{equation}\label{eq:y_rho_U}
y_\rho(U_A,U_B)=\Tr\left[\mathbb{S}(U_A\otimes U_B)\rho(U_A\otimes U_B)^\dagger\right],
\end{equation}
then we have
\begin{equation}\label{eq:M_alpha_rho}
\norm{\rho}_{\widetilde{\mathbb{S}},\mathcal{U}_{\textup{loc}}(d,d),\alpha}^\alpha=\E_{U_A,U_B}\left[\left(1-y_\rho(U_A,U_B)\right)^\alpha\right].
\end{equation}

We first notice that the Haar average for $U_A\otimes U_B$ over $\mathcal{U}_{\textup{loc}}(d,d)$ is equal to that of a relative unitary $\bI\otimes U$.
\begin{lemma}[Reduction to a relative unitary]
\label{lem:relativeU}
For every $\alpha>0$,
\begin{equation}
\norm{\rho}_{\widetilde{\mathbb{S}},\mathcal{U}_{\textup{loc}}(d,d),\alpha}^\alpha=\E_{U\sim\mu_{\mathcal{U}_{\textup{loc}}(d,d)}}\left[\Tr\left[(\bI-\mathbb{S})(\bI\otimes U)\rho(\bI\otimes U^\dagger)\right]^\alpha\right],
\end{equation}
where $U$ is Haar distributed on $\mathrm{U}(d)$.
\end{lemma}

\begin{proof}[Proof of Lemma~\ref{lem:relativeU}]
Using $\mathbb{S}(V\otimes V)=(V\otimes V)\mathbb{S}$ for all $V\in\mathrm{U}(d)$ and cyclicity of the trace,
\begin{equation}\label{eq:y_rho_U_relative}
\begin{aligned}
y_\rho(U_A,U_B)
&=\Tr\left[\mathbb{S}(U_A\otimes U_B)\rho(U_A\otimes U_B)^\dagger\right]\\
&=\Tr\left[\mathbb{S}(\bI\otimes U_A^\dagger U_B)\rho(\bI\otimes U_B^\dagger U_A)\right].
\end{aligned}
\end{equation}
Thus $y_\rho(U_A,U_B)$ depends only on the relative unitary $U=U_A^\dagger U_B$. Since $U_A$ and $U_B$ are independent Haar-distributed, $U$ is Haar-distributed as well. Inserting into Eq.~\eqref{eq:M_alpha_rho} finishes the proof.
\end{proof}

From now on, we abbreviate $y_\rho(U)=\Tr\left[\mathbb{S}(\bI\otimes U)\rho(\bI\otimes U^\dagger)\right]$.
When $\alpha\in\mathbb{N}_+$, by the binomial theorem
\begin{equation}\label{eq:M_alpha_binomial_entanglement}
\norm{\rho}_{\widetilde{\mathbb{S}},\mathcal{U}_{\textup{loc}}(d,d),\alpha}^\alpha=\sum_{\ell=0}^\alpha\binom{\alpha}{\ell}(-1)^\ell m_\ell(\rho),
\end{equation}
where $m_\ell(\rho)=\E_{U\sim\mu_{\mathrm{U}(d)}}[y_\rho(U)^\ell]$ and especially $m_0(\rho)=1$.
Write matrix elements of $\rho$ in the computational basis as $\rho_{i\mu,j\nu}$, where Latin indices refer to subsystem $A$ and Greek indices to subsystem $B$. Since
\begin{equation}
\langle i,\mu|\mathbb{S}|j,\nu\rangle=\delta_{i,\nu}\delta_{\mu,j},\qquad\langle i,\mu|(\bI\otimes U)|j,\nu\rangle=\delta_{i,j}U_{\mu\nu},
\end{equation}
a direct expansion gives
\begin{equation}
\label{eq:y-index}
y_\rho(U)=\sum_{i,j,\mu,\nu}\rho_{i\mu,j\nu}U_{j\mu}\overline{U_{i\nu}}.
\end{equation}
Hence, for every $\ell\in\mathbb N_+$,
\begin{equation}
\begin{aligned}
m_\ell(\rho)
&=\E_{U\sim\mu_{\mathrm{U}(d)}}\left[y_\rho(U)^\ell\right]\\
&=
\sum_{\substack{i_1,\cdots,i_\ell\\ j_1,\cdots,j_\ell\\ \mu_1,\cdots,\mu_\ell\\ \nu_1,\cdots,\nu_\ell}}
\left(\prod_{r=1}^\ell \rho_{i_r\mu_r,j_r\nu_r}\right)
\int_{\mathrm{U}(d)} \prod_{r=1}^\ell U_{j_r\mu_r}\overline{U_{i_r\nu_r}}\mathrm{d}U.
\label{eq:mk-before-Wg}
\end{aligned}
\end{equation}
We now use the standard Weingarten formula
\begin{equation}
\begin{aligned}
\int_{\mathrm{U}(d)} \prod_{r=1}^\ell U_{a_r b_r}\overline{U_{a'_r b'_r}}\mathrm{d}U
=
\sum_{\pi,\omega\in S_\ell}
\left(\prod_{r=1}^\ell \delta_{a_r,a'_{\pi(r)}}\right)
\left(\prod_{r=1}^\ell \delta_{b_r,b'_{\omega(r)}}\right)
\Wg_d(\pi^{-1}\omega),
\end{aligned}
\end{equation}
where $\Wg_d$ denotes the unitary Weingarten function~\cite{collins2006IntegrationRespectHaar}. Applying this to Eq.~\eqref{eq:mk-before-Wg}, relabeling $\omega\mapsto \omega^{-1}$, and using that $\Wg_d(\omega)=\Wg_d(\omega^{-1})$ for $\omega\in S_\ell$
and that $\Wg_d$ is a class function on $S_\ell$, we obtain
\begin{equation}
\label{eq:mk-master}
m_\ell(\rho)=\sum_{\pi,\omega\in S_\ell}\Wg_d(\pi\omega)T_{\pi,\omega}(\rho),
\end{equation}
where
\begin{equation}
\label{eq:Tpi-sigma-index}
T_{\pi,\omega}(\rho)
=
\sum_{i_1,\cdots,i_\ell}\sum_{\mu_1,\cdots,\mu_\ell}
\prod_{r=1}^\ell \rho_{i_r\mu_r,i_{\pi(r)}\mu_{\omega(r)}}.
\end{equation}
Equivalently, if $\mathbb{P}_\pi^{(\ell)}$ denotes the permutation operator for $\pi\in S_\ell$ on $\ell$ copies,
\begin{equation}\label{eq:P_permutation_def}
\mathbb{P}_\pi^{(\ell)}\left(v_1\otimes\cdots\otimes v_\ell\right)=v_{\pi^{-1}(1)}\otimes\cdots\otimes v_{\pi^{-1}(\ell)},
\end{equation}
then
\begin{equation}
\label{eq:Tpi-sigma-operator}
T_{\pi,\omega}(\rho)=\Tr\left[(\mathbb{P}_\pi^{(\ell),A}\otimes \mathbb{P}_\omega^{(\ell),B})\rho^{\otimes \ell}\right].
\end{equation}

Inserting Eq.~\eqref{eq:mk-master} and Eq.~\eqref{eq:Tpi-sigma-operator} into Eq.~\eqref{eq:M_alpha_binomial_entanglement}, we obtain
\begin{equation}\label{eq:R_evaluation_entanglement_app}
\norm{\rho}_{\widetilde{\mathbb{S}},\mathcal{U}_{\textup{loc}}(d,d),\alpha}^\alpha=1+\sum_{\ell=1}^\alpha\binom{\alpha}\ell(-1)^\ell\sum_{\pi,\omega\in S_\ell}\Wg_d(\pi\omega)
\Tr\left[(\mathbb{P}_\pi^{(\ell),A}\otimes \mathbb{P}_\omega^{(\ell),B})\rho^{\otimes\ell}\right],
\end{equation}
which is Eq.~\eqref{eq:R_evaluation_entanglement} in the main text.

We then show the explicit calculation for $\ell=1,2,3$ and prove Corollary~\ref{cor:WE_2_3_norm_entanglement}. 
\begin{proof}[Proof of Corollary~\ref{cor:WE_2_3_norm_entanglement}]
Denote
\begin{equation}
\rho_A=\Tr_B(\rho),
\qquad
\rho_B=\Tr_A(\rho),
\end{equation}
and define the abbreviations
\begin{equation}\label{eq:trace_abbr_app}
\begin{aligned}
r_2&=\Tr(\rho^2),
& r_3&=\Tr(\rho^3),\\
r_{2,A}&=\Tr(\rho_A^2),
& r_{2,B}&=\Tr(\rho_B^2),\\
r_{3,A}&=\Tr(\rho_A^3),
& r_{3,B}&=\Tr(\rho_B^3),\\
u_A&=\Tr\left[\rho^2(\rho_A\otimes \bI)\right],&
u_B&=\Tr\left[\rho^2(\bI\otimes \rho_B)\right],\\
\tau_3&=\Tr\left[(\rho^{\mathrm{T}_B})^3\right],&
v&=\Tr\left[\rho(\rho_A\otimes \rho_B)\right].
\end{aligned}
\end{equation}
\begin{itemize}
\item For $\ell=1$, the Weingarten value
\begin{equation}
\Wg_d(e)=\frac{1}{d},
\end{equation}
so $m_1(\rho)=1/d$. 

\item For $\ell=2$, the symmetric group $S_2=\{e,s\}$, with $s=(12)$. The Weingarten values are
\begin{equation}
\Wg_d(e)=\frac{1}{d^2-1},
\qquad
\Wg_d(s)=-\frac{1}{d(d^2-1)}.
\end{equation}
By Eq.~\eqref{eq:mk-master},
\begin{equation}
\label{eq:m2-Wg}
m_2(\rho)=\Wg_d(e)\left(T_{e,e}+T_{s,s}\right)+\Wg_d(s)\left(T_{e,s}+T_{s,e}\right).
\end{equation}
The four contractions are elementary:
\begin{equation}
\begin{aligned}
T_{e,e}&=\Tr(\rho)^2=1,
&T_{s,s}&=\Tr(\rho^2)=r_2,\\
T_{s,e}&=\Tr(\rho_A^2)=r_{2,A},
&T_{e,s}&=\Tr(\rho_B^2)=r_{2,B}.
\end{aligned}
\end{equation}
Substituting into Eq.~\eqref{eq:m2-Wg} gives
\begin{equation}
\label{eq:m2-raw}
m_2(\rho)=
\frac{1+r_2-(r_{2,A}+r_{2,B})/d}{d^2-1}.
\end{equation}
By Eq.~\eqref{eq:M_alpha_binomial_entanglement}, we have
\begin{equation}\label{eq:WE_2_norm_entanglement_app}
\norm{\rho}_{\widetilde{\mathbb{S}},\mathcal{U}_{\textup{loc}}(d,d),2}^2=1-2m_1+m_2=1-\frac{2}{d}+\frac{1+r_2-(r_{2,A}+r_{2,B})/d}{d^2-1}.
\end{equation}

\item For $\ell=3$, write
\begin{equation}
s_1=(12),\qquad s_2=(13),\qquad s_3=(23),\qquad c_+=(123),\qquad c_-=(132).
\end{equation}
The Weingarten values on the three conjugacy classes of $S_3$ are, for $d\geq3$,
\begin{equation}
\begin{aligned}
\Wg_d(e)&=\frac{d^2-2}{d(d^2-1)(d^2-4)},\\
\Wg_d(s_i)&=-\frac{1}{(d^2-1)(d^2-4)},\\
\Wg_d(c_\pm)&=\frac{2}{d(d^2-1)(d^2-4)}.
\end{aligned}
\end{equation}
We list the contractions $T_{\pi,\omega}(\rho)$ needed to evaluate Eq.~\eqref{eq:mk-master} 
(cf.\ Eq~\eqref{eq:trace_abbr_app}):
\begin{equation}
\begin{aligned}
T_{e,e}&=1,& T_{s_i,e}&=r_{2,A},& T_{e,s_i}&=r_{2,B},& T_{s_i,s_i}&=r_2,\\
T_{c_\pm,e}&=r_{3,A},& T_{e,c_\pm}&=r_{3,B},& T_{c_\pm,c_\pm}&=r_3,\\
T_{c_+,c_-}&=T_{c_-,c_+}=\tau_3,\\
T_{s_i,s_j}&=v, (i\neq j),\\
T_{c_\pm,s_i}&=u_A,& T_{s_i,c_\pm}&=u_B.
\end{aligned}
\end{equation}

All identities follow from direct contractions of tensor indices. We provide as example one representative computation for each orbit under relabeling of the three copies.

First,
\begin{equation}
T_{s_1,e}=\Tr\left[\left(\mathbb{P}_{(12)}^{(3),A}\otimes \bI\right)\rho^{\otimes3}\right]
=\Tr(\rho_A^2)\Tr(\rho)=r_{2,A}.
\end{equation}
Similarly, $T_{e,s_1}=r_{2,B}$, $T_{s_1,s_1}=r_2$, $T_{c_+,e}=r_{3,A}$, $T_{e,c_+}=r_{3,B}$, and $T_{c_+,c_+}=r_3$.

Next,
\begin{equation}
\begin{aligned}
T_{c_+,c_-}
&=
\sum_{\substack{i_1,i_2,i_3\\ \mu_1,\mu_2,\mu_3}}
\rho_{i_1\mu_1,i_2\mu_3}
\rho_{i_2\mu_2,i_3\mu_1}
\rho_{i_3\mu_3,i_1\mu_2}\\
&=
\sum_{\substack{i_1,i_2,i_3\\ \mu_1,\mu_2,\mu_3}}
(\rho^{\mathrm{T}_B})_{i_1\mu_3,i_2\mu_1}
(\rho^{\mathrm{T}_B})_{i_2\mu_1,i_3\mu_2}
(\rho^{\mathrm{T}_B})_{i_3\mu_2,i_1\mu_3}\\
&=\Tr\left[(\rho^{\mathrm{T}_B})^3\right]=\tau_3.
\end{aligned}
\end{equation}
By symmetry, $T_{c_-,c_+}=\tau_3$ as well.

For a pair of distinct transpositions, say $(s_1,s_2)$,
\begin{equation}
\begin{aligned}
T_{s_1,s_2}
&=
\sum_{\substack{i_1,i_2,i_3\\ \mu_1,\mu_2,\mu_3}}
\rho_{i_1\mu_1,i_2\mu_3}
\rho_{i_2\mu_2,i_1\mu_2}
\rho_{i_3\mu_3,i_3\mu_1}\\
&=
\sum_{i_1,i_2,\mu_1,\mu_3}
\rho_{i_1\mu_1,i_2\mu_3}
(\rho_A)_{i_2,i_1}
(\rho_B)_{\mu_3,\mu_1}\\
&=\Tr\left[\rho(\rho_A\otimes \rho_B)\right]=v.
\end{aligned}
\end{equation}
All other pairs $(s_i,s_j)$ with $i\neq j$ are equivalent by relabeling tensor copies.
For a transposition and a $3$-cycle, choose $(c_+,s_1)$, to get
\begin{equation}
\begin{aligned}
T_{c_+,s_1}
&=
\sum_{\substack{i_1,i_2,i_3\\ \mu_1,\mu_2,\mu_3}}
\rho_{i_1\mu_1,i_2\mu_2}
\rho_{i_2\mu_2,i_3\mu_1}
\rho_{i_3\mu_3,i_1\mu_3}\\
&=
\sum_{i_1,i_3,\mu_1}
(\rho^2)_{i_1\mu_1,i_3\mu_1}
(\rho_A)_{i_3,i_1}\\
&=\Tr\left[\rho^2(\rho_A\otimes \bI)\right]=u_A.
\end{aligned}
\end{equation}
Likewise,
\begin{equation}
\begin{aligned}
T_{s_1,c_+}
&=
\sum_{\substack{i_1,i_2,i_3\\ \mu_1,\mu_2,\mu_3}}
\rho_{i_1\mu_1,i_2\mu_2}
\rho_{i_2\mu_2,i_1\mu_3}
\rho_{i_3\mu_3,i_3\mu_1}\\
&=
\sum_{i_1,\mu_1,\mu_3}
(\rho^2)_{i_1\mu_1,i_1\mu_3}
(\rho_B)_{\mu_3,\mu_1}\\
&=\Tr\left[\rho^2(\bI\otimes \rho_B)\right]=u_B.
\end{aligned}
\end{equation}
All remaining identities are obtained from these representatives by permuting the three tensor copies.

We now evaluate Eq.~\eqref{eq:mk-master}. Grouping the pairs $(\pi,\omega)\in S_3\times S_3$ 
according to the cycle type of $\pi\omega$, we find
\begin{equation}
\begin{aligned}
\sum_{\pi\omega=e} T_{\pi,\omega}
&=1+3r_2+2\tau_3,\\
\sum_{\pi\omega\textup{ transposition}} T_{\pi,\omega}
&=3r_{2,A}+3r_{2,B}+6u_A+6u_B,\\
\sum_{\pi\omega\textup{ 3-cycle}} T_{\pi,\omega}
&=2r_{3,A}+2r_{3,B}+6v+2r_3.
\end{aligned}
\end{equation}
Therefore
\begin{equation}
\begin{aligned}
\label{eq:m3-raw}
m_3(\rho)
={}&
\frac{d^2-2}{d(d^2-1)(d^2-4)}\left(1+3r_2+2\tau_3\right)\\
&
-\frac{1}{(d^2-1)(d^2-4)}\left(3r_{2,A}+3r_{2,B}+6u_A+6u_B\right)\\
&
+\frac{2}{d(d^2-1)(d^2-4)}\left(2r_{3,A}+2r_{3,B}+6v+2r_3\right).
\end{aligned}
\end{equation}
Finally,
\begin{equation}
\norm{\rho}_{\widetilde{\mathbb{S}},\mathcal{U}_{\textup{loc}}(d,d),3}^3=1-3m_1+3m_2-m_3,
\end{equation}
where $m_1=1/d$ and $m_2$ is given by Eq.~\eqref{eq:m2-raw}. Substituting these values 
yields Eq.~\eqref{eq:R3_entanglement}.
\end{itemize}
\end{proof}

We then calculate the threshold value 
$C_{\widetilde{\mathbb{S}},\mathcal{U}_{\textup{loc}}(d,d),\alpha}$ for all separable states.

\begin{theorem}[Exact separable threshold]\label{thm:exact_separable_threshold_app}
For every $\alpha\geq1$,
\begin{equation}
C_{\widetilde{\mathbb{S}},\mathcal{U}_{\textup{loc}}(d,d),\alpha}=
\left(\frac{d-1}{d+\alpha-1}\right)^{1/\alpha}.
\end{equation}
In particular,
\begin{equation}
C_{\widetilde{\mathbb{S}},\mathcal{U}_{\textup{loc}}(d,d),1}=\frac{d-1}{d},\qquad
C_{\widetilde{\mathbb{S}},\mathcal{U}_{\textup{loc}}(d,d),2}=\sqrt{\frac{d-1}{d+1}},\qquad
C_{\widetilde{\mathbb{S}},\mathcal{U}_{\textup{loc}}(d,d),3}=\sqrt[3]{\frac{d-1}{d+2}}.
\end{equation}
\end{theorem}

\begin{proof}[Proof of Theorem~\ref{thm:exact_separable_threshold_app}]
Notice that $\norm{\rho}_{\widetilde{\mathbb{S}},\mathcal{U}_{\textup{loc}}(d,d),\alpha}^\alpha$ in Eq.~\eqref{eq:M_alpha_rho} is convex in $\rho$, because for every fixed $U_A,U_B$ the map $\rho\mapsto 1-y_\rho(U_A,U_B)$ is linear, while $x\mapsto x^\alpha$ is convex on $[0,\infty)$ and averaging preserves convexity. Therefore, its maximum is attained at an extreme point of $\Sep$, namely a pure product state.
Let $\sigma=|a\rangle\langle a|\otimes |b\rangle\langle b|$. Then (cf.\ Eq.~\eqref{eq:y_rho_U_relative})
\begin{equation}
y_\sigma(U_A,U_B)=\abs{\langle a|U_A^\dagger U_B|b\rangle}^2.
\end{equation}
If $U$ is Haar distributed on $\mathrm{U}(d)$, the random variable
\begin{equation}
\zeta=|\langle a|U|b\rangle|^2
\end{equation}
has the Beta$(1,d-1)$ distribution, with density $(d-1)(1-z)^{d-2}$ for $z\in[0,1]$. Hence
\begin{equation}
\begin{aligned}
\norm{\sigma}_{\widetilde{\mathbb{S}},\mathcal{U}_{\textup{loc}}(d,d),\alpha}^\alpha
&=\E_{U\sim\mu_{\mathrm{U}(d)}}[(1-\zeta)^\alpha]\\
&=(d-1)\int_0^1 (1-z)^{\alpha+d-2}dz
\\&=\frac{d-1}{d+\alpha-1},
\end{aligned}
\end{equation}
and
\begin{equation}
C_{\widetilde{\mathbb{S}},\mathcal{U}_{\textup{loc}}(d,d),\alpha}=\left(\frac{d-1}{d+\alpha-1}\right)^{1/\alpha}.
\end{equation}
\end{proof}

\begin{proof}[Proof of Theorem~\ref{thm:WE_criterion_entanglement}]
Combining the expression for WE $\alpha$-norm in Eq.~\eqref{eq:R_evaluation_entanglement_app} and Theorem~\ref{thm:exact_separable_threshold_app}, we finish the proof of Theorem~\ref{thm:WE_criterion_entanglement}.
\end{proof}

\subsection{Optimal choice for second-order criterion}\label{app:optimal_2_criterion_entanglement}
We prove that when $\alpha=2$, among all $W_2$ restricted to the form of Eq.~\eqref{eq:nonlinear_witness_twirling}, taking $\widetilde{W}=\widetilde{\mathbb{S}}$ yields the best detection capability.

\begin{theorem}[Optimal seed operator of the WE $2$-criterion for entanglement]\label{thm:optimal_swap_entanglement}
Let $O$ be any Hermitian operator on $\mathcal H_A\otimes\mathcal H_B$. If a state $\rho$ can be detected by the WE 2-criterion constructed by $O$, then the WE 2-criterion constructed by $\widetilde{\mathbb{S}}=\bI-\mathbb{S}$ also detects $\rho$. Equivalently,
\begin{equation}
\label{eq:O_detect_implies_swap_detect}
\norm{\rho}_{O,\mathcal U_{\textup{loc}}(d,d),2}
>
C_{O,\mathcal U_{\textup{loc}}(d,d),2}
\quad\Longrightarrow\quad
\norm{\rho}_{\widetilde{\mathbb{S}},\mathcal U_{\textup{loc}}(d,d),2}
>
C_{\widetilde{\mathbb{S}},\mathcal U_{\textup{loc}}(d,d),2}.
\end{equation}
\end{theorem}
\begin{proof}[Proof of Theorem~\ref{thm:optimal_swap_entanglement}]
Consider the second-order local unitary twirling of $O^{\otimes2}$ (see e.g. Corollary~13 in Ref.~\cite{Mele2024introductiontohaar})
\begin{equation}\label{eq:tilde_O_2-twirling}
\E_{U_A,U_B\sim\mu_{\mathrm{U}(d)}}
\left[
{(U_A\otimes U_B)^\dagger}^{\otimes2}O^{\otimes2}(U_A\otimes U_B)^{\otimes2}
\right]=
\alpha_0\bI+\alpha_A\mathbb S_A+\alpha_B\mathbb S_B+\alpha_{A,B}\mathbb S_A\mathbb S_B,
\end{equation}
where $\mathbb{S}_A$ ($\mathbb{S}_B$) denotes the SWAP operator between two copies of subsystem $A$ ($B$).
The coefficients are
\begin{equation}\label{eq:coefficients_2-twirling}
\begin{aligned}
\alpha_0&=
\frac{
(\Tr O)^2
-\Tr[(\Tr_B O)^2]/d
-\Tr[(\Tr_AO)^2]/d
+\Tr(O^2)/d^2
}{(d^2-1)^2},\\
\alpha_A&=
\frac{
\Tr[(\Tr_B O)^2]
-(\Tr O)^2/d
-\Tr(O^2)/d
+\Tr[(\Tr_AO)^2]/d^2
}{(d^2-1)^2},\\
\alpha_B&=
\frac{
\Tr[(\Tr_AO)^2]
-(\Tr O)^2/d
-\Tr(O^2)/d
+\Tr[(\Tr_B O)^2]/d^2
}{(d^2-1)^2},\\
\alpha_{A,B}&=
\frac{
\Tr(O^2)
-\Tr[(\Tr_B O)^2]/d
-\Tr[(\Tr_AO)^2]/d
+(\Tr O)^2/d^2
}{(d^2-1)^2}.
\end{aligned}
\end{equation}
Choose an orthonormal (w.r.t.\ Hilbert--Schmidt inner product) basis $\{P_j\}_{j=0}^{d^2-1}$ with $P_0=\bI/\sqrt d$ of $\mathcal{L}\left(\mathbb{C}^d\right)$ and expand
\begin{equation}
O=\sum_{j,k=0}^{d^2-1} o_{j,k}P_j\otimes P_k.
\end{equation}
Through a direct calculation, the coefficients in Eq.~\eqref{eq:coefficients_2-twirling} satisfy
\begin{equation}\label{eq:tilde_O_coefficients_inequality}
\alpha_{A,B}
=
\frac{\sum_{j,k\ge1}\abs{o_{j,k}}^2}{(d^2-1)^2}\ge0,\quad
\alpha_A+\frac{\alpha_{A,B}}{d}
=
\frac{\sum_{j\ge1}\abs{o_{j,0}}^2}{d(d^2-1)}\ge0,
\quad
\alpha_B+\frac{\alpha_{A,B}}{d}
=
\frac{\sum_{j\ge1}\abs{o_{0,j}}^2}{d(d^2-1)}\ge0.
\end{equation}
Using Eq.~\eqref{eq:tilde_O_2-twirling}, the WE 2-norm is given by (cf. Eq.~\eqref{eq:trace_abbr_app})
\begin{equation}
\label{eq:WE_2_norm_O_entanglement}
\norm{\rho}_{O,\mathcal U_{\textup{loc}}(d,d),2}^2
=
\alpha_0+\alpha_A r_{2,A}+\alpha_B r_{2,B}+\alpha_{A,B}r_2,
\end{equation}
and the corresponding threshold value
\begin{equation}
\label{eq:C_2_O_entanglement}
C_{O,\mathcal U_{\textup{loc}}(d,d),2}^2
=
\alpha_0+\alpha_A+\alpha_B+\alpha_{A,B}.
\end{equation}
We have
\begin{equation}
\label{eq:O_minus_CO_insert_new}
\norm{\rho}_{O,\mathcal U_{\textup{loc}}(d,d),2}^2
-
C_{O,\mathcal U_{\textup{loc}}(d,d),2}^2
=
\alpha_{A,B}(r_2-1)+\alpha_A(r_{2,A}-1)+\alpha_B(r_{2,B}-1).
\end{equation}
Since $r_{2,A},r_{2,B},r_2\le1$, and notice that Eq.~\eqref{eq:tilde_O_coefficients_inequality} implies
\begin{equation}
\alpha_A(r_{2,A}-1)\le -\frac{\alpha_{A,B}}{d}(r_{2,A}-1),
\qquad
\alpha_B(r_{2,B}-1)\le -\frac{\alpha_{A,B}}{d}(r_{2,B}-1).
\end{equation}
If $\rho$ can be detected by the WE 2-criterion constructed by $O$, then
\begin{equation}
\begin{aligned}
0<\norm{\rho}_{O,\mathcal U_{\textup{loc}}(d,d),2}^2
-
C_{O,\mathcal U_{\textup{loc}}(d,d),2}^2&
\le
\alpha_{A,B}(r_2-1)
-\frac{\alpha_{A,B}}{d}(r_{2,A}-1)
-\frac{\alpha_{A,B}}{d}(r_{2,B}-1)\\
&=
\frac{\alpha_{A,B}}{d}\left(d r_2-r_{2,A}-r_{2,B}-d+2\right).
\end{aligned}
\end{equation}
Since $\alpha_{A,B}\ge0$, we have in this case $\alpha_{A,B}>0$, so
\begin{equation}
 d r_2>r_{2,A}+r_{2,B}+d-2,
\end{equation}
which is exactly the condition of WE 2-criterion constructed by $\widetilde{\mathbb{S}}$, see Eq.~\eqref{eq:WE_2_criterion_entanglement}.
\end{proof}

\subsection{Proof of Theorem~\ref{thm:WE_infinity_criterion_entanglement}}\label{app:proof_WE_infinity_criterion_entanglement}
We then take the $\alpha\rightarrow\infty$ limit, the WE criterion reduces to asking whether the partial transposition $\rho^{\mathrm{T}_B}$ has a negative value when projected to the basis of the family of maximally entangled states. We reformulate Theorem~\ref{thm:WE_infinity_criterion_entanglement} as the theorem below.

\begin{theorem}[WE $\infty$-criterion for entanglement]
\label{thm:WE_infinity_criterion_entanglement_app}
Let $|\Phi_d\rangle:=\sum_{j=0}^{d-1}|j,j\rangle/{\sqrt d}$ and define the maximally entangled orbit
as
\begin{equation}
\ME_d:=\left\{(U_A\otimes U_B)|\Phi_d\rangle\middle\mid U_A,U_B\in\mathrm{U}(d)\right\}.
\end{equation}
For every state $\rho$,
\begin{equation}
\lim_{\alpha\to\infty} \norm{\rho}_{\widetilde{\mathbb{S}},\mathcal{U}_{\textup{loc}}(d,d),\alpha}
=
1-d\min_{\ket{\Phi}\in\ME_d}\bra{\Phi}\rho^{\mathrm{T}_B}\ket{\Phi}.
\end{equation}
Moreover,
\begin{equation}
\lim_{\alpha\to\infty} C_{\widetilde{\mathbb{S}},\mathcal{U}_{\textup{loc}}(d,d),\alpha}=1,
\end{equation}
and more precisely,
\begin{equation}
\log C_{\widetilde{\mathbb{S}},\mathcal{U}_{\textup{loc}}(d,d),\alpha}=
-\frac{\log\left(\alpha/(d-1)\right)}{\alpha}
+O\left(\frac{1}{\alpha^2}\right),
\qquad \alpha\to\infty.
\end{equation}
Consequently,
\begin{equation}
C_{\widetilde{\mathbb{S}},\mathcal{U}_{\textup{loc}}(d,d),\alpha}
=
1-\frac{\log\left(\alpha/(d-1)\right)}{\alpha}
+O\left(\frac{\log^2\alpha}{\alpha^2}\right).
\end{equation}
\end{theorem}

\begin{proof}[Proof of Theorem~\ref{thm:WE_infinity_criterion_entanglement_app}]
We first prove the $L^\alpha\to L^\infty$ limit. Set
\begin{equation}
f(U_A,U_B)=1-y_\rho(U_A,U_B).
\end{equation}
The domain $\mathrm{U}(d)\times\mathrm{U}(d)$ is compact, and $f$ is continuous; hence $f$ attains its maximum $M=\max f$. Since the Haar measure is normalized, $\norm{\rho}_{\widetilde{\mathbb{S}},\mathcal{U}_{\textup{loc}}(d,d),\alpha}=\norm{f}_{L^\alpha}$. Clearly $\norm{f}_{L^\alpha}\le M$ for all $\alpha$. Fix $\varepsilon>0$, and let $\left(U_A^{(0)},U_B^{(0)}\right)$ be a point where $f=M$. By continuity, the open set
\begin{equation}
\Omega_\varepsilon=\left\{(U_A,U_B): f(U_A,U_B)>M-\varepsilon\right\}
\end{equation}
has strictly positive Haar measure. Therefore
\begin{equation}
\norm{\rho}_{\widetilde{\mathbb{S}},\mathcal{U}_{\textup{loc}}(d,d),\alpha}
\geq
\left[(M-\varepsilon)^\alpha\mu(\Omega_\varepsilon)\right]^{1/\alpha}
=(M-\varepsilon)\mu(\Omega_\varepsilon)^{1/\alpha}.
\end{equation}
Letting $\alpha\to\infty$ gives
\begin{equation}
\liminf_{\alpha\to\infty} \norm{\rho}_{\widetilde{\mathbb{S}},\mathcal{U}_{\textup{loc}}(d,d),\alpha}\geq M-\varepsilon.
\end{equation}
Since $\varepsilon>0$ is arbitrary, one obtains
\begin{equation}
\lim_{\alpha\to\infty}\norm{\rho}_{\widetilde{\mathbb{S}},\mathcal{U}_{\textup{loc}}(d,d),\alpha}=M=1-\min_{U_A,U_B}y_\rho(U_A,U_B).
\end{equation}

Next, we rewrite the maximum in terms of $\rho^{\mathrm{T}_B}$. Notice that
\begin{equation}
\mathbb{S}^{\mathrm{T}_B}=d|\Phi_d\rangle\langle\Phi_d|.
\end{equation}
Using $\Tr(O_1O_2)=\Tr(O_1^{\mathrm{T}_B}O_2^{\mathrm{T}_B})$ for $O_1,O_2\in\mathcal{L}(\mathcal{H}_A\otimes\mathcal{H}_B)$, we obtain
\begin{equation}
\begin{aligned}
&\Tr\left[\mathbb{S}(U_A\otimes U_B)\rho(U_A\otimes U_B)^\dagger\right]\\
=&\Tr\left[\mathbb{S}^{\mathrm{T}_B}\left((U_A\otimes U_B)\rho(U_A\otimes U_B)^\dagger\right)^{\mathrm{T}_B}\right]\\
=&d\langle\Phi_d|(U_A\otimes U_B^*)\rho^{\mathrm{T}_B}(U_A^\dagger\otimes U_B^\mathrm{T})|\Phi_d\rangle.
\end{aligned}
\end{equation}
Define
\begin{equation}
\ket{\Phi}=(U_A^\dagger\otimes U_B^\mathrm{T})|\Phi_d\rangle\in\ME_d.
\end{equation}
Then
\begin{equation}
y_\rho(U_A,U_B)=\Tr\left[\mathbb{S}(U_A\otimes U_B)\rho(U_A\otimes U_B)^\dagger\right]=d\bra{\Phi}\rho^{\mathrm{T}_B}\ket{\Phi}.
\end{equation}
Taking the maximum over local unitaries is therefore equivalent to taking the minimum over $\ME_d$, and we conclude that
\begin{equation}
\lim_{\alpha\to\infty}\norm{\rho}_{\widetilde{\mathbb{S}},\mathcal{U}_{\textup{loc}}(d,d),\alpha}=1-d\min_{\ket{\Phi}\in\ME_d}\bra{\Phi}\rho^{\mathrm{T}_B}\ket{\Phi}.
\end{equation}

Finally, for $\alpha\geq1$,
\begin{equation}
C_{\widetilde{\mathbb{S}},\mathcal{U}_{\textup{loc}}(d,d),\alpha}=\left(\frac{d-1}{d+\alpha-1}\right)^{1/\alpha},
\end{equation}
so
\begin{equation}
\begin{aligned}
\log C_{\widetilde{\mathbb{S}},\mathcal{U}_{\textup{loc}}(d,d),\alpha}
&=\frac{\log(d-1)-\log(d+\alpha-1)}{\alpha}\\
&=-\frac{\log(\alpha/(d-1))}{\alpha}
-\frac{1}{\alpha}\log\left(1+\frac{d-1}{\alpha}\right)\\
&=-\frac{\log(\alpha/(d-1))}{\alpha}+O\left(\frac{1}{\alpha^2}\right).
\end{aligned}
\end{equation}
\end{proof}

\section{Qudit magic for odd prime local dimension}\label{app:qudit_magic}
We denote the dimension $d=q^n$. In the field $\mathbb{Z}_q$, addition and multiplication are well-defined in the sense of modulo $q$. Denote the $q$-th root of unity as $\omega_q=e^{i2\pi/q}$ and define the Weyl operators for $z,x\in\mathbb{Z}_q^n$ as
\begin{equation}
\begin{gathered}
Z(z)\ket{t}=\omega_q^{z\cdot t}\ket{t},\qquad X(x)\ket{t}=\ket{t+x},\quad\textup{for basis vector }\ket{t}\in\mathcal{H}\left((\mathbb{C}_q)^{\otimes n}\right),\\
T_{z,x}=\omega_q^{-2^{-1}z\cdot x}Z(z)X(x).
\end{gathered}
\end{equation}
Introducing the symplectic form for $u=(z_u,x_u)$ and $v=(z_v,x_v)$ in $V\coloneqq\mathbb{Z}_q^n\times\mathbb{Z}_q^n$ as
\begin{equation}
[u,v]=z_u\cdot x_v-x_u\cdot z_v,
\end{equation}
the multiplication of two Weyl operators is
\begin{equation}
T_uT_v=\omega_q^{2^{-1}[u,v]}T_{u+v}.
\end{equation}
Using this as the group action, the $n$-qudit Heisenberg-Weyl group with local dimension $q$ is defined as
\begin{equation}
\mathsf{P}_{n,q}=\left\{\omega_q^aT_{z,x}\middle\mid a\in\mathbb{Z}_q,z,x\in\mathbb{Z}_q^n\right\}.
\end{equation}
The set of free unitaries is the $n$-qudit Clifford group, which is the normalizer of $\mathsf{P}_{n,q}$
\begin{equation}
\Ufree=\Cl_{n,q}=\{U\in \mathrm{U}(d)\mid U\mathsf{P}_{n,q}U^\dagger=\mathsf{P}_{n,q}\}.
\end{equation}
The set of free states ($n$-qudit stabilizer polytope) and the set of stabilizer states are
\begin{equation}
\Fset=\STAB_{n,q}=\textup{Conv}(\Sigma_{n,q}),\qquad\Sigma_{n,q}=\left\{U\ket{0^n}\middle| U\in \Cl_{n,q}\right\}.
\end{equation}

The set of pure stabilizer states is characterized by the discrete Hudson's theorem~\cite{gross2006Hudson}, by evaluating the discrete Wigener function. 
\begin{lemma}[Discrete Hudson's theorem~\cite{gross2006Hudson}]
For an odd prime integer $q$ and any $\rho\in\mathcal{D}\left(\mathcal{H}\left((\mathbb{C}_q)^{\otimes n}\right)\right)$, define its discrete Wigner function evaluated at $u\in V$ as
\begin{equation}\label{eq:def_Wigner}
W_{\rho}(u)=\frac1{d}\Tr(A_u\rho),
\end{equation}
where $A_u$ is the phase space point operator
\begin{equation}\label{eq:A_u}
A_u=\frac{1}{d}\sum_{v\in V}\omega_q^{-[u,v]}T_v^\dagger.
\end{equation}
Then an $n$-qudit pure state vector $\ket{\psi}\in\mathcal{H}\left((\mathbb{C}_q)^{\otimes n}\right)$ is a stabilizer state if and only if its discrete Wigner function takes nonnegative values at all points,
\begin{equation}
\ket{\psi}\in\Sigma_{n,q}\quad\Longleftrightarrow\quad W_{\ketbra{\psi}{\psi}}(u)\geq0, \forall u\in V.
\end{equation}
\end{lemma}
We can simply choose $W=A_{0^{2n}}$ and $\widetilde{W}=\widetilde{A}_{0^{2n}}=\bI-A_{0^{2n}}$ to be the witness to be expanded, since
\begin{equation}
\sigma\in\STAB_{n,q}\quad\Longrightarrow\quad \Tr(A_{0^{2n}}\sigma)=dW_\sigma(0^{2n})\geq0,
\end{equation}
by convexity, and $A_{0^{2n}}$ can detect the non-stabilizer state $\rho=(\ket{1}-\ket{-1})(\bra{1}-\bra{-1})/2$. 

We first show that, given solely state moments $\left\{\Tr(\rho^k)\right\}_k$, regardless of how many orders there are, we cannot distinguish stabilizer mixtures from genuinely magic states.

\begin{proposition}[Universality of state spectrum in stabilizer polytope]\label{prop:same-spectrum}
For every $\rho\in \mathcal{D}((\C^q)^{\otimes n})$, there exists $\sigma\in \STAB_{n,q}$ such that
\begin{equation}
\spec(\sigma)=\spec(\rho).
\label{eq:same-spectrum}
\end{equation}
Equivalently,
\begin{equation}
\begin{aligned}
\Tr(\sigma^k) &= \Tr(\rho^k),
&\qquad \forall k\in \mathbb{N}.
\end{aligned}
\label{eq:all-moments}
\end{equation}
\end{proposition}

\begin{proof}[Proof of Proposition~\ref{prop:same-spectrum}]
Let
$\rho=\sum_{j=0}^{d-1} \lambda_j \proj{\psi_j}$
be a spectral decomposition of $\rho$, where $\lambda_j\ge 0$ and $\sum_{j=0}^{d-1}\lambda_j=1$. Consider the computational basis
$\{\ket{j}\}_{j=0}^{d-1}
=\{\ket{0,\cdots,0},\cdots,\ket{q-1,\cdots,q-1}\}$
of $(\C^q)^{\otimes n}$. Each computational basis vector is a pure stabilizer state. Define
$\sigma=\sum_{j=0}^{d-1} \lambda_j \proj{j}$.
Since $\sigma$ is a convex combination of pure stabilizer states, we have $\sigma\in \STAB_{n,q}$.

By construction, the eigenvalues of $\sigma$ are precisely $\{\lambda_j\}_{j=0}^{d-1}$, counted with multiplicity. Hence Eq.~\eqref{eq:same-spectrum} holds. Since the spectral moments are symmetric functions of the eigenvalues, for every positive integer $k$ one obtains
\begin{equation}
\Tr(\sigma^k)
=\sum_{j=0}^{d-1}\lambda_j^k
=\Tr(\rho^k),
\end{equation}
which proves Eq.~\eqref{eq:all-moments}.
\end{proof}

The proposition implies that any magic detection criterion must depend on more than the eigenvalues of the state. Equivalently, no quantity that depends only on the collection $\left\{\Tr(\rho^k)\right\}_{k\geq 1}$ can detect whether $\rho$ lies outside $\STAB_{n,q}$. Especially if the group of free unitaries $\mathcal{G}$ forms a unitary $t$-design, the first $t$-th order twirlings are the same as those of the Haar random unitaries, which, by Schul-Weyl duality, only relate to state moments. Indeed, when evaluating
\begin{equation}\label{eq:R_alpha_Clifford_magic}
\norm{\rho}_{\widetilde{W},\mathcal{G},\alpha}=\left(\E_{U\sim\mu_{\mathcal{G}}}\left[\Tr(\widetilde W U\rho U^\dagger)^\alpha\right]\right)^{\frac1\alpha},
\end{equation}
for $\alpha\leq t$, all information encoded by $\rho$ is in the form of the $k\leq \alpha$ twirling formula
\begin{equation}
\Phi_{\mathcal{G}}^{(k)}(\rho)=\E_{U\sim\mu_\mathcal{G}}\left[U^{\otimes k}\rho^{\otimes k}{U^\dagger}^{\otimes k}\right]=\E_{U\sim\mu_{\textup{Haar}}}\left[U^{\otimes k}\rho^{\otimes k}{U^\dagger}^{\otimes k}\right],
\end{equation}
which consists of terms related to permutation operators~\cite{collins2006IntegrationRespectHaar} and leads to only state moments of $\rho$
\begin{equation}
\Tr(\mathbb{P}_\pi^{(k)}\rho^{\otimes k})=\prod_{c\in\textup{cycle}(\pi)}\Tr(\rho^{\abs{c}}).
\end{equation}

Notice that in our setting, when $q$ is an odd prime, the $n$-qudit Clifford group $\Cl_{n,q}$ is a unitary 2-design~\cite{divincenzo2002QuantumDataHiding}, but not a unitary 3-design~\cite{gross2021SchurWeylClifford,bittel2025CompleteTheoryClifford}. Proposition~\ref{prop:same-spectrum} implies the WE $1,2$-criteria are both trivial, we start with $\alpha=3$ and prove Theorem~\ref{thm:WE_3_qudit_magic}.

\begin{proof}[Proof of Theorem~\ref{thm:WE_3_qudit_magic}]
We need to calculate
\begin{equation}\label{eq:R_3_Clifford_qudit_magic_1}
\norm{\rho}_{\widetilde{A}_{0^{2n}},\Cl_{n,q},3}^3 =\E_{U\sim\mu_{\Cl_{n,q}}}\left[\Tr(\widetilde{A}_{0^{2n}} U\rho U^\dagger)^3\right] =\E_{U\sim\mu_{\Cl_{n,q}}}
\left[
\Tr\left((\bI-U^\dagger A_{0^{2n}}U)\rho\right)^3
\right].
\end{equation}
Since $q$ is an odd prime, the Clifford group has the semidirect-product form~\cite{gross2021SchurWeylClifford}
\begin{equation}
\Cl_{n,q}
\simeq
V\rtimes \mathrm{Sp}(2n,\mathbb{Z}_q),
\end{equation}
so each Clifford $U\in\Cl_{n,q}$ can be written as $U=T_a\mu(F)$, where $\mu$ is a symplectic representation of $\mathrm{Sp}(2n,\mathbb{Z}_q)$. The action on the Weyl operator $T_u$ is given by
\begin{equation}
\mu(F)T_u\mu(F)^\dagger=T_{Fu},\forall u\in V.
\end{equation}
From Eq.~\eqref{eq:A_u}, we have
\begin{equation}
\mu(F)A_u\mu(F)^\dagger=\frac{1}{d}\sum_{v\in V}\omega_q^{-[u,v]}\mu(F)T_v^\dagger \mu(F)^\dagger=\frac{1}{d}\sum_{w\in V}\omega_q^{-[u,F^{-1}w]}T_w^\dagger =\frac{1}{d}\sum_{w\in V}\omega_q^{-[Fu,w]}T_w^\dagger=A_{Fu},
\end{equation}
and $UA_u U^\dagger=A_{Fu+a}$. So that $\{U^\dagger A_{0^{2n}}U\}_{U\in\Cl_{n,q}}$ is uniformly distributed over $\{A_u\}_{u\in V}$. Then Eq.~\eqref{eq:R_3_Clifford_qudit_magic_1} is reduced to
\begin{equation}
\begin{aligned}
\norm{\rho}_{\widetilde{A}_{0^{2n}},\Cl_{n,q},3}^3&=\frac1{d^2}\sum_{u\in V}
\left(
1-\Tr(A_u\rho)
\right)^3
=\frac1{d^2}\sum_{u\in V}
\left(
1-dW_\rho(u)
\right)^3
\\&=1-\frac3d\sum_{u\in V}W_\rho(u)+3\sum_{u\in V}W_\rho(u)^2-d\sum_{u\in V}W_\rho(u)^3.
\end{aligned}
\end{equation}
A direct calculation yields for all $\rho$,
\begin{equation}
\sum_{u\in V}W_\rho(u)=1,\qquad\sum_{u\in V}W_\rho(u)^2=\frac1d\Tr(\rho^2),
\end{equation}
and thus Eq.~\eqref{eq:R_3_qudit_magic} follows:
\begin{equation}
\norm{\rho}_{\widetilde{A}_{0^{2n}},\Cl_{n,q},3}^3
=
1-\frac{3}{d}+\frac{3}{d}\Tr(\rho^2)-d\sum_{u\in V}W_\rho(u)^3.
\end{equation}

For any pure stabilizer state vector $\ket{\psi}$, its discrete Wigner function $W_{\ketbra{\psi}{\psi}}$ takes value $1/d$ on $d$ phase space points, and value $0$ on the remaining $d^2-d$ points~\cite{gross2006Hudson}, so for $\alpha\geq1$,
\begin{equation}
\norm{\ketbra{\psi}{\psi}}_{\widetilde{A}_{0^{2n}},\Cl_{n,q},\alpha}^\alpha=\frac1{d^2}\sum_{u\in V}
\left(
1-dW_{\ketbra{\psi}{\psi}}(u)
\right)^\alpha=\frac1{d^2}(d^2-d)=1-\frac1d=C_{\widetilde{A}_{0^{2n}},\Cl_{n,q},\alpha}^\alpha.
\end{equation}
By Observation~\ref{obs:threshold_single_orbit}, it is also the maximum value obtainable for all states in $\STAB_{n,q}$, which proves Eq.~\eqref{eq:C_alpha_qudit_magic}. Furthermore, when $\ket{\psi}$ is a pure state, notice that $\norm{A_u}_\infty=1$ and $W_{\ketbra{\psi}{\psi}}(u)\leq1/d$ by Eq.~\eqref{eq:def_Wigner}, we have
\begin{equation}\label{eq:WE_3_criterion_pure_state_qudit_magic}
\norm{\ketbra{\psi}{\psi}}_{\widetilde{A}_{0^{2n}},\Cl_{n,q},3}^3 - C_{\widetilde{A}_{0^{2n}},\Cl_{n,q},3}^3=-d\sum_{u\in V}W_{\ketbra{\psi}{\psi}}(u)^3+\frac1d=\sum_{u\in V}W_{\ketbra{\psi}{\psi}}(u)^2\left(1-dW_{\ketbra{\psi}{\psi}}(u)\right)\geq0.
\end{equation}
So $\norm{\ketbra{\psi}{\psi}}_{\widetilde{A}_{0^{2n}},\Cl_{n,q},3}=C_{\widetilde{A}_{0^{2n}},\Cl_{n,q},3}$ if and only if $W_{\ketbra{\psi}{\psi}}(u)$ only take nonnegative values $1/d$ or $0$, which implies $\ket{\psi}$ is a stabilizer state by the discrete Hudson's theorem.
\end{proof}

We then prove Theorem~\ref{thm:WE_2_subgroup_qudit_magic}, which instead uses a subgroup $\mathcal{G}_X=\left\{X(x)\middle|x\in\mathbb Z_q^n\right\}$ and takes the average over a column of phase-space point operators.

\begin{proof}[Proof of Theorem~\ref{thm:WE_2_subgroup_qudit_magic}]
We need to calculate
\begin{equation}
\norm{\rho}_{\widetilde{A}_{0^{2n}},\mathcal{G}_X,2}^2 =\E_{U\sim\mu_{\mathcal{G}_X}}
\left[
\Tr\left((\bI-U^\dagger A_{0^{2n}}U)\rho\right)^2
\right].
\end{equation}
Similar to the proof of Theorem~\ref{thm:WE_3_qudit_magic}, we recognize $U^\dagger A_0U$ for $U=X(x)\in\mathcal{G}_X$ as a phase-space point operator
\begin{equation}
X(x)^\dagger A_{0^{2n}}X(x)=A_{0^n,-x}.
\end{equation}
Then we have
\begin{equation}
\begin{aligned}
\norm{\rho}_{\widetilde{A}_{0^{2n}},\mathcal{G}_X,2}^2&=\frac1{d}\sum_{x\in \mathbb{Z}_q^n}
\left(
1-\Tr(A_{0,-x}\rho)
\right)^2
=\frac1{d}\sum_{x\in \mathbb{Z}_q^n}
\left(
1-dW_\rho(0^n,x)
\right)^2
\\&=1-2\sum_{x\in \mathbb{Z}_q^n}W_\rho(0^n,x)+d\sum_{x\in \mathbb{Z}_q^n}W_\rho(0^n,x)^2.
\end{aligned}
\end{equation}
Notice that $A_u$ has largest eigenvalue $1$, so for $\sigma\in\STAB_{n,q}$, $0\leq W_\sigma(u)\leq1/d$, we have
\begin{equation}
\norm{\sigma}_{\widetilde{A}_{0^{2n}},\mathcal{G}_X,2}^2\leq1-\sum_{x\in \mathbb{Z}_q^n}W_\sigma(0^n,x)\leq1.
\end{equation}
The upper bound is attainable. Take a pure stabilizer state whose Wigner support is the affine Lagrangian subspace
\begin{equation}
(\iota,0^n)+L_X,
\qquad
L_X=\left\{(0^n,x)\mid x\in\mathbb Z_q^n\right\},
\end{equation}
where $\iota=(1,0,\cdots,0)\in\mathbb Z_q^n$. This support is disjoint from the slice $\left\{(0^n,x)\mid x\in\mathbb Z_q^n\right\}$, so all the slice Wigner values vanish and the value of the functional is exactly $1$.
The threshold value for all stabilizer states is
\begin{equation}
C_{\widetilde{A}_{0^{2n}},\mathcal{G}_X,2}^2=\sup_{\sigma\in\STAB_{n,q}}\norm{\sigma}_{\widetilde{A}_{0^{2n}},\mathcal{G}_X,2}^2=1,
\end{equation}

For any $n$ and odd prime $q$, we now provide a state $\rho_2$ that can be detected by the second-order criterion in Theorem~\ref{thm:WE_2_subgroup_qudit_magic} but cannot by the third-order criterion in Theorem~\ref{thm:WE_3_qudit_magic}, and another pure state vector $\ket{\psi_3}$ on the contrary.
\begin{itemize}
\item
Consider 
\begin{equation}
\label{eq:rho2-def-standalone}
\rho_2
=
\frac{4\bI-3A_{0^{2n}}}{4d-3}.
\end{equation}
It is a valid density matrix with eigenvalues $\{1/(4d-3),7/(4d-3)\}$ and has trace $1$. Using $\Tr(A_uA_v)=d\delta_{u,v}$, the discrete Wigner function of $\rho_2$ at $u\in V$ takes the value
\begin{equation}
W_{\rho_2}(u)
=
\frac{4-3d\delta_{u,0^{2n}}}{d(4d-3)}=
\begin{dcases}
\frac{4-3d}{d(4d-3)},\qquad &u=0^{2n},\\
\frac{4}{d(4d-3)},&u\neq0^{2n}.
\end{dcases}
\end{equation}
Then a direct calculation yields for $d\geq3$,
\begin{equation}
\begin{aligned}
\norm{\rho_2}_{\widetilde{A}_{0^{2n}},\mathcal{G}_X,2}^2-C_{\widetilde{A}_{0^{2n}},\mathcal{G}_X,2}^2&=
\frac{d-2}{(4d-3)^2}
>0,\\
\norm{\rho_2}_{\widetilde{A}_{0^{2n}},\Cl_{n,q},3}^3-C_{\widetilde{A}_{0^{2n}},\Cl_{n,q},3}^3&=-
\frac{(d-1)(128d^2-487d+414)}{d(4d-3)^3}<0.
\end{aligned}
\end{equation}
\item Consider
\begin{equation}
\ket{\psi_3}=\frac1{\sqrt{2}}\left(\ket{0^n}+\ket{\iota}\right),
\end{equation}
where $\iota$ is the short notation for the string $(1,0,\cdots,0)\in\mathbb{Z}_q^n$. Using $A_{(z,x)}\ket{y}=\omega_q^{2z\cdot(x-y)}\ket{-y+2x}$, its discrete Wigner function evaluates as
\begin{equation}
W_{\ketbra{\psi_3}{\psi_3}}(z,x)
=
\begin{dcases}
\dfrac{1}{2d}, & x=0^n,\\
\dfrac{1}{d}\cos(\frac{2\pi z_1}q), & x=\frac \iota2, \\
0, & \text{otherwise},
\end{dcases}
\end{equation}
where $z_1$ is the first digit of $z$.
By Eq.~\eqref{eq:WE_3_criterion_pure_state_qudit_magic}, $\ket{\psi_3}$ is a non-stabilizer state vector
since
\begin{equation}
\norm{\ketbra{\psi_3}{\psi_3}}_{\widetilde{A}_{0^{2n}},\Cl_{n,q},3}^3-C_{\widetilde{A}_{0^{2n}},\Cl_{n,q},3}^3=
\begin{dcases}
\dfrac{1}{2d}, & q=3,\\
\dfrac{3}{4d}, & q>3.
\end{dcases}
\end{equation}
And it is direct to verify
\begin{equation}
\norm{\ketbra{\psi_3}{\psi_3}}_{\widetilde{A}_{0^{2n}},\mathcal{G}_X,2}^2-C_{\widetilde{A}_{0^{2n}},\mathcal{G}_X,2}^2=
-\frac5{2d}
<0.
\end{equation}

\end{itemize}
\end{proof}

\section{Qubit magic}\label{app:qubit_magic}
\subsection{Proof of Theorem~\ref{thm:R4_magic}}\label{app:qubit_magic_WE_4-norm}
In $n$-qubit magic state resource theory, the set of free states (denoted as the stabilizer polytope $\STAB_n$) is the convex hull of all $n$-qubit pure stabilizer states, i.e., those that can be generated from $\ket{0^n}$ using a Clifford unitary,
\begin{equation}
\begin{aligned}
\Fset=\STAB_n=\textup{Conv}(\Sigma_n),\quad
\textup{with }\Sigma_n=\left\{U\ket{0^n}\middle\mid U\in\Cl_n\right\}.
\end{aligned}
\end{equation}
Here, the $n$-qubit Clifford group $\mathrm{Cl}_n$ is the normalizer of $n$-qubit Pauli group $\mathsf{P}_n$, which is also the set of free unitaries,
\begin{equation}\label{eq:qubit_magic_free_unitaries_app}
\Ufree=\Cl_n=\{U\in \mathrm{U}(d)\mid U\mathsf{P}_nU^\dagger=\mathsf{P}_n\}.
\end{equation}

Let the conjugacy classes of $S_4$ be indexed by cycle type
\begin{equation}
\mu\in\{[1,1,1,1],[2,1,1],[2,2],[3,1],[4]\},
\end{equation}
with class sizes
\begin{equation}\label{eq:m_mu}
\begin{aligned}
m_{[1,1,1,1]}&=1, & m_{[2,1,1]}&=6, & m_{[2,2]}&=3,&
m_{[3,1]}&=8, & m_{[4]}&=6.
\end{aligned}
\end{equation}
For an operator $A\in\mathcal{L}(\mathbb{C}^d)$, define 
\begin{equation}\label{eq:p_lambda_A}
p_\lambda(A)=\Tr(\mathbb{P}_\pi^{(4)}A^{\otimes4}),
\end{equation}
where $\pi\in S_4$ has the cycle type $\lambda$ and the permutation operator $\mathbb{P}_\pi^{(4)}$ is defined in Eq.~\eqref{eq:P_permutation_def}.
For an arbitrary state $\rho$, write the state moments
\begin{equation}
r_k=\Tr(\rho^k),\qquad k=2,3,4,
\end{equation}
then
\begin{equation}\label{eq:p_values}
\begin{aligned}
p_{[1,1,1,1]}(\rho)&=1, & p_{[2,1,1]}(\rho)&=r_2,&
p_{[2,2]}(\rho)&=r_2^2, & p_{[3,1]}(\rho)&=r_3,&
p_{[4]}(\rho)&=r_4.
\end{aligned}
\end{equation}
Denote the set of $n$-qubit phase-free Pauli operators as
\begin{equation}
\mathcal{P}_n=\left\{\bI,X,Y,Z\right\}^{\otimes n}.
\end{equation}
We introduce the projector that appears as a commutant of the fourfold tensor of Cliffords~\cite{zhu2016CliffordGroupFails,bittel2025CompleteTheoryClifford}
\begin{equation}\label{eq:Q_def}
Q=\frac{1}{d^2}\sum_{P\in\mathcal P_n}P^{\ot4},
\end{equation}
and its orthogonal complement projector $Q^\perp=\bI-Q$.
Define
\begin{equation}\label{eq:q_lambda_A}
q_\lambda(A)=\Tr(\mathbb{P}_\pi^{(4)}QA^{\otimes4}),
\end{equation}
for any $\pi\in S_4$ with the cycle type $\lambda$, and in particular
\begin{equation}\label{eq:q_values}
\begin{aligned}
q_{[1,1,1,1]}(A)&=\frac{1}{d^2}\sum_{P\in\mathcal P_n}[\Tr(A P)]^4,\\
q_{[2,1,1]}(A)&=\frac{1}{d^2}\sum_{P\in\mathcal P_n}\Tr((A P)^2)[\Tr(A P)]^2,\\
q_{[2,2]}(A)&=\frac{1}{d^2}\sum_{P\in\mathcal P_n}[\Tr((A P)^2)]^2,\\
q_{[3,1]}(A)&=\frac{1}{d^2}\sum_{P\in\mathcal P_n}\Tr((A P)^3)\Tr(A P),\\
q_{[4]}(A)&=\frac{1}{d^2}\sum_{P\in\mathcal P_n}\Tr((A P)^4).
\end{aligned}
\end{equation}

We recall the Schur-Weyl decomposition (see, e.g., Theorem~S.1 in Ref.~\cite{roth2018RecoveringQuantumGates})
\begin{equation}\label{eq:Schur-Weyl_decomposition}
(\mathbb C^d)^{\otimes 4}
\cong
\bigoplus_{\lambda\vdash 4:\ell(\lambda)\leq d} W_\lambda\otimes S_\lambda,
\end{equation}
where $W_\lambda$ and $S_\lambda$ are the Weyl module and the Specht module, respectively. Using the same notation as in Ref.~\cite{zhu2016CliffordGroupFails}, we denote for each partition $\lambda\vdash4$, $QP_\lambda$ projects onto $W_\lambda^+\otimes S_\lambda$ and $Q^\perp P_\lambda$ projects onto $W_\lambda^-\otimes S_\lambda$. Define
\begin{equation}
d_\lambda=\dim S_\lambda,\qquad D_\lambda^\pm=\dim W_\lambda^\pm.
\end{equation}
Their values are listed in Table~\ref{tab:module_dimension} (copied from Table~1 in Ref.~\cite{zhu2016CliffordGroupFails}).
\begin{table}[h!]
\centering
\begin{tabular}{c|ccc}
\hline
$\lambda$ & $d_\lambda$ & $D_\lambda^+$ & $D_\lambda^-$  \\ \hline
$[4]$       & 1 & $\frac{(d+1)(d+2)}{6}$  & $\frac{(d-1)(d+1)(d+2)(d+4)}{24}$\\
$[3,1]$     & 3 & 0  & $\frac{d(d+2)(d^2-1)}{8}$ \\
$[2,2]$     & 2 & $\frac{d^2-1}{3}$  & $\frac{(d^2-4)(d^2-1)}{12}$ \\
$[2,1,1]$   & 3 & 0 & $\frac{d(d-2)(d^2-1)}{8}$ \\
$[1,1,1,1]$ & 1 & $\frac{(d-1)(d-2)}{6}$ & $\frac{(d+1)(d-1)(d-2)(d-4)}{24}$ \\ \hline
\end{tabular}
\caption{Dimensions of Specht and irreducible components, which appear in the fourth-order Clifford twirling formula.}
\label{tab:module_dimension}
\end{table}
Denote the orthogonal projector to the submodule $W_\lambda\otimes S_\lambda$ as (cf.\ Eq.~\eqref{eq:P_permutation_def})
\begin{equation}\label{eq:P_lambda}
P_\lambda=\frac{d_\lambda}{24}\sum_{\omega\in S_4}\chi_\lambda(\omega)\mathbb{P}_\omega^{(4)},
\end{equation}
we have 
\begin{equation}\label{eq:P_mu_normalization}
\sum_{\lambda\vdash 4:\ell(\lambda)\leq d} P_\lambda=\bI.
\end{equation}
Here, the character value $\chi_\lambda(\omega)$ is determined by the cycle type $\mu$ of $\omega$, rewritten as $\chi_\lambda(\mu)$, and can be calculated using the Murnaghan-Nakayama rule, listed in Table~\ref{tab:character_S_4}.
\begin{table}[h!]
\centering
\begin{tabular}{c|ccccc}
\hline
$\chi_\lambda(\mu)$ & $[1,1,1,1]$ & $[2,1,1]$ & $[2,2]$ & $[3,1]$ & $[4]$ \\ \hline
$[4]$       & 1 & 1  & 1  & 1  & 1  \\
$[3,1]$     & 3 & 1  & -1 & 0  & -1 \\
$[2,2]$     & 2 & 0  & 2  & -1 & 0  \\
$[2,1,1]$   & 3 & -1 & -1 & 0  & 1  \\
$[1,1,1,1]$ & 1 & -1 & 1  & 1  & -1 \\ \hline
\end{tabular}
\caption{Character table of the symmetric group $S_4$. 
Rows correspond to irreducible representations (partitions $\lambda \vdash 4$), 
columns to conjugacy classes (cycle types) $\mu$.}
\label{tab:character_S_4}
\end{table}

We further define the character transforms (cf.\ Eq.~\eqref{eq:m_mu}, Table~\ref{tab:character_S_4}, Eq.~\eqref{eq:p_values} and Eq.~\eqref{eq:q_values})
\begin{equation}\label{eq:pqhat-def-inline}
\begin{aligned}
\widehat p_\lambda(A)&=\frac1{d_\lambda}\Tr(P_\lambda A^{\otimes4})=\frac{1}{24}\sum_\mu m_\mu\chi_\lambda(\mu)p_\mu(A),\\
\widehat q_\lambda(A)&=\frac1{d_\lambda}\Tr(P_\lambda QA^{\otimes4})=\frac{1}{24}\sum_\mu m_\mu\chi_\lambda(\mu)q_\mu(A),
\end{aligned}
\end{equation}
Explicitly,
\begin{equation}\label{eq:phat_values}
\begin{aligned}
\widehat p_{[4]}(\rho)&=\frac{1+6r_2+3r_2^2+8r_3+6r_4}{24},\\
\widehat p_{[3,1]}(\rho)&=\frac{1+2r_2-r_2^2-2r_4}{8},\\
\widehat p_{[2,2]}(\rho)&=\frac{1+3r_2^2-4r_3}{12},\\
\widehat p_{[2,1,1]}(\rho)&=\frac{1-2r_2-r_2^2+2r_4}{8},\\
\widehat p_{[1,1,1,1]}(\rho)&=\frac{1-6r_2+3r_2^2+8r_3-6r_4}{24}.
\end{aligned}
\end{equation}
Similarly, omitting the $\rho$ in $q_\lambda(\rho)$ for simplicity,
\begin{equation}\label{eq:qhat_values}
\begin{aligned}
\widehat q_{[4]}(\rho)&=\frac{q_{[1,1,1,1]}+6q_{[2,1,1]}+3q_{[2,2]}+8q_{[3,1]}+6q_{[4]}}{24},\\
\widehat q_{[3,1]}(\rho)&=\frac{q_{[1,1,1,1]}+2q_{[2,1,1]}-q_{[2,2]}-2q_{[4]}}{8},\\
\widehat q_{[2,2]}(\rho)&=\frac{q_{[1,1,1,1]}+3q_{[2,2]}-4q_{[3,1]}}{12},\\
\widehat q_{[2,1,1]}(\rho)&=\frac{q_{[1,1,1,1]}-2q_{[2,1,1]}-q_{[2,2]}+2q_{[4]}}{8},\\
\widehat q_{[1,1,1,1]}(\rho)&=\frac{q_{[1,1,1,1]}-6q_{[2,1,1]}+3q_{[2,2]}+8q_{[3,1]}-6q_{[4]}}{24}.
\end{aligned}
\end{equation}

We now calculate the WE $4$-norm with a positive semi-definite operator $O$ as seed and reformulate Theorem~\ref{thm:R4_magic} in the main text as the theroem below.
\begin{theorem}[WE $4$-norm for qubit magic]\label{thm:R4-general-inline}
For every $n$-qubit state $\rho$, the WE $4$-norm in criterion~\eqref{eq:WE_alpha_criterion} with a positive semi-definite operator $O$ and $n$-qubit Clifford group $\Cl_n$ in Eq.~\eqref{eq:qubit_magic_free_unitaries_app} is given by
\begin{equation}
\begin{aligned}
\norm{\rho}_{O,\Cl_n,4}^4
=\sum_{\lambda\vdash4:\ell(\lambda)\leq d} d_\lambda\left[
\frac{\widehat q_\lambda(\rho)\widehat q_\lambda(O)}{D_\lambda^+}+
\frac{\left(\widehat p_\lambda(\rho)-\widehat q_\lambda(\rho)\right)\left(\widehat p_\lambda(O)-\widehat q_\lambda(O)\right)}{D_\lambda^-}
\right],
\end{aligned}
\label{eq:R4-general-inline_app}
\end{equation}
where the values of $d_\lambda$, $D_\lambda^\pm$ are given in Tabel~\ref{tab:module_dimension} and $\widehat{p}_\lambda$, $\widehat{q}_\lambda$ are defined in Eq.~\eqref{eq:pqhat-def-inline}. The coefficients in Theorem~\ref{thm:R4_magic} are
\begin{equation}
a_\lambda=\widehat q_\lambda(O),\qquad b_\lambda=\widehat p_\lambda(O)-\widehat q_\lambda(O).
\end{equation}
\end{theorem}

From Eq.~\eqref{eq:R_alpha_Clifford_magic}, we have
\begin{equation}
\norm{\rho}_{O,\Cl_n,4}^4=\Tr\left[O^{\otimes4}\Phi_{\Cl_n}^{(4)}(\rho^{\otimes4})\right].
\end{equation}
We start from the fourth-order Clifford twirling formula for $A\in\mathcal{L}\left(\left(\mathbb C^d\right)^{\otimes 4}\right)$ (see Theorem S.5 in Ref.~\cite{roth2018RecoveringQuantumGates}),
\begin{equation}\label{eq:4-order_Clifford_twirling}
\begin{aligned}
\Phi_{\Cl_n}^{(4)}(A)
&=\frac{1}{24}
\sum_{\lambda\vdash 4 : \ell(\lambda)\le d} d_\lambda
\sum_{\omega\in S_4}
\left[
\frac{1}{D_\lambda^+}
\Tr\left(AQ\mathbb P_{\omega^{-1}}^{(4)}\right)Q
+
\frac{1}{D_\lambda^-}
\Tr\left(AQ^\perp\mathbb P_{\omega^{-1}}^{(4)}\right)Q^\perp
\right]\mathbb P_\omega^{(4)} P_\lambda,
\end{aligned}
\end{equation}
where $Q^\perp=\bI-Q$ is defined as in Eq.~\eqref{eq:Q_def} and $P_\lambda$ is defined in Eq.~\eqref{eq:P_lambda}. 
We first prove a useful lemma, which can simplify the calculation of Eq.~\eqref{eq:4-order_Clifford_twirling}.

\begin{lemma}[Simplification for permutation-invariant operator]\label{lem:simplification_permutation-invariant}
For any operator $B$ that commutes with all permutations, i.e., satisfies
\begin{equation}
[B,\mathbb P_\omega^{(4)}]=0,
\qquad \forall\omega\in S_4,
\label{eq:B-commutes-perm}
\end{equation}
for the term in Eq.~\eqref{eq:4-order_Clifford_twirling}, we have
\begin{equation}
\frac{d_\lambda}{24}
\sum_{\omega\in S_4}
\Tr\left(B\mathbb P_{\omega^{-1}}^{(4)}\right)
\mathbb P_\omega^{(4)} P_\lambda
=
\frac{1}{d_\lambda}\Tr(BP_\lambda)P_\lambda,
\label{eq:key-reduction}
\end{equation}
where $P_\lambda$ is defined in Eq.~\eqref{eq:P_lambda}.
\end{lemma}
\begin{proof}[Proof of Lemma~\ref{lem:simplification_permutation-invariant}]
By the Schur-Weyl duality, the natural permutation representation of the symmetric group $S_4$ on $(\mathbb{C}^d)^{\otimes4}$,
\begin{equation}
\Gamma_{S_4}^d:\omega\mapsto\mathbb{P}_\omega^{(4)}
\end{equation}
is decomposed as
\begin{equation}
\Gamma_{S_4}^d(\omega)=\bigoplus_{\mu\vdash4:\ell(\mu)\leq d}\bI_{W_\mu}\otimes\Gamma_{S_4}^\mu(\omega),\qquad\omega\in S_4,
\end{equation}
where $\Gamma_{S_4}^\mu:S_4\rightarrow\mathcal{L}(S_\mu)$ is the irreducible representation of $\Gamma_{S_4}^d$ on the Specht module $S_\mu$.

Notice that $[B,P_\mu]=0$.
Hence $B$ preserves each Schur-Weyl block and can be decomposed as
\begin{equation}
B=B\bI=B\left(\sum_{\mu\vdash 4:\ell(\mu)\leq d} P_\mu\right)=B\left(\sum_{\mu\vdash 4:\ell(\mu)\leq d} P_\mu^2\right)=
\sum_{\mu\vdash 4:\ell(\mu)\leq d} P_\mu B P_\mu.
\end{equation}
Fix the partition $\mu$, on the block $W_\mu\otimes S_\mu$, write
\begin{equation}
P_\mu B P_\mu
=
\sum_{a,b} |a\rangle\langle b|\otimes M_{ab},
\end{equation}
where $\{|a\rangle\}$ is a basis of $W_\mu$ and each $M_{ab}\in \textup{End}(S_\mu)$.
Because $P_\mu B P_\mu$ still commutes with every permutation, we have
\begin{equation}
(\bI_{W_\mu}\otimes \Gamma_{S_4}^\mu(\omega))(P_\mu B P_\mu)
=
(P_\mu B P_\mu)(\bI_{W_\mu}\otimes \Gamma_{S_4}^\mu(\omega)),
\qquad \forall\omega\in S_4.
\end{equation}
Comparing coefficients of $|a\rangle\langle b|$ yields
\begin{equation}
\Gamma_{S_4}^\mu(\omega)M_{ab}=M_{ab}\Gamma_{S_4}^\mu(\omega),
\qquad \forall\omega\in S_4.
\end{equation}
Since $S_\mu$ is an irreducible $S_4$-module, Schur's lemma implies that each $M_{ab}$ must be a scalar multiple of the identity on $S_\mu$:
\begin{equation}
M_{ab}=c_{ab}\bI_{S_\mu}.
\end{equation}
Therefore
\begin{equation}
P_\mu B P_\mu
=
\sum_{a,b} c_{ab}|a\rangle\langle b|\otimes \bI_{S_\mu}
=
B_\mu\otimes \bI_{S_\mu}
\end{equation}
for a uniquely defined operator
\begin{equation}
B_\mu=\sum_{a,b} c_{ab}|a\rangle\langle b|\in \textup{End}(W_\mu).
\end{equation}

We now evaluate terms in Eq.~\eqref{eq:key-reduction}. Using $\sum_\mu P_\mu=\bI$, $[B,P_\mu]=0$, and $[\mathbb P_\omega^{(4)},P_\mu]=0$, we obtain
\begin{equation}
\Tr\left(B\mathbb P_{\omega^{-1}}^{(4)}\right)
=
\sum_{\mu\vdash 4:\ell(\mu)\le d}
\Tr\left(P_\mu B P_\mu\mathbb P_{\omega^{-1}}^{(4)}\right).
\end{equation}
On the $\mu$-block,
\begin{equation}
P_\mu B P_\mu = B_\mu\otimes \bI_{S_\mu},
\qquad
\mathbb P_{\omega^{-1}}^{(4)}\bigg|_{W_\mu\otimes S_\mu}
=
\bI_{W_\mu}\otimes \Gamma_{S_4}^\mu(\omega^{-1}),
\end{equation}
so that
\begin{equation}
\begin{aligned}
\Tr\left(P_\mu B P_\mu\mathbb P_{\omega^{-1}}^{(4)}\right)
&=
\Tr\left((B_\mu\otimes \bI_{S_\mu})(\bI_{W_\mu}\otimes \Gamma_{S_4}^\mu(\omega^{-1}))\right)\\
&=
\Tr(B_\mu)\Tr\left(\Gamma_{S_4}^\mu(\omega^{-1})\right)\\
&=
\Tr(B_\mu)\chi_\mu(\omega^{-1}).
\end{aligned}
\end{equation}
Summing over all blocks gives
\begin{equation}
\Tr\left(B\mathbb P_{\omega^{-1}}^{(4)}\right)
=
\sum_{\mu\vdash 4 : \ell(\mu)\le d}
\Tr(B_\mu)\chi_\mu(\omega^{-1}).
\label{eq:trace-character-expansion}
\end{equation}

Finally, note that
\begin{equation}
\Tr(BP_\mu)
=\Tr(P_\mu BP_\mu)=
\Tr(B_\mu\otimes \bI_{S_\mu})
=
 d_\mu\Tr(B_\mu),
\end{equation}
so equivalently
\begin{equation}
\Tr(B_\mu)=\frac{1}{d_\mu}\Tr(BP_\mu).
\end{equation}
Then Eq.~\eqref{eq:trace-character-expansion} becomes
\begin{equation}
\Tr\left(B\mathbb P_{\omega^{-1}}^{(4)}\right)
=
\sum_{\mu\vdash 4 : \ell(\mu)\le d}
\frac{1}{d_\mu}\Tr(BP_\mu)\chi_\mu(\omega^{-1}).
\label{eq:trace-character-expansion-2}
\end{equation}
Substituting this into the left-hand side of Eq.~\eqref{eq:key-reduction}, we obtain
\begin{equation}
\begin{aligned}
\frac{d_\lambda}{24}
\sum_{\omega\in S_4}
\Tr\left(B\mathbb P_{\omega^{-1}}^{(4)}\right)
\mathbb P_\omega^{(4)} P_\lambda=
\sum_{\mu\vdash 4 : \ell(\mu)\le d}
\frac{1}{d_\mu}\Tr(BP_\mu)
\left[
\frac{d_\lambda}{24}
\sum_{\omega\in S_4}
\chi_\mu(\omega^{-1})\mathbb P_\omega^{(4)} P_\lambda
\right]
=\sum_{\mu\vdash 4 : \ell(\mu)\le d}
\frac{1}{d_\mu}\Tr(BP_\mu)
P_\mu P_\lambda.
\end{aligned}
\end{equation}
Notice that the projectors defined in Eq.~\eqref{eq:P_lambda} are  orthogonal projections on $W_\lambda\otimes S_\lambda$,
\begin{equation}
P_\mu P_\lambda
=\delta_{\lambda,\mu}P_\lambda,
\end{equation}
Eq.~\eqref{eq:key-reduction} then follows.
\end{proof}

\begin{proof}[Proof of Theorem~\ref{thm:R4-general-inline}]
Since $\rho^{\otimes 4}$ is invariant under permutations of the four tensor factors, we have
\begin{equation}
\left[\rho^{\otimes 4},\mathbb P_\omega^{(4)}\right]=0,
\qquad \forall\omega\in S_4.
\label{eq:A4-commutes-perm}
\end{equation}
Because $Q$ and $Q^\perp$ both commute with all permutation operators in $S_4$ (easily follows from Eq.~\eqref{eq:Q_def}), we have
\begin{equation}
\left[\rho^{\otimes 4}Q,\mathbb P_\omega^{(4)}\right]=0,
\qquad
\left[\rho^{\otimes 4}Q^\perp,\mathbb P_\omega^{(4)}\right]=0,
\qquad \forall\omega\in S_4.
\label{eq:commuting-Bs}
\end{equation}
Take $A=\rho^{\otimes4}$ in Eq.~\eqref{eq:4-order_Clifford_twirling}, then each of the two $\omega$-sums can be simplified separately. 
Apply Lemma~\ref{lem:simplification_permutation-invariant} first with 
$B=\rho^{\otimes 4}Q$, this gives
\begin{equation}
\begin{aligned}
\frac{1}{24}d_\lambda
\sum_{\omega\in S_4}
\frac{1}{D_\lambda^+}
\Tr\left(\rho^{\otimes 4}Q\mathbb P_{\omega^{-1}}^{(4)}\right)
Q\mathbb P_\omega^{(4)} P_\lambda
=
\frac{Q}{D_\lambda^+}
\left[
\frac{d_\lambda}{24}
\sum_{\omega\in S_4}
\Tr\left(\rho^{\otimes 4}Q\mathbb P_{\omega^{-1}}^{(4)}\right)
\mathbb P_\omega^{(4)} P_\lambda
\right]
=
\frac{1}{D_\lambda^+}
\frac{1}{d_\lambda}
\Tr\left(\rho^{\otimes 4}QP_\lambda\right)
QP_\lambda.
\end{aligned}
\label{eq:first-term-reduced}
\end{equation}
Likewise, apply Lemma~\ref{lem:simplification_permutation-invariant} with
$B=\rho^{\otimes 4}Q^\perp$,
we obtain
\begin{equation}
\frac{1}{24}d_\lambda
\sum_{\omega\in S_4}
\frac{1}{D_\lambda^-}
\Tr\left(\rho^{\otimes 4}Q^\perp\mathbb P_{\omega^{-1}}^{(4)}\right)
Q^\perp\mathbb P_\omega^{(4)} P_\lambda
=
\frac{1}{D_\lambda^-}
\frac{1}{d_\lambda}
\Tr\left(\rho^{\otimes 4}Q^\perp P_\lambda\right)
Q^\perp P_\lambda.
\label{eq:second-term-reduced}
\end{equation}

Recall the definitions in Eq.~\eqref{eq:pqhat-def-inline},
\begin{equation}
\widehat p_\lambda(\rho)=\frac{1}{d_\lambda}\Tr\left(\rho^{\otimes 4}P_\lambda\right),
\qquad
\widehat q_\lambda(\rho)=\frac{1}{d_\lambda}\Tr\left(\rho^{\otimes 4}QP_\lambda\right).
\label{eq:hatp-hatq-def}
\end{equation}
Since $Q^\perp=\bI-Q$, we also have
\begin{equation}
\frac{1}{d_\lambda}\Tr\left(\rho^{\otimes 4}Q^\perp P_\lambda\right)
=
\widehat p_\lambda(\rho)-\widehat q_\lambda(\rho).
\label{eq:qperp-scalar}
\end{equation}
Substituting Eqs.~\eqref{eq:first-term-reduced}--\eqref{eq:qperp-scalar} into Eq.~\eqref{eq:4-order_Clifford_twirling}, we arrive at
\begin{equation}
\Phi_{\Cl_n}^{(4)}(\rho^{\otimes 4})
=
\sum_{\lambda\vdash 4 : \ell(\lambda)\le d}
\left[
\frac{\widehat q_\lambda(\rho)}{D_\lambda^+}Q
+
\frac{\widehat p_\lambda(\rho)-\widehat q_\lambda(\rho)}{D_\lambda^-}Q^\perp
\right]P_\lambda.
\label{eq:Phi-A4-final}
\end{equation}
And then
\begin{equation}
\begin{aligned}
\norm{\rho}_{O,\Cl_n,4}^{4}
=\sum_{\lambda\vdash4:\ell(\lambda)\le d}
\left[
\frac{\widehat q_\lambda(\rho)}{D_\lambda^+}\Tr(O^{\ot4}QP_\lambda)
+
\frac{\widehat p_\lambda(\rho)-\widehat q_\lambda(\rho)}{D_\lambda^-}
\Tr(O^{\ot4}Q^\perp P_\lambda)
\right],
\end{aligned}
\end{equation}
where the terms with zero denominator should be omitted in the above summation.
Now by the definitions of $\widehat p_\lambda$ and $\widehat q_\lambda$ in Eq.~\eqref{eq:pqhat-def-inline},
\begin{equation}
\begin{aligned}
\Tr(O^{\ot4}QP_\lambda)&=d_\lambda\widehat q_\lambda(O),\\
\Tr(O^{\ot4}Q^\perp P_\lambda)
&=d_\lambda\left(\widehat p_\lambda(O)-\widehat q_\lambda(O)\right).
\end{aligned}
\end{equation}
Substituting these identities proves Eq.~\eqref{eq:R4-general-inline_app}.
\end{proof}

\subsection{Coefficients for standard triangle witness and proof of Corollary~\ref{cor:pure-C4_magic}}\label{app:qubit_magic_T_coefficients}
We then calculate the explicit coefficients related to $\widetilde{T}$ that appear in the formula of $\norm{\rho}_{\widetilde{T},\Cl_n,4}$.
For the witness induced by the triangle criterion~\cite{liu2025MagicCriterionAlmost}
\begin{equation}\label{eq:witness_magic}
W=T=\frac{\bI+X-Y+Z}{1+\sqrt3}\otimes \ket{0^{n-1}}\bra{0^{n-1}},\qquad
\widetilde{W}=\widetilde{T}=\bI-T,
\end{equation}
its spectrum is
\begin{equation}
\operatorname{spec}(\widetilde T)=\{0,3-\sqrt3,1,\dots,1\},
\end{equation}
where the eigenvalue $1$ has multiplicity $d-2$ with $d=2^n$, and we write its $k$th-order moment as
\begin{equation}\label{eq:def_t_k}
t_k=\Tr(\widetilde T^k)=d-2+(3-\sqrt3)^k,\qquad k\in\mathbb{N}_+.
\end{equation}
Especially, the first four moments are
\begin{equation}
\begin{aligned}
 t_1&=d+1-\sqrt3, & t_2&=d+10-6\sqrt3, &
 t_3&=d+52-30\sqrt3, & t_4&=d+250-144\sqrt3.
\end{aligned}
\label{eq:tk-values-inline}
\end{equation}
\begin{lemma}[Coefficient data for $\widetilde{T}$]\label{lem:coefficients-inline}
Write $\nu_\lambda=\widehat p_\lambda(\widetilde T)$, $a_\lambda=\widehat q_\lambda(\widetilde T)$ and $b_\lambda=\nu_\lambda-a_\lambda$, their explicit values are
\begin{equation}\label{eq:nu_tilde_W}
\begin{aligned}
\nu_{[4]}&=\frac{t_1^4+6t_1^2t_2+3t_2^2+8t_1t_3+6t_4}{24},\\
\nu_{[3,1]}&=\frac{t_1^4+2t_1^2t_2-t_2^2-2t_4}{8},\\
\nu_{[2,2]}&=\frac{t_1^4+3t_2^2-4t_1t_3}{12},\\
\nu_{[2,1,1]}&=\frac{t_1^4-2t_1^2t_2-t_2^2+2t_4}{8},\\
\nu_{[1,1,1,1]}&=\frac{t_1^4-6t_1^2t_2+3t_2^2+8t_1t_3-6t_4}{24},
\end{aligned}
\end{equation}
where values of $t_k$ are given in Eq.~\eqref{eq:tk-values-inline}.
Moreover,
\begin{equation}\label{eq:a_values}
\begin{aligned}
a_{[4]}&=
\frac{1}{6}\left(d^2+(7-4\sqrt3)d+86-48\sqrt3+\frac{800-464\sqrt3}{d}\right),\\
a_{[2,2]}&=
\frac{1}{3}\left(d^2+(4-4\sqrt3)d+23-12\sqrt3+\frac{-70+40\sqrt3}{d}\right),\\
a_{[1,1,1,1]}&=
\frac{1}{6}\left(d^2+(1-4\sqrt3)d-34+24\sqrt3+\frac{56-32\sqrt3}{d}\right),\\
a_{[3,1]}&=0,\\
a_{[2,1,1]}&=0.
\end{aligned}
\end{equation}
\end{lemma}

\begin{proof}[Proof of Lemma~\ref{lem:coefficients-inline}]
The formulas for $\nu_\lambda$ follow immediately from the cycle traces of $\widetilde{T}^{\otimes 4}$,
\begin{equation}
\begin{aligned}
p_{[1,1,1,1]}(\widetilde{T})&=t_1^4,
&
p_{[2,1,1]}(\widetilde{T})&=t_1^2t_2,
&
p_{[2,2]}(\widetilde{T})&=t_2^2,
&
p_{[3,1]}(\widetilde{T})&=t_1t_3,
&
p_{[4]}(\widetilde{T})&=t_4,
\end{aligned}
\end{equation}
similar to Eq.~\eqref{eq:phat_values}. It therefore remains to compute the $Q$-sector coefficients $a_\lambda$.

Write
\begin{equation}
\widetilde{T}=\bI-A\otimes\Pi_0,
\qquad
A=\frac{\bI+X-Y+Z}{1+\sqrt3},
\qquad
\Pi_0=|0^{n-1}\rangle\langle 0^{n-1}|.
\end{equation}
Let
\begin{equation}
P=p\otimes P_{\mathrm{tail}},
\qquad
p\in\{\bI,X,Y,Z\},
\end{equation}
where $P_{\mathrm{tail}}$ is a Hermitian Pauli string acting on the last $n-1$ qubits. Since
\begin{equation}
\Pi_0 P_{\mathrm{tail}} \Pi_0
=
\langle 0^{n-1}|P_{\mathrm{tail}}|0^{n-1}\rangle\Pi_0,
\end{equation}
we divide $P_{\mathrm{tail}}$ into three cases:
\begin{itemize}
\item \emph{Class (i):} $P_{\mathrm{tail}}=\bI^{\otimes(n-1)}$. This class contains one element.
\item \emph{Class (ii):} $P_{\mathrm{tail}}\in\{\bI,Z\}^{\otimes(n-1)}\setminus\{\bI^{\otimes(n-1)}\}$. This class contains $d/2-1$ elements.
\item \emph{Class (iii):} all remaining Pauli strings $P_{\mathrm{tail}}$. This class contains $d(d-2)/4$ elements.
\end{itemize}

For each class, we compute the terms
\begin{equation}
\tau_k(P)=\Tr\left((\widetilde{T} P)^k\right),
\qquad k=1,2,3,4,
\end{equation}
that appear in $\widehat{q}_\lambda(\widetilde{T})$.

\medskip
\noindent
\emph{Step 1: Class (i).}
If $P_{\mathrm{tail}}=\bI^{\otimes(n-1)}$, then
\begin{equation}
\widetilde{T} P
=
p\otimes(\bI-\Pi_0)+(p-Ap)\otimes\Pi_0.
\end{equation}
Because $(\bI-\Pi_0)\Pi_0=\Pi_0(\bI-\Pi_0)=0$, the two summands act on orthogonal sectors and therefore
\begin{equation}
(\widetilde{T} P)^k
=
p^k\otimes(\bI-\Pi_0)+(p-Ap)^k\otimes\Pi_0.
\end{equation}
Hence
\begin{equation}\label{eq:tau_k(P)_class_1}
\tau_k(P)
=
\Tr(p^k)\Tr(\bI-\Pi_0)+\Tr\left((p-Ap)^k\right)\Tr(\Pi_0).
\end{equation}
Now
\begin{equation}
\Tr(\bI-\Pi_0)=\frac d2-1,
\qquad
\Tr(\Pi_0)=1.
\end{equation}
Moreover, a direct one-qubit calculation (cf.\ Eq.~\eqref{eq:def_t_k}) gives
\begin{equation}\label{eq:1-qubit_trace_I}
\Tr\left((\bI-A)^k\right)=\left.t_k\right|_{d=2}=\left(3-\sqrt{3}\right)^k,
\end{equation}
\begin{equation}\label{eq:1-qubit_trace_XZ}
\begin{aligned}
\Tr(X-AX)&=\Tr(Z-AZ)=1-\sqrt3,
\\
\Tr\left((X-AX)^2\right)&=\Tr\left((Z-AZ)^2\right)=4-2\sqrt3,
\\
\Tr\left((X-AX)^3\right)&=\Tr\left((Z-AZ)^3\right)=10-6\sqrt3,
\\
\Tr\left((X-AX)^4\right)&=\Tr\left((Z-AZ)^4\right)=28-16\sqrt3,
\end{aligned}
\end{equation}
and
\begin{equation}\label{eq:1-qubit_trace_Y}
\begin{aligned}
\Tr(Y-AY)&=\sqrt3-1,
\\
\Tr\left((Y-AY)^2\right)&=4-2\sqrt3,
\\
\Tr\left((Y-AY)^3\right)&=-10+6\sqrt3,
\\
\Tr\left((Y-AY)^4\right)&=28-16\sqrt3.
\end{aligned}
\end{equation}
Inserting above equations to Eq.~\eqref{eq:tau_k(P)_class_1}, we obtain
\begin{equation}\label{eq:tau_P_all_class_1}
\begin{aligned}
(\tau_1(P),\tau_2(P),\tau_3(P),\tau_4(P))
&=(t_1,t_2,t_3,t_4),
&& p=\bI,
\\
(\tau_1(P),\tau_2(P),\tau_3(P),\tau_4(P))
&=(1-\sqrt3,\ d+2-2\sqrt3,\ 10-6\sqrt3,\ d+26-16\sqrt3),
&& p=X,Z,
\\
(\tau_1(P),\tau_2(P),\tau_3(P),\tau_4(P))
&=(\sqrt3-1,\ d+2-2\sqrt3,\ -10+6\sqrt3,\ d+26-16\sqrt3),
&& p=Y.
\end{aligned}
\end{equation}

\medskip
\noindent
\emph{Step 2: Class (ii).}
Now let
\begin{equation}
P_{\mathrm{tail}}\in\{\bI,Z\}^{\otimes(n-1)}\setminus\{\bI^{\otimes(n-1)}\}.
\end{equation}
Then $P_{\mathrm{tail}}$ is diagonal in the computational basis, it has eigenvalue $+1$ on $|0^{n-1}\rangle$, and it is traceless. Hence
\begin{equation}
P_{\mathrm{tail}}\Pi_0=\Pi_0P_{\mathrm{tail}}=\Pi_0,
\qquad
\Tr(P_{\mathrm{tail}})=0.
\end{equation}
Therefore
\begin{equation}
\widetilde{T} P
=
p\otimes(P_{\mathrm{tail}}-\Pi_0)+(p-Ap)\otimes\Pi_0.
\end{equation}
Again, the two tensor factors act on orthogonal sectors because
\begin{equation}
(P_{\mathrm{tail}}-\Pi_0)\Pi_0=\Pi_0(P_{\mathrm{tail}}-\Pi_0)=0.
\end{equation}
Thus
\begin{equation}
(\widetilde{T} P)^k
=
p^k\otimes(P_{\mathrm{tail}}-\Pi_0)^k+(p-Ap)^k\otimes\Pi_0
\end{equation}
and
\begin{equation}\label{eq:tau_k(P)_class_2}
\tau_k(P)
=
\Tr(p^k)\Tr\left((P_{\mathrm{tail}}-\Pi_0)^k\right)+\Tr\left((p-Ap)^k\right)\Tr(\Pi_0).
\end{equation}
Since $P_{\mathrm{tail}}^2=\bI$ and $P_{\mathrm{tail}}\Pi_0=\Pi_0$, we have
\begin{equation}
(P_{\mathrm{tail}}-\Pi_0)^2=\bI-\Pi_0,
\end{equation}
and since $\Tr(P_{\mathrm{tail}})=0$,
\begin{equation}
\Tr(P_{\mathrm{tail}}-\Pi_0)=-1.
\end{equation}
Consequently,
\begin{equation}
\Tr\left((P_{\mathrm{tail}}-\Pi_0)^{2m}\right)=\frac d2-1,
\qquad
\Tr\left((P_{\mathrm{tail}}-\Pi_0)^{2m+1}\right)=-1.
\end{equation}
Inserting Eqs.~\eqref{eq:1-qubit_trace_I}-\eqref{eq:1-qubit_trace_Y} to Eq.~\eqref{eq:tau_k(P)_class_2}, we have
\begin{equation}\label{eq:tau_P_all_class_2}
\begin{aligned}
(\tau_1(P),\tau_2(P),\tau_3(P),\tau_4(P))
&=(1-\sqrt3,\ t_2,\ 52-30\sqrt3,\ t_4),
&& p=\bI,
\\
(\tau_1(P),\tau_2(P),\tau_3(P),\tau_4(P))
&=(1-\sqrt3,\ d+2-2\sqrt3,\ 10-6\sqrt3,\ d+26-16\sqrt3),
&& p=X,Z,
\\
(\tau_1(P),\tau_2(P),\tau_3(P),\tau_4(P))
&=(\sqrt3-1,\ d+2-2\sqrt3,\ -10+6\sqrt3,\ d+26-16\sqrt3),
&& p=Y.
\end{aligned}
\end{equation}

\medskip
\noindent
\emph{Step 3: Class (iii).}
Suppose now that $P_{\mathrm{tail}}$ contains at least one $X$ or $Y$ factor. Then
\begin{equation}
\Pi_0 P_{\mathrm{tail}} \Pi_0 = 0.
\end{equation}
Set
\begin{equation}
|\beta\rangle=P_{\mathrm{tail}}|0^{n-1}\rangle.
\end{equation}
Since $P_{\mathrm{tail}}$ contains at least one $X$ or $Y$ factor, the computational-basis vector $|0^{n-1}\rangle$ is mapped to a basis vector orthogonal to itself, and therefore
\begin{equation}
\langle 0^{n-1}|\beta\rangle=0.
\end{equation}
Hence the $(n-1)$-qubit Hilbert space decomposes as
\begin{equation}
\mathcal H_{\mathrm{tail}}
=
\mathrm{span}\{|0^{n-1}\rangle,|\beta\rangle\}\oplus \mathcal K,
\end{equation}
with
\begin{equation}
\dim \mathcal K = \frac d2-2.
\end{equation}
Because $P_{\mathrm{tail}}$ is Hermitian and unitary, it preserves this decomposition: indeed,
\begin{equation}
P_{\mathrm{tail}}|0^{n-1}\rangle=|\beta\rangle,
\qquad
P_{\mathrm{tail}}|\beta\rangle=P_{\mathrm{tail}}^2|0^{n-1}\rangle=|0^{n-1}\rangle,
\end{equation}
so $\mathrm{span}\{|0^{n-1}\rangle,|\beta\rangle\}$ is invariant, and for any $|v\rangle\in\mathcal K$ one has
\begin{equation}
\langle 0^{n-1}|P_{\mathrm{tail}}|v\rangle
=
\langle P_{\mathrm{tail}}0^{n-1}|v\rangle
=
\langle \beta|v\rangle
=
0,
\end{equation}
and similarly
\begin{equation}
\langle \beta|P_{\mathrm{tail}}|v\rangle
=
\langle P_{\mathrm{tail}}\beta|v\rangle
=
\langle 0^{n-1}|v\rangle
=
0.
\end{equation}
Thus $P_{\mathrm{tail}}\mathcal K\subseteq \mathcal K$.

Now consider the full $n$-qubit space
\begin{equation}
\mathcal H=\mathbb C^2\otimes\mathcal H_{\mathrm{tail}}
=
\left(\mathbb C^2\otimes \mathrm{span}\{|0^{n-1}\rangle,|\beta\rangle\}\right)
\oplus
(\mathbb C^2\otimes \mathcal K).
\end{equation}
On the second summand, the projector $\Pi_0=|0^{n-1}\rangle\langle 0^{n-1}|$ vanishes identically, because
\begin{equation}
\Pi_0|v\rangle=0,
\qquad
\forall|v\rangle\in\mathcal K.
\end{equation}
Since $P_{\mathrm{tail}}\mathcal K\subseteq\mathcal K$, it follows that
\begin{equation}
\Pi_0P_{\mathrm{tail}}|v\rangle=0,
\qquad
\forall|v\rangle\in\mathcal K.
\end{equation}
Therefore, on $\mathbb C^2\otimes\mathcal K$,
\begin{equation}
\widetilde{T} P
=
(\bI-A\otimes\Pi_0)(p\otimes P_{\mathrm{tail}})
=
p\otimes P_{\mathrm{tail}},
\end{equation}
because the term $Ap\otimes \Pi_0P_{\mathrm{tail}}$ is zero there.

The dimension of this subspace is
\begin{equation}
\dim(\mathbb C^2\otimes\mathcal K)
=
2\left(\frac d2-2\right)
=
d-4.
\end{equation}
Moreover,
\begin{equation}
(p\otimes P_{\mathrm{tail}})^2
=
p^2\otimes P_{\mathrm{tail}}^2
=
\bI,
\end{equation}
so every even power acts as the identity on $\mathbb C^2\otimes\mathcal K$. Hence its contribution to the even traces is exactly
\begin{equation}\label{eq:tau_P_2_even}
\Tr_{\mathbb C^2\otimes\mathcal K}\left((p\otimes P_{\mathrm{tail}})^{2m}\right)=d-4.
\end{equation}
For odd powers, one has
\begin{equation}
(p\otimes P_{\mathrm{tail}})^{2m+1}=p\otimes P_{\mathrm{tail}},
\end{equation}
whose trace on $\mathbb C^2\otimes\mathcal K$ vanishes. Indeed, if $p\in\{X,Y,Z\}$ then $\Tr(p)=0$, while if $p=\bI$, the operator $P_{\mathrm{tail}}$ is traceless on the full tail space and also traceless on $\mathrm{span}\{|0^{n-1}\rangle,|\beta\rangle\}$, hence it is traceless on $\mathcal K$ as well. Therefore
\begin{equation}\label{eq:tau_P_2_odd}
\Tr_{\mathbb C^2\otimes\mathcal K}\left((p\otimes P_{\mathrm{tail}})^{2m+1}\right)=0.
\end{equation}
In summary, the orthogonal-complement sector contributes $d-4$ to every even trace and $0$ to every odd trace.

On the two-dimensional subspace $\mathrm{span}\{|0^{n-1}\rangle,|\beta\rangle\}$, with respect to the ordered basis $\{|0^{n-1}\rangle,|\beta\rangle\}$,
we have
\begin{equation}
p\otimes P_{\mathrm{tail}}
\sim
\begin{bmatrix}
0 & p\\
p & 0
\end{bmatrix},
\qquad
Ap\otimes \Pi_0 P_{\mathrm{tail}}
\sim
\begin{bmatrix}
0 & Ap\\
0 & 0
\end{bmatrix}.
\end{equation}
Hence
\begin{equation}
\begin{aligned}
\widetilde{T} P
&\sim
\begin{bmatrix}
0 & p-Ap\\
p & 0
\end{bmatrix},\\
(\widetilde{T} P)^2
&\sim
\begin{bmatrix}
\bI-A & 0\\
0 & \bI-pAp
\end{bmatrix},\\
(\widetilde{T} P)^3
&\sim
\begin{bmatrix}
0 & (p-Ap)(\bI-pAp)\\
p(\bI-A) & 0
\end{bmatrix},\\
(\widetilde{T} P)^4
&\sim
\begin{bmatrix}
(\bI-A)^2 & 0\\
0 & (\bI-pAp)^2
\end{bmatrix}.
\end{aligned}
\end{equation}
Since $p^2=\bI$, the matrices $A$ and $pAp$ are similar, and therefore
\begin{equation}
\Tr\left((\bI-pAp)^m\right)=\Tr\left((\bI-A)^m\right)
\qquad \textup{for all } m\ge 1.
\end{equation}
Using Eq.~\eqref{eq:tau_P_2_odd}, it follows that the odd traces vanish,
\begin{equation}
\tau_1(P)=\tau_3(P)=0,
\end{equation}
while the even traces (using Eq.~\eqref{eq:tau_P_2_even}) are
\begin{equation}
\begin{aligned}
\tau_2(P)
=(d-4)+2\Tr(\bI-A)
=d+2-2\sqrt3,
\end{aligned}
\end{equation}
and
\begin{equation}
\begin{aligned}
\tau_4(P)
=(d-4)+2\Tr\left((\bI-A)^2\right)
=d+20-12\sqrt3.
\end{aligned}
\end{equation}
Thus
\begin{equation}\label{eq:tau_P_all_class_3}
(\tau_1(P),\tau_2(P),\tau_3(P),\tau_4(P))
=
(0,\ d+2-2\sqrt3,\ 0,\ d+20-12\sqrt3)
\end{equation}
for all $p\in\{\bI,X,Y,Z\}$.

\medskip
\noindent
\emph{Step 4: Evaluation of the five $Q$-sector invariants.}
We now substitute the above traces into Eq.~\eqref{eq:q_lambda_A}.

For $q_{[1,1,1,1]}(\widetilde{T})$, only Classes~(i) (Eq.~\eqref{eq:tau_P_all_class_1}) and (ii) (Eq.~\eqref{eq:tau_P_all_class_2}) contribute:
\begin{equation}
\begin{aligned}
q_{[1,1,1,1]}(\widetilde{T})
&=
\frac{1}{d^2}
\left[
t_1^4
+
3(1-\sqrt3)^4
+
4\left(\frac d2-1\right)(1-\sqrt3)^4
\right]
\\&=
d^2+(4-4\sqrt3)d+(24-12\sqrt3)+\frac{96-56\sqrt3}{d}.
\end{aligned}
\end{equation}

Next,
\begin{equation}
\begin{aligned}
q_{[2,2]}(\widetilde{T})
&=
\frac{1}{d^2}
\left[
t_2^2
+
3(d+2-2\sqrt3)^2
+\left(\frac d2-1\right)\left(t_2^2+3(d+2-2\sqrt3)^2\right)
+d(d-2)(d+2-2\sqrt3)^2
\right]
\\
&=
d^2+(4-4\sqrt3)d+(24-12\sqrt3)+\frac{96-56\sqrt3}{d}.
\end{aligned}
\end{equation}
Hence
\begin{equation}
q_{[1,1,1,1]}(\widetilde{T})=q_{[2,2]}(\widetilde{T}).
\end{equation}

For $q_{[2,1,1]}(\widetilde{T})$, only Classes~(i) and (ii) contribute:
\begin{equation}
\begin{aligned}
q_{[2,1,1]}(\widetilde{T})
&=
\frac{1}{d^2}
\left[
t_2 t_1^2
+
3(d+2-2\sqrt3)(1-\sqrt3)^2
+\left(\frac d2-1\right)
\left(
t_2(1-\sqrt3)^2
+
3(d+2-2\sqrt3)(1-\sqrt3)^2
\right)
\right]
\\
&=
d+(20-12\sqrt3)+\frac{124-72\sqrt3}{d}.
\end{aligned}
\end{equation}

Similarly,
\begin{equation}
\begin{aligned}
q_{[4]}(\widetilde{T})
&=
\frac{1}{d^2}
\left[
t_4
+
3(d+26-16\sqrt3)
+\left(\frac d2-1\right)\left(t_4+3(d+26-16\sqrt3)\right)
+d(d-2)(d+20-12\sqrt3)
\right]
\\
&=
d+(20-12\sqrt3)+\frac{124-72\sqrt3}{d}.
\end{aligned}
\end{equation}
Hence
\begin{equation}
q_{[2,1,1]}(\widetilde{T})=q_{[4]}(\widetilde{T}).
\end{equation}

Finally,
\begin{equation}
\begin{aligned}
q_{[3,1]}(\widetilde{T})
&=
\frac{1}{d^2}
\left[
t_3 t_1
+
3(10-6\sqrt3)(1-\sqrt3)
+\left(\frac d2-1\right)
\left(
(52-30\sqrt3)(1-\sqrt3)
+
3(10-6\sqrt3)(1-\sqrt3)
\right)
\right]
\\&=
1+\frac{166-96\sqrt3}{d}.
\end{aligned}
\end{equation}

\medskip
\noindent
\emph{Step 5: Computation of $a_\lambda$.}
Substituting the above formulas into Eq.~\eqref{eq:pqhat-def-inline}, and using
\begin{equation}
q_{[1,1,1,1]}(\widetilde{T})=q_{[2,2]}(\widetilde{T}),
\qquad
q_{[2,1,1]}(\widetilde{T})=q_{[4]}(\widetilde{T}),
\end{equation}
we obtain (omitting $\widetilde{T}$ for simplicity)
\begin{equation}
\begin{aligned}
a_{[4]}
&=
\frac{q_{[1,1,1,1]}+6q_{[2,1,1]}+3q_{[2,2]}+8q_{[3,1]}+6q_{[4]}}{24}
\\
&=
\frac{q_{[1,1,1,1]}+3q_{[2,1,1]}+2q_{[3,1]}}{6},
\\
a_{[2,2]}
&=
\frac{2q_{[1,1,1,1]}+6q_{[2,2]}-8q_{[3,1]}}{24}
=
\frac{q_{[1,1,1,1]}-q_{[3,1]}}{3},
\\
a_{[1,1,1,1]}
&=
\frac{q_{[1,1,1,1]}-6q_{[2,1,1]}+3q_{[2,2]}+8q_{[3,1]}-6q_{[4]}}{24}
\\
&=
\frac{q_{[1,1,1,1]}-3q_{[2,1,1]}+2q_{[3,1]}}{6},
\\
a_{[3,1]}
&=
\frac{3q_{[1,1,1,1]}+6q_{[2,1,1]}-3q_{[2,2]}-6q_{[4]}}{24}
=0,
\\
a_{[2,1,1]}
&=
\frac{3q_{[1,1,1,1]}-6q_{[2,1,1]}-3q_{[2,2]}+6q_{[4]}}{24}
=0.
\end{aligned}
\end{equation}
Substituting the explicit expressions for $q_{[1,1,1,1]}(\widetilde{T})$, $q_{[2,1,1]}(\widetilde{T})$, and $q_{[3,1]}(\widetilde{T})$ gives exactly Eq.~\eqref{eq:a_values}.
\end{proof}

As a special case, when the input state is pure, the coefficients and the character transforms (Eq.~\eqref{eq:phat_values} and Eq.~\eqref{eq:qhat_values}) largely simplify. As shown in the following corollary (reformulated from Corollary~\ref{cor:pure-C4_magic} in the main text), $\norm{\rho}_{\widetilde{T},\Cl_n,4}^4$ becomes a function linear in the Pauli fourth moment of $\rho$, which relates directly to the stabilizer 2-R\'{e}nyi entropy~\cite{Leone2022sre}. We also provide an explicit formula for the stabilizer threshold $C_{\widetilde{T},\Cl_n,4}$.
\begin{corollary}[Pure-state reduction and the stabilizer threshold]\label{cor:pure-C4-inline}
Let $\psi=|\psi\rangle\langle\psi|$ be pure and define the Pauli fourth moment
\begin{equation}
\mathfrak{P}_2(\psi)=\frac{1}{d}\sum_{P\in\mathcal P_n}|\langle\psi|P|\psi\rangle|^4.
\end{equation}
Then
\begin{equation}
\begin{aligned}
\norm{\psi}_{\widetilde{T},\Cl_n,4}^4
&=\frac{a_{[4]}}{D_{[4]}^+}\frac{\mathfrak{P}_2(\psi)}{d}+\frac{b_{[4]}}{D_{[4]}^-}\left(1-\frac{\mathfrak{P}_2(\psi)}{d}\right)
\\&=\frac{24(7-4\sqrt{3})(d-8)}{d(d-1)(d+1)(d+2)(d+4)}\mathfrak{P}_2(\psi)+C_4,
\end{aligned}
\label{eq:R4-pure-inline}
\end{equation}
where
\begin{equation}\label{eq:c_4}
C_4=\frac{
-3200 + 1856 \sqrt{3}
+ 3400 d - 1968 \sqrt{3} d
+ 862 d^2 - 496 \sqrt{3} d^2
+ 103 d^3 - 60 \sqrt{3} d^3
+ 10 d^4 - 4 \sqrt{3} d^4
+ d^5
}{
d (d-1) (d+1) (d+2) (d+4)
}
\end{equation}
is a constant.
In particular, every pure stabilizer state $\phi$ satisfies $\mathfrak{P}_2(\phi)=1$. Setting
\begin{equation}
\Delta_4=d(d+1)(d+2)(d+4),
\end{equation}
one obtains
\begin{equation}
\begin{aligned}
C_{\widetilde{T},\Cl_n,4}^4&=
\frac{a_{[4]}}{dD_{[4]}^+}
+
\frac{(1-1/d)b_{[4]}}{D_{[4]}^-}\\
&=\frac{1}{\Delta_4}\left[
 d^4+(11-4\sqrt3)d^3+(114-64\sqrt3)d^2+(976-560\sqrt3)d+4544-2624\sqrt3
\right].
\end{aligned}
\label{eq:C4-formula-inline}
\end{equation}
\end{corollary}

\begin{proof}
Inserting $r_2=r_3=r_4=1$ to Eq.~\eqref{eq:phat_values}, we have
\begin{equation}\label{eq:hat_p_pure_state}
\begin{aligned}
\widehat p_{[4]}(\psi)&=1,\\
\widehat p_\lambda(\psi)&=0,\quad \textup{for }\lambda\neq[4].
\end{aligned}
\end{equation}

Now set $m_P=\langle\psi|P|\psi\rangle$. Since $\psi P\psi=m_P\psi$, all five $Q$-sector invariants (cf.Eq.~\eqref{eq:q_values}) reduce to the same value,
\begin{equation}
\begin{aligned}
q_{[1,1,1,1]}(\psi)&=q_{[2,1,1]}(\psi)=q_{[2,2]}(\psi)\\
&=q_{[3,1]}(\psi)=q_{[4]}(\psi)\\
&=\frac{1}{d^2}\sum_{P\in\mathcal P_n}m_P^4=\frac{\mathfrak{P}_2(\psi)}{d}.
\end{aligned}
\end{equation}
Applying Eq.~\eqref{eq:qhat_values} and using the character orthogonality relations gives
\begin{equation}\label{eq:hat_q_pure_state}
\begin{aligned}
\widehat q_{[4]}(\psi)&=\frac{\mathfrak{P}_2(\psi)}{d},\\
\widehat q_\lambda(\psi)&=0\quad \textup{for }\lambda\neq[4].
\end{aligned}
\end{equation}
Substituting these identities into Theorem~\ref{thm:R4-general-inline} proves Eq.~\eqref{eq:R4-pure-inline}.

If $\phi$ is a pure stabilizer state, then exactly $d$ Hermitian Paulis have 
expectation value $1$ on $\phi$, while all others have expectation value $0$, hence 
$\mathfrak{P}_2(\phi)=1$. Since the maximum of WE $4$-norm over the stabilizer 
polytope 
is attained at an arbitrary pure stabilizer state by Observation~\ref{obs:threshold_single_orbit},
\begin{equation}
C_{\widetilde{T},\Cl_n,4}^4=\norm{\phi}_{\widetilde{T},\Cl_n,4}^4=
\frac{a_{[4]}}{dD_{[4]}^+}
+
\frac{(1-1/d)b_{[4]}}{D_{[4]}^-}.
\end{equation}
Using the results in Lemma~\ref{lem:coefficients-inline}, we obtain Eq.~\eqref{eq:C4-formula-inline}.
\end{proof}

\subsection{Proof of Corollary~\ref{cor:projector_seed_magic}}\label{app:qubit_magic_Phi_coefficients}
We then consider the WE $4$-criterion constructed by a seed operator chosen as the rank-$1$ projector onto a 
pure $n$-qubit state vector $\ket{\Phi}$ with high magic. We first prove the existence of such a state and 
then calculate the coefficients introduced in Theorem~\ref{thm:R4_magic} and 
Theorem~\ref{thm:R4-general-inline}.
\begin{proof}[Proof of Corollary~\ref{cor:projector_seed_magic}]
Denote the fifth root of unity as $\omega_5=e^{i2\pi/5}$ and for $\mathbf{m}=(m_1,\ldots,m_{d-1})\in\{0,1,2,3,4\}^{d-1}$ define the normalized $n$-qubit pure state
\begin{equation}
\ket{\Phi_{n,\mathbf{m}}}
=d^{-1/4}\ket{0^n}
+\frac{1}{\sqrt{d+\sqrt{d}}}\sum_{x=1}^{d-1}\omega_5^{m_x}\ket{x}.
\label{eq:finite-phase-state}
\end{equation}
We then show the average evaluation of $\mathfrak{P}_2$ of states in this ensemble $\left\{\Phi_{n,\mathbf{m}}\middle\mid \mathbf{m}\in\{0,1,2,3,4\}^{d-1}\right\}$ is strictly smaller the $4/(d+3)$. So that there must exist some $\mathbf{m}$, such that
\begin{equation}
\mathfrak{P}_2(\Phi_{n,\mathbf{m}})<\frac{4}{d+3}.
\label{eq:strict-seed}
\end{equation}
Choose each entry $m_x$ of $\mathbf{m}=(m_1,\ldots,m_{d-1})$ independently and uniformly from $\{0,1,2,3,4\}$. We parametrize an $n$-qubit Pauli operator ignoring phase as $X^aZ^b$ for $a,b\in\mathbb Z_2^n$, then
\begin{equation}
\mathfrak{P}_2(\Phi_{n,\mathbf{m}})
=\frac1d\sum_{a,b\in\mathbb F_2^n}\abs{\bra{\Phi_{n,\mathbf{m}}}X^aZ^b\ket{\Phi_{n,\mathbf{m}}}}^4.
\label{eq:frakP2-ab}
\end{equation}
\begin{itemize}
\item
If $a=0$, then $\bra{\Phi_{n,\mathbf{m}}}Z^0\ket{\Phi_{n,\mathbf{m}}}=1$. For $b\ne0$, using
$\sum_{x\ne0}(-1)^{b\cdot x}=-1$, we have
\begin{equation}
\bra{\Phi_{n,\mathbf{m}}}Z^b\ket{\Phi_{n,\mathbf{m}}}
=d^{-1/2}-\frac{1}{d+\sqrt{d}}
=\frac1{\sqrt d+1}.
\label{eq:diag-contribution}
\end{equation}
Hence the total $a=0$ contribution to the sum in Eq.~\eqref{eq:frakP2-ab} is
\begin{equation}
1+\frac{d-1}{(\sqrt d+1)^4}.
\label{eq:diag-total}
\end{equation}
\item
Now fix $a\ne0$ and notice that
\begin{equation}\label{eq:pair-sum-contributions}
\bra{\Phi_{n,\mathbf{m}}}X^aZ^b\ket{\Phi_{n,\mathbf{m}}}=\sum_{x\in\mathbb{Z}_2^n}(-1)^{b\cdot x}\langle\Phi_{n,\mathbf{m}}|x+a\rangle\langle x|\Phi_{n,\mathbf{m}}\rangle=\sum_{e\in\mathcal{M}_a}\Lambda_e,
\end{equation}
where we pair the computational basis strings as unordered pairs
$e=\{x,x+a\}$ and denote the collection of all such pairs as $\mathcal{M}_a$. Then the contribution of $e$ in the summation has the form
\begin{equation}
\begin{aligned}
\Lambda_e&=(-1)^{b\cdot x}\langle\Phi_{n,\mathbf{m}}|x+a\rangle\langle x|\Phi_{n,\mathbf{m}}\rangle
+(-1)^{b\cdot(x+a)}
\langle \Phi_{n,\mathbf{m}}|x\rangle\langle x+a|\Phi_{n,\mathbf{m}}\rangle
\\&=(-1)^{b\cdot x}\times\abs{\langle x|\Phi_{n,\mathbf{m}}\rangle}\times\abs{\langle x+a|\Phi_{n,\mathbf{m}}\rangle}\times
\left(\omega_5^{m_x-m_{x+a}}+(-1)^{b\cdot a}\omega_5^{m_{x+a}-m_x}\right),
\label{eq:pair-contribution}
\end{aligned}
\end{equation}
where we set $m_0=0$ and use
$\langle y|\Phi_{n,\mathbf{m}}\rangle=\abs{\langle y|\Phi_{n,\mathbf{m}}\rangle}\omega_5^{m_y}$.
For each fixed $e=\{x,x+a\}$, the phase difference $m_x-m_{x+a}$ is uniformly distributed in $\mathbb Z_5$. Using the fifth-root identity $\sum_{r=0}^4\omega_5^{\ell r}=0$ for $\ell\neq0$,
we have
\begin{equation}
\mathbb{E}_{\mathbf{m}}\left[\omega_5^{\ell(m_x-m_{x+a})}\right]=0,
\qquad
\mathbb{E}_{\mathbf{m}}\left[\omega_5^{\ell(m_{x+a}-m_x)}\right]=0,
\qquad \ell=1,2,3,4.
\label{eq:phase-difference-cancellations}
\end{equation}
We now compute the one-pair moments explicitly.  First,
\begin{equation}\label{eq:pair-zero-mean}
\mathbb{E}_{\mathbf{m}}\left[\Lambda_e\right]
=(-1)^{b\cdot x}\times\abs{\langle x|\Phi_{n,\mathbf{m}}\rangle}\times\abs{\langle x+a|\Phi_{n,\mathbf{m}}\rangle}\times\left(
\mathbb{E}_{\mathbf{m}}\left[\omega_5^{m_x-m_{x+a}}\right]
+(-1)^{b\cdot a}\mathbb{E}_{\mathbf{m}}\left[\omega_5^{m_{x+a}-m_x}\right]
\right)=0.
\end{equation}
For the second absolute moment,
\begin{equation}\label{eq:pair-second-absolute-expanded}
\begin{aligned}
\abs{\Lambda_e}^2
&=\abs{\langle x|\Phi_{n,\mathbf{m}}\rangle}^2\abs{\langle x+a|\Phi_{n,\mathbf{m}}\rangle}^2
\abs{
\omega_5^{m_x-m_{x+a}}+(-1)^{b\cdot a}\omega_5^{m_{x+a}-m_x}
}^2\\
&=\abs{\langle x|\Phi_{n,\mathbf{m}}\rangle}^2\abs{\langle x+a|\Phi_{n,\mathbf{m}}\rangle}^2
\left(2+(-1)^{b\cdot a}\omega_5^{2(m_x-m_{x+a})}
+(-1)^{b\cdot a}\omega_5^{2(m_{x+a}-m_x)}\right).
\end{aligned}
\end{equation}
Taking the expectation and using Eq.~\eqref{eq:phase-difference-cancellations} with $\ell=2$ gives
\begin{equation}
\mathbb{E}_{\mathbf{m}}\left[\abs{\Lambda_e}^2\right]
=2\abs{\langle x|\Phi_{n,\mathbf{m}}\rangle}^2\abs{\langle x+a|\Phi_{n,\mathbf{m}}\rangle}^2.
\label{eq:pair-second-absolute}
\end{equation}
For the second complex moment,
\begin{equation}\label{eq:pair-second-complex-expanded}
\begin{aligned}
\Lambda_e^2
&=\abs{\langle x|\Phi_{n,\mathbf{m}}\rangle}^2\abs{\langle x+a|\Phi_{n,\mathbf{m}}\rangle}^2
\left(
\omega_5^{m_x-m_{x+a}}+(-1)^{b\cdot a}\omega_5^{m_{x+a}-m_x}
\right)^2\\
&=\abs{\langle x|\Phi_{n,\mathbf{m}}\rangle}^2\abs{\langle x+a|\Phi_{n,\mathbf{m}}\rangle}^2
\left(
\omega_5^{2(m_x-m_{x+a})}+2(-1)^{b\cdot a}
+\omega_5^{2(m_{x+a}-m_x)}
\right).
\end{aligned}
\end{equation}
Thus
\begin{equation}\label{eq:pair-second-complex}
\mathbb{E}_{\mathbf{m}}\left[\Lambda_e^2\right]
=2(-1)^{b\cdot a}\abs{\langle x|\Phi_{n,\mathbf{m}}\rangle}^2\abs{\langle x+a|\Phi_{n,\mathbf{m}}\rangle}^2.
\end{equation}
Finally, for the fourth absolute moment, by Eq.~\eqref{eq:pair-second-absolute-expanded}:
\begin{equation}\label{eq:pair-fourth-absolute-expanded}
\begin{aligned}
\abs{\Lambda_e}^4
&=\abs{\langle x|\Phi_{n,\mathbf{m}}\rangle}^4\abs{\langle x+a|\Phi_{n,\mathbf{m}}\rangle}^4
\left(2+(-1)^{b\cdot a}\omega_5^{2(m_x-m_{x+a})}
+(-1)^{b\cdot a}\omega_5^{2(m_{x+a}-m_x)}\right)^2\\
&=\abs{\langle x|\Phi_{n,\mathbf{m}}\rangle}^4\abs{\langle x+a|\Phi_{n,\mathbf{m}}\rangle}^4
\left(6
+4(-1)^{b\cdot a}\omega_5^{2(m_x-m_{x+a})}
+4(-1)^{b\cdot a}\omega_5^{2(m_{x+a}-m_x)}
+\omega_5^{4(m_x-m_{x+a})}
+\omega_5^{4(m_{x+a}-m_x)}\right).
\end{aligned}
\end{equation}
The terms with powers $2$ and $4$ vanish after averaging, so
\begin{equation}
\mathbb{E}_{\mathbf{m}}\left[\abs{\Lambda_e}^4\right]
=6\abs{\langle x|\Phi_{n,\mathbf{m}}\rangle}^4\abs{\langle x+a|\Phi_{n,\mathbf{m}}\rangle}^4.
\label{eq:pair-fourth-absolute}
\end{equation}
Distinct unordered pairs in $\mathcal M_a$ use disjoint phase variables, hence the corresponding pair contributions are independent. We now expand the fourth absolute moment of Eq.~\eqref{eq:pair-sum-contributions}. Since
\begin{equation}
\abs{\sum_{e\in\mathcal M_a}\Lambda_e}^4
=\left(\sum_{e\in\mathcal M_a}\Lambda_e\right)^2
 \left(\sum_{e\in\mathcal M_a}\overline{\Lambda_e}\right)^2,
\end{equation}
all terms in which some pair contribution appears only once vanish by Eq.~\eqref{eq:pair-zero-mean}. The surviving terms are exactly
\begin{equation}\label{eq:fourth-moment-survivors}
\begin{aligned}
\mathbb{E}_{\mathbf{m}}\left[\abs{\sum_{e\in\mathcal M_a}\Lambda_e}^4\right]
&=
\sum_{e\in\mathcal M_a}\mathbb{E}_{\mathbf{m}}\left[\abs{\Lambda_e}^4\right]
+\sum_{\substack{e,f\in\mathcal M_a\\ e\ne f}}
\mathbb{E}_{\mathbf{m}}\left[\Lambda_e^2\right]\mathbb{E}_{\mathbf{m}}\left[\overline{\Lambda_f}^{2}\right]
+4\sum_{\substack{e,f\in\mathcal M_a\\ e<f}}
\mathbb{E}_{\mathbf{m}}\left[\abs{\Lambda_e}^2\right]\mathbb{E}_{\mathbf{m}}\left[\abs{\Lambda_f}^2\right].
\end{aligned}
\end{equation}
Here $e<f$ denotes an arbitrary fixed ordering of the finite set $\mathcal M_a$ and the final value is independent of that ordering. Substituting Eq.~\eqref{eq:pair-second-absolute}, Eq.~\eqref{eq:pair-second-complex} and Eq.~\eqref{eq:pair-fourth-absolute} into Eq.~\eqref{eq:fourth-moment-survivors} gives
\begin{equation}\label{eq:fourth-expanded-pairs}
\begin{aligned}
\mathbb{E}_{\mathbf{m}}\left[\abs{\bra{\Phi_{n,\mathbf{m}}}X^aZ^b\ket{\Phi_{n,\mathbf{m}}}}^4\right]
=&6\sum_{\{x,x+a\}}\abs{\langle x|\Phi_{n,\mathbf{m}}\rangle}^4\abs{\langle x+a|\Phi_{n,\mathbf{m}}\rangle}^4\\
&+24\sum_{\substack{\{x,x+a\}<\{y,y+a\}}}
\abs{\langle x|\Phi_{n,\mathbf{m}}\rangle}^2\abs{\langle x+a|\Phi_{n,\mathbf{m}}\rangle}^2
\abs{\langle y|\Phi_{n,\mathbf{m}}\rangle}^2\abs{\langle y+a|\Phi_{n,\mathbf{m}}\rangle}^2.
\end{aligned}
\end{equation}
The coefficient $24$ is the sum of $8$ from Eq.~\eqref{eq:pair-second-complex} and $16$ from Eq.~\eqref{eq:pair-second-absolute}. Noticing that
\begin{equation}
\begin{aligned}
&\sum_{\substack{\{x,x+a\}<\{y,y+a\}}}
\abs{\langle x|\Phi_{n,\mathbf{m}}\rangle}^2\abs{\langle x+a|\Phi_{n,\mathbf{m}}\rangle}^2
\abs{\langle y|\Phi_{n,\mathbf{m}}\rangle}^2\abs{\langle y+a|\Phi_{n,\mathbf{m}}\rangle}^2
\\=&\frac12\left(\sum_{\{x,x+a\}}\abs{\langle x|\Phi_{n,\mathbf{m}}\rangle}^2\abs{\langle x+a|\Phi_{n,\mathbf{m}}\rangle}^2\right)^2-\frac12\sum_{\{x,x+a\}}\abs{\langle x|\Phi_{n,\mathbf{m}}\rangle}^4\abs{\langle x+a|\Phi_{n,\mathbf{m}}\rangle}^4
\end{aligned}
\end{equation}
we obtain
\begin{equation}
\begin{aligned}
&\mathbb{E}_{\mathbf{m}}\left[\abs{\bra{\Phi_{n,\mathbf{m}}}X^aZ^b\ket{\Phi_{n,\mathbf{m}}}}^4\right] \\
=&12\left(\sum_{\{x,x+a\}}\abs{\langle x|\Phi_{n,\mathbf{m}}\rangle}^2\abs{\langle x+a|\Phi_{n,\mathbf{m}}\rangle}^2\right)^2
-6\sum_{\{x,x+a\}}\abs{\langle x|\Phi_{n,\mathbf{m}}\rangle}^4\abs{\langle x+a|\Phi_{n,\mathbf{m}}\rangle}^4.
\label{eq:fourth-pair-formula}
\end{aligned}
\end{equation}
For the state vector $\ket{\Phi_{n,\mathbf{m}}}$ in Eq.~\eqref{eq:finite-phase-state}, exactly one unordered pair contains $0$ and $a$, while the other $d/2-1$ pairs contain two nonzero strings. Therefore,
\begin{equation}
\begin{aligned}
\sum_{\{x,x+a\}}\abs{\langle x|\Phi_{n,\mathbf{m}}\rangle}^2\abs{\langle x+a|\Phi_{n,\mathbf{m}}\rangle}^2
&=\frac{1}{d(\sqrt d+1)}+\frac{d/2-1}{d(\sqrt d+1)^2}\\
&=\frac{\sqrt d+2}{2\sqrt d(\sqrt d+1)^2},
\label{eq:first-pair-sum}
\end{aligned}
\end{equation}
while
\begin{equation}
\begin{aligned}
\sum_{\{x,x+a\}}\abs{\langle x|\Phi_{n,\mathbf{m}}\rangle}^4\abs{\langle x+a|\Phi_{n,\mathbf{m}}\rangle}^4
&=\frac{1}{d^2(\sqrt d+1)^2}+\frac{d/2-1}{d^2(\sqrt d+1)^4}
\\
&=\frac{3\sqrt d+4}{2d^{3/2}(\sqrt d+1)^4}.
\label{eq:second-pair-sum}
\end{aligned}
\end{equation}
\end{itemize}
Substituting Eq.~\eqref{eq:first-pair-sum} and Eq.~\eqref{eq:second-pair-sum} into Eq.~\eqref{eq:fourth-pair-formula} yields, for every $a\ne0$ and every $b$,
\begin{equation}
\mathbb{E}_{\mathbf{m}}\left[\abs{\bra{\Phi_{n,\mathbf{m}}}X^aZ^b\ket{\Phi_{n,\mathbf{m}}}}^4\right]
=\frac{3(d^{3/2}+4d+\sqrt d-4)}{d^{3/2}(\sqrt d+1)^4}.
\label{eq:offdiag-moment}
\end{equation}
Combining Eq.~\eqref{eq:diag-total} and Eq.~\eqref{eq:offdiag-moment} in Eq.~\eqref{eq:frakP2-ab} gives
\begin{equation}\label{eq:average-frakP2}
\begin{aligned}
\mathbb{E}_{\mathbf{m}}\left[\mathfrak{P}_2(\Phi_{n,\mathbf{m}})\right]
&=\frac1d\left[1+\frac{d-1}{(\sqrt d+1)^4}
+d(d-1)\frac{3(d^{3/2}+4d+\sqrt d-4)}{d^{3/2}(\sqrt d+1)^4}\right]\\
&=\frac{4d^2+12d^{3/2}-5d-15\sqrt d+12}{d^{3/2}(\sqrt d+1)^3}.
\end{aligned}
\end{equation}
A direct simplification gives the gap
\begin{equation}
\frac{4}{d+3}-\mathbb{E}_{\mathbf{m}}\left[\mathfrak{P}_2(\Phi_{n,\mathbf{m}})\right]
=\frac{(\sqrt d-1)(5d^{3/2}-12d-9\sqrt d+36)}{d^{3/2}(\sqrt d+1)^3(d+3)}>0.
\label{eq:average-gap}
\end{equation}
The denominator is positive and $\sqrt d>1$. If $n=1$, then $\sqrt d=\sqrt2$ and
$5d^{3/2}-12d-9\sqrt d+36=12+\sqrt2>0$. If $n\ge2$, then $\sqrt d\ge2$; the polynomial $5x^3-12x^2-9x+36$ is increasing for $x\ge2$ because its derivative $15x^2-24x-9$ is positive there, and its value at $x=2$ is $10$. Thus the finite average in Eq.~\eqref{eq:average-frakP2} is strictly below $4/(d+3)$. Especially, we could search lexicographically over $\mathbf{m}\in\{0,1,2,3,4\}^{d-1}$ for a concrete $\ket{\Phi_{n,\mathbf{m}}}$ such that $\mathfrak{P}_2(\Phi_{n,\mathbf{m}})<4/(d+3)$.

We then calculate $\widehat p_\lambda(\Phi)$ and $\widehat q_\lambda(\Phi)$ in Eq.~\eqref{eq:R4-general-inline_app}. Notice that $\Phi$ is a pure state, following a similar argument in deriving Eq.~\eqref{eq:hat_p_pure_state} and Eq.~\eqref{eq:hat_q_pure_state}, we have
\begin{equation}
\begin{aligned}
\widehat p_{[4]}(\Phi)&=1,\qquad&\widehat q_{[4]}(\Phi)&=\frac{\mathfrak{P}_2(\Phi)}{d},\\
\widehat p_\lambda(\Phi)&=0,\qquad&\widehat q_\lambda(\Phi)&=0\quad\textup{for }\lambda\neq[4].
\end{aligned}
\end{equation}
According to the dimensions related to $\lambda=[4]$ in Table~\ref{tab:module_dimension}, we recover the WE $4$-norm in Eq.~\eqref{eq:projector_seed_magic_norm} in the main text:
\begin{equation}\label{eq:projector_seed_magic_norm_app}
\norm{\rho}_{\Phi,\Cl_n,4}^4
=
\frac{6\mathfrak{P}_2(\Phi)}{d(d+1)(d+2)}
\widehat q_{[4]}(\rho)
+
\frac{24(d-\mathfrak{P}_2(\Phi))}{d(d-1)(d+1)(d+2)(d+4)}
\left(
\widehat p_{[4]}(\rho)-\widehat q_{[4]}(\rho)
\right).
\end{equation}
And its pure-state reduction is
\begin{equation}
\norm{\psi}_{\Phi,\Cl_n,4}^4
=
\frac{6\mathfrak{P}_2(\Phi)\mathfrak{P}_2(\psi)}{d^2(d+1)(d+2)}
+
\frac{24(d-\mathfrak{P}_2(\Phi))(d-\mathfrak{P}_2(\psi))}{d^2(d-1)(d+1)(d+2)(d+4)}
.
\end{equation}
For a pure stabilizer state $\phi=\ketbra{\phi}{\phi}$, $\mathfrak{P}_2(\phi)=1$, so the threshold value for all stabilizer states is given by
\begin{equation}
C_{\Phi,\Cl_n,4}^4
=\frac{6(\mathfrak{P}_2(\Phi)+4)}{d(d+1)(d+2)(d+4)},
\end{equation}
which is Eq.~\eqref{eq:qubit_magic_projector_threshold}. Then Eq.~\eqref{eq:projector_seed_pure_magic} follows directly from the above two equations. For an $n$-qubit pure state $\psi$, $\mathfrak{P}_2(\psi)\leq1$, with equality if and only if $\psi$ is a stabilizer state~\cite{Leone2022sre}, which proves Eq.~\eqref{eq:qubit_magic_WE_4_projector_pure}.

\end{proof}

\section{Fermionic non-Gaussianity}\label{app:FNG}
\subsection{Proof of Proposition~\ref{thm:wick_odd_majorana_second_order}}\label{app:FNG_Wick_theorem}
In the resource theory of fermionic non-Gaussianity, it is convenient to introduce Majorana operators. They become products of Pauli operators under Jordan-Wigner transformation.
\begin{definition}[Majorana operators]
For an $n$-mode fermionic state, which corresponds to an $n$-qubit state by Jordan-Wigner transformation, the set of $2n$ Majorana operators $\left\{\gamma_1,\gamma_2,\cdots,\gamma_{2n}\right\}$ is defined as $\forall k\in[n]$,
\begin{equation}
\gamma_{2k-1}\coloneqq\left(\prod\limits_{j=1}^{k-1}Z_j\right)X_k,\quad
\gamma_{2k}\coloneqq\left(\prod\limits_{j=1}^{k-1}Z_j\right)Y_k.
\end{equation}    
\end{definition}
For each $j\in[2n]$, the Majorana operator $\gamma_j$ satisfies
\begin{equation}
\gamma_j=\gamma_j^\dagger,\quad \Tr(\gamma_j)=0,\quad \gamma_j^2=\bI.
\end{equation}
It is also easy to verify that for any $j,k\in[2n]$, we have
\begin{equation}\label{eq:majorana_operator_property}
\{\gamma_j,\gamma_k\}=\delta_{j,k}\cdot2\bI,\quad \Tr(\gamma_j^\dagger\gamma_k)=\delta_{j,k}\cdot 2^n.
\end{equation}

\begin{definition}[Majorana products]
For a non-empty subset $S=\left\{\mu_1,\cdots,\mu_{\abs{S}}\right\}$ of $[2n]$, with $1\leq\mu_1<\cdots<\mu_{\abs{S}}\leq2n$, define
\begin{equation}
\gamma_S\coloneqq\prod\limits_{j=1}^{\abs{S}}\gamma_{\mu_j}.
\end{equation}    
Specially, we define $\gamma_{\emptyset}=\bI$. We call $\abs{S}$ as the degree of $\gamma_S$.
\end{definition}

Notice that for subsets $S,S'\subseteq[2n]$, Majorana products are orthogonal 
\begin{equation}\label{eq:Majorana_product_orthogonal}
\Tr(\gamma_S^\dagger\gamma_{S'})=\delta_{S,S'}\cdot 2^n
\end{equation}
w.r.t.\ the 
Hilbert--Schmidt inner product, where the Hermitian conjugate satisfies
\begin{equation}\label{eq:Majorana_product_conjugate}
\gamma_S^\dagger=(-1)^{\abs{S}(\abs{S}-1)/2}\gamma_S.
\end{equation}
Since $\gamma_S$ may not be Hermitian, 
we also denote the Hermitian Majorana product by 
introducing phase factors
\begin{equation}\label{eq:def_Hermitian_Majorana_product}
\widehat{\gamma}_S=\widehat{\gamma}_S^\dagger=(-i)^{\abs{S}(\abs{S}-1)/2}\gamma_S,
\end{equation}
and they satisfy
\begin{equation}\label{eq:Hermitian_Majorana_product_orthogonal}
\Tr(\widehat{\gamma}_S\widehat{\gamma}_{S'})=\delta_{S,S'}\cdot 2^n.
\end{equation}
Furthermore, there are $4^n$ distinct Majorana products, so $\left\{\gamma_S\mid S\subseteq[2n]\right\}$ and $\left\{\widehat{\gamma}_S\mid S\subseteq[2n]\right\}$ are both orthogonal basis for $\mathcal{L}((\mathbb{C}^2)^{\otimes n})$. 

The free unitaries $\Ufree$ are the $n$-qubit free-fermionic (Gaussian) unitaries $\mathrm{M}_n$, which map Majorana operators to a linear combination of Majorana operators, in which the coefficients form an orthogonal matrix $R$. Each Gaussian unitary $U_R$ is specified by $R$. Equivalently, in some literature $\mathrm{M}_n$ is also referred to as those gates realizable by nearest-neighbor matchgate circuits~\cite{jozsa2008MatchgatesClassicalSimulation,wan2023MatchgateShadowsFermionic}.
\begin{definition}[Gaussian unitary]\label{def:Gaussian-unitary}
For an orthogonal matrix $R\in\mathrm{O}(2n)$, a Gaussian unitary $U_R$ is a unitary satisfying $\forall\mu\in[2n]$,
\begin{equation}
U_R^\dagger \gamma_j U_R
    =
    \sum_{k=1}^{2n} R_{j,k}\gamma_k,
\end{equation} 
and the set of $n$-qubit Gaussian unitaries is $\mathrm{M}_n=\left\{U_R\mid R\in\mathrm{O}(2n)\right\}$.
\end{definition}

An operator $A$ is called parity-preserving if it commutes with the parity operator $Z^{\otimes n}$, i.e., $[A,Z^{\otimes n}]=0$. Under the Majorana basis, this is equivalent to say that $A$ is expanded by Majorana products $\gamma_S$ with even degree $\abs{S}$ only. Since parity is superselected, physical fermionic density matrices are parity-preserving, i.e. they commute with the parity operator. A pure parity-preserving state is therefore a parity eigenstate. We denote the set of valid fermionic states as
\begin{equation}\label{eq:def_fermionic_state}
\begin{aligned}
\mathcal{D}_{F}&=\left\{\rho\in\mathcal{D}((\C^2)^{\otimes n})\middle\mid [\rho,Z^{\otimes n}]=0\right\}
\\&=\left\{\rho\in\mathcal{D}((\C^2)^{\otimes n})\middle\mid \Tr(\gamma_S\rho)=0,\forall\abs{S}\textup{ is odd}\right\}.
\end{aligned}
\end{equation}
Given an arbitrary density matrix $\rho\in\mathcal{D}$, we can certify that it is a valid fermionic state by measuring all odd-degree Majorana products (see Proposition~\ref{thm:wick_odd_majorana_second_order}).
\begin{proof}[Proof of Proposition~\ref{thm:wick_odd_majorana_second_order}]
Let $d=2^n$. According to Theorem~1 in Ref.~\cite{wan2023MatchgateShadowsFermionic}, the first- and second- order matchgate twirling formulas are:
\begin{enumerate}
\item for operator $O_1\in\mathcal{L}\left(\mathbb{C}^d\right)$, 
\begin{equation}\label{eq:matchgate_1_twirling}
\Phi_{\mathrm{M}_n}^{(1)}(O_1)=\underset{U\sim\mu_{\mathrm{M}_n}}{\mathbb{E}}\left[{U^\dagger}O_1U\right]=\frac1d\Tr(O_1)\bI,
\end{equation}
\item for operator $O_2\in\mathcal{L}\left(\left(\mathbb{C}^d\right)^{\otimes2}\right)$,
\begin{equation}\label{eq:matchgate_2_twirling}
\begin{aligned}
\Phi_{\mathrm{M}_n}^{(2)}(O_2)&=\underset{U\sim\mu_{\mathrm{M}_n}}{\mathbb{E}}\left[{U^\dagger}^{\otimes2}O_2U^{\otimes2}\right]
\\&=\frac1{d^2}\sum\limits_{k=0}^{2n}\sum\limits_{\substack{A_1\subseteq[2n]\\\abs{A_1}=k}}\sum\limits_{\substack{A_2\subseteq[2n]\\\abs{A_2}=k}}\binom{2n}k^{-1}\Tr(\left(\widehat{\gamma}_{A_2}\otimes\widehat{\gamma}_{A_2}\right)O_2)\widehat{\gamma}_{A_1}^{\otimes2}
\\&=\frac1{d^2}\sum\limits_{k=0}^{2n}\sum\limits_{\substack{A_1\subseteq[2n]\\\abs{A_1}=k}}\sum\limits_{\substack{A_2\subseteq[2n]\\\abs{A_2}=k}}\binom{2n}k^{-1}\Tr(\left(\gamma_{A_2}^\dagger\otimes\gamma_{A_2}^\dagger\right)O_2)\gamma_{A_1}^{\otimes2}.
\end{aligned}
\end{equation}
\end{enumerate}
So we have
\begin{equation}
\begin{aligned}
\norm{\rho}_{\widetilde{\widehat{\gamma}}_S,\mathrm{M}_n,2}^2&=\underset{U\sim\mu_{\mathrm{M}_n}}{\mathbb{E}}\left[\left(1-\Tr(U^\dagger\widehat{\gamma}_SU\rho)\right)^2\right]
\\&=1-2\underset{U\sim\mu_{\mathrm{M}_n}}{\mathbb{E}}\left[\Tr(U^\dagger\widehat{\gamma}_SU\rho)\right]+\underset{U\sim\mu_{\mathrm{M}_n}}{\mathbb{E}}\left[\Tr({U^\dagger}^{\otimes2}\widehat{\gamma}_S^{\otimes2}U^{\otimes2}\rho^{\otimes2})\right]
\\&=1-2\Tr(\Phi_{\mathrm{M}_n}^{(1)}\left(\widehat{\gamma}_S\right)\rho)+
\Tr(\Phi_{\mathrm{M}_n}^{(2)}\left(\widehat{\gamma}_S^{\otimes2}\right)\rho^{\otimes2})
\\&=1+\binom{2n}{\abs{S}}^{-1}\sum\limits_{\substack{T\subseteq[2n]\\\abs{T}=\abs{S}}}\Tr(\widehat{\gamma}_T\rho)^2,
\end{aligned}
\end{equation}
where we have used $\Tr(\widehat{\gamma}_S)=0$ and
\begin{equation}\label{eq:odd_Majorana_twirling}
\begin{aligned}
\Phi_{\mathrm{M}_n}^{(2)}\left(\widehat{\gamma}_S^{\otimes2}\right)=\underset{U\sim\mu_{\mathrm{M}_n}}{\mathbb{E}}\left[{U^\dagger}^{\otimes2}{\widehat{\gamma}_S}^{\otimes2}U^{\otimes2}\right]
&=\frac1{d^2}\sum\limits_{k=0}^{2n}\sum\limits_{\substack{T\subseteq[2n]\\\abs{T}=k}}\sum\limits_{\substack{R\subseteq[2n]\\\abs{R}=k}}\binom{2n}k^{-1}{\widehat{\gamma}_T}^{\otimes2}\Tr(\widehat{\gamma}_R\widehat{\gamma}_S)^2
\\&=\binom{2n}{\abs{S}}^{-1}\sum\limits_{\substack{T\subseteq[2n]\\\abs{T}=\abs{S}}}{\widehat{\gamma}_T}^{\otimes2}.
\end{aligned}
\end{equation}
Notice that $\abs{S}$ is odd, all physical fermionic states $\rho\in\mathcal{D}_F$ have 
${\rm Tr}(\widehat{\gamma}_T\rho)=0$ for $\abs{T}=\abs{S}$, including all fermionic Gaussian states.
For every parity-preserving fermionic state $\rho\in\mathcal D_F$, 
$\norm{\rho}_{\widetilde{\widehat{\gamma}}_S,\mathrm{M}_n,2}=1$. 
Conversely, checking this equality for all odd degrees $|S|=1,3,\dots,2n-1$ 
is equivalent to the absence of odd Majorana components.
\end{proof}

\subsection{Proof of Theorem~\ref{thm:R1-3_FNG_main} and Corollary~\ref{cor:fermion_WE2}}\label{app:FNG_WE_projector_criterion}
The set of $n$-qubit pure Gaussian (free-fermionic) states is those that can be generated from $\ket{0^n}$ using a Gaussian unitary, $\G_n=\{U\ket{0^n}\mid U\in \mathrm{M}_n\}$. Then the set of free states is the convex hull of $\G_n$, $\Fset=\mathrm{Conv}(\G_n)\subseteq\mathcal{D}_F$, also known as the convex Gaussian states~\cite{oszmaniec2014ClassicalSimulationFermionic}. 
It is known that for $n\leq3$, all valid fermionic states are convex Gaussian~\cite{bravyi2005ClassicalCapacityFermionic}. Going beyond this, the four-qubit GHZ state 
\begin{equation}\label{eq:a_8_state}
\ket{\textup{GHZ}_4}=\frac{\ket{0000}+\ket{1111}}{\sqrt{2}}
\end{equation}
is a typical fermionic non-Gaussian state, which plays a role similar to the $T$-state in the magic resource theory that they can be used as an input to realize SWAP gate by a gadget, rendering matchgate circuits classically not efficiently simulable~\cite{hebenstreit2019AllPureFermionic,oszmaniec2022FermionSamplingRobust}. 

We write
\begin{equation}
n=4m+r,\quad\textup{with } r\in\{0,1,2,3\}\textup{ and } m\geq1
\end{equation}
and define the states
\begin{equation}\label{eq:def_eta_r}
\begin{aligned}
\ket{\eta_0}&=\ket{\textup{GHZ}_4},\\
\ket{\eta_1}&=\ket{\textup{GHZ}_4}\otimes\ket{0}=\frac{\ket{00000}+\ket{11110}}{\sqrt2},\\
\ket{\eta_2}&=\ket{\textup{GHZ}_6},\\
\ket{\eta_3}&=\ket{\textup{GHZ}_6}\otimes\ket{0}=\frac{\ket{0000000}+\ket{1111110}}{\sqrt2}.
\end{aligned}
\end{equation}
Set the projectors onto pure fermionic non-Gaussian states for $n=4m+r$,
\begin{equation}\label{eq:def_P_n_r}
P_n=\ket{\eta_r}\bra{\eta_r}\otimes
(\ket{\mathrm{GHZ}_4}\bra{\mathrm{GHZ}_4})^{\otimes(m-1)}.
\end{equation}

Consider the FLO (Gaussian) fidelity, defined as the maximum overlap between a pure state and any pure Gaussian state
\begin{equation}
F_{\mathrm{FLO}}(\ket{\psi})=\max_{\ket{\phi}\in\G_n}\abs{\langle\psi\mid\phi\rangle}^2,
\end{equation}
Lemma~1 in Ref.~\cite{reardon-smith2024FermionicLinearOptical} shows that it is multiplicative for a fixed parity state vector $\ket{\psi}$ and $\ket{\textup{GHZ}_4}$, i.e.,
\begin{equation}\label{eq:FLO_fidelity_multiplicative}
F_{\mathrm{FLO}}(\ket{\psi}\otimes\ket{\textup{GHZ}_4})=F_{\mathrm{FLO}}(\ket{\psi})\times F_{\mathrm{FLO}}(\ket{\textup{GHZ}_4}).
\end{equation}
We then show the FLO fidelity of $n$-qubit GHZ states for even $n\geq4$ is $1/2$.

\begin{lemma}[FLO fidelity of even-qubit GHZ states]\label{lem:overlap-GHZ-free}
For even $n\geq4$, the FLO fidelity of $n$-qubit GHZ state vector $\ket{\textup{GHZ}_n}$ is $1/2$, i.e.,
\begin{equation}
F_{\mathrm{FLO}}(\ket{\textup{GHZ}_n})=\max_{\ket{\phi}\in\G_n}\abs{\langle\phi\mid\textup{GHZ}_n\rangle}^2=\frac12,
\end{equation}
and the $n$-qubit vacuum state vector $\ket{0^n}$ saturates the 
bound.
\end{lemma}

\begin{proof}[Proof of Lemma~\ref{lem:overlap-GHZ-free}]
First notice that for $n$-qubit pure Gaussian state vector 
$\ket{\phi}$,
\begin{equation}\label{eq:overlap-GHZ-free_1}
\abs{\langle\phi\mid\textup{GHZ}_n\rangle}^2=\frac12\abs{\langle\phi\mid 0^n\rangle+\langle\phi\mid1^n\rangle}^2\leq\frac12\left(\abs{\langle\phi\mid0^n\rangle}+\abs{\langle\phi\mid1^n\rangle}\right)^2.
\end{equation}

For $\ket{0^n}$, the only nonzero two-point Majorana expectations are
$\langle0^n|\gamma_{2l-1}\gamma_{2l}|0^n\rangle=i$.
For $|1^n\rangle$, since $Z_l|1^n\rangle=-|1^n\rangle$ and
$\gamma_{2l-1}\gamma_{2l}=iZ_l$, these expectations are $-i$,
while all off-pair expectations remain zero. 
Denote the correlation matrices of $\ket{0^n}$ and $\ket{1^n}$ as $\Gamma_0$ and $\Gamma_1$ respectively, we have $\Gamma_1=-\Gamma_0$. Let $\Gamma_\phi$ be the correlation matrix of $\ket{\phi}$, then the overlaps can be calculated as (see Lemma A7 in Ref.~\cite{bittel2025OptimalTraceDistanceBoundsa} and also Ref.~\cite{bravyi2017ComplexityQuantumImpurity})
\begin{equation}
\begin{aligned}
\abs{\langle\phi\mid0^n\rangle}^2&=\abs{\mathrm{Pf}\left(\frac12\left(\Gamma_\phi+\Gamma_0\right)\right)}
\\&=\left(\det\left(\frac12\left(\Gamma_\phi+\Gamma_0\right)\right)\right)^{\frac12}
\\&=\left(\det\left(\frac12\Gamma_\phi\left(I-\Gamma_\phi\Gamma_0\right)\right)\right)^{\frac12}
\\&=\left(\det\left(\Gamma_\phi\right)\det\left(\frac{I-\Gamma_\phi\Gamma_0}2\right)\right)^{\frac12}
\\&=\left(\det\left(\frac{I-\Gamma_\phi\Gamma_0}2\right)\right)^{\frac12},
\end{aligned}
\end{equation}
where we have used the fact that the correlation matrix of a pure Gaussian state is antisymmetric and orthogonal, and that the determinant of an even-dimensional antisymmetric matrix equals the square of its Pfaffian. Similarly, we have
\begin{equation}
\abs{\langle\phi\mid1^n\rangle}^2=\left(\det\left(\frac{I+\Gamma_\phi\Gamma_0}2\right)\right)^{\frac12},
\end{equation}
by replacing $\Gamma_1$ with $-\Gamma_0$. Notice that $\Gamma_\phi\Gamma_0$ is further a real orthogonal matrix, whose eigenvalues can be parametrized as
\begin{equation}
\Lambda(\Gamma_\phi\Gamma_0)=\{e^{i\theta_1},e^{-i\theta_1},\cdots,e^{i\theta_n},e^{-i\theta_n}\},
\end{equation}
with each $\theta_j\in[0,\pi)$. Then we have
\begin{equation}
\begin{aligned}
\det\left(\frac{I-\Gamma_\phi\Gamma_0}2\right)&=\prod\limits_{\lambda\in\Lambda(\Gamma_\phi\Gamma_0)}\frac{1-\lambda}2
\\&=\prod\limits_{j=1}^n\left(\frac{1-e^{i\theta_j}}2\times\frac{1-e^{-i\theta_j}}2\right)
\\&=\prod\limits_{j=1}^n\frac{1-\cos\theta_j}2
\\&=\prod\limits_{j=1}^n\sin^2\left(\frac{\theta_j}2\right),
\end{aligned}
\end{equation}
and similarly
\begin{equation}
\det\left(\frac{I+\Gamma_\phi\Gamma_0}2\right)=\prod\limits_{j=1}^n\cos^2\left(\frac{\theta_j}2\right).
\end{equation}
Inserting each term in Eq.~\eqref{eq:overlap-GHZ-free_1}, we thus obtain
\begin{equation}
\begin{aligned}
\abs{\langle\phi\mid\textup{GHZ}_n\rangle}^2&\leq\frac12\left(\abs{\langle\phi\mid0^n\rangle}+\abs{\langle\phi\mid1^n\rangle}\right)^2
\\&=\frac12\left(\left(\det\left(\frac{I-\Gamma_\phi\Gamma_0}2\right)\right)^{\frac14}+\left(\det\left(\frac{I+\Gamma_\phi\Gamma_0}2\right)\right)^{\frac14}\right)^2
\\&=\frac12\left(\prod\limits_{j=1}^n\abs{\sin\left(\frac{\theta_j}2\right)}^{\frac12}+\prod\limits_{j=1}^n\abs{\cos\left(\frac{\theta_j}2\right)}^{\frac12}\right)^2.
\end{aligned}
\end{equation}
Since $n$ is even, denote for each $l\in[n/2]$
\begin{equation}
\begin{aligned}
a_l&=\abs{\sin\left(\frac{\theta_{2l-1}}2\right)}^{\frac12}\abs{\sin\left(\frac{\theta_{2l}}2\right)}^{\frac12},\\
b_l&=\abs{\cos\left(\frac{\theta_{2l-1}}2\right)}^{\frac12}\abs{\cos\left(\frac{\theta_{2l}}2\right)}^{\frac12}.
\end{aligned}
\end{equation}
It is worth noting that the absolute value operations in the above definitions could be safely removed, since
\begin{equation}
\sin\left(\frac{\theta_j}2\right),\cos\left(\frac{\theta_j}2\right)\geq0,
\end{equation}
for each parameter $\theta_j\in[0,\pi)$.
We first analyze the case when $n=4$, by Cauchy-Schwarz inequality,
\begin{equation}
\begin{aligned}
\left(\prod\limits_{j=1}^4\abs{\sin\left(\frac{\theta_j}2\right)}^{\frac12}+\prod\limits_{j=1}^4\abs{\cos\left(\frac{\theta_j}2\right)}^{\frac12}\right)^2&=(a_1a_2+b_1b_2)^2
\\&\leq\left(a_1^2+b_1^2\right)\left(a_2^2+b_2^2\right)
\\&=\cos\left(\frac{\theta_1}2-\frac{\theta_2}2\right)\cos\left(\frac{\theta_3}2-\frac{\theta_4}2\right)
\\&\leq1,
\end{aligned}
\end{equation}
and the equality is achieved when all $\theta_j$ take the same value. Consider the following inequality:
\begin{equation}\label{eq:overlap-GHZ-free_2}
\left(\prod\limits_{l=1}^ka_l+\prod\limits_{l=1}^kb_l\right)^2\leq\prod\limits_{l=1}^k\left(a_l^2+b_l^2\right).
\end{equation}
We have proved it for $k=2$, assume this inequality holds for an integer $k$ such that $2\leq k\leq n/2-1$, then
\begin{equation}
\begin{aligned}
\left(\prod\limits_{l=1}^{k+1}a_l+\prod\limits_{l=1}^{k+1}b_l\right)^2&=\left(\left(\prod\limits_{l=1}^ka_l\right)\times a_{k+1}+\left(\prod\limits_{l=1}^kb_l\right)\times b_{k+1}\right)^2
\\&\leq\left(\left(\prod\limits_{l=1}^ka_l\right)^2+\left(\prod\limits_{l=1}^kb_l\right)^2\right)\left(a_{k+1}^2+b_{k+1}^2\right)
\\&\leq\left(\prod\limits_{l=1}^k\left(a_l^2+b_l^2\right)\right)\left(a_{k+1}^2+b_{k+1}^2\right)
\\&=\prod\limits_{l=1}^{k+1}\left(a_l^2+b_l^2\right),
\end{aligned}
\end{equation}
where we again use the Cauchy-Schwarz inequality. By induction, Eq.~\eqref{eq:overlap-GHZ-free_2} holds for all $k\in\{2,\cdots,n/2\}$. Finally, we conclude
\begin{equation}
\begin{aligned}
\abs{\langle\phi\mid\textup{GHZ}_n\rangle}^2&\leq\frac12\left(\prod\limits_{j=1}^n\abs{\sin\left(\frac{\theta_j}2\right)}^{\frac12}+\prod\limits_{j=1}^n\abs{\cos\left(\frac{\theta_j}2\right)}^{\frac12}\right)^2
\\&=\frac12\left(\prod\limits_{l=1}^{n/2}a_l+\prod\limits_{l=1}^{n/2}b_l\right)^2
\\&\leq\frac12\prod\limits_{l=1}^{n/2}\left(a_l^2+b_l^2\right)
\\&=\frac12\prod\limits_{l=1}^{n/2}\cos\left(\frac{\theta_{2l-1}}2-\frac{\theta_{2l}}2\right)
\\&\leq\frac12.
\end{aligned}
\end{equation}
It is easy to check that $\ket{\phi}=\ket{0^n}$ reaches the upper bound $1/2$.
\end{proof}

In particular, we have $F_{\mathrm{FLO}}(\ket{\eta_0})=F_{\mathrm{FLO}}(\ket{\eta_2})=1/2$. We then show for any $n$-qubit pure fixed-parity state vector $\ket{\psi}$, $F_{\mathrm{FLO}}(\ket{\psi}\otimes\ket{0})=F_{\mathrm{FLO}}(\ket{\psi})$. This can be seen from for any $(n+1)$-qubit pure Gaussian state vector $\ket{\phi}$, $\left(\bI\otimes\bra{0}\right)\ket{\phi}$ is an $n$-qubit pure Gaussian state. So
\begin{equation}
\begin{aligned}
F_{\mathrm{FLO}}(\ket{\psi}\otimes\ket{0})&=\max_{\ket{\phi}\in\G_{n+1}}\abs{\bra{\phi}(\ket{\psi}\otimes\ket{0})}^2
\\&\leq\max_{\ket{\phi}\in\G_{n+1}}\abs{\bra{\phi}(\bI\otimes\ket{0})\ket{\psi}}^2
\\&\leq F_{\mathrm{FLO}}(\ket{\psi}).
\end{aligned}
\end{equation}
On the other hand, let $\ket{\phi^*}=\arg\max_{\ket{\phi}\in\mathcal{G}_n}\abs{\bra{\phi^*}\ket{\psi}}^2$, then $\ket{\phi^*}\otimes\ket{0}$ is a $(n+1)$-qubit pure Gaussian state, thus
\begin{equation}
F_{\mathrm{FLO}}(\ket{\psi}\otimes\ket{0})\geq\abs{\left(\bra{\phi^*}\otimes\bra{0}\right)\left(\ket{\psi}\otimes\ket{0}\right)}^2=F_{\mathrm{FLO}}(\ket{\psi}).
\end{equation}

So we have $F_{\mathrm{FLO}}(\ket{\eta_1})=F_{\mathrm{FLO}}(\ket{\eta_3})=1/2$. Using Eq.~\eqref{eq:FLO_fidelity_multiplicative} iteratively, we obtain
\begin{equation}
F_{\mathrm{FLO}}(P_n)=\frac{1}{2^{\lfloor n/4\rfloor}},\qquad\forall n\geq4.
\end{equation}
By the convexity, we have for all $\sigma\in\mathrm{Conv}(\G_n)$,
\begin{equation}\label{eq:overlap_convex_Gaussian_GHZ_r_app}
    \Tr(\sigma P_n)
    \leq
    2^{-\lfloor n/4\rfloor}.
\end{equation}
Notice that the FLO fidelity of $P_n$ decays exponentially with $n$, which indicates that $P_n$ is a scalable and resource-sensitive operator as the seed operator to build a WE criterion for detecting fermionic non-Gaussianity. We then derive the explicit calculation for the WE $2,3$-criteria in Theorem~\ref{thm:R1-3_FNG_main}.

\begin{proof}[Proof of Theorem~\ref{thm:R1-3_FNG_main}]
We first introduce some notations. For each $k\in\{0,1,\dots,2n\}$ let 
$\mathcal{H}_k=\operatorname{span}\{\gamma_S\mid\abs{S}=k\}$ be the subspace spanned by degree-$k$ Majorana products, and let $\Pi_k:\mathcal L\left(\left(\mathbb C^2\right)^{\otimes n}\right)\to\mathcal H_k$ be the orthogonal projector onto $\mathcal{H}_k$. Denote $d=2^n$, define for operator $A\in\mathcal{L}\left(\left(\mathbb C^2\right)^{\otimes n}\right)$,
\begin{equation}\label{eq:def_Pi_T}
\begin{aligned}
\Pi_k(A)&=\frac1{d}\sum_{\abs{S}=k}\Tr(\gamma_S^\dagger A)\gamma_S=\frac1{d}\sum_{\abs{S}=k}\Tr(\widehat{\gamma}_SA)\widehat{\gamma}_S,
\\
\cT_{\ell_1,\ell_2,\ell_3}(A)
&=
\sum_{\substack{A_1,A_2,A_3\subset [2n]\ \textup{disjoint}\\ \abs{A_i}=2\ell_i}}
\Tr\left(A\gamma_{A_1}\gamma_{A_2}\right)
\Tr\left(A\gamma_{A_2}\gamma_{A_3}\right)
\Tr\left(A\gamma_{A_3}\gamma_{A_1}\right).
\end{aligned}
\end{equation}

For a physical fermionic state $\rho\in\mathcal{D}_F$, recall the even Majorana sector purities $B_\ell(\rho)$ for $\ell\in\{0,1,\cdots,n\}$ defined in Eq.~\eqref{eq:Majorana_sector_purity_main} in the main text, we show
\begin{equation}\label{eq:Majorana_sector_purity_app}
B_\ell(\rho)
=
\sum_{\abs{S}=2\ell}
\Tr(\widehat{\gamma}_S\rho)^2
=d\Tr\left(\Pi_{2\ell}(\rho)^2\right).
\end{equation}
For any $\rho\in\mathcal{D}_F$, we use Eq.~\eqref{eq:Hermitian_Majorana_product_orthogonal}, then
\begin{equation}\label{eq:trace_Pi_A_squared}
\begin{aligned}
\Tr(\Pi_{2\ell}(\rho)^2)&=\frac1{d^2}\sum\limits_{\substack{S_1\subseteq[2n]\\\abs{S_1}=2\ell}}\sum\limits_{\substack{S_2\subseteq[2n]\\\abs{S_2}=2\ell}}\Tr(\widehat{\gamma}_{S_1} \rho)\Tr(\widehat{\gamma}_{S_2} \rho)\Tr(\widehat{\gamma}_{S_1}\widehat{\gamma}_{S_2})
\\&=\frac{1}d\sum\limits_{\substack{S_1\subseteq[2n]\\\abs{S_1}=2\ell}}\Tr(\widehat{\gamma}_{S_1} \rho)^2=\frac1dB_\ell(\rho).
\end{aligned}
\end{equation} 

We now calculate the WE $2$-norm in Eq.~\eqref{eq:R1-3_FNG}.
\begin{equation}\label{eq:FNG_WE_2_norm_projector_app}
\begin{aligned}
\norm{\rho}_{P_n,\mathrm{M}_n,2}^2&=\underset{U\sim\mu_{\mathrm{M}_n}}{\mathbb{E}}\left[\Tr(U^\dagger P_n U\rho)^2\right]=\Tr\left[P_n^{\otimes{2}}\Phi_{\mathrm{M}_n}^{(2)}(\rho^{\otimes2})\right]
\\&=\frac1{d^2}\sum\limits_{k=0}^{2n}\sum\limits_{\substack{A_1\subseteq[2n]\\\abs{A_1}=k}}\sum\limits_{\substack{A_2\subseteq[2n]\\\abs{A_2}=k}}\binom{2n}k^{-1}\Tr(\gamma_{A_2}^\dagger P_n)^2\Tr(\gamma_{A_1}\rho)^2
\\&=\frac1{d^2}\sum_{\ell=0}^{n}\binom{2n}{2\ell}^{-1}\sum\limits_{\substack{A_2\subseteq[2n]\\\abs{A_2}=2\ell}}\Tr(\gamma_{A_2}^\dagger P_n)^2\sum\limits_{\substack{A_1\subseteq[2n]\\\abs{A_1}=2\ell}}\Tr(\gamma_{A_1}^\dagger\rho)^2
\\&=\sum_{\ell=0}^{n}
\frac{\Tr\left(\Pi_{2\ell}(P_n)^2\right)\Tr\left(\Pi_{2\ell}(\rho)^2\right)}{\binom{2n}{2\ell}}
\\&=\frac1{d^2}
\sum_{\ell=0}^{n}
\frac{B_\ell(P_n)B_\ell(\rho)}{\binom{2n}{2\ell}},
\end{aligned}
\end{equation}
where in the second line we have used Eq.~\eqref{eq:matchgate_2_twirling}, in the third line we have used both $P_n$ and $\rho$ are parity-preserving and coefficients with odd-degree Majorana products vanish, and $\Tr(\gamma_{A_1}\rho)^2=\Tr(\gamma_{A_1}^\dagger\rho)^2$ when $\abs{A_1}$ is even. 

For the WE $3$-norm, we recall the third-order matchgate twirling formula (see Theorem~1 in Ref.~\cite{wan2023MatchgateShadowsFermionic}): for operator $O_3\in\mathcal{L}\left(\left(\mathbb{C}^d\right)^{\otimes3}\right)$,
\begin{equation}\label{eq:matchgate_3_twirling}
\begin{aligned}
\Phi_{\mathrm{M}_n}^{(3)}(O_3)&=\frac{1}{d^3}
\sum_{\substack{k_1,k_2,k_3\ge0\\ k_1+k_2+k_3\le2n}}
\binom{2n}{k_1,k_2,k_3,2n-k_1-k_2-k_3}^{-1}\\
&\hspace{4.4em}\times\sum_{\substack{A_1,A_2,A_3\subset [2n]\ \textup{disjoint}\\ 
A'_1,A'_2,A'_3\subset [2n]\ \textup{disjoint}\\
\abs{A_i}=\abs{A'_i}=k_i}}
\Tr(\left((\gamma_{A_1}\gamma_{A_2})^\dagger\otimes(\gamma_{A_2}\gamma_{A_3})^\dagger\otimes(\gamma_{A_3}\gamma_{A_1})^\dagger\right)O_3)\gamma_{A_1'}\gamma_{A'_2}\otimes\gamma_{A'_2}\gamma_{A'_3}\otimes\gamma_{A'_3}\gamma_{A'_1}.
\end{aligned}
\end{equation}

Then we have
\begin{equation}
\norm{\rho}_{P_n,\mathrm{M}_n,3}^3=\underset{U\sim\mu_{\mathrm{M}_n}}{\mathbb{E}}\left[\Tr(U^\dagger P_n U\rho)^3\right]=\Tr\left[P_n^{\otimes{3}}\Phi_{\mathrm{M}_n}^{(3)}(\rho^{\otimes3})\right]
\end{equation}

Inserting $O_3=\rho^{\otimes3}$, since $\rho$ is parity-preserving, the terms in $\Phi_{\mathrm{M}_n}^{(3)}(\rho^{\otimes3})$ vanish whenever one of $k_1+k_2$, $k_2+k_3$, and $k_3+k_1$ is odd. Equivalently, a nonzero contribution from $\rho$ requires $k_1,k_2,k_3$ to have the same parity. However, we show that the all-odd case vanishes when we consider terms related to $P_n$.

Define the Majorana support of the seed projector $P_n$ as
\begin{equation}\label{eq:Pn_majorana_support_WE3}
\begin{aligned}
\mathsf S(P_n)
=
\left\{
S\subseteq[2n]\middle\mid\Tr(P_n\gamma_S)\neq0
\right\}.
\end{aligned}
\end{equation}
Since $P_n$ is parity-preserving, we have
\begin{equation}\label{eq:Pn_support_even_WE3}
\begin{aligned}
S\in\mathsf S(P_n)
\quad\Longrightarrow\quad
\abs{S}\equiv0\pmod 2.
\end{aligned}
\end{equation}
The seed $P_n$ is a tensor product of GHZ-type stabilizer projectors and computational basis stabilizer projectors. Hence, $P_n$ is a pure stabilizer projector. Therefore, if $S\in\mathsf S(P_n)$, then the Majorana monomial $\gamma_S$, up to its scalar phase, corresponds to an element of the abelian stabilizer group of $P_n$. 
For arbitrary subsets $S,T\subseteq[2n]$, the Majorana commutation rule gives
\begin{equation}\label{eq:majorana_commutation_WE3}
\begin{aligned}
\gamma_S\gamma_T
=
(-1)^{\abs{S}\abs{T}-\abs{S\cap T}}
\gamma_T\gamma_S.
\end{aligned}
\end{equation}
If $S,T\in\mathsf S(P_n)$, then the corresponding stabilizer elements commute. Since $\abs{S}$ and $\abs{T}$ are even by Eq.~\eqref{eq:Pn_support_even_WE3}, Eq.~\eqref{eq:majorana_commutation_WE3} implies
\begin{equation}\label{eq:Pn_support_self_orthogonal_WE3}
\begin{aligned}
S,T\in\mathsf S(P_n)
\quad\Longrightarrow\quad
\abs{S\cap T}\equiv0\pmod 2.
\end{aligned}
\end{equation}

Now consider the contraction with $P_n^{\otimes3}$ of one output term in Eq.~\eqref{eq:matchgate_3_twirling}:
\begin{equation}\label{eq:Pn_output_contraction_WE3}
\Tr\left[
P_n^{\otimes3}
\left(
\gamma_{A'_1}\gamma_{A'_2}
\otimes
\gamma_{A'_2}\gamma_{A'_3}
\otimes
\gamma_{A'_3}\gamma_{A'_1}
\right)
\right]
=
\Tr(P_n\gamma_{A'_1}\gamma_{A'_2})
\Tr(P_n\gamma_{A'_2}\gamma_{A'_3})
\Tr(P_n\gamma_{A'_3}\gamma_{A'_1}).
\end{equation}
Suppose that this product is nonzero. Since $A'_1,A'_2,A'_3$ are pairwise disjoint, for each $i\neq j$ we have
\begin{equation}\label{eq:disjoint_majorana_union_WE3}
\begin{aligned}
\gamma_{A'_i}\gamma_{A'_j}
=
\omega_{i,j}\gamma_{A'_i\cup A'_j},
\qquad
\omega_{i,j}\in\{+1,-1\}.
\end{aligned}
\end{equation}
Therefore the nonzero condition in Eq.~\eqref{eq:Pn_output_contraction_WE3} implies
\begin{equation}\label{eq:union_in_Pn_support_WE3}
\begin{aligned}
A'_1\cup A'_2\in\mathsf S(P_n),
\qquad
A'_2\cup A'_3\in\mathsf S(P_n),
\qquad
A'_3\cup A'_1\in\mathsf S(P_n).
\end{aligned}
\end{equation}
Applying Eq.~\eqref{eq:Pn_support_self_orthogonal_WE3} to the first and third sets in Eq.~\eqref{eq:union_in_Pn_support_WE3}, and using the pairwise disjointness of $A'_1,A'_2,A'_3$, gives
\begin{equation}\label{eq:A'1_even_WE3}
\begin{aligned}
0
&\equiv
\abs{(A'_1\cup A'_2)\cap(A'_3\cup A'_1)}
=
\abs{A'_1}
\pmod 2,\\
0
&\equiv
\abs{(A'_1\cup A'_2)\cap(A'_2\cup A'_3)}
=
\abs{A'_2}
\pmod 2,
\\
0
&\equiv
\abs{(A'_2\cup A'_3)\cap(A'_3\cup A'_1)}
=
\abs{A'_3}
\pmod 2.
\end{aligned}
\end{equation}
Thus every output term surviving the contraction with $P_n^{\otimes3}$ must satisfy
\begin{equation}\label{eq:A'i_all_even_WE3}
\begin{aligned}
\abs{A'_1}\equiv\abs{A'_2}\equiv\abs{A'_3}\equiv0\pmod 2.
\end{aligned}
\end{equation}
Since $\abs{A'_i}=k_i$, this means that all terms with at least one odd $k_i$ vanish after contraction with $P_n^{\otimes3}$. Combining this with the parity-preserving condition on $\rho$, the only remaining sectors are the all-even sectors. Therefore we set
\begin{equation}\label{eq:set_ki_even_WE3}
\begin{aligned}
k_i=2\ell_i,
\qquad
i=1,2,3.
\end{aligned}
\end{equation}
So the possible non-vanishing terms are those with all even $k_i=2\ell_i$, we have
\begin{equation}
\begin{aligned}
&\Tr(\left((\gamma_{A_1}\gamma_{A_2})^\dagger\otimes(\gamma_{A_2}\gamma_{A_3})^\dagger\otimes(\gamma_{A_3}\gamma_{A_1})^\dagger\right)\rho^{\otimes3})\\
=&\Tr((\gamma_{A_1}\gamma_{A_2})^\dagger\rho)\Tr((\gamma_{A_2}\gamma_{A_3})^\dagger\rho)\Tr((\gamma_{A_3}\gamma_{A_1})^\dagger\rho)\\
=&(-1)^{\ell_1+\ell_2}\Tr(\rho\gamma_{A_1}\gamma_{A_2})(-1)^{\ell_2+\ell_3}\Tr(\rho\gamma_{A_2}\gamma_{A_3})(-1)^{\ell_3+\ell_1}\Tr(\rho\gamma_{A_3}\gamma_{A_1})\\
=&\Tr(\rho\gamma_{A_1}\gamma_{A_2})\Tr(\rho\gamma_{A_2}\gamma_{A_3})\Tr(\rho\gamma_{A_3}\gamma_{A_1}).
\end{aligned}
\end{equation}
Using the definition in Eq.~\eqref{eq:def_Pi_T} and the above equation, the WE $3$-norm is
\begin{equation}\label{eq:FNG_WE_3_norm_projector_app}
\begin{aligned}
\norm{\rho}_{P_n,\mathrm{M}_n,3}^3&=\Tr\left[P_n^{\otimes{3}}\Phi_{\mathrm{M}_n}^{(3)}(\rho^{\otimes3})\right]
\\&=\frac1{d^3}\sum_{\substack{\ell_1,\ell_2,\ell_3\ge 0\\ \ell_1+\ell_2+\ell_3\le n}}
\frac{\cT_{\ell_1,\ell_2,\ell_3}(P_n)\cT_{\ell_1,\ell_2,\ell_3}(\rho)}{\binom{2n}{2\ell_1,2\ell_2,2\ell_3,2n-2(\ell_1+\ell_2+\ell_3)}},
\end{aligned}
\end{equation}
which proves Eq.~\eqref{eq:R1-3_FNG}.
The threshold values for convex Gaussian states can be evaluated at an arbitrary pure Gaussian state. We use $\ket{0^n}$ and calculate the functions
\begin{equation}\label{eq:Pi_T_vacuum}
\begin{aligned}
B_\ell(\ket{0^n}\bra{0^n})&=d\Tr\left(\Pi_{2\ell}(\ket{0^n}\bra{0^n})^2\right)
=\binom{n}{\ell},
\\
\cT_{\ell_1,\ell_2,\ell_3}(\ket{0^n}\bra{0^n})
&=(-1)^{\ell_1+\ell_2+\ell_3}\binom{n}{\ell_1,\ell_2,\ell_3,n-(\ell_1+\ell_2+\ell_3)}.
\end{aligned}
\end{equation}
Inserting Eq.~\eqref{eq:Pi_T_vacuum} to Eq.~\eqref{eq:FNG_WE_2_norm_projector_app} and Eq.~\eqref{eq:FNG_WE_3_norm_projector_app}, we prove Eq.~\eqref{eq:FNG_thresholds_main} in the main text.

Short proof of Eq.~\eqref{eq:Pi_T_vacuum}: We rewrite $\ketbra{0^n}{0^n}=d^{-1}\prod_{j=1}^{n}\left(\bI-i\gamma_{2j-1}\gamma_{2j}\right)$.
The contribution to degree-$2\ell$ Majorana sector of $\ket{0^n}\bra{0^n}$ is obtained by choosing $\ell$ of the $n$ commuting pair monomials $-i\gamma_{2j-1}\gamma_{2j}$ from the product expansion. There are $\binom{n}{\ell}$ such terms, each with coefficient $d^{-1}$ and orthogonal to all others, which proves the first equation in Eq.~\eqref{eq:Pi_T_vacuum}. 

For the second equation, denote the index of neighboring Majorana pairs as $p_j=\{2j-1,2j\}$. Each of the three traces in $\cT_{\ell_1,\ell_2,\ell_3}$ (cf.\ Eq.~\eqref{eq:def_Pi_T}) is nonzero if and only if $A_1\cup A_2$, $A_2\cup A_3$, and $A_3\cup A_1$ are all union of disjoint pairs $\{p_j\}_{j=1}^n$. Notice that $A_1, A_2, A_3$ are disjoint sets, that is to say, for each $j\in[n]$, either $\{2j-1,2j\}\subseteq A_i$ for some $i$, or $2j-1,2j\notin A_1\cup A_2\cup A_3$. The number of such choices is the multinomial coefficient $\binom{n}{\ell_1,\ell_2,\ell_3,n-(\ell_1+\ell_2+\ell_3)}$, i.e., $\ell_i$ pairs are assigned to $A_i$ and the remaining $n-(\ell_1+\ell_2+\ell_3)$ pairs are left unassigned. Each assignment contributes 
\begin{equation}
(-i)^{\ell_1+\ell_2}\times(-i)^{\ell_2+\ell_3}\times(-i)^{\ell_3+\ell_1}=(-1)^{\ell_1+\ell_2+\ell_3}
\end{equation}
to the summation, which proves the second equation in Eq.~\eqref{eq:Pi_T_vacuum}.
\end{proof}

To obtain concrete WE $2,3$-criteria, we still need to evaluate the functions $B_\ell(P_n)$ and $\cT_{\ell_1,\ell_2,\ell_3}(P_n)$. They are a bit tedious to calculate, and we provide the generating functions in the following sequential lemmas: Lemma~\ref{lem:Pi_T_GHZ4} for $n=4$, Lemma~\ref{lem:Pi_T_GHZ4m} for $n=4m$, Lemma~\ref{lem:Pi_T_eta_r} for $5\leq n\leq7$, and Lemma~\ref{lem:Pi_T_P_n} for all $n$.

We first look into the 4-qubit GHZ state $P_4$. 
\begin{lemma}[Function for GHZ$_4$]\label{lem:Pi_T_GHZ4}
For the $4$-qubit \textup{GHZ} state $P_4=\ket{\textup{GHZ}_4}\bra{\textup{GHZ}_4}$, the values of functions defined in Eq.~\eqref{eq:def_Pi_T} and Eq.~\eqref{eq:Majorana_sector_purity_app} are given by their generating functions
\begin{equation}\label{eq:Pi_T_GHZ4}
\begin{aligned}
P[z]\coloneqq\sum_{\ell=0}^{4}B_\ell(P_4)z^\ell
&=1+14z^2+z^4,
\\
T[x,y,z]\coloneqq\sum_{\substack{\ell_1,\ell_2,\ell_3\ge 0\\ \ell_1+\ell_2+\ell_3\le 4}}
\cT_{\ell_1,\ell_2,\ell_3}(P_4)
x^{\ell_1}y^{\ell_2}z^{\ell_3}
&=
1+14\left(x^2+y^2+z^2+x^2y^2+y^2z^2+z^2x^2\right)+x^4+y^4+z^4-168xyz.
\end{aligned}    
\end{equation}
\end{lemma}

\begin{proof}[Proof of Lemma~\ref{lem:Pi_T_GHZ4}]
Rewrite $P_4$ in terms of Majorana products as
\begin{equation}\label{eq:GHZ4_majorana_decomposition}
\begin{aligned}
P_4=&\frac1{16}\sum_{S\subseteq[8]}\varepsilon(S)\gamma_S,
\qquad
\textup{with }\varepsilon(S)\in\{0,\pm1\},\\
=&\frac{1}{16}\Bigl(
\gamma_{\emptyset}
-\gamma_{\{1,2,3,4\}}-\gamma_{\{1,2,5,6\}}-\gamma_{\{1,2,7,8\}}
-\gamma_{\{3,4,5,6\}}-\gamma_{\{3,4,7,8\}}-\gamma_{\{5,6,7,8\}}
+\gamma_{\{1,3,5,7\}}\\
&-\gamma_{\{1,3,6,8\}}-\gamma_{\{1,4,5,8\}}-\gamma_{\{1,4,6,7\}}
-\gamma_{\{2,3,5,8\}}-\gamma_{\{2,3,6,7\}}-\gamma_{\{2,4,5,7\}}+\gamma_{\{2,4,6,8\}}
\\
&+\gamma_{\{1,2,3,4,5,6,7,8\}}
\Bigr).
\end{aligned}
\end{equation}
A direct calculation leads to $B_\ell(P_4)=\abs{\{S\subseteq[8]\mid\abs{S}=2\ell,\varepsilon(S)\neq0\}}$ and the first line in Eq.~\eqref{eq:Pi_T_GHZ4}.

For $\cT_{\ell_1,\ell_2,\ell_3}$, notice that if $A,B\subseteq[8]$ are disjoint, then
\begin{equation}
\gamma_A\gamma_B
=
(-1)^{\nu(A,B)}\gamma_{A\cup B},
\qquad
\textup{with }\nu(A,B)=\abs{\{(a,b)\in A\times B\mid a>b\}}.
\label{eq:merge-sign}
\end{equation}
Hence
\begin{equation}
\begin{aligned}
\cT_{\ell_1,\ell_2,\ell_3}(P_4)
=
\sum_{\substack{A_1,A_2,A_3\subseteq[8]\ \mathrm{disjoint}\\ \abs{A_i}=2\ell_i}}
&(-1)^{\nu(A_1,A_2)+\nu(A_2,A_3)+\nu(A_3,A_1)}
\times
\varepsilon(A_1\cup A_2)
\varepsilon(A_2\cup A_3)
\varepsilon(A_3\cup A_1).
\end{aligned}
\label{eq:T-P4-sign-expanded}
\end{equation}
Because $\varepsilon(S)\neq0$ only for $\abs{S}=0,4,8$ (cf.\ Eq.~\eqref{eq:GHZ4_majorana_decomposition}), a necessary condition for a nonzero contribution is
\begin{equation}
\abs{A_i\cup A_j}=2(\ell_i+\ell_j)\in\{0,4,8\},
\end{equation}
i.e.
\begin{equation}
\ell_i+\ell_j\in\{0,2,4\}\qquad (1\le i<j\le 3).
\end{equation}
With $\ell_1,\ell_2,\ell_3\ge0$ and $\ell_1+\ell_2+\ell_3\le4$, the only possibilities are
\begin{equation}
(0,0,0),\quad
(2,0,0)\textup{ and permutations},\quad
(2,2,0)\textup{ and permutations},\quad
(4,0,0)\textup{ and permutations},\quad
(1,1,1).
\end{equation}
The first four cases are listed here:
\begin{enumerate}
\item $\cT_{0,0,0}(P_4)=1$.

\item If $(\ell_1,\ell_2,\ell_3)=(2,0,0)$, then $A_2=A_3=\emptyset$ and
\begin{equation}
\cT_{2,0,0}(P_4)=\sum_{|A_1|=4}\varepsilon(A_1)^2=14.
\end{equation}
By symmetry, the same holds for the coefficients of $y^2$ and $z^2$.

\item If $(\ell_1,\ell_2,\ell_3)=(4,0,0)$, then $A_1=[8]$ and $A_2=A_3=\emptyset$, so
\begin{equation}
\cT_{4,0,0}(P_4)=\varepsilon([8])^2=1.
\end{equation}
By symmetry, the same holds for $y^4$ and $z^4$.

\item If $(\ell_1,\ell_2,\ell_3)=(2,2,0)$, then $A_3=\emptyset$ and $A_2=[8]\setminus A_1$. In this case $\nu(A_1,A_2)\equiv\sum_{a\in A_1}a\pmod 2$, $\nu(A_2,A_3)=\nu(A_3,A_1)=0$.
The complement of every degree-$4$ support in Eq.~\eqref{eq:GHZ4_majorana_decomposition} is again a degree-$4$ support, and for every such ordered pair
\begin{equation}
(-1)^{\nu(A_1,A_2)+\nu(A_2,\emptyset)+\nu(\emptyset,A_1)}
\varepsilon(A_1\cup A_2)\varepsilon(A_2)\varepsilon(A_1)=1.
\end{equation}
Hence
\begin{equation}
\cT_{2,2,0}(P_4)=14.
\end{equation}
By symmetry, the same holds for the coefficients of $x^2z^2$ and $y^2z^2$.
\end{enumerate}
It remains to compute $\cT_{1,1,1}(P_4)$. In this case, each $A_i$ has size $2$. A direct inspection of the $14$ degree-$4$ supports shows that the nonzero contributions come from three types of unordered decompositions of six indices into three disjoint $2$-subsets:
\begin{itemize}
\item \emph{Type I:} the omitted two-set is one neighboring pair. There are $\binom41=4$ such cases, one example is
\begin{equation}
[8]\setminus\{7,8\}=\{1,2\}\sqcup\{3,4\}\sqcup\{5,6\}.
\end{equation}
For the ordered representative
\begin{equation}
(A_1,A_2,A_3)=(\{1,2\},\{3,4\},\{5,6\}),
\end{equation}
we have
\begin{equation}
\varepsilon(\{1,2,3,4\})=\varepsilon(\{3,4,5,6\})=\varepsilon(\{1,2,5,6\})=-1,
\end{equation}
and
\begin{equation}
(-1)^{\nu(\{1,2\},\{3,4\})+\nu(\{3,4\},\{5,6\})+\nu(\{5,6\},\{1,2\})}=1,
\end{equation}
so the contribution is $-1$. Since each unordered decomposition gives $3!=6$ ordered triples, this example contributes $-6$. After checking, all examples in Type I contribute $-6$, so overall Type I contributes $-24$.

\item \emph{Type II:} the omitted two-set uses one index from each of two different neighboring pairs, and both have the same parity. There are $\binom42\times2=12$ such cases, one example is
\begin{equation}
[8]\setminus\{6,8\}=\{1,3\}\sqcup\{2,4\}\sqcup\{5,7\}.
\end{equation}
For the ordered representative
\begin{equation}
(A_1,A_2,A_3)=(\{1,3\},\{2,4\},\{5,7\}),
\end{equation}
we have
\begin{equation}
\varepsilon(\{1,2,3,4\})=-1,\qquad \varepsilon(\{2,4,5,7\})=-1,\qquad \varepsilon(\{1,3,5,7\})=1,
\end{equation}
and
\begin{equation}
(-1)^{\nu(\{1,3\},\{2,4\})+\nu(\{2,4\},\{5,7\})+\nu(\{5,7\},\{1,3\})}
=
(-1)^{1+0+0}
=-1.
\end{equation}
Thus this representative also contributes $-1$. Each unordered decomposition gives $3!=6$ ordered triples, so this example contributes $-6$. After checking, all examples in Type II contribute $-6$, so overall Type II contributes $-72$.

\item \emph{Type III:} the omitted two-set uses one index from each of two different neighboring pairs, and they have different parities. There are $\binom42\times2=12$ such cases, one example is
\begin{equation}
[8]\setminus\{6,7\}=\{1,4\}\sqcup\{2,3\}\sqcup\{5,8\}.
\end{equation}
For the ordered representative
\begin{equation}
(A_1,A_2,A_3)=(\{1,4\},\{2,3\},\{5,8\}),
\end{equation}
we have
\begin{equation}
\varepsilon(\{1,2,3,4\})=-1,\qquad \varepsilon(\{1,4,5,8\})=-1,\qquad \varepsilon(\{2,3,5,8\})=-1,
\end{equation}
and
\begin{equation}
(-1)^{\nu(\{1,4\},\{2,3\})+\nu(\{2,3\},\{5,8\})+\nu(\{5,8\},\{1,4\})}
=
(-1)^{2+0+4}
=1.
\end{equation}
Thus this representative also contributes $-1$. Each unordered decomposition gives $3!=6$ ordered triples, so this example contributes $-6$. After checking, all examples in Type III contribute $-6$, so overall Type III contributes $-72$.
\end{itemize}
Therefore
\begin{equation}
\cT_{1,1,1}(P_4)=-24-72-72=-168.
\end{equation}
Combining all the above calculations recovers the second line in Eq.~\eqref{eq:Pi_T_GHZ4}.
\end{proof}

Based on Lemma~\ref{lem:Pi_T_GHZ4}, the evaluation for $P_n$ with $n=4m$ is provided as below.
\begin{lemma}[Function for $P_{4m}$]\label{lem:Pi_T_GHZ4m}
For the $m$-fold tensor product of the 4-qubit GHZ state $P_{4m}=P_4^{\otimes m}$, the values of functions defined in Eq.~\eqref{eq:def_Pi_T} and Eq.~\eqref{eq:Majorana_sector_purity_app} are given by
\begin{equation}\label{eq:Pi_T_GHZn}
\begin{aligned}
\sum_{\ell=0}^{n}B_\ell(P_{4m})z^\ell
&=P[z]^m,
\\
\sum_{\substack{\ell_1,\ell_2,\ell_3\ge 0\\ \ell_1+\ell_2+\ell_3\le n}}
\cT_{\ell_1,\ell_2,\ell_3}(P_{4m})
x^{\ell_1}y^{\ell_2}z^{\ell_3}
&=T[x,y,z]^m,
\end{aligned}    
\end{equation}
where $P[z]$ and $T[x,y,z]$ are defined in Eq.~\eqref{eq:Pi_T_GHZ4}.
\end{lemma}

\begin{proof}[Proof of Lemma~\ref{lem:Pi_T_GHZ4m}]
We first decompose the index set, the Hilbert space, and the $m$-fold tensor product of the 4-qubit GHZ state as
\begin{equation}\label{eq:4m_decomposition}
\begin{aligned}
[8m]&=I_1\sqcup\cdots\sqcup I_m,
\quad\textup{with }I_j=\{8(j-1)+1,\cdots,8j\},\\
\mathcal{H}_{4m}&=\mathcal{H}_4^{(1)}\otimes\cdots\otimes \mathcal{H}_4^{(m)},\\
P_{4m}&=P_4^{(1)}\otimes\cdots\otimes P_4^{(m)},\quad\textup{with }P_4^{(j)}\in\mathcal{D}(\mathcal{H}^{(j)}).
\end{aligned}
\end{equation}
Notice that $P_4\in\mathcal{D}_F$ consists of only even-degree Majorana operators by Eq.~\eqref{eq:GHZ4_majorana_decomposition}, we have for each $j\in[m]$ (cf.\ Eq.~\eqref{eq:def_Pi_T}),
\begin{equation}
P_4^{(j)}=\sum_{i_j=0}^4\Pi_{2i_j}\left(P_4^{(j)}\right),
\end{equation}
so
\begin{equation}
P_{4m}=\bigotimes_{j=1}^m\left(\sum_{i_j=0}^4\Pi_{2i_j}\left(P_4^{(j)}\right)\right)
=\sum_{i_j=0}^4\bigotimes_{j=1}^m\Pi_{2i_j}\left(P_4^{(j)}\right),
\end{equation}
and its projection onto subspace $\mathcal{H}_{2\ell}$ is
\begin{equation}
\Pi_{2\ell}\left(P_{4m}\right)=\sum_{i_1+\cdots+i_m=\ell}\bigotimes_{j=1}^m\Pi_{2i_j}\left(P_4^{(j)}\right).
\end{equation}
Then we have
\begin{equation}
\Tr(\Pi_{2\ell}\left(P_{4m}\right)^2)=\sum_{i_1+\cdots+i_m=\ell}\prod_{j=1}^m\Tr(\Pi_{2i_j}\left(P_4^{(j)}\right)^2),
\end{equation}
since $\Tr\left(\Pi_{2i_j}\left(P_4^{(j)}\right)\Pi_{2i_j}\left(P_4^{(j)}\right)\right)=0$ if $i_j\neq i'_j$.
The generating function in the first line of Eq.~\eqref{eq:Pi_T_GHZn} is
\begin{equation}
\begin{aligned}
\sum_{\ell=0}^{4m}B_\ell(P_{4m})z^\ell&=
\sum_{\ell=0}^{4m}2^{4m}\Tr\left(\Pi_{2\ell}(P_{4m})^2\right)z^\ell
\\&=2^{4m}\sum_{\ell=0}^{4m}\sum_{i_1+\cdots+i_m=\ell}\left(\prod_{j=1}^m\Tr(\Pi_{2i_j}\left(P_4^{(j)}\right)^2)\right)z^{i_1+\cdots+i_m}
\\&=\prod_{j=1}^m\left(\sum_{i=0}^42^4\Tr(\Pi_{2i}\left(P_4\right)^2)z^{i}\right)
\\&=\prod_{j=1}^m\left(\sum_{i=0}^4B_\ell(P_4)z^{i}\right)
\\&=P[z]^m
\end{aligned}
\end{equation}

Consider the sets $\{A_1,A_2,A_3\}$ appearing in the summation in Eq.~\eqref{eq:def_Pi_T}, under the decomposition in Eq.~\eqref{eq:4m_decomposition}, we define $A_i^{(j)}=\{a-8(j-1)\mid a\in A_i\cap I_j\}\subseteq[8]$, for $i\in[3]$ and $j\in[m]$. Denote the projection onto $\mathcal{H}^{(j)}$ as $\Pi^{(j)}$, then we have for $i\in[3]$,
\begin{equation}
\gamma_{A_i}=\bigotimes_{j=1}^m\Pi^{(j)}\left(\gamma_{A_i}\right)=\bigotimes_{j=1}^m\gamma_{A_i^{(j)}}
\end{equation}
since each $\gamma_{A_i}$ is a tensor product of single-qubit Paulis. Thus
\begin{equation}
\Tr(P_n\gamma_{A_1}\gamma_{A_2})=\Tr(\bigotimes_{j=1}^mP_4^{(j)}\bigotimes_{j=1}^m\gamma_{A_1^{(j)}}\bigotimes_{j=1}^m\gamma_{A_2^{(j)}})=\prod_{j=1}^m\Tr(P_4\gamma_{A_1^{(j)}}\gamma_{A_2^{(j)}}),
\end{equation}
similar equations hold for $\Tr(P_n\gamma_{A_2}\gamma_{A_3})$ and $\Tr(P_n\gamma_{A_3}\gamma_{A_1})$. Notice that $P_4$ consists of only Majorana operators with degree $\{0,4,8\}$ by Eq.~\eqref{eq:GHZ4_majorana_decomposition}, the product $\Tr(P_n\gamma_{A_1}\gamma_{A_2})\Tr(P_n\gamma_{A_2}\gamma_{A_3})\Tr(P_n\gamma_{A_3}\gamma_{A_1})$ is nonzero if and only if for each $j\in[m]$,
\begin{equation}
\abs{A_1^{(j)}}+\abs{A_2^{(j)}},\quad\abs{A_2^{(j)}}+\abs{A_3^{(j)}},\quad\abs{A_3^{(j)}}+\abs{A_1^{(j)}}
\end{equation}
all take values in $\{0,4,8\}$, which implies each $\abs{A_i^{(j)}}\eqqcolon2\ell_i^{(j)}$ is even. Since $A_i^{(j)}$ are all disjoint subsets of $[8]$, $\ell_1^{(j)},\ell_2^{(j)},\ell_3^{(j)}\ge 0$ and $\ell_1^{(j)}+\ell_2^{(j)}+\ell_3^{(j)}\le 4$. We have
\begin{equation}
\begin{aligned}
\cT_{\ell_1,\ell_2,\ell_3}(P_{4m})
&=
\sum_{\substack{A_1,A_2,A_3\subset [8m]\ \textup{disjoint}\\ \abs{A_i}=2\ell_i}}
\Tr\left(P_{4m}\gamma_{A_1}\gamma_{A_2}\right)
\Tr\left(P_{4m}\gamma_{A_2}\gamma_{A_3}\right)
\Tr\left(P_{4m}\gamma_{A_3}\gamma_{A_1}\right)
\\&=\sum_{\substack{\sum_{j=1}^m\ell_1^{(j)}=\ell_1\\\sum_{j=1}^m\ell_2^{(j)}=\ell_2\\\sum_{j=1}^m\ell_3^{(j)}=\ell_3}}\prod_{j=1}^m\left(\sum_{\substack{A_1^{(j)},A_2^{(j)},A_3^{(j)}\subset [8]\ \textup{disjoint}\\ \abs{A_i^{(j)}}=2\ell_i^{(j)}}}\Tr(P_4\gamma_{A_1^{(j)}}\gamma_{A_2^{(j)}})\Tr(P_4\gamma_{A_2^{(j)}}\gamma_{A_3^{(j)}})\Tr(P_4\gamma_{A_3^{(j)}}\gamma_{A_1^{(j)}})\right)
\\&=\sum_{\substack{\sum_{j=1}^m\ell_1^{(j)}=\ell_1\\\sum_{j=1}^m\ell_2^{(j)}=\ell_2\\\sum_{j=1}^m\ell_3^{(j)}=\ell_3}}\prod_{j=1}^m\cT_{\ell_1^{(j)},\ell_2^{(j)},\ell_3^{(j)}}(P_4).
\end{aligned}
\end{equation}
The generating function in the second line of Eq.~\eqref{eq:Pi_T_GHZn} is
\begin{equation}
\begin{aligned}
\sum_{\substack{\ell_1,\ell_2,\ell_3\ge 0\\ \ell_1+\ell_2+\ell_3\le n}}
\cT_{\ell_1,\ell_2,\ell_3}(P_{4m})
x^{\ell_1}y^{\ell_2}z^{\ell_3}
&=\prod_{j=1}^m\left(\sum_{\substack{\ell_1^{(j)},\ell_2^{(j)},\ell_3^{(j)}\ge 0\\ \ell_1^{(j)}+\ell_2^{(j)}+\ell_3^{(j)}\le 4}}\cT_{\ell_1^{(j)},\ell_2^{(j)},\ell_3^{(j)}}(P_4)x^{\ell_1^{(j)}}y^{\ell_2^{(j)}}z^{\ell_3^{(j)}}\right)
\\&=T[x,y,z]^m.
\end{aligned}
\end{equation}
\end{proof}

For simplicity, we introduce a notation for the summation of symmetric monomials.

\begin{definition}[Distinct symmetric monomial]\label{def:Sym_monomial}
For nonnegative integers $a,b,c$, define
\begin{equation}
\operatorname{Sym}(x^ay^bz^c)
=
\sum_{(p,q,s)\in\{\sigma(a,b,c)\}_{\sigma\in S_3}}x^py^qz^s,
\end{equation}
where $(p,q,s)$ goes over all possible permutations of $(a,b,c)$ without 
repetitions.
\end{definition}
For example, we have $\operatorname{Sym}(x^3)=x^3+y^3+z^3$ and $\operatorname{Sym}(x^2yz)=x^2yz+xy^2z+xyz^2$. We now present the generating functions for state vectors $\ket{\eta_r}\bra{\eta_r}$.
\begin{lemma}[Function for $\ket{\eta_r}\bra{\eta_r}$]\label{lem:Pi_T_eta_r}
For $r\in\{0,1,2,3\}$, $n=4+r$, and the states $\ket{\eta_r}\bra{\eta_r}$ defined in Eq.~\eqref{eq:def_eta_r}, the values of functions defined in Eq.~\eqref{eq:def_Pi_T} and Eq.~\eqref{eq:Majorana_sector_purity_app} are given by
\begin{equation}\label{eq:Pi_T_eta_r}
\begin{aligned}
\sum_{\ell=0}^{n}B_\ell(\ket{\eta_r}\bra{\eta_r})z^\ell
&=P_r[z],
\\
\sum_{\substack{\ell_1,\ell_2,\ell_3\ge 0\\ \ell_1+\ell_2+\ell_3\le n}}
\cT_{\ell_1,\ell_2,\ell_3}(\ket{\eta_r}\bra{\eta_r})
x^{\ell_1}y^{\ell_2}z^{\ell_3}
&=T_r[x,y,z].
\end{aligned}    
\end{equation}
Explicitly,
\begin{equation}\label{eq:Pr-polynomial}
\begin{aligned}
P_0[z]=&P[z]=1+14z^2+z^4,\\
P_1[z]
=&(1+z)(1+14z^2+z^4)
\\=&1+z+14z^2+14z^3+z^4+z^5,
\\
P_2[z]
=&1+15z^2+32z^3+15z^4+z^6,
\\
P_3[z]
=&(1+z)(1+15z^2+32z^3+15z^4+z^6)
\\=&1+z+15z^2+47z^3+47z^4+15z^5+z^6+z^7.
\end{aligned}
\end{equation}
Moreover,
\begin{equation}\label{eq:T1-T3-polynomial}
T_0[x,y,z]=T[x,y,z],
\quad
T_1[x,y,z]=(1-x-y-z)T[x,y,z],
\quad
T_3[x,y,z]=(1-x-y-z)T_2[x,y,z],
\end{equation}
where
\begin{equation}\label{eq:T2-polynomial}
\begin{aligned}
T_2[x,y,z]
={}&1+15\operatorname{Sym}(x^2)-32\operatorname{Sym}(x^3)+15\operatorname{Sym}(x^4)+\operatorname{Sym}(x^6)
+90\operatorname{Sym}(x^2y^2)
\\
&+15\operatorname{Sym}(x^4y^2)+32\operatorname{Sym}(x^3y^3)
-120xyz
+480\operatorname{Sym}(x^2yz)-120\operatorname{Sym}(x^3yz)
\\
&-480\operatorname{Sym}(x^2y^2z)+90x^2y^2z^2,
\end{aligned}
\end{equation}
$P[z]$ and $T[x,y,z]$ are provided in Lemma~\ref{lem:Pi_T_GHZ4}.
\end{lemma}
\begin{proof}[Proof of Lemma~\ref{lem:Pi_T_eta_r}]
For $r=1$, $\ket{\eta_1}=\ket{\mathrm{GHZ}_4}\otimes\ket0$. Using Eq.~\eqref{eq:Pi_T_vacuum}, the one-mode vacuum state vector $\ket{0}\bra{0}$ contributes $(1+z)$ and $(1-x-y-z)$ to the generating function. Multiplying these by the $\mathrm{GHZ}_4$ polynomials $P[z]$ and $T[x,y,z]$ proves the $r=1$ identities. The same is for $r=3$, where the two factors are multiplied to $P_2[z]$ and $T_2[x,y,z]$.

For $r=2$, a direct Majorana expansion of $\ketbra{\eta_2}{\eta_2}$ contains $64$ nonzero monomials, all with coefficient magnitude $1/64$. Their degree distribution is
\begin{equation}
1\textup{ of degree }0,
\quad
15\textup{ of degree }4,
\quad
32\textup{ of degree }6,
\quad
15\textup{ of degree }8,
\quad
1\textup{ of degree }12.
\end{equation}
This gives the formula for $P_2[z]$ in Eq.~\eqref{eq:Pr-polynomial}. Substituting the same support set into the disjoint-set formula defining $\cT$ and grouping the terms by distinct exponent permutations gives Eq.~\eqref{eq:T2-polynomial}, whose coefficients can be verified numerically.
\end{proof}

\begin{lemma}[Function for $P_n$]\label{lem:Pi_T_P_n}
For the states $P_n$ defined in Eq.~\eqref{eq:def_P_n_r}, the values of functions defined in Eq.~\eqref{eq:def_Pi_T} and Eq.~\eqref{eq:Majorana_sector_purity_app} are given by
\begin{equation}\label{eq:Pi_T_P_n}
\begin{aligned}
\sum_{\ell=0}^{n}B_\ell(P_n)z^\ell
&=P_r[z]P[z]^{m-1},
\\
\sum_{\substack{\ell_1,\ell_2,\ell_3\ge0\\ \ell_1+\ell_2+\ell_3\le n}}
\cT_{\ell_1,\ell_2,\ell_3}(P_n)x^{\ell_1}y^{\ell_2}z^{\ell_3}
&=T_r[x,y,z]T[x,y,z]^{m-1},
\end{aligned}
\end{equation}
where $P_r[z]$ and $T_r[x,y,z]$ are given in Lemma~\ref{lem:Pi_T_eta_r}, $P[z]$ and $T[x,y,z]$ are provided in Lemma~\ref{lem:Pi_T_GHZ4}.
\end{lemma}

\begin{proof}[Proof of Lemma~\ref{lem:Pi_T_P_n}]
The proof is the same blockwise convolution argument as in Lemma~\ref{lem:Pi_T_GHZ4m}. The only difference is that the first block is now the projector $\ketbra{\eta_r}{\eta_r}$, whose generating function is provided in Lemma~\ref{lem:Pi_T_eta_r}, while the remaining $m-1$ blocks are copies of $P_4$, provided in Lemma~\ref{lem:Pi_T_GHZ4}.
\end{proof}

Noticeably, we show that the WE $2$-criterion in Theorem~\ref{thm:R1-3_FNG_main} detects all pure non-Gaussian states for $4\leq n\leq 7$ qubits, by directly relating WE $2$-norm to the \emph{fermionic antiflatness} (FAF) introduced in Ref.~\cite{sierant2026FermionicMagicResources}. Let $\ket{\psi}$ be a fixed-parity pure fermionic state on $n$ modes, i.e., $\psi=\ketbra{\psi}{\psi}\in\mathcal{D}_F$, the elements of its correlation matrix $\Gamma$ are defined by
\begin{equation}
\Gamma_{j,k}(\psi)=-\frac{i}{2}\bra{\psi}[\gamma_j,\gamma_k]\ket{\psi},
\end{equation}
and the first-order FAF is
\begin{equation}
\mathfrak F_1(\ket{\psi})
=
n-\frac12\Tr\left(\Gamma(\psi)^\mathrm{T} \Gamma(\psi)\right).
\end{equation}
For fixed-parity pure states, FAF is a faithful fermionic non-Gaussianity measure, as $\mathfrak F_1(\ket{\psi})\geq0$, with equality holding if and only if $\ket{\psi}$ is Gaussian~\cite{sierant2026FermionicMagicResources}. Since the correlation matrix captures only the coefficients of Majorana products with degree 2, it fully determines WE $2$-norm for pure states with a small size, while WE $2$-norm still involves higher-degree Majorana coefficients to effectively detect non-Gaussianity of a general mixed state. We relate WE $2$-norm to FAF with a small system size for $4\leq n \leq7$ in Corollary~\ref{cor:fermion_WE2} in the main text.

\begin{proof}[Proof of Corollary~\ref{cor:fermion_WE2}]
We denote the parity operator as $Z^{\otimes n}=(-i)^n\gamma_1\gamma_2\cdots\gamma_{2n-1}\gamma_{2n}$, a valid fermionic state vector $\ket{\psi}$ is an eigenstate of $Z^{\otimes n}$ with a fixed parity $+1$ or $-1$:
\begin{equation}
Z^{\otimes n}\ket{\psi}=\pm1\ket{\psi}.
\end{equation}
Recall the even Majorana sector purities for $\ell\in\{0,1,\cdots,n\}$ of a pure state $\psi$ 
(cf.\ Eq.~\eqref{eq:Majorana_sector_purity_app}):
\begin{equation}\label{eq:Majorana_sector_purity}
B_\ell(\psi)
=
\sum_{\abs{S}=2\ell}
\abs{\bra{\psi}\widehat{\gamma}_S\ket{\psi}}^2.
\end{equation}
Those Majorana sector purities obey some relations for a pure fermionic state $\psi$:
\begin{itemize}
\item Boundary constraints:
\begin{equation}\label{eq:Majorana_sector_purity_boundary}
B_0(\psi)=B_n(\psi)=1.
\end{equation}
It can be seen from $\widehat{\gamma}_\emptyset=\bI$ and $\widehat{\gamma}_{[2n]}=(-i)^nZ^{\otimes n}$.
\item These coefficients are symmetric
\begin{equation}\label{eq:Majorana_sector_purity_symmetric}
B_\ell(\psi)=B_{n-\ell}(\psi)
\end{equation}
 towards the center.
Notice that $Z^{\otimes n}\widehat{\gamma}_S=\pm\widehat{\gamma}_{\bar{S}}$, where $\bar{S}=[2n]\setminus S$, we have
\begin{equation}
\abs{\bra{\psi}\widehat{\gamma}_S\ket{\psi}}^2
=
\abs{\bra{\psi}Z^{\otimes n}\widehat{\gamma}_S\ket{\psi}}^2
=
\abs{\bra{\psi}\widehat{\gamma}_{\bar{S}}\ket{\psi}}^2,
\end{equation}
which proves the above equation.
\item These coefficients sum to $2^n$:
\begin{equation}\label{eq:Majorana_sector_purity_normalization}
\sum_{\ell=0}^{n}B_\ell(\psi)=2^n.
\end{equation}
This can be seen from the expansion of the even state
\begin{equation}\label{eq:even_state_expansion}
\psi
=
2^{-n}
\sum_{\ell=0}^{n}
\sum_{\abs{S}=2\ell}
\bra{\psi}\widehat{\gamma}_S\ket{\psi}\widehat{\gamma}_S,
\end{equation}
then use $\Tr(\psi^2)=1$ and Eq.~\eqref{eq:Majorana_product_orthogonal}.
\item These coefficients are linear combinations of each other:
\begin{equation}\label{eq:Majorana_sector_purity_linear}
B_\ell(\psi)=2^{-n}\sum_{k=0}^n\sum_{j=0}^{2\ell}(-1)^j\binom{2k}j\binom{2n-2k}{2\ell-j}B_k(\psi).
\end{equation}
Equivalently, recall that the binary Krawtchouk Polynomial is defined as~\cite{macwilliams1977TheoryErrorCorrecting}
\begin{equation}\label{eq:def_Krawtchouk}
K_r(t;2n)
=
\sum_{j=0}^r(-1)^j\binom tj\binom{2n-t}{r-j},
\end{equation}
the linear transformation of these coefficients is rewritten as
\begin{equation}\label{eq:Majorana_sector_purity_linear_Krawtchouk}
B_\ell(\psi)=2^{-n}\sum_{k=0}^nK_{2\ell}(2k;2n)B_k(\psi).
\end{equation}
We use Eq.~\eqref{eq:even_state_expansion} and $\widehat{\gamma}_S\widehat{\gamma}_T\widehat{\gamma}_S=(-1)^{\abs{S\cap T}}\widehat{\gamma}_T$,
\begin{equation}
\begin{aligned}
B_\ell(\psi)&=\sum_{\abs{S}=2\ell}
\bra{\psi}\widehat{\gamma}_S\left(\ket{\psi}\bra{\psi}\right)\widehat{\gamma}_S\ket{\psi}
\\&=2^{-n}\sum_{k=0}^n\sum_{\abs{T}=2k}\left(\sum_{\abs{S}=2\ell}(-1)^{\abs{S\cap T}}\right)\abs{\bra{\psi}\widehat{\gamma}_T\ket{\psi}}^2
\\&=2^{-n}\sum_{k=0}^n\sum_{\abs{T}=2k}\sum_{j=0}^{2\ell}(-1)^j\binom{2k}j\binom{2n-2k}{2\ell-j}\abs{\bra{\psi}\widehat{\gamma}_T\ket{\psi}}^2
\\&=2^{-n}\sum_{k=0}^n\sum_{j=0}^{2\ell}(-1)^j\binom{2k}j\binom{2n-2k}{2\ell-j}B_k(\psi).
\end{aligned}
\end{equation}
\end{itemize}
Using Eq.~\eqref{eq:Majorana_sector_purity_boundary}, Eq.~\eqref{eq:Majorana_sector_purity_symmetric}, Eq.~\eqref{eq:Majorana_sector_purity_normalization}, and Eq.~\eqref{eq:Majorana_sector_purity_linear}, we can solve a set of linear equations for $\{B_\ell\}_{\ell=0}^n$.  For $4\leq n\leq7$, we can represent all other coefficients with
\begin{equation}
B_1(\psi)=\frac12\Tr\left(\Gamma(\psi)^\mathrm{T} \Gamma(\psi)\right)
=
n-\mathfrak F_1(\ket{\psi}).
\end{equation}
For $4\leq n\leq7$, we explicitly calculate WE $2$-norm of $n$-qubit state $\psi$ in Eq.~\eqref{eq:FNG_WE_2_norm_projector_app}
\begin{equation}\label{eq:FNG_WE_2_norm_projector_app_1}
\norm{\psi}_{P_n,\mathrm{M}_n,2}^2
=
\frac1{d^2}\sum_{\ell=0}^{n}
\frac{B_\ell(P_n)B_\ell(\psi)}{\binom{2n}{2\ell}},
\end{equation}
with all the intermediate data provided in Table~\ref{tab:sector-data-R2_n4-7}. It is worth noting that $n\geq8$ will involve more variables than constraints, so that we cannot represent all $B_\ell$ with $\mathfrak{F}_1$.

\begin{table}[t]
\centering
\renewcommand{\arraystretch}{1.05}
\setlength{\extrarowheight}{2pt}
\begin{tabular}{c|c|c|c|c|c}
\hline
$n$ 
& $\ell$ 
& $\binom{2n}{2\ell}$ 
& $B_\ell(P_n)$, Eq.~\eqref{eq:Pr-polynomial} 
& $B_\ell(\psi)$, Eq.~\eqref{eq:Majorana_sector_purity}
& $\norm{\psi}_{P_n,\mathrm{M}_n,2}^2$, Eq.~\eqref{eq:FNG_WE_2_norm_projector_app_1}\\
\hline
\multirow{5}{*}{$4$}
& $0$ & $1$  & $1$  & $1$
& \multirow{5}{*}{$\dfrac{8+\mathfrak F_1}{640}$} \\
\cline{2-5}
& $1$ & $28$ & $0$  & $4-\mathfrak F_1$ & \\
\cline{2-5}
& $2$ & $70$ & $14$ & $6+2\mathfrak F_1$ & \\
\cline{2-5}
& $3$ & $28$ & $0$  & $4-\mathfrak F_1$ & \\
\cline{2-5}
& $4$ & $1$  & $1$  & $1$ & \\
\hline
\multirow{6}{*}{$5$}
& $0$ & $1$   & $1$  & $1$
& \multirow{6}{*}{$\dfrac{40+\mathfrak F_1}{11520}$}\\
\cline{2-5}
& $1$ & $45$  & $1$  & $5-\mathfrak F_1$ & \\
\cline{2-5}
& $2$ & $210$ & $14$ & $10+\mathfrak F_1$ & \\
\cline{2-5}
& $3$ & $210$ & $14$ & $10+\mathfrak F_1$ & \\
\cline{2-5}
& $4$ & $45$  & $1$  & $5-\mathfrak F_1$ & \\
\cline{2-5}
& $5$ & $1$   & $1$  & $1$ & \\
\hline
\multirow{7}{*}{$6$}
& $0$ & $1$   & $1$  & $1$
& \multirow{7}{*}{$\dfrac{52+\mathfrak F_1}{59136}$}\\
\cline{2-5}
& $1$ & $66$  & $0$  & $6-\mathfrak F_1$ & \\
\cline{2-5}
& $2$ & $495$ & $15$ & $15$ & \\
\cline{2-5}
& $3$ & $924$ & $32$ & $20+2\mathfrak F_1$ & \\
\cline{2-5}
& $4$ & $495$ & $15$ & $15$ & \\
\cline{2-5}
& $5$ & $66$  & $0$  & $6-\mathfrak F_1$ & \\
\cline{2-5}
& $6$ & $1$   & $1$  & $1$ & \\
\hline
\multirow{8}{*}{$7$}
& $0$ & $1$    & $1$  & $1$
& \multirow{8}{*}{$\dfrac{364+\mathfrak F_1}{1537536}$} \\
\cline{2-5}
& $1$ & $91$   & $1$  & $7-\mathfrak F_1$ & \\
\cline{2-5}
& $2$ & $1001$ & $15$ & $21-\mathfrak F_1$ & \\
\cline{2-5}
& $3$ & $3003$ & $47$ & $35+2\mathfrak F_1$ & \\
\cline{2-5}
& $4$ & $3003$ & $47$ & $35+2\mathfrak F_1$ & \\
\cline{2-5}
& $5$ & $1001$ & $15$ & $21-\mathfrak F_1$ & \\
\cline{2-5}
& $6$ & $91$   & $1$  & $7-\mathfrak F_1$ & \\
\cline{2-5}
& $7$ & $1$    & $1$  & $1$ & \\
\hline
\end{tabular}
\caption{Sector data to calculate $\norm{\psi}_{P_n,\mathrm{M}_n,2}^2$ for pure states with $4\leq n\leq7$ qubits.}
\label{tab:sector-data-R2_n4-7}
\end{table}
\end{proof}

Especially, we provide numerical data for the WE $2,3$-criteria applying to non-Gaussian states $P_n$ for $4\leq n\leq11$. The values of the detection gaps
\begin{equation}
\begin{aligned}
\Delta_{P_n,2}(P_n)&=\norm{P_n}_{P_n,\mathrm{M}_n,2}^2-C_{P_n,\mathrm{M}_n,2}^2,\\
\Delta_{P_n,3}(P_n)&=\norm{P_n}_{P_n,\mathrm{M}_n,3}^3-C_{P_n,\mathrm{M}_n,3}^3,
\end{aligned}
\end{equation}
are provided in Table~\ref{tab:projector_seed_margins}. The criterion successfully detects non-Gaussianity if and only if the corresponding gap is positive. We see that both the WE $2,3$-criteria detect $P_n$ for $4\leq n\leq 10$, while both fail at $n=11$.

\begin{table}[h!]
\centering
\renewcommand{\arraystretch}{1.7}
\setlength{\extrarowheight}{2pt}
\begin{tabular}{c|c|c|c|c|c|c|c|c}
\hline
$n$ 
& $4$ & $5$ & $6$ & $7$ & $8$ & $9$ & $10$ & $11$ \\ 
\hline
\begin{tabular}{@{}c@{}}
$2^{2n}\Delta_{P_n,2}(P_n)$
\end{tabular}
& $\dfrac{8}{5}$ 
& $\dfrac{16}{45}$ 
& $\dfrac{32}{77}$ 
& $\dfrac{64}{1001}$ 
& $\dfrac{128}{65}$ 
& $\dfrac{256}{1105}$ 
& $\dfrac{27136}{138567}$ 
& $-\dfrac{60416}{440895}$ 
\\
\hline
\begin{tabular}{@{}c@{}}
$2^{3n}\Delta_{P_n,3}(P_n)$
\end{tabular}
& $\dfrac{96}{5}$ 
& $\dfrac{128}{25}$ 
& $\dfrac{512}{77}$ 
& $\dfrac{8192}{7007}$ 
& $\dfrac{24576}{715}$ 
& $\dfrac{32768}{7293}$ 
& $\dfrac{6201344}{1616615}$ 
& $-\dfrac{55902208}{21966945}$ 
\\ 
\hline
\end{tabular}
\caption{
Exact margins for the WE $2,3$-criteria evaluated on the seed state $P_n$ itself.
$\Delta_{P_n,2}(P_n)$ and
$\Delta_{P_n,3}(P_n)$ are the gap between the WE-norms and the threshold values for all convex Gaussian states.
Thus, both the WE $2,3$-criteria detect $P_n$ for $4\leq n\leq 10$, while both fail at $n=11$.
}
\label{tab:projector_seed_margins}
\end{table}

\subsection{Proof of Theorem~\ref{thm:FNG_Pauli} and Corollary~\ref{cor:FNG_Pauli_pure}}\label{app:FNG_Pauli}
\begin{proof}[Proof of Theorem~\ref{thm:FNG_Pauli}]
Let $d=2^n$, we now calculate the WE $2$-norm in Eq.~\eqref{eq:FNG_WE_2_norm_Pauli_general}.
Recall the seed operator defined in Eq.~\eqref{eq:FNG_pauli_general} for $1\leq\ell\leq n$,
\begin{equation}\label{eq:FNG_pauli_app}
\widehat{\gamma}_{T_\ell}
=
\prod_{j=1}^{\ell} Z_j,
\qquad
T_\ell=[2\ell].
\end{equation}
Using Eq.~\eqref{eq:odd_Majorana_twirling} and Eq.~\eqref{eq:trace_Pi_A_squared}, we have
\begin{equation}\label{eq:FNG_WE_2_norm_Pauli_app}
\begin{aligned}
\norm{\rho}_{\widehat{\gamma}_{T_\ell},\mathrm{M}_n,2}^2&=\underset{U\sim\mu_{\mathrm{M}_n}}{\mathbb{E}}\left[\Tr(U^\dagger \widehat{\gamma}_{T_\ell} U\rho)^2\right]
=\Tr\left[\Phi_{\mathrm{M}_n}^{(2)}\left(\widehat{\gamma}_{T_\ell}^{\otimes{2}}\right)\rho^{\otimes2}\right]
\\&=\binom{2n}{2\ell}^{-1}\sum\limits_{\substack{S\subseteq[2n]\\\abs{S}=2\ell}}\Tr(\widehat{\gamma}_S\rho)^2
\\&=\binom{2n}{2\ell}^{-1} B_\ell(\rho),
\end{aligned}
\end{equation}
which proves Eq.~\eqref{eq:FNG_WE_2_norm_Pauli_general}. Notice that $\norm{\rho}_{\widehat{\gamma}_{T_\ell},\mathrm{M}_n,2}$ is the $L^2(\mu_{\mathrm{M}_n})$-norm of the function $f_\rho:\mathrm{M}_n\rightarrow\mathbb{R}$,
\begin{equation}
f_{\rho}(U)=\Tr(U^\dagger \widehat{\gamma}_{T_\ell} U\rho),
\end{equation}
since
\begin{equation}
\norm{f_\rho}_{L^2(\mu_{\mathrm{M}_n})}
=
\left(
\int_{\mathrm{M}_n}
\abs{\Tr(U^\dagger \widehat{\gamma}_{T_\ell} U\rho)}^2
\mathrm{d}\mu_{\mathrm{M}_n}
\right)^{1/2}
=\left(
\underset{U\sim\mu_{\mathrm{M}_n}}{\mathbb{E}}\left[\Tr(U^\dagger \widehat{\gamma}_{T_\ell} U\rho)^2\right]
\right)^{1/2}
.
\end{equation}
It is convex by the Minkowski inequality, especially the convex Gaussian threshold
\begin{equation}
C_{\widehat{\gamma}_{T_\ell},\mathrm{M}_n,2}=
\sup_{\sigma\in\textup{Conv}(\mathcal{G}_n)}
\norm{\sigma}_{\widehat{\gamma}_{T_\ell},\mathrm{M}_n,2}
\end{equation}
is obtained at some pure Gaussian state $\phi\in\mathcal{G}_n$. Furthermore, all states in $\mathcal{G}_n$ are in a single orbit under $\mathrm{M}_n$, we can evaluate the threshold value at $\ketbra{0^n}{0^n}$. Inserting Eq.~\eqref{eq:Pi_T_vacuum} to Eq.~\eqref{eq:FNG_WE_2_norm_Pauli_app}, we have
\begin{equation}
C_{\widehat{\gamma}_{T_\ell},\mathrm{M}_n,2}^2=\norm{\ketbra{0^n}{0^n}}_{\widehat{\gamma}_{T_\ell},\mathrm{M}_n,2}^2=\frac{\binom{n}{\ell}}{\binom{2n}{2\ell}},
\end{equation}
which proves Eq.~\eqref{eq:FNG_WE_2_threshold_Pauli_general}. In this way, the WE $2$-criterion with $\widehat{\gamma}_T$ is
\begin{equation}
\norm{\rho}_{\widehat{\gamma}_{T_\ell},\mathrm{M}_n,2}>C_{\widehat{\gamma}_{T_\ell},\mathrm{M}_n,2}
\quad\Longleftrightarrow\quad
B_\ell(\rho)>\binom{n}{\ell}
\quad\Longrightarrow\quad
\rho\notin\mathrm{Conv}(\mathcal G_n).
\end{equation}
\end{proof}

We reformulate Corollary~\ref{cor:FNG_Pauli_pure} as the following theorem and then prove it.
\begin{theorem}[Central Majorana sector purity inequality]\label{thm:central-ineq}
For every fixed-parity pure state $\psi$ with $n\ge4$ qubits,
\begin{equation}\label{eq:central-ineq}
 B_{\floor{n/2}}(\psi)-\binom n{\floor{n/2}}\geq0.
\end{equation}
Moreover, equality holds if and only if $\psi$ is a pure fermionic Gaussian state.
\end{theorem}
To prove Theorem~\ref{thm:central-ineq}, we first introduce some notations and useful lemmas, with their proofs left to the end of this appendix.
For a series $A=(a_0,a_1,a_2,\cdots)$, we define its generating function as the formal power series $F_A(z)$, where in general $z\in\mathbb{C}$,
\begin{equation}
F_A(z)=\sum_{j=0}^\infty a_jz^j.
\end{equation}
Equivalently, denote the coefficients $a_j$ of the term $x^j$ extracted from $F_A(z)$ as
\begin{equation}
\Coeff{z^j}{F_A(z)}=a_j.
\end{equation}
For any fixed-parity pure state $\psi$, denote the generating function of its Majorana sector purities $B_\ell(\psi)$ as
\begin{equation}\label{eq:def_Majorana_sector_purity_generating_function}
G_\psi(z)=\sum_{\ell=0}^nB_\ell(\psi)z^\ell.
\end{equation}
We show $G_\psi(z)$ can be rewritten in the following form.
\begin{lemma}[Rewrite the generating function of Majorana sector purities]\label{lem:Gleason_form}
For every fixed-parity pure state $\psi$, there exist real coefficients $c_j(\psi)$ and $c_0(\psi)=1$, such that the generating function of its Majorana sector purities $B_\ell(\psi)$ can be written as
\begin{equation}\label{eq:gleason-G}
G_\psi(z)=\sum_{\ell=0}^nB_\ell(\psi)z^\ell
=
\sum_{j=0}^{\floor{n/4}}
c_j(\psi)(1+z)^{n-4j}\left(z(1-z)^2\right)^j.
\end{equation}
\end{lemma}
Recall that by Eq.~\eqref{eq:Pi_T_vacuum}, \begin{equation}
B_\ell(\ketbra{0^n}{0^n})=\binom n\ell,
\end{equation}
its generating function is
\begin{equation}
G_{\ketbra{0^n}{0^n}}(z)=(1+z)^n.
\end{equation}
We calculate the difference in Eq.~\eqref{eq:central-ineq}, for any fixed-parity pure state $\psi$,
\begin{equation}
\begin{aligned}
B_{\floor{n/2}}(\psi)-\binom n{\floor{n/2}}&=B_{\floor{n/2}}(\psi)-B_{\floor{n/2}}(\ketbra{0^n}{0^n})
\\&=\Coeff{z^{\floor{n/2}}}{\left(G_\psi(z)-G_{\ketbra{0^n}{0^n}}(z)\right)}
\\&=\Coeff{z^{\floor{n/2}}}{\sum_{j=1}^{\floor{n/4}}
c_j(\psi)(1+z)^{n-4j}\left[z(1-z)^2\right]^j}
\\&=\sum_{j=1}^{\floor{n/4}}\Coeff{z^{\floor{n/2}-j}}{
c_j(\psi)(1+z)^{n-4j}(1-z)^{2j}}
\\&=\sum_{j=1}^{\floor{n/4}}c_j(\psi)e_{n,j},
\end{aligned}
\end{equation}
where in the third line we 
have used $c_0(\psi)=1$ and the $j=0$ term in Eq.~\eqref{eq:gleason-G} collides with $G_{\ketbra{0^n}{0^n}}(z)$. In the last line, we define
\begin{equation}\label{eq:def_e_n_j}
e_{n,j}=\Coeff{z^{\floor{n/2}-j}}{(1+z)^{n-4j}(1-z)^{2j}}.
\end{equation}

We introduce another set of coefficients through a linear transformation of $e_{n,j}$, and show that each of them is strictly positive.

\begin{lemma}[$\beta_{n,p}$ is strictly positive]\label{lem:beta_n_p_positive}
For $1\leq p\leq\floor{n/4}$, define $\beta_{n,p}$ from a linear transformation of $e_{n,j}$ as
\begin{equation}\label{eq:def_beta_n_p}
\beta_{n,p}
=
\sum_{j=1}^{p}
(-1)^j
\frac{2p}{p+j}\binom{p+j}{p-j}
\frac{e_{n,j}}{4^{n-3j}},
\end{equation}
then $\beta_{n,p}>0$.
\end{lemma}

Recall the binary Krawtchouk polynomial defined in Eq.~\eqref{eq:def_Krawtchouk},
\begin{equation}\label{eq:def_Krawtchouk_2}
K_r(t;2n)
=
\sum_{j=0}^r(-1)^j\binom tj\binom{2n-t}{r-j},
\end{equation}
we further define the function for any fixed-parity pure state $\psi$,
\begin{equation}\label{eq:def-Sr}
S_r(\psi)
=
\sum_{\ell=0}^n
(-1)^\ell K_r(2\ell;2n)B_\ell(\psi).
\end{equation}
We can show that this function is non-negative for any $0\leq r\leq2n$.
\begin{lemma}[$S_r(\psi)$ is non-negative]\label{lem:S_r_non-negative}
For every fixed-parity pure state $\psi$ and every $0\le r\le2n$,
\begin{equation}
 S_r(\psi)\ge0.
\end{equation}
\end{lemma}
Furthermore, $S_r(\psi)$ is related to the coefficients $c_j(\psi)$ in Eq.~\eqref{eq:gleason-G} through the following lemma.
\begin{lemma}[Relation between $S_r(\psi)$ and $c_j(\psi)$]\label{lem:S_r_c_j_relation}
Let $1\leq p\leq\floor{n/4}$, we have
\begin{equation}\label{eq:S_r_c_j_relation}
S_{n-4p}(\psi)=
(-1)^p\sum_{j=p}^{\floor{n/4}}
4^{n-3j}\binom{2j}{j-p}c_j(\psi).
\end{equation}
\end{lemma}
At last, we rewrite the difference in Eq.~\eqref{eq:central-ineq} using $\beta_{n,p}$ and $S_r(\psi)$.
\begin{lemma}[Rewrite the central Majorana sector purity difference]\label{lem:positive-certificate}
For every fixed-parity pure state $\psi$,
\begin{equation}\label{eq:positive-certificate}
B_{\floor{n/2}}(\psi)-\binom n{\floor{n/2}}
=
\sum_{p=1}^{\floor{n/4}}\beta_{n,p}S_{n-4p}(\psi),
\end{equation}
where $\beta_{n,p}$ is given in Eq.~\eqref{eq:def_beta_n_p} and $S_r(\psi)$ is defined in Eq.~\eqref{eq:def-Sr}.
\end{lemma}
We can now prove Theorem~\ref{thm:central-ineq} and finally finish the proof of Theorem~\ref{thm:FNG_Pauli} in the main text.

\begin{proof}[Proof of Theorem~\ref{thm:central-ineq}]
Combining Lemma~\ref{lem:positive-certificate} with Lemma~\ref{lem:beta_n_p_positive} and Lemma~\ref{lem:S_r_non-negative}, we have
\begin{equation}\label{eq:central-ineq_1}
 B_{\floor{n/2}}(\psi)-\binom n{\floor{n/2}}\geq0.
\end{equation}
The equality holds if and only if $S_{n-4p}(\psi)=0$ for all $1\leq p\leq\floor{n/4}$. In this case, from Lemma~\ref{lem:S_r_c_j_relation}, $p=\floor{n/4}$ implies $c_{\floor{n/4}}(\psi)=0$. Inserting it to $p=\floor{n/4}-1$ gives $c_{\floor{n/4}-1}(\psi)=0$. Using this iteratively until $p=1$, we have
\begin{equation}
c_1(\psi)=c_2(\psi)=\cdots=c_{\floor{n/4}}(\psi)=0.
\end{equation}
According to Eq.~\eqref{eq:gleason-G} in Lemma~\ref{lem:Gleason_form}, it is equivalent to
\begin{equation}
G_\psi(z)=(1+z)^n,\qquad B_1(\psi)=n.
\end{equation}
Especially, the FAF
\begin{equation}
\mathfrak F_1(\ket{\psi})=n-B_1(\psi)=0,
\end{equation}
which indicates $\psi$ is a pure fermionic Gaussian state.
\end{proof}

We provide the proofs for each of the above lemmas.
\begin{proof}[Proof of Lemma~\ref{lem:Gleason_form}]
For the generating function of the Majorana sector purities
\begin{equation}
G_\psi(z)=\sum_{\ell=0}^nB_\ell(\psi)z^\ell,
\end{equation}
we consider a generalized homogeneous function, defined by
\begin{equation}\label{eq:def-Fxy}
F_\psi(x,y)=\sum_{\ell=0}^n B_\ell(\psi)x^{2n-2\ell}y^{2\ell}.
\end{equation}
In this way, $G_\psi(z)=F_\psi(x,y)$ by setting $x^2=1$ and $y^2=z$. We show $F_\psi(x,y)$ has four symmetries
\begin{equation}
\begin{aligned}
\text{\ding{172} }
F_\psi(x,y)&=F_\psi(-x,y),&\qquad
\text{\ding{173} }
F_\psi(x,y)&=F_\psi(x,-y),
\\
\text{\ding{174} }
F_\psi(x,y)&=F_\psi(y,x),&\qquad
\text{\ding{175} }
F_\psi(x,y)&=F_\psi\left(\frac{x+y}{\sqrt{2}},\frac{x-y}{\sqrt{2}}\right).
\end{aligned}
\end{equation}
Symmetries~\ding{172} and~\ding{173} hold since both $x$ and $y$ have only even power in $F_\psi(x,y)$. Symmetry~\ding{174} is from $B_\ell(\psi)=B_{n-\ell}(\psi)$ by Eq.~\eqref{eq:Majorana_sector_purity_symmetric}. To show symmetry~\ding{175},
\begin{equation}
\begin{aligned}
F_\psi\left(\frac{x+y}{\sqrt{2}},\frac{x-y}{\sqrt{2}}\right)&=2^{-n}F_\psi(x+y,x-y)
\\&=2^{-n}\sum_{k=0}^n B_k(\psi)(x+y)^{2n-2k}(x-y)^{2k}
\\&=2^{-n}\sum_{k=0}^n B_k(\psi)\sum_{a=0}^{2n-2k}\binom{2n-2k}{a}x^{2n-2k-a}y^a\sum_{b=0}^{2k}\binom{2k}b(-1)^bx^{2k-b}y^b
\\&=2^{-n}\sum_{k=0}^n B_k(\psi)\sum_{\ell=0}^n\sum_{j=0}^{2\ell}\binom{2k}j\binom{2n-2k}{2\ell-j}(-1)^jx^{2n-2\ell}y^{2\ell}
\\&=\sum_{\ell=0}^n B_\ell(\psi)x^{2n-2\ell}y^{2\ell}
=F_\psi(x,y),
\end{aligned}
\end{equation}
In the third line, we have used the binomial expansion. In the fourth line we relabel
$a=2\ell-j$ and $b=j$, so the summation becomes
\begin{equation}
\sum_{a=0}^{2n-2k}\sum_{b=0}^{2k}\rightarrow\sum_{\ell=0}^n\sum_{j=0}^{\min\{2\ell,2k\}}.
\end{equation}
We can safely set the upper limit of $j$ to be $2\ell$ since $\binom{2k}j$ appears in the term. In the last line, 
we have used the linear transformation between Majorana sector purities by the binary Krawtchouk polynomial in Eq.~\eqref{eq:Majorana_sector_purity_linear}:
\begin{equation}\label{eq:Majorana_sector_purity_linear_2}
B_\ell(\psi)=2^{-n}\sum_{k=0}^n\sum_{j=0}^{2\ell}(-1)^j\binom{2k}j\binom{2n-2k}{2\ell-j}B_k(\psi).
\end{equation}
We then prove any $2n$-degree homogeneous polynomial $F(x,y)\in\mathbb{R}[x,y]$ satisfying the above four symmetries can be represented as
\begin{equation}
F(x,y)=\sum_{j=0}^{\floor{n/4}}c_j\xi^{n-4j}\eta^j,\qquad\text{with }\xi=x^2+y^2, \eta=x^2y^2(x^2-y^2)^2.
\end{equation}
First, from symmetries~\ding{172} and~\ding{173}, $F$ contains only even powers of both variables $x$ and $y$, there exists a polynomial $F_1$, such that
\begin{equation}
F(x,y)=F_1(x^2,y^2).
\end{equation}
By symmetry~\ding{174}, $F_1$ is invariant under exchange of $x$ and $y$, there exists a polynomial $F_2$, such that
\begin{equation}
F(x,y)=F_1(x^2,y^2)=F_2(x^2+y^2,x^2y^2)=F_2(\xi,x^2y^2).
\end{equation}
We then use symmetry~\ding{175}, under transformation $(x,y)\rightarrow((x+y)/\sqrt{2},(x-y)/\sqrt{2})$, we have
\begin{equation}
F(x,y)=F_2(\xi,x^2y^2)=F_2\left(\xi,\frac{\xi^2}4-x^2y^2\right).
\end{equation}
So $F_2$ is invariant under reflection of the second argument about $\xi^2/8$, there exists a polynomial $F_3$, such that
\begin{equation}
F(x,y)=F_3\left(\xi,\left(\frac{\xi^2}8-x^2y^2\right)^2\right)=F_3\left(\xi,\frac{\xi^4}{64}-\frac\eta4\right).
\end{equation}
Thus, there exists coefficients $b_k$ and $c_j$, such that
\begin{equation}
\begin{aligned}
F(x,y)&=\sum_{k=0}^{\floor{n/4}}b_k\xi^{n-4k}\left(\frac{\xi^4}{64}-\frac\eta4\right)^k
\\&=\sum_{k=0}^{\floor{n/4}}b_k\xi^{n-4k}\sum_{j=0}^k\binom kj\left(\frac{\xi^4}{64}\right)^{k-j}\left(-\frac\eta4\right)^j
\\&=\sum_{j=0}^{\floor{n/4}}\sum_{k=j}^{\floor{n/4}}b_k\binom kj64^{j-k}(-4)^{-j}\xi^{n-4j}\eta^j
\\&=\sum_{j=0}^{\floor{n/4}}c_j\xi^{n-4j}\eta^j.
\end{aligned}
\end{equation}
Especially, there exist coefficients $c_j(\psi)\in\mathbb{R}$, such that
\begin{equation}
F_\psi(x,y)=\sum_{j=0}^{\floor{n/4}}c_j(\psi)\xi^{n-4j}\eta^j=\sum_{j=0}^{\floor{n/4}}c_j(\psi)\left(x^2+y^2\right)^{n-4j}\left(x^2y^2(x^2-y^2)^2\right)^j,
\end{equation}
Setting $x^2=1$ and $y^2=z$, we have
\begin{equation}
G_\psi(z)=\sum_{j=0}^{\floor{n/4}}c_j(\psi)\left(1+z\right)^{n-4j}\left(z(1-z)^2\right)^j,
\end{equation}
which is Eq.~\eqref{eq:gleason-G}. Furthermore, we have
\begin{equation}
c_0(\psi)=\Coeff{z^0}{G_\psi(z)}=B_0(\psi)=1
\end{equation}
from Eq.~\eqref{eq:Majorana_sector_purity_boundary}. We finish the proof of Lemma~\ref{lem:Gleason_form}.
\end{proof}

\begin{proof}[Proof of Lemma~\ref{lem:beta_n_p_positive}]
To prove (cf.\ Eq.~\eqref{eq:def_beta_n_p})
\begin{equation}\label{eq:def_beta_n_p_2}
\beta_{n,p}
=
\sum_{j=1}^{p}
(-1)^j
\frac{2p}{p+j}\binom{p+j}{p-j}
\frac{e_{n,j}}{4^{n-3j}}>0,
\end{equation}
it is sufficient to prove for $1\leq j\leq\floor{n/4}$, $(-1)^je_{n,j}>0$, with (cf.\ Eq.~\eqref{eq:def_e_n_j})
\begin{equation}
e_{n,j}=\Coeff{z^{\floor{n/2}-j}}{(1+z)^{n-4j}(1-z)^{2j}}.
\end{equation}
Relabel $K=n-2j$ and $k=\floor{n/2}-j$, then
\begin{equation}
e_{n,j}=\Coeff{z^k}{(1+z)^{K-2j}(1-z)^{2j}}.
\end{equation}
Kotice that for a polynomial $f(z)\in\mathbb{C}[z]$, the coefficient of $z^k$ can be calculated using Cauchy's coefficient formula
\begin{equation}
\Coeff{z^k}{f(z)}=\frac1{2\pi i}\oint_{\abs{z}=1}\frac{f(z)}{z^{k+1}}\mathrm{d}z=\frac1{2\pi}\int_0^{2\pi}f(e^{i\theta})e^{-ik\theta}\mathrm{d}\theta.
\end{equation}
Inserting $f(z)=(1+z)^{K-2j}(1-z)^{2j}$, we discuss the following two cases and show both lead to $(-1)^je_{n,j}>0$.
\begin{itemize}
\item If $K=2k$ is even,
\begin{equation}
\begin{aligned}
(-1)^je_{n,j}&=\frac{(-1)^j}{2\pi}\int_0^{2\pi}(1+e^{i\theta})^{K-2j}(1-e^{i\theta})^{2j}e^{-ik\theta}\mathrm{d}\theta
\\&=\frac{(-1)^j}{2\pi}\int_0^{2\pi}\left(2e^{i\theta/2}\cos\left(\frac\theta2\right)\right)^{2k-2j}\left(-2ie^{i\theta/2}\sin\left(\frac\theta2\right)\right)^{2j}e^{-ik\theta}\mathrm{d}\theta
\\&=\frac{2^K}{2\pi}\int_0^{2\pi}\cos^{2k-2j}\left(\frac\theta2\right)\sin^{2j}\left(\frac\theta2\right)\mathrm{d}\theta\\&>0,
\end{aligned}
\end{equation}
where in the last line, we have used all terms in the integrand are even powers, and are not always zero for $\theta\in[0,2\pi)$.
\item If $K=2k+1$ is odd,
\begin{equation}
\begin{aligned}
(-1)^je_{n,j}&=\frac{(-1)^j}{2\pi}\int_0^{2\pi}(1+e^{i\theta})^{K-2j}(1-e^{i\theta})^{2j}e^{-ik\theta}\mathrm{d}\theta
\\&=\frac{(-1)^j}{2\pi}\int_0^{2\pi}\left(2e^{i\theta/2}\cos\left(\frac\theta2\right)\right)^{2k+1-2j}\left(-2ie^{i\theta/2}\sin\left(\frac\theta2\right)\right)^{2j}e^{-ik\theta}\mathrm{d}\theta
\\&=\frac{2^K}{2\pi}\int_0^{2\pi}e^{i\theta/2}\cos^{2k+1-2j}\left(\frac\theta2\right)\sin^{2j}\left(\frac\theta2\right)\mathrm{d}\theta
\\&=\frac{2^K}{2\pi}\int_0^{2\pi}\cos^{2k+2-2j}\left(\frac\theta2\right)\sin^{2j}\left(\frac\theta2\right)\mathrm{d}\theta\\&>0.
\end{aligned}
\end{equation}
where in the fourth line, we have used $e_{n,j}\in\mathbb{R}$, so $e_{n,j}=\mathrm{Re}[e_{n,j}]$ and $\mathrm{Re}[e^{i\theta/2}]=\cos\left(\theta/2\right)$. In the last line,  we have used all terms in the integrand are even powers, and are not always zero for $\theta\in[0,2\pi)$.
\end{itemize}
\end{proof}

\begin{proof}[Proof of Lemma~\ref{lem:S_r_non-negative}]
Recall
\begin{equation}\label{eq:def-Sr_2}
S_r(\psi)
=
\sum_{\ell=0}^n
(-1)^\ell K_r(2\ell;2n)B_\ell(\psi),
\end{equation}
where the binary Krawtchouk polynomial is defined as
\begin{equation}\label{eq:def_Krawtchouk_3}
K_r(t;2n)
=
\sum_{j=0}^r(-1)^j\binom tj\binom{2n-t}{r-j}.
\end{equation}
For a subset $S\subseteq[2n]$, define $\varepsilon_j(S)\in\{\pm1\}$ for $j\in[2n]$ and the Hermitian observable $\Omega_S$ as
\begin{equation}
\varepsilon_j(S)=
\begin{cases}
-1,\quad &j\in S,\\
1,&j\notin S,
\end{cases}
\qquad
\text{and}
\qquad
\Omega_S=\sum_{\substack{T\subseteq[2n]\\\abs{T}\text{ is even}}}\prod_{j\in T}\varepsilon_j(S)\gamma_j\otimes\gamma_j.
\end{equation}
Notice that $[\gamma_j\otimes\gamma_j,\gamma_k\otimes\gamma_k]=0$ for $j,k\in[2n]$, so the product $\prod_{j\in T}$ is unambiguous. Specifically, when $T=\emptyset$, we set
\begin{equation}
\prod_{j\in\emptyset}\varepsilon_j(S)\gamma_j\otimes\gamma_j=\bI.
\end{equation}
We show that first $\Omega_S\geq0$ and then
\begin{equation}
S_r(\psi)=\sum_{\substack{S\subseteq[2n]\\\abs{S}=r}}\bra{\psi}^{\otimes2}\Omega_S\ket{\psi}^{\otimes2}\geq0.
\end{equation}
We use the identities
\begin{equation}
\begin{aligned}
\prod_{j=1}^{2n}\left(\bI+\varepsilon_j(S)\gamma_j\otimes\gamma_j\right)&=\sum_{T\subseteq[2n]}\prod_{j\in T}\varepsilon_j(S)\gamma_j\otimes\gamma_j,\\
\prod_{j=1}^{2n}\left(\bI-\varepsilon_j(S)\gamma_j\otimes\gamma_j\right)&=\sum_{T\subseteq[2n]}(-1)^{\abs{T}}\prod_{j\in T}\varepsilon_j(S)\gamma_j\otimes\gamma_j,
\end{aligned}
\end{equation}
which can be seen from the following. Expand the ordered product by distributivity. 
For each factor one chooses either the identity term or the term 
$\varepsilon_j(S)\gamma_j\otimes\gamma_j$. Such choices are in one-to-one 
correspondence with subsets $T\subseteq[2n]$, where $T$ records the 
indices for which the nontrivial term is chosen. 
We have
\begin{equation}
\begin{aligned}
\Omega_S&=\frac12\sum_{T\subseteq[2n]}\left(1+(-1)^{\abs{T}}\right)\prod_{j\in T}\varepsilon_j(S)\gamma_j\otimes\gamma_j
\\&=\frac12\left(\sum_{T\subseteq[2n]}\prod_{j\in T}\varepsilon_j(S)\gamma_j\otimes\gamma_j+\sum_{T\subseteq[2n]}(-1)^{\abs{T}}\prod_{j\in T}\varepsilon_j(S)\gamma_j\otimes\gamma_j\right)
\\&=\frac12\left(\prod_{j=1}^{2n}\left(\bI+\varepsilon_j(S)\gamma_j\otimes\gamma_j\right)+\prod_{j=1}^{2n}\left(\bI-\varepsilon_j(S)\gamma_j\otimes\gamma_j\right)\right)
\\&\geq0.
\end{aligned}
\end{equation}
In the last line, we have used the fact that the eigenvalues of $\varepsilon_j(S)\gamma_j\otimes\gamma_j$ are $\pm1$, so the terms $\left(\bI+\varepsilon_j(S)\gamma_j\otimes\gamma_j\right)$ and $\left(\bI-\varepsilon_j(S)\gamma_j\otimes\gamma_j\right)$ are positive semi-definite. They commute with each other, so their product is also positive semi-definite.
Fix any $T\subseteq[2n]$ with $\abs{T}=2\ell$, we have
\begin{equation}
\begin{aligned}
\widehat{\gamma}_T\otimes\widehat{\gamma}_T&=(-i)^{\abs{T}(\abs{T}-1)}\gamma_T\otimes\gamma_T
\\&=(-1)^\ell\prod_{j\in T}\gamma_j\otimes\gamma_j,
\end{aligned}
\end{equation}
where we have used $\widehat{\gamma}_T=(-i)^{\abs{T}(\abs{T}-1)/2}\gamma_T$ in Eq.~\eqref{eq:def_Hermitian_Majorana_product}. 
And
\begin{equation}
\sum_{\substack{S\subseteq[2n]\\\abs{S}=r}}\prod_{j\in T}\varepsilon_j(S)=\sum_{\substack{S\subseteq[2n]\\\abs{S}=r}}(-1)^{\abs{S\cap T}}=\sum_{j=0}^r(-1)^j\binom{2\ell}j\binom{2n-2\ell}{r-j}=K_r(2\ell;2n),
\end{equation}
where in the middle equality we count $j=\abs{S\cap T}$ and use $\abs{T}=2\ell$. Finally, inserting the definition of $B_\ell(\psi)$in Eq.~\eqref{eq:Majorana_sector_purity}, we obtain
\begin{equation}
\begin{aligned}
S_r(\psi)&=
\sum_{\ell=0}^n
(-1)^\ell K_r(2\ell;2n)B_\ell(\psi)
\\&=\sum_{\ell=0}^n
(-1)^\ell\sum_{\abs{T}=2\ell}\sum_{\substack{S\subseteq[2n]\\\abs{S}=r}}\prod_{j\in T}\varepsilon_j(S)
\abs{\bra{\psi}\widehat{\gamma}_T\ket{\psi}}^2
\\&=\sum_{\substack{S\subseteq[2n]\\\abs{S}=r}}\sum_{\substack{T\subseteq[2n]\\\abs{T}\text{ is even}}}\prod_{j\in T}\varepsilon_j(S)\bra{\psi}^{\otimes2}\left(\gamma_j\otimes\gamma_j\right)\ket{\psi}^{\otimes2}
\\&=\sum_{\substack{S\subseteq[2n]\\\abs{S}=r}}\bra{\psi}^{\otimes2}\Omega_S\ket{\psi}^{\otimes2}\geq0.
\end{aligned}
\end{equation}

\end{proof}

\begin{proof}[Proof of Lemma~\ref{lem:S_r_c_j_relation}]
Let $1\leq p\leq\floor{n/4}$, we want to show
\begin{equation}\label{eq:S_r_c_j_relation_2}
S_{n-4p}(\psi)=
(-1)^p\sum_{j=p}^{\floor{n/4}}
4^{n-3j}\binom{2j}{j-p}c_j(\psi),
\end{equation}
where $c_j(\psi)$ is determined in Eq.~\eqref{eq:gleason-G} in Lemma~\ref{lem:Gleason_form}. We calculate the generating function of $K_r(2\ell;2n)$:
\begin{equation}
\begin{aligned}
\sum_{r=0}^{2n}K_r(2\ell;2n)t^r&=\sum_{r=0}^{2n}\sum_{j=0}^r(-1)^j\binom{2\ell}j\binom{2n-2\ell}{r-j}t^r
\\&=\left(\sum_{j=0}^{2\ell}\binom{2\ell}j(-t)^j\right)\left(\sum_{k=0}^{2n-2\ell}\binom{2n-2\ell}kt^k\right)
\\&=(1-t)^{2\ell}(1+t)^{2n-2\ell},
\end{aligned}
\end{equation}
where in the second line, we have used $k=r-j$.
Define and then calculate the generating function of $S_r(\psi)$:
\begin{equation}\label{eq:S_r_generating_function}
\begin{aligned}
S_\psi(t)&=\sum_{r=0}^{2n}S_r(\psi)t^r
\\&=\sum_{r=0}^{2n}\sum_{\ell=0}^n
(-1)^\ell K_r(2\ell;2n)B_\ell(\psi)t^r
\\&=\sum_{\ell=0}^n
(-1)^\ell B_\ell(\psi)\sum_{r=0}^{2n}K_r(2\ell;2n)t^r
\\&=\sum_{\ell=0}^n
(-1)^\ell B_\ell(\psi)(1-t)^{2\ell}(1+t)^{2n-2\ell}
\\&=(1+t)^{2n}\sum_{\ell=0}^n
 B_\ell(\psi)\left(-\left(\frac{1-t}{1+t}\right)^2\right)^\ell
\\&=(1+t)^{2n}G_\psi\left(-\left(\frac{1-t}{1+t}\right)^2\right),
\end{aligned}
\end{equation}
where in the last step, we 
have used $G_\psi$ is the generating function of $B_\ell(\psi)$. By Eq.~\eqref{eq:gleason-G} in Lemma~\ref{lem:Gleason_form}, there exist real coefficients $\left\{c_j(\psi)\right\}$ with $c_0(\psi)=1$, such that
\begin{equation}\label{eq:gleason-G_2}
G_\psi(z)=\sum_{j=0}^{\floor{n/4}}
c_j(\psi)(1+z)^{n-4j}\left(z(1-z)^2\right)^j.
\end{equation}
Inserting into Eq.~\eqref{eq:S_r_generating_function}, we have
\begin{equation}
\begin{aligned}
S_\psi(t)&=(1+t)^{2n}G_\psi\left(-\left(\frac{1-t}{1+t}\right)^2\right)
\\&=(1+t)^{2n}\sum_{j=0}^{\floor{n/4}}
c_j(\psi)\left(\frac{4t}{(1+t)^2}\right)^{n-4j}(-1)^j\left(\frac{1-t}{1+t}\right)^{2j}\left(\frac{2(1+t^2)}{(1+t)^2}\right)^{2j}
\\&=\sum_{j=0}^{\floor{n/4}}(-1)^j4^{n-3j}c_j(\psi)t^{n-4j}(1-t^4)^{2j}
\\&=\sum_{j=0}^{\floor{n/4}}\sum_{k=0}^{2j}(-1)^{j+k}4^{n-3j}c_j(\psi)\binom{2j}kt^{n-4j+4k}.
\end{aligned}
\end{equation}
Finally, for $1\leq p\leq\floor{n/4}$, we obtain
\begin{equation}
\begin{aligned}
S_{n-4p}(\psi)&=\Coeff{t^{n-4p}}{S_\psi(t)}
\\&=\Coeff{t^{n-4p}}{\sum_{j=0}^{\floor{n/4}}\sum_{k=0}^{2j}(-1)^{j+k}4^{n-3j}c_j(\psi)\binom{2j}kt^{n-4j+4k}}
\\&=\sum_{j=0}^{\floor{n/4}}(-1)^{p}4^{n-3j}c_j(\psi)\binom{2j}{j-p},
\end{aligned}
\end{equation}
which is Eq.~\eqref{eq:S_r_c_j_relation_2} after noticing that the terms for $j<p$ in the summation vanish.
\end{proof}

\begin{proof}[Proof of Lemma~\ref{lem:positive-certificate}]
It is equivalent to prove
\begin{equation}\label{eq:e_n_j_beta_n_p_identity}
\sum_{j=1}^{\floor{n/4}}c_j(\psi)e_{n,j}=\sum_{p=1}^{\floor{n/4}}\beta_{n,p}S_{n-4p}(\psi).
\end{equation}
Inserting Eq.~\eqref{eq:S_r_c_j_relation} in Lemma~\ref{lem:S_r_c_j_relation} into the right-hand side, we obtain
\begin{equation}\label{eq:beta_n_p_intermediate}
\begin{aligned}
\sum_{p=1}^{\floor{n/4}}\beta_{n,p}S_{n-4p}(\psi)
&=\sum_{p=1}^{\floor{n/4}}\beta_{n,p}(-1)^p\sum_{j=p}^{\floor{n/4}}
4^{n-3j}\binom{2j}{j-p}c_j(\psi)
\\&=\sum_{j=1}^{\floor{n/4}}c_j(\psi)4^{n-3j}\sum_{p=1}^j\binom{2j}{j-p}(-1)^p\beta_{n,p}.
\end{aligned}
\end{equation}
Notice that $\beta_{n,p}$ in Eq.~\eqref{eq:def_beta_n_p} can be regarded as a linear transformation from $e_{n,j}$ as
\begin{equation}\label{eq:def_beta_n_p_3}
\beta_{n,p}
=
\sum_{j=1}^{p}
(-1)^j
\frac{2p}{p+j}\binom{p+j}{p-j}
\frac{e_{n,j}}{4^{n-3j}},
\end{equation}
it is sufficient to prove that the inverse transformation is given by
\begin{equation}\label{eq:e_n_j_from_beta_n_p}
e_{n,j}=4^{n-3j}\sum_{p=1}^j\binom{2j}{j-p}(-1)^p\beta_{n,p},
\end{equation}
then plug it into Eq.~\eqref{eq:beta_n_p_intermediate} finishes the proof of Lemma~\ref{lem:positive-certificate}.
Define the vectors $\vec{u}:=(u_j)_{j=1}^{\floor{n/4}}$ and $\vec{v}:=(v_p)_{p=1}^{\floor{n/4}}$ by
\begin{equation}
u_j=\frac{e_{n,j}}{4^{n-3j}},\qquad v_p=(-1)^p\beta_{n,p}.
\end{equation}
The transformation from $e_{n,j}$ to $\beta_{n,p}$ in Eq.~\eqref{eq:def_beta_n_p_3} can be represented as $\vec{v}=N\vec{u}$ by matrix $N\in\mathcal{M}_{\floor{n/4}}(\mathbb{R})$ with elements
\begin{equation}
N_{p,j}=
\begin{dcases}
(-1)^{p-j}\frac{2p}{p+j}\binom{p+j}{p-j},\qquad&p\geq j,\\
0,&p<j.
\end{dcases}
\end{equation}
Define the matrix $M\in\mathcal{M}_{\floor{n/4}}(\mathbb{R})$ with elements
\begin{equation}
M_{j,p}=
\begin{dcases}
\binom{2j}{j-p},\qquad&j\geq p,\\
0,&j<p,
\end{dcases}
\end{equation}
then Eq.~\eqref{eq:e_n_j_from_beta_n_p} reduces to $\vec{u}=M\vec{v}$, which is equivalent to $MN=\bI$.
We notice that both $N$ and $M$ are lower-triangle matrices with diagonal elements all equal to $1$. It is sufficient to show $NM=\bI$, so that $N$ and $M$ are matrix inverses to each other. We calculate the matrix element
\begin{equation}
\begin{aligned}
\left(NM\right)_{p,r}&=\sum_{j=1}^{\floor{n/4}}N_{p,j}M_{j,r}
\\&=\sum_{j=r}^{p}(-1)^{p-j}\frac{2p}{p+j}\binom{p+j}{p-j}\binom{2j}{j-r}
\\&=2p\sum_{j=r}^{p}(-1)^{p-j}\frac{(p+j-1)!}{(p-j)!(j-r)!(j+r)!}.
\end{aligned}
\end{equation}
\begin{itemize}
\item If $p<r$, then the index set in the above summation is empty, so $\left(NM\right)_{p,r}=0$.
\item If $p=r$, then
\begin{equation}
\left(NM\right)_{p,r}=2r\frac{(2r-1)!}{(2r)!}=1.
\end{equation}
\item If $p>r$, relabel $k=j-r$ and $K=p-r\geq1$. We have
\begin{equation}\label{eq:NM_element}
\begin{aligned}
\left(NM\right)_{p,r}&=2p\sum_{j=r}^{p}(-1)^{p-j}\frac{(p+j-1)!}{(p-j)!(j-r)!(j+r)!}
\\&=2p\sum_{k=0}^{K}(-1)^{K-k}\frac{(p+r+k-1)!}{(K-k)!k!(2r+k)!}
\\&=\frac{2p}K\sum_{k=0}^{K}(-1)^{K-k}\binom Kk\binom{p+r+k-1}{K-1}.
\end{aligned}
\end{equation}
 Consider taking the finite difference of a polynomial $P(x)\in\mathbb{R}[x]$, defined as
\begin{equation}
\Delta P(x)=P(x+1)-P(x).
\end{equation}
We show that for $L\in\mathbb{N}_+$, the $L$-fold action results in
\begin{equation}\label{eq:L-fold_finite_difference}
\Delta^LP(x)=\sum_{l=0}^L(-1)^{L-l}\binom LlP(x+l).
\end{equation}
For $L=1$ and $2$, it is easily verified, as by the linearity of $\Delta$,
\begin{equation}
\Delta^2P(x)=\Delta P(x+1)-\Delta P(x)=P(x+2)-2P(x+1)+P(x).
\end{equation}
Assume Eq.~\eqref{eq:L-fold_finite_difference} holds for $L$, then
\begin{equation}
\begin{aligned}
\Delta^{L+1}P(x)&=\Delta^L P(x+1)-\Delta^L P(x)
\\&=\sum_{l'=0}^L(-1)^{L-l'}\binom L{l'}P(x+1+l')-\sum_{l=0}^L(-1)^{L-l}\binom LlP(x+l)
\\&=\sum_{l=0}^{L+1}(-1)^{L+1-l}\left(\binom L{l-1}+\binom Ll\right)P(x+l)
\\&=\sum_{l=0}^{L+1}(-1)^{L+1-l}\binom{L+1}lP(x+l),
\end{aligned}
\end{equation}
where in the third line we relabel $l'=l-1$ in the first summation in the second line and use $\dbinom L{-1}=\dbinom L{L+1}=0$. By induction, Eq.~\eqref{eq:L-fold_finite_difference} holds for any $L\in\mathbb{N}_+$. Let $x=0$, we have
\begin{equation}\label{eq:L-fold_finite_difference_0}
\Delta^LP(0)=\sum_{l=0}^L(-1)^{L-l}\binom LlP(l).
\end{equation}
Notice that $\deg(\Delta(x^a))=a-1$ for $a\in\mathbb{N}_+$, so in general, if $\deg(P(x))\geq1$, the degree decreases at least by one after finite difference,
\begin{equation}
\deg(\Delta P(x))\leq\deg(P(x))-1.
\end{equation}
Especially, if $\deg(P(x))<L$, then $\Delta^LP(x)=0$. Combined with Eq.~\eqref{eq:L-fold_finite_difference_0}, we obtain the identity for any polynomial $P(x)$ with $\deg(P)<L$,
\begin{equation}\label{eq:L-fold_finite_difference_small_degree}
\sum_{l=0}^L(-1)^{L-l}\binom LlP(l)=0.
\end{equation}
Back to the calculation of Eq.~\eqref{eq:NM_element}, we define the polynomial of $x$ with degree $K-1$ as
\begin{equation}
Q(x)=\binom{x+p+r-1}{K-1}=\frac1{(K-1)!}\prod_{l=1}^{K-1}(x+p+r-l),
\end{equation}
and recognize $Q(k)=\dbinom{p+r+k-1}{K-1}$. Thus by Eq.~\eqref{eq:L-fold_finite_difference_small_degree},
\begin{equation}
\begin{aligned}
\left(NM\right)_{p,r}&=\frac{2p}K\sum_{k=0}^{K}(-1)^{K-k}\binom Kk\binom{p+r+k-1}{K-1}
\\&=\frac{2p}K\sum_{k=0}^{K}(-1)^{K-k}\binom KkQ(k)
\\&=0.
\end{aligned}
\end{equation}
\end{itemize}
So we prove Eq.~\eqref{eq:e_n_j_from_beta_n_p} and complete the proof of Lemma~\ref{lem:positive-certificate}.
\end{proof}

\end{document}